\documentclass[]{usiinfthesis}
\pdfoutput=1
\usepackage[draft]{fixme}
\usepackage{tikz}

\usepackage[utf8]{inputenc}

\usepackage{amsmath}
\usepackage{amsthm}
\usepackage{amsfonts}

\usepackage{tablefootnote}

\usepackage{enumerate}

\usepackage{tikz,pgfplots}
\usetikzlibrary{patterns}

\usepackage{graphicx}
\usepackage{caption}
\usepackage{subcaption}

\usepackage{xspace}

\usepackage{fetamont}

\paragraphfont{\textbf}

\newtheorem{theorem}{Theorem}
\newtheorem{lemma}[theorem]{Lemma}
\newtheorem{definition}{Definition}
\newtheorem{corollary}[theorem]{Corollary}

\newtheorem*{remark}{Remark}
\newtheorem{proposition}[theorem]{Proposition}

\newcommand{\eps}{\varepsilon}
\newcommand{\R}{\mathcal{R}}
\newcommand{\polylog}{\operatorname{poly\,log}}

\newcommand{\NP}{\textbf{NP}}

\DeclareMathOperator*{\E}{\mathbb{E}}

\newcommand{\cT}{{\cal T}}
\newcommand{\cR}{{\cal R}}

\newcommand{\tdk}{2DGK\xspace}
\newcommand{\tdkr}{2DGKR\xspace}
\newcommand{\epst}{\varepsilon_{box}}
\newcommand{\epss}{\varepsilon_{small}}
\newcommand{\epsl}{\varepsilon_{large}}

\newcommand{\epsb}{\varepsilon_{box}}
\newcommand{\epsr}{\varepsilon_{ring}}

\newcommand{\epsau}{\varepsilon_{ra}}

\newcommand{\apx}{APX}

\newcommand{\opt}{OPT}
\newcommand{\optco}{OPT_{corr}}
\newcommand{\optla}{OPT_{large}}
\newcommand{\optsm}{OPT_{small}}
\newcommand{\optsk}{OPT_{skew}}
\newcommand{\optho}{OPT_{hor}}
\newcommand{\optve}{OPT_{ver}}

\newcommand{\optfa}{OPT_{fat}}
\newcommand{\optth}{OPT_{thin}}

\newcommand{\optki}{OPT_{kill}}
\newcommand{\optin}{OPT_{int}}
\newcommand{\optbo}{OPT_{box}}
\newcommand{\optrc}{OPT_{\fontL\&C}}
\newcommand{\Rsm}{R_{small}}
\newcommand{\Rla}{R_{large}}
\newcommand{\Rho}{R_{hor}}
\newcommand{\Rve}{R_{ver}}
\newcommand{\Rsk}{R_{skew}}
\newcommand{\Rin}{R_{int}}

\renewcommand{\L}{\mathcal{L}}

\newcommand{\fontL}{L}
\newcommand{\T}{OPT_T}
\newcommand{\il}{I_{\mathrm{long}}}
\newcommand{\ilopt}{OPT_{\mathrm{long}}}
\newcommand{\is}{I_{\mathrm{short}}}
\newcommand{\ilong}{\il}
\newcommand{\ishort}{\is}
\newcommand{\isopt}{OPT_{\mathrm{short}}}
\newcommand{\efl}{\varepsilon_{\fontL}}
\newcommand{\ilthin}{\ilopt \cap \T}
\newcommand{\ilfat}{\ilopt \setminus \T}
\newcommand{\isfat}{\isopt \setminus \T}
\newcommand{\isthin}{\isopt \cap \T}
\newcommand{\isfhor}{(\isopt \setminus \T)_{hor}}
\newcommand{\isfver}{(\isopt \setminus \T)_{ver}}

\newcommand{\F}{OPT_F}
\newcommand{\OK}{OPT_K}
\newcommand{\LF}{OPT_{LF}}
\newcommand{\ST}{OPT_{ST}}
\newcommand{\LT}{OPT_{LT}}
\newcommand{\SF}{OPT_{SF}}

\global\long\def\cell{C}
\global\long\def\C{\mathcal{C}}
\global\long\def\K{\mathcal{K}}

\newcommand{\bottomc}{bottom}
\newcommand{\topc}{top}

\newcommand{\height}{h}
\newcommand{\width}{w}
\newcommand{\profit}{p}
\newcommand{\shift}{shift}
\newcommand{\base}{base}

\ifdefined\DEBUG

\newcommand{\sal}[1]{\textcolor{green}{#1}}

\def\rem#1{{\marginpar{\raggedright\scriptsize #1}}}
\newcommand{\fabr}[1]{\rem{\textcolor{red}{$\bullet$ #1}}}
\newcommand{\salr}[1]{\rem{\textcolor{green}{$\bullet$ #1}}}
\newcommand{\arir}[1]{\rem{\textcolor{blue}{$\bullet$ #1}}}
\newcommand{\walr}[1]{\rem{\textcolor{magenta}{$\bullet$ #1}}}
\newcommand{\andyr}[1]{\rem{\textcolor{cyan}{$\bullet$ #1}}}

\else

\newcommand{\sal}[1]{#1}

\newcommand{\fabr}[1]{}
\newcommand{\salr}[1]{}
\newcommand{\arir}[1]{}
\newcommand{\walr}[1]{}
\newcommand{\andyr}[1]{}
\fi

\title{Approximation Algorithms\\for Rectangle Packing Problems} 
\author{Salvatore Ingala} 
\advisor{Luca Maria Gambardella} 
\coadvisor{Fabrizio Grandoni} 
\Day{27} 
\Month{October} 
\Year{2017} 
\place{Lugano} 
\programDirector{Walter Binder} 

\committee{%
  \committeeMember{Antonio Carzaniga}{Università della Svizzera Italiana, Lugano, Switzerland} 
  \committeeMember{Evanthia Papadopoulou}{Università della Svizzera Italiana, Lugano, Switzerland}  \committeeMember{Nikhil Bansal}{Eindhoven University of Technology, Nederlands}  \committeeMember{Roberto Grossi}{University of Pisa, Italy}  

} 

\begin{document}

\maketitle 

\frontmatter 

\begin{abstract}
\vspace{-10pt}
\small
In rectangle packing problems we are given the task of positioning some axis-aligned rectangles inside a given plane region, so that they do not overlap with each other. In the Maximum Weight Independent Set of Rectangles (MWISR) problem, their position is already fixed and we can only select which rectangles to choose, while trying to maximize their total weight. In the Strip Packing problem, we have to pack \emph{all} the given rectangles in a rectangular region of fixed width, while minimizing its height. In the $2$-Dimensional Geometric Knapsack (\tdk) problem, the target region is a square of a given size, and our goal is to select and pack a subset of the given rectangles of maximum weight.

All of the above problems are $\NP$-hard, and a lot of research has approached them by trying to find efficient \emph{approximation algorithms}. Besides their intrinsic interest as natural mathematical problems, geometric packing has numerous applications in settings like map labeling, resource allocation, data mining, cutting stock, VLSI design, advertisement placement, and so on.

We study a generalization of MWISR and use it to obtain improved approximation algorithms for a resource allocation problem called bagUFP.

We revisit some classical results on Strip Packing and \tdk, by proposing a framework based on smaller \emph{containers} that are packed with simpler rules; while variations of this scheme are indeed a standard technique in this area, we abstract away some of the problem-specific differences, obtaining simpler and cleaner algorithms that work unchanged for different problems. In this framework, we obtain improved approximation algorithms for a variant of Strip Packing where one is allowed pseudo-polynomial time, and for a variant of $\tdk$ where one is allowed to rotate the given rectangles by $90^\circ$ (thereby swapping width and height). For the latter, we propose the first algorithms with approximation factor better than $2$.

For the main variant of \tdk (without $90^\circ$ rotations), a container-based approach seems to face a natural barrier of $2$ in the approximation factor. Thus, we consider a generalized kind of packing that combines container packings with another packing problem that we call $\fontL$-packing problem, where we have to pack rectangles in an $\fontL$-shaped region of the plane. By finding a $(1+\eps)$-approximation for this problem and exploiting the combinatorial structure of the \tdk problem, we obtain the first algorithms that break the barrier of $2$ for the approximation factor of this problem.
\end{abstract}

\begin{acknowledgements}
	I want to express my gratitude to Fabrizio Grandoni for being a constant source of support and advise throughout my research, always bringing encouragement and many good ideas. It is difficult to imagine a better mentorship.
	
	Thank you to other my IDSIA coauthors Sumedha, Arindam and Waldo for sharing part of this journey: you have been fundamental to help me unravel the (rec)tangles of research, day after day. I wish you all success and always better approximation ratios.
	
	A deserved special thank to Roberto Grossi, Luigi Laura, and all the other tutors from the Italian Olympiads in Informatics, as they created the fertile ground for my interest in algorithms and data structures, prior to my PhD. Keep up the good work!
	
	Goodbye to all the friends and colleagues at IDSIA with whom I shared many desks, lunches and coffee talks. I am hopeful and confident that there will be many other occasions to share time (possibly with better coffee).
	
	To all the friend in Lugano with whom I shared a beer, a hike, a dive into the lake, a laughter: thanks for making my life better. I thank in particular Arne and all the rest of the rock climbers, as I could not have finished my PhD, had they let go of the other end of the rope!
	
	I am grateful to my family for always being close and supportive, despite the physical distance. Thank you, mom and dad! Thank you, Angelo and Concetta!
	
	Finally, to the one that was always sitting next to me, even when I failed to make enough room for her in the knapsack of my life: thank you, Joice. You always make me better.
\end{acknowledgements}

\makeatletter
\renewcommand\tableofcontents{%
	\ifthenelse{\boolean{@hypermode}}{\phantomsection}{}
	\chapter*{\contentsname
		\@mkboth{\contentsname}{\contentsname}}%
	\@starttoc{toc}%
	\cleardoublepage
}
\makeatletter

\tableofcontents

\mainmatter

\chapter{Introduction}\label{chap:introduction}

The development of the theory of $\mathbf{NP}$-hardness led for the first time to the understanding that many natural discrete optimization problems are \emph{really} difficult to solve. In particular, assuming that $\mathbf{P} \neq \mathbf{NP}$, it is impossible to provide a correct algorithm that is (1) efficient and (2) optimal (3) in the worst case.

Yet, this knowledge is not fully satisfactory: industry still needs solutions, and theorists want to know what can be achieved, if an efficient exact solution is impossible. Thus, a lot of research has been devoted to techniques to solve $\mathbf{NP}$-hard problem while relaxing at least one of the above conditions.

One approach is to relax constraint (1), and design algorithms that require super-polynomial time, but are \emph{as fast as possible}. For example, while a naive solution for the Traveling Salesman Problem (TSP) requires time at least $\Omega(n!) = \sal{2^{\Omega({n \log n})}}$ in order to try all the possible permutations, the celebrated Bellman-Held-Karp algorithm uses Dynamic Programming to reduce the running time to $O(n^2 2^n)$. See \cite{fk10} for a compendium of this kind of algorithms.

A prolific line of research, first studied systematically in \cite{df99}, looks for algorithms whose running time is exponential in the size of the input, but it is polynomial if some \emph{parameter} of the instance is fixed. For example, $k$-Vertex Cover can be solved in time $O(2^k n)$, which is exponential if $k$ is unbounded (since $k$ could be as big as $\Omega(n)$), but is polynomial for a fixed $k$. A modern and comprehensive textbook on this field is \cite{cfklmpps15}.

For the practical purpose of solving real-world instances, heuristic techniques (e.g.: simulated annealing, tabu search, ant colony optimization, genetic algorithms; see for example \cite{t09}), or search algorithms like branch and bound (\cite{h03}) often provide satisfactory results. While these techniques inherently fail to yield provable worst-case guarantees (either in the quality of the solution, or in the running time, or both), the fact that real-world instances are not ``the hardest possible'', combined with the improvements on both general and problem-specific heuristics and the increase in computational power of modern CPUs, allows one to handle input sizes that were not approachable just a few years back.

Researchers also consider special cases of the problem at hand: sometimes, in fact, the problem might be hard, but there could be interesting classes of instances that are easier to solve; for example, many hard problems on graphs are significantly easier if the given instances are guaranteed to be planar graphs. Moreover, real-world instances might not present some pathological features that make the problem hard to solve in the worst case, allowing for algorithms that behave reasonably well in practice.

The field of approximation algorithms (see \cite{h96,v01,ws11}) is obtained by relaxing requirement (2): we seek for algorithms that are correct and run in polynomial time on all instances; moreover, while these algorithms might not find the optimal solution, they are guaranteed to return solutions that are \emph{not too bad}, meaning that they are \emph{proven} to be not far from the optimal solution (a formal definition is provided below).

There are several reasons to study approximation algorithms, which include the following:
\begin{itemize}
	\item Showing worst-case guarantees (by means of formally proving theorems) gives a strong theoretical justification to an heuristic; in fact, heuristics could behave badly on some types of instances, and it is hard to be convinced that an algorithm will \emph{never} have worse results than what experiments have exhibited, unless a formal proof is provided.
	\item To better understand \emph{the problem}: even when an approximation algorithm is not practical, the ideas and the techniques employed to analyze it require a better understanding of the problem and its properties, that could eventually lead to better practical algorithms.
	\item To better understand \emph{the algorithms}: proving worst-case guarantees often leads to identifying what kind of instances are \emph{the hardest} for a given algorithm, allowing to channel further efforts on solving the hard core of the problem.
	\item To assess \emph{how hard} a problem is. There are $\mathbf{NP}$-hard optimization problems for which it is provably hard to obtain \emph{any} approximation algorithm, and, vice versa, other ones that admit arbitrarily good approximations.
\end{itemize}
The study of approximation algorithms has led to a rich and complex taxonomy of $\mathbf{NP}$-hard problems, and stimulated the development of new powerful algorithmic techniques and analytical tools in the attempt to \emph{close the gap} between the best negative (i. e. hardness of approximation) and positive (approximation algorithms) results.

While the above discussion identifies some directions that are established and well developed in the research community, no hard boundaries exist between the above mentioned fields. In fact, especially in recent times, many interesting results have been published that combine several approaches, for example: parameterized approximation algorithms, approximation algorithms in super-polynomial time, fast approximation algorithm for polynomial time problems, and so on.

\section{Approximation algorithms}\label{sec:approximation_algorithms}

In this section, we introduce some fundamental definitions and concepts in the field of approximation algorithms.

Informally, an $\NP$-Optimization problem $\Pi$ is either a minimization or a maximization problem. $\Pi$ defines a set of valid instances, and for each instance $I$, a non-empty set of \emph{feasible} solutions. Moreover an \emph{objective function value} is defined that, for each feasible solution $S$, returns a non-negative real number, which is generally intended as a measure of the quality of the solution. The goal is either to maximize or minimize the objective function value, and we call $\Pi$ a \emph{maximization problem} or a \emph{minimization problem}, accordingly. A solution $OPT$ that maximizes (respectively minimizes) the objective function value is called \emph{optimal solution}; for the problems that we are interested in, finding an optimal solution is $\NP$-hard. For simplicity, from now on we simply talk about \emph{optimization problems}. We refer the reader, for example, to \cite{v01} for a formal definition of $\NP$-optimization problems.

\begin{definition}\label{def:apx}
	A polynomial-time algorithm $\mathcal{A}$ for an optimization problem $\Pi$ is an \emph{$\alpha$-approximation algorithm} if it returns a feasible solution whose value is at most a factor $\alpha$ away from the value of the optimal solution, for any input instance.
\end{definition}

In this work, we follow the convention that $\alpha \geq 1$, both for minimization and maximization problems.

Thus, with our convention, if $opt$ is the objective function value of the optimal solution, an algorithm is an $\alpha$-approximation if it \emph{always} returns a solution whose value is at most $\alpha\cdot opt$ for a minimization problem, or at least $opt/\alpha$ for a maximization problem.\footnote{For the case of maximization problems, another convention which is common in literature is to enforce $\alpha \leq 1$, and say that an algorithm is an $\alpha$-approximation if the returned solution has value at least $\alpha \cdot opt$.}

Note that $\alpha$ could, in general, be a growing function of the input size, and not necessarily a constant number.

For some problems, it is possible to find approximate solutions that are arbitrarily close to the optimal solution. More formally:

\begin{definition}\label{def:ptas}
	We say that an algorithm $\mathcal{A}$ is a \emph{polynomial time approximation scheme} (PTAS) for an optimization problem $\Pi$ if for every fixed $\eps > 0$ and for any instance $I$, the running time of $\mathcal{A}(\eps, I)$ is polynomial in the input size $n$, and it returns a  solution whose value is at most a $1 + \eps$ factor away from the value of the optimal solution.
\end{definition}

Definitions \ref{def:apx} and \ref{def:ptas} can be generalized in the obvious way to allow for non-deterministic algorithms. In particular, we call an algorithm an \emph{expected} $\alpha$-approximation if the expected value of the output solution satisfies the above constraints.

Observe that the running time is polynomial for a \emph{fixed} $\eps$, but the dependency on $\eps$ can be arbitrarily large; in fact, running times of the form $f(1/\eps)n^{O(g(1/\eps))}$ for some functions $f$ and $g$ that are \sal{super-polynomial} with respect to $1/\eps$ are indeed very common. If the function $g$ is a constant not depending on $\eps$, the algorithm is called \emph{Efficient}~PTAS (EPTAS); if, moreover, $f$ is polynomial in $1/\eps$, then it is called a \emph{Fully Polynomial Time Approximation Scheme (FPTAS)}. In some sense, $\mathbf{NP}$-hard problems admitting an FPTAS can be thought as \emph{the easiest} hard problems; one such example is the Knapsack Problem.

It is also interesting to consider a relaxation of Definition~\ref{def:ptas} that allows for a slightly larger running time: a Quasi-Polynomial Time Approximation Scheme (QPTAS) is define exactly as above, except that $\mathcal{A}$ is allowed quasi-polynomial time $O_\eps(n^{\polylog n})$\footnote{The notation $O_\eps(f(n))$ means that the implicit constant hidden by the big O notation can depend on $\eps$.}

The class of problems that admit a constant factor approximation is called \textbf{APX}. Clearly, all problems that admit a PTAS are in \textbf{APX}, but the converse is not true if $\textbf{P} \neq \NP$.

For some problems, better approximation ratios can be found if one assumes that the solution is large enough. Formally, for a minimization problem $\Pi$, the asymptotic approximation ratio $\rho_{\mathcal{A}}^{\infty}$ of an algorithm $\mathcal{A}$ is defined as:

\[
	\rho_{\mathcal{A}}^{\infty} = \lim_{n \to \infty} \sup_I \left\{\frac{apx(I)}{opt(I)} : opt(I) \geq n \right\}
\]
where $apx(I)$ and $opt(I)$ are, respectively, the objective function value of the solution returned by $\mathcal{A}$ on the instance $I$, and that of an optimal solution to $I$.

Similarly to the definition of PTAS, we say that an algorithm $\mathcal{A}$ is an \emph{Asymptotic PTAS} or \emph{APTAS} for problem $\Pi$ if $\mathcal{A}(\eps, \cdot)$ is an asymptotic $(1 + \eps)$-approximation for any fixed $\eps > 0$.

\section{Lower bounds}\label{sec:lower_bounds}

An approximation algorithm, by definition, proves that a certain approximation factor is possible for a problem; thus, it provides an \emph{upper bound} on the best possible approximation factor.

Starting with the seminal work of \cite{sg76}, a significant amount of research has been devoted to a dual kind of result, that is, proofs that certain approximation ratios are not possible (under the assumption that $\textbf{P} \neq \textbf{NP}$ \sal{or analogous assumptions}). These results provide a \emph{lower bound} on the approximation factor.

Some problems are known to be \emph{APX-hard}. If a problem is APX-hard, then the existence of a PTAS for it would imply the existence of a PTAS for \emph{every problem in \textbf{APX}}. Thus, being APX-hard is considered a strong evidence that the problem does \emph{not} admit a PTAS.

For many problems, lower bounds are known that rely only on the assumption that $\textbf{P} \neq \NP$. For example, the Bin~Packing problem
can easily be showed to be impossible to approximate to a factor $\frac32 - \eps$ for any \sal{constant} $\eps > 0$, which follows from the $\NP$-hardness of the Partition problem.

Starting from the late 1990s, many other inapproximability results were proven as consequences of the celebrated PCP theorem and its more powerful versions; for example, \cite{ds05} proved that it is not possible to obtain a $1.3606$-approximation for Vertex Cover, unless $\textbf{P} = \NP$.

Stronger results can often be provided by making stronger assumptions; for example, many hardness results have been proven under the assumption that $\NP \not\subseteq \textbf{DTIME}\left(n^{O(\log \log n)}\right)$, or the stronger $\NP \not\subseteq \textbf{DTIME}\left(n^{\polylog n}\right)$ (that is, $\NP$-hard problems cannot be solved in quasi-polynomial time); more recently, hardness results have been proved under the \emph{Exponential Time Hypothesis} (ETH), that is, the assumption that SAT cannot be solved in sub-exponential time. Another strong hardness assumption is the \emph{Unique Games Conjecture} (UGC), proposed by~\cite{k02}. Clearly, stronger assumptions might lead to stronger results, but they are also considered less likely to be true; hence, there is a strong push for results that are as strong as possible, while relying on the weakest possible assumption (ideally, $\textbf{P} \neq \NP$).

Table~\ref{tab:upper_and_lower_bounds} presents some classical problems together with the best known upper and lower bounds.

\begin{table}[]
	\centering
	\caption{Examples of known upper and lower bounds. The $\tilde{O}$ notation hides poly-logarithmic factors.}
	\label{tab:upper_and_lower_bounds}
	\begin{tabular}{lll}
		Problem         & Upper bound    & Lower bound\\
		\hline
		Knapsack        & FPTAS          & Weakly $\NP$-hard\\
		Vertex Cover    & $2$ & $1.3606$ if $\textbf{P} \neq \NP$\\
		~               & ~   & $2$ under UGC\\
		Set Cover       & $(1 - o(1))\ln n$ & $(1 - \eps)\ln n$ if $\NP \not\subseteq \textbf{DTIME}\left(n^{O(\log \log n)}\right)$\\
		Independent Set & $\tilde{O}(n)$ & $O(n^{1 - \eps})$ if $\textbf{ZPP} \neq \NP$\tablefootnote{$\textbf{ZPP}$ is the class of problems admitting randomized algorithms in expected polynomial time.}
	\end{tabular}
\end{table}

The existence of a QPTAS for a problem is sometimes seen as a hint that a PTAS might exist: in fact, it implies that the problem is not APX-hard unless $\mathbf{NP} \subseteq \textbf{DTIME}(n^{\polylog n})$.

\section{Pseudo-polynomial time}

One of the ways to cope with $\NP$-hard problems is to allow running times that are not strictly polynomial. For problems that have numeric values in the input (that we consider to be integers for the sake of simplicity), one can consider \emph{pseudo-polynomial time} algorithms, whose running time is polynomial in the \emph{values} of the input instance, instead of their size. More formally:
\begin{definition}\label{def:PPT}
	An algorithm is said to run in \emph{pseudo-polynomial time} (PPT) if its running time is bounded by $O\left((nW)^{O(1)}\right)$, where $W$ is the maximum absolute value of the integers in the instance, and $n$ is the size of the instance.
\end{definition}
Equivalently, an algorithm runs in PPT if it runs in polynomial time in the size of the input when all the number in the input instance are represented in unary notation.

Clearly, this is a relaxation of the polynomial time requirement; thus, a problem might become significantly easier if PPT is allowed. For example, the knapsack problem is $\NP$-hard, but a classical Dynamic Programming approach can solve it exactly in PPT.

A problem that is $\NP$-hard but can be solved in PPT is called \emph{weakly $\NP$-hard}. A problem that does not admit an exact PPT algorithm unless $\textbf{P} = \NP$ is called \emph{strongly $\NP$-hard}.

Studying PPT algorithms is interesting also in the context of approximation algorithms, and it has recently been a more frequent trend in the research community. Note that the standard hardness result do not always apply. Some problem, nonetheless, are hard to solve even in this relaxed model; one such problem is the Strip Packing problem, which we introduce in the next section and will be discussed in detail in Chapter~\ref{chap:StripPacking}.

\section{Rectangle packing problems}\label{sec:rectangle_packing_problems}

In this section we introduce several geometric packing problems involving rectangles. In all of them, we are given as input a set of rectangles, and we are assigned the task of placing all or a subset of them into one or more target regions, while making sure that they are completely contained in the assigned region and they do not overlap with each other.

Given a set $\R$ of rectangles and a rectangular box of size $w\times h$, we call a \emph{packing} of $\R$ a pair $(x_i,y_i)\in \mathbb{N}\times \mathbb{N}$ for each $R_i \in \R$, with $0\leq x_i\leq w-w_i$ and $0\leq x_i\leq h-h_i$, meaning that the left-bottom corner of $R_i$ is placed in position $(x_i,y_i)$ and its right-top corner in position $(x_i+w_i,y_i+h_i)$. This packing is feasible if the interior of the rectangles is disjoint in this embedding (or equivalently rectangles are allowed to overlap on their boundary only). More formally, the packing is feasible if for every two distinct rectangles $R_i, R_j \in \R$, we have that $(x_i, x_i + w_i) \cap (x_j, x_j + w_j) = \emptyset$ or $(y_i, y_i + h_i) \cap (y_j, y_j + h_j) = \emptyset$.

Note that the above definition only admits \emph{orthogonal} packings, that is, every rectangle is always axis-parallel in the packing, as in Figure~\ref{fig:orthogonal-packing}. It has been long known that orthogonal packings are not necessarily optimal, even in very restricted cases like in packing equal squares into bigger squares; see Figure~\ref{fig:non-orthogonal-packing} for an example. Nevertheless, orthogonal packings are much simpler to handle mathematically, and the additional constraint is meaningful in many possible applications of the corresponding packing problems, and it is the focus of most of the literature on such optimization problems.

There is only one case of rectangle rotations that is compatible with orthogonal packings: $90^\circ$ rotations. In fact, for the packing problems that we consider, there are two variations: one where rotations are not allowed, as in the definition of \emph{packing} given above; and one where $90^\circ$ rotations are allowed, that is, one is allowed to swap width and height of a rectangle in the packing.

\begin{figure}
	\captionsetup[subfigure]{justification=centering}
	\begin{subfigure}[b]{.5\textwidth}
		\centering
		\includegraphics[width=5cm]{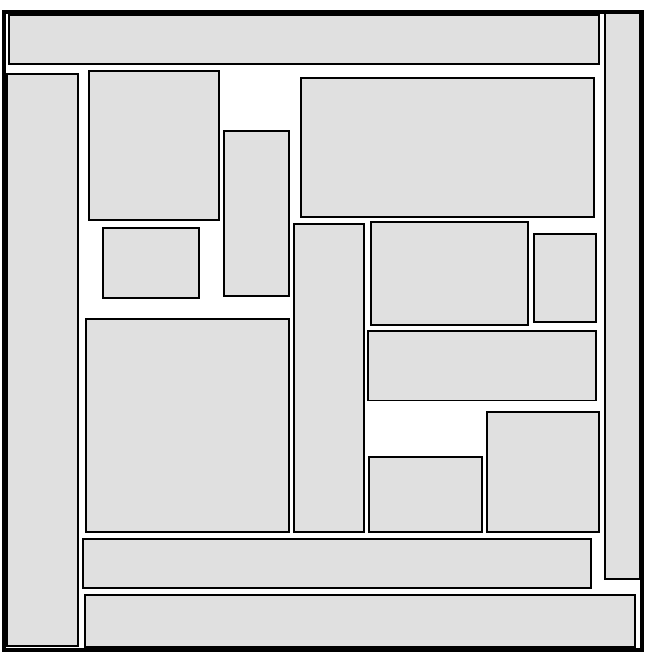}
		\caption{Example of an orthogonal packing. All the packed rectangles are axis-parallel and non-overlapping.}
		\label{fig:orthogonal-packing}
	\end{subfigure}%
	\begin{subfigure}[b]{.5\textwidth}
		\centering
		\includegraphics[width=5cm]{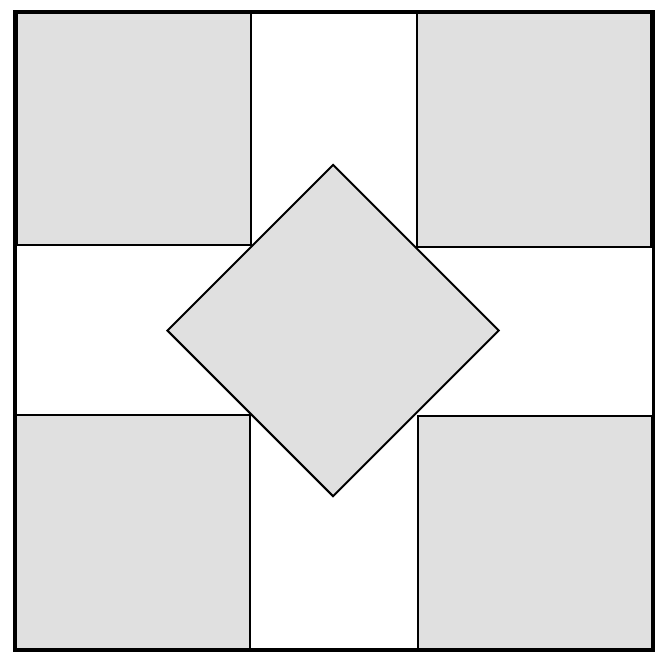}
		\caption{A non-orthogonal packing. With an orthogonal packing, only $4$ of the $5$ rectangles can be placed instead.}
		\label{fig:non-orthogonal-packing}
	\end{subfigure}
	\caption{Orthogonal and non-orthogonal packings.}
\end{figure}

There is another interesting variant of these packing problem, where we are interested in packings that can be separated via the so called \emph{guillotine cuts}, that is, edge-to-edge cuts parallel to an edge of the box. See Figure~\ref{fig:guillotinve-vs-normal} for a comparison between a packing with guillotine cuts and a generic packing. Such constraints are common in scenarios where the packed rectangles are patches of a material that must be cut, where the availability of a guillotine cutting sequence simplify the cutting procedure, thereby reducing costs; see for example \cite{prk04} and \cite{s88}.

\begin{figure}
	\centering
	\includegraphics[width=10.7cm]{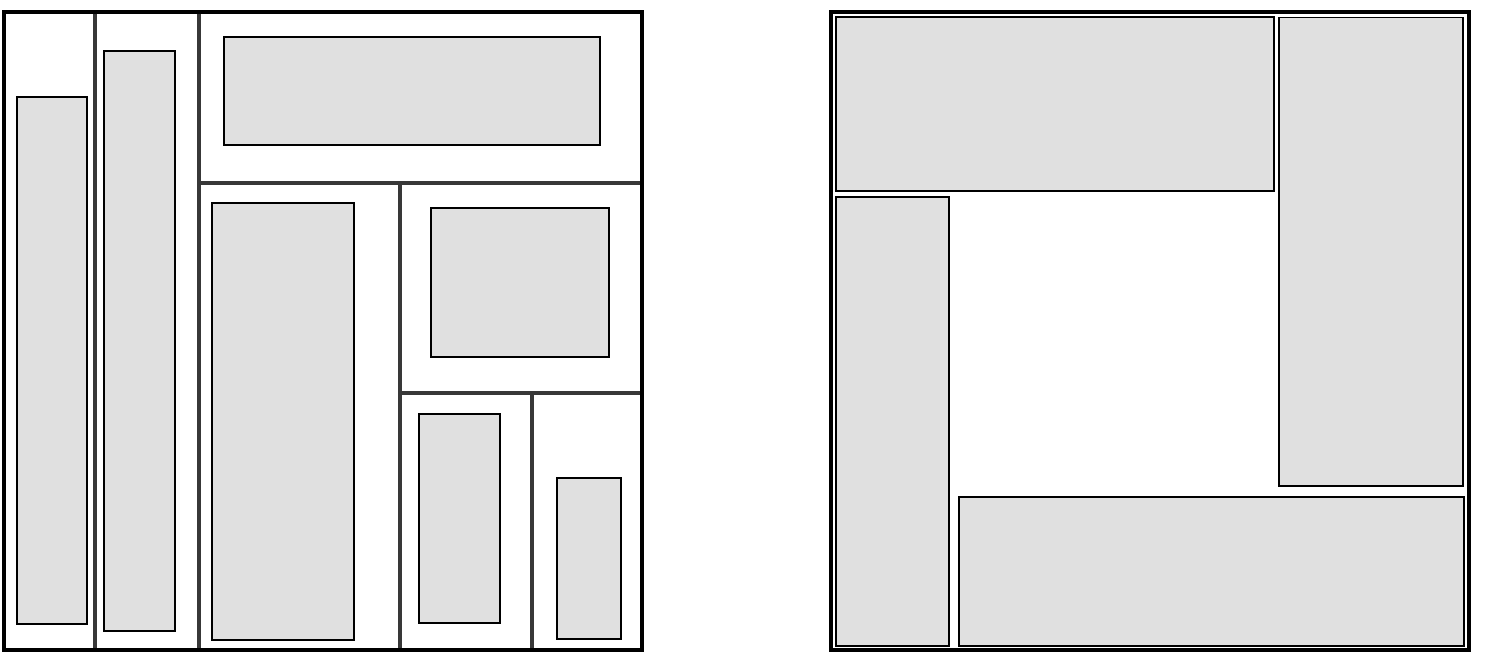}
	\caption{A guillotine packing on the left, and packing that cannot be a guillotine packing on the right.}
	\label{fig:guillotinve-vs-normal}
\end{figure}

\subsection{Strip Packing}

In the \emph{Strip Packing} problem, we are given a parameter $W\in \mathbb{N}$ and a set $\cR=\{R_1,\ldots,R_n\}$ of $n$ rectangles, each one characterized by a positive integer width $w_i\in \mathbb{N}$, $w_i\leq W$, and a positive integer height $h_i\in \mathbb{N}$. Our goal is to find a value $h$ and a feasible packing of all the rectangles in $\R$ in a rectangle of size $W \times h$, while minimizing $h$.

The version with $90^\circ$ rotations is also considered.

\begin{figure}
	\centering
	\includegraphics[width=10.7cm]{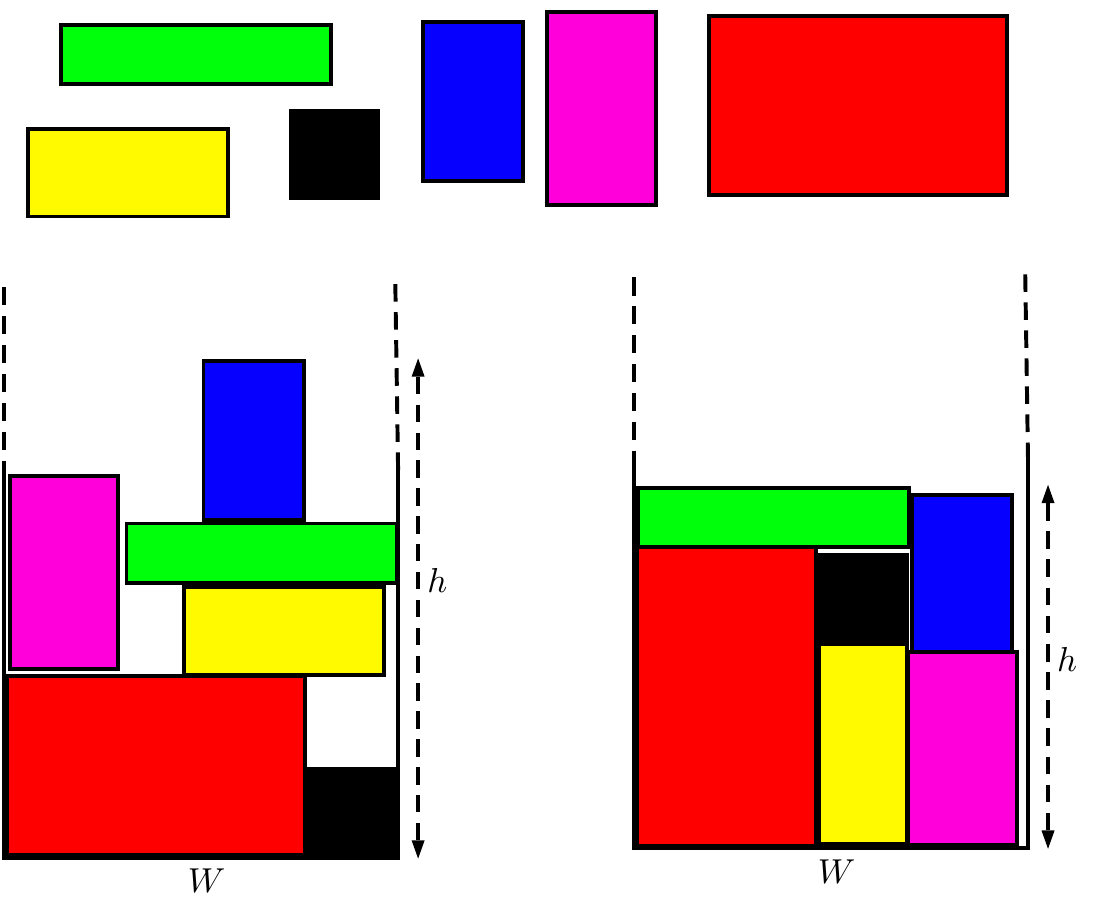}
	\caption{An instance of Strip Packing and corresponding feasible packings for the cases without (left) and with rotations (right).}
	\label{fig:StripPacking}
\end{figure}

Strip packing is a natural generalization of \emph{one-dimensional bin packing} (obtained when all the rectangles have the same height; see \cite{cecgmv13}) and \emph{makespan minimization} (obtained when all the rectangles have the same width; see \cite{cb76}).
The problem has lots of applications in industrial engineering and computer science, specially in cutting stock, logistics and scheduling (\cite{kr00,hjpv14}).
Recently, there have been several applications of strip packing in electricity allocation and peak demand reductions in smart-grids; see for example \cite{thlw13}, \cite{kskl13} and \cite{rks15}.

A simple reduction from the \sal{Partition} problem shows that the problem cannot be approximated within a factor $\frac{3}{2}-\eps$ for any $\eps>0$ in polynomial-time unless $\mathbf{P}=\mathbf{NP}$.
This reduction relies on exponentially large (in $n$) rectangle widths, and it also applies for the case with rotations.


Let $OPT=OPT(\cR)$ denote the optimal height for the considered strip packing instance $(\cR, W)$, and $h_{\max}=h_{\max}(\cR)$ (respectively, $w_{\max}=w_{\max}(\cR)$) be the largest height (respectively, width) of any rectangle in $\cR$. Most of the literature on this problem is devoted to the case without rotations. Observe that, for this case, $OPT\geq h_{\max}$, and we can assume that $W \leq n w_{\max}$ without loss of generality. The first non-trivial approximation algorithm for strip packing, with approximation ratio 3, was given by \cite{bcr80}. 
The First-Fit-Decreasing-Height algorithm (FFDH) by \cite{cgjt80} gives a 2.7-approximation. \cite{s80} gave an algorithm that generates a packing of height $2OPT+\frac{h_{max}}{2}$, hence achieving a 2.5-approximation. 
Afterwards, \cite{s97} and \cite{s94} independently improved the approximation ratio to 2. \cite{hv09} first broke the barrier of 2 with their 1.9396 approximation.
The present best $(\frac53+\eps)$-approximation is due to \cite{hjpv14}.

\cite{nw16} overcame the $\frac{3}{2}$-inapproximability barrier by presenting a $(\frac{7}{5}+\epsilon)$-approximation algorithm running in pseudo-polynomial-time (PPT). More specifically, they provided an algorithm with running time $O((Nn)^{O(1)})$, where $N=\max\{w_{max},h_{max}\}$\footnote{For the case without rotations, the polynomial dependence on $h_{max}$ can indeed be removed with standard techniques.}. In \citet*{ggik16}, we improved the approximation factor to $4/3 + \eps$, also generalizing it to the case with $90^\circ$ rotations; this result is described in Chapter~\ref{chap:StripPacking}. For the case without rotations, an analogous result was independently obtained by \cite{jr17}.

As strip packing is strongly NP-hard (see \cite{gj78}), it does not admit an exact pseudo-polynomial-time algorithm. Moreover, very recently \cite{akpp17} proved that it is $\mathbf{NP}$-hard to approximate Strip Packing within a factor $12/11 - \eps$ for any constant $\eps > 0$ in PPT, and this lower bound was further improved to $5/4 - \eps$ in \cite{hjrs17}. This rules out the possibility of a PPT approximation scheme.

\medskip

The problem has been also studied in terms of \emph{asymptotic} approximation, where the approximability is much better understood. For the case without rotations, \cite{cgjt80} analyzed two algorithms called Next Fit Decreasing Height (NFDH) and First Fit Decreasing Height (FFDH), and proved that their asymptotic approximation ratios are 2 and 1.7, respectively. \cite{bcr80} described another heuristic called Bottom Leftmost Decreasing Width (BLWD), and proved that it is an asymptotic 2-approximation. Then \cite{kr00} provided an AFPTAS, that is, an asymptotic $(1+\eps)$-approximation for any $\eps > 0$. The additive constant of $O\left(\frac{h_{max}}{\eps^2}\right)$ was improved by \cite{js09}, who achieved an AFPTAS with additive constant only $h_{max}$. \cite{s12} obtained a polynomial time algorithm that returns a solution with height at most $OPT + O(\sqrt{OPT \log OPT})$.

The analysis for the asymptotic 2-approximations (for example by NFDH) also apply to the case with $90^\circ$ rotations. \cite{mw04} improved this to a 1.613-approximation, and then \cite{ev04a} proved a 1.5-approximation. Finally, \cite{jv05} obtained an AFPTAS also for this variation of the problem.

The natural generalization of Strip Packing in 3 dimensions has been considered for the first time by \cite{lc90}, who obtained an asymptotic $3.25$-approximation. In \cite{lc92}, they further improved it to an asymptotic ratio $T_\infty^2 \approx 2.89$, where $T_\infty \approx 1.69$ is the so called \emph{harmonic constant} in the context of bin packing. \cite{bhisz07} provided an asymptotic $T_\infty$-approximation. More recently, \cite{jp14} obtained the currently best known asymptotic $1.5$-approximation. In terms of absolute approximation, the above mentioned result of \cite{lc90} also implies a $45/4$-approximation, which was improved to $25/4$ in \cite{dhjtt08}.

\subsection{2-Dimensional Geometric Knapsack}

2-Dimensional Geometric Knapsack (\tdk) is a geometric generalization of the well studied
one-dimensional knapsack problem. We are given a set of rectangles $\R=\{R_1,\ldots,R_n\}$, where
each $R_i$ is an axis-parallel rectangle with an integer
width $w_{i}>0$, height $h_{i}>0$ and profit $p_{i}>0$, and a \emph{knapsack} that we assume
to be a square of size $N\times N$ for some integer $N>0$.

A feasible solution is any axis-aligned packing of a subset
$\R'\subseteq \R$ into the knapsack, and the goal is to maximize the profit $p(\R')=\sum_{R_i\in \R'}p_i$.

Like for Strip Packing, the variant \emph{with rotations} of the problem is also considered, where one is allowed to rotate the rectangles by $90^\circ$.

\begin{figure}
	\centering
	\includegraphics[width=10.7cm]{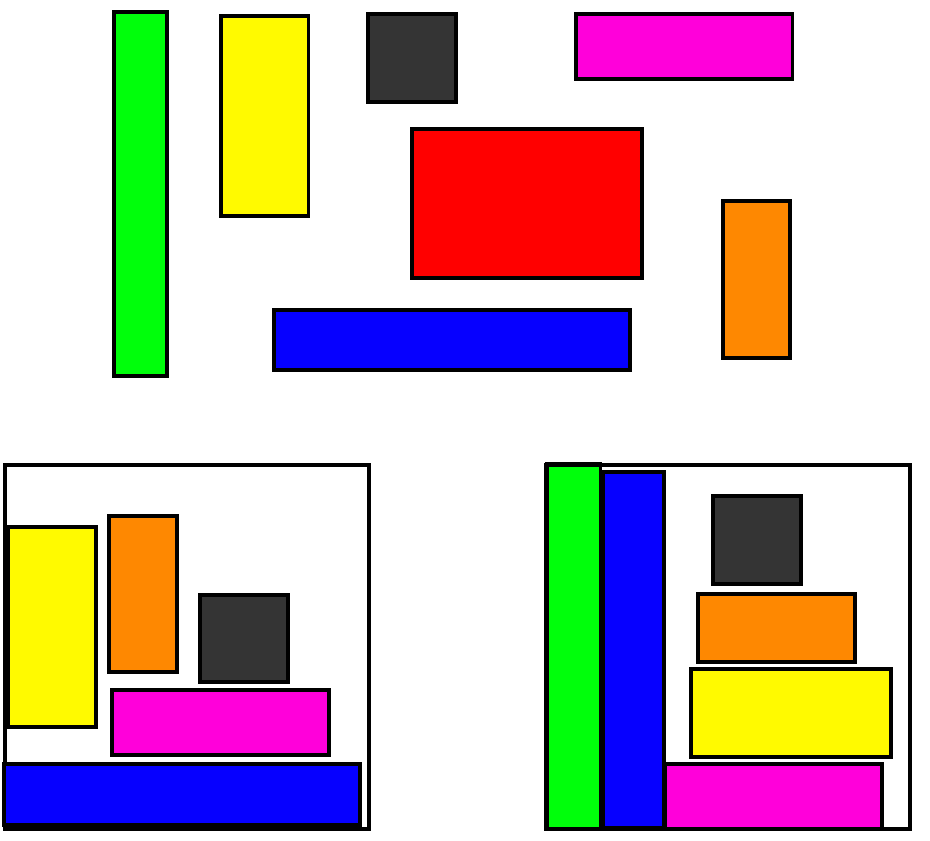}
	\caption{An instance of 2-Dimensional Geometric Knapsack, and corresponding feasible packing without (left) and with rotations (right); for simplicity, rectangle profits are not represented. Note that, unlike Strip Packing, it is not required that all the rectangles are packed.}
	\label{fig:2DGK}
\end{figure}

The problem is motivated by several practical applications. For instance, one might want to place advertisements on a board or a website, or cut rectangular pieces from a sheet of some material (in this context, the variant with rotations is useful when the material does not have a texture, or the texture itself is rotation invariant). It can also model a scheduling setting where each rectangle corresponds to a job that needs some ``contiguous amount'' of a given resource (memory storage, frequencies, etc.). In all these cases, dealing only with rectangular shapes is a reasonable simplification.

\cite{cm04} gave the first non-trivial approximation for the problem, obtaining ratio $3 + \eps$; this result builds on the $2$-approximation for Strip Packing in \cite{s97}. \cite{jz07} obtained an algorithm with approximation factor $2+\eps$, which is currently the best known in polynomial time; in \cite{jz04} they also gave a simpler and faster $(2+\eps)$-approximation for the special case in which all the rectangles have profit $1$ (cardinality case).

On the other hand, the only known hardness was given in \cite{ltwyc90}, who proved that an FPTAS is impossible even in the special case of packing squares into a square. Note that this does not rule out a PTAS.

Better results are known for many special cases. If all the rectangles are squares, \cite{h06} obtained a $(5/4 + \eps)$-approximation, and \cite{js08} showed that a PTAS is possible; very recently, \cite{hw17} obtained an EPTAS. \cite{fgjs08} gave a PTAS for rectangles in the relaxed model with \emph{resource augmentation}, that is, where one is allowed to enlarge the \sal{width and the height of the knapsack} by a $1 + \eps$ factor. \cite{js09} showed that the same can be achieved even if only the width (or only the height) is augmented; we prove a slightly modified version of this result in Lemma~\ref{lem:structural_lemma_augm} in Section~\ref{sec:2dgk-resource-augmentation}. \cite{bcjps09} proved that a PTAS is possible for the special case in which the profit of each rectangle equals its area, that is, $p_i := w_i h_i$, for both the cases with or without rotations. \cite{fgj05} presented a PTAS for the special case where the height of all the rectangles is much smaller than the height of the knapsack.

By using a structural result for independent set of rectangles, \cite{aw15} gave a QPTAS for the case without rotations, with the assumption that widths and heights are quasi-polynomially bounded. \cite{acckpsw15} extended the technique to obtain a similar QPTAS for the version with guillotine cuts (with the same assumption on rectangle sizes).

In Chapter~\ref{chap:2dgk-norot} we show how to obtain a polynomial-time algorithm with approximation factor $17/9 + \eps$ for the case without rotations, which we further improve to $\frac{558}{325} + \eps < 1.72$ for the cardinality case. This is the first polynomial time algorithm that breaks the barrier of $2$ for this problem.

In Chapter~\ref{chap:2dgk-rot} we consider the case with rotations, and obtain a $(3/2 + \eps)$-approximation for the general case, and a $(4/3 + \eps)$-approximation for the cardinality case.

The results in Chapters~\ref{chap:2dgk-norot} and~\ref{chap:2dgk-rot} are published in \citet*{gghikw17}.

\subsection{Related packing problems}

\paragraph{Maximum Weight Independent Set of Rectangles (MWISR)}\hspace{-10pt} is the restriction of the well known Independent Set problem to the intersection graphs of 2D rectangles. We are given a set $\R$ of axis-aligned rectangles as above, but this time the coordinates $x_i, y_i$ of the bottom-left corner of each rectangle $R_i$ are given as part of the input. The goal is to select a subset of rectangles $\R' \subseteq \R$ of pairwise non-overlapping rectangles, while maximizing the total profit $p(\R') = \sum_{R_i \in \R} p_j$ of the selected rectangles.

\begin{figure}
	\centering
	\includegraphics[width=14cm]{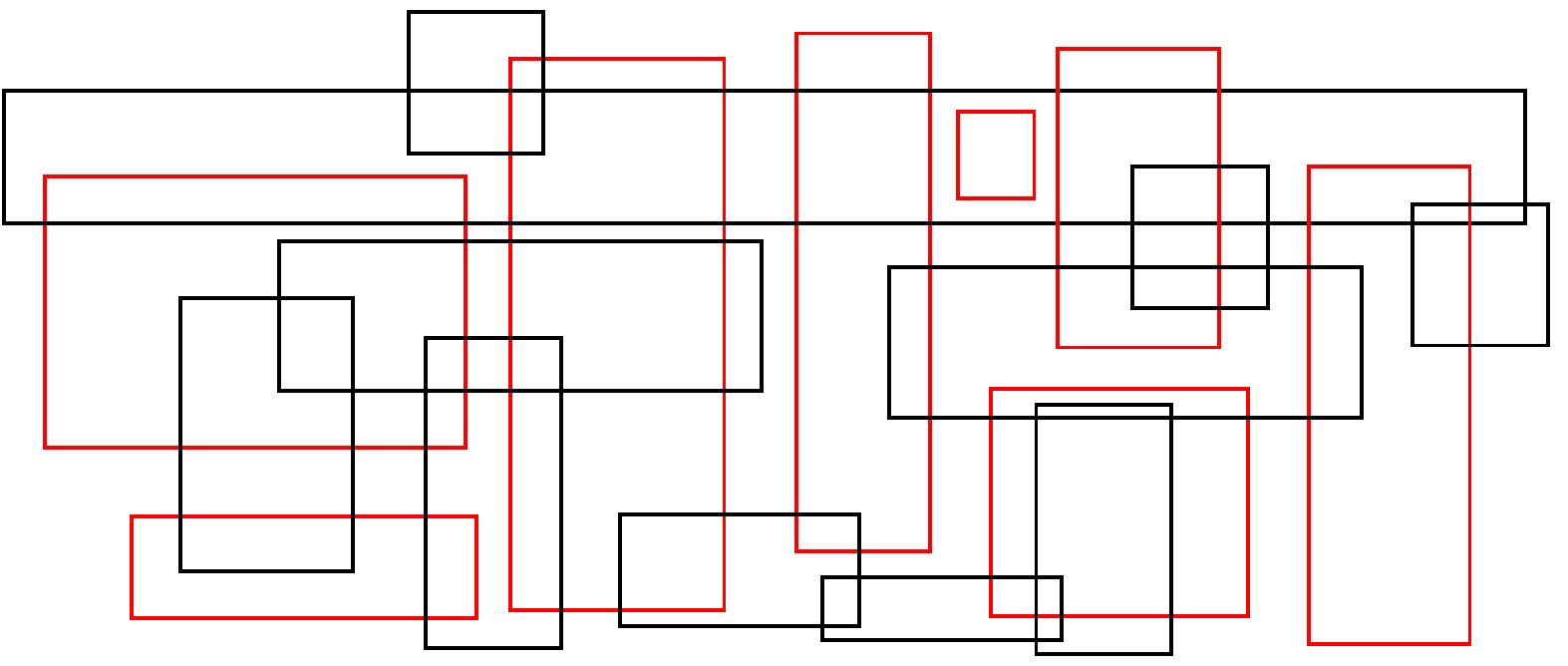}
	\caption{An instance of the Independent Set of Rectangles problem. Since the red rectangles are pairwise non-overlapping, they are a feasible solution to the problem.}
	\label{fig:independent-set-of-rectangles}
\end{figure}

While the Independent Set problem on general graphs is hard to approximate to $n^{1-\eps}$ for any fixed $\eps>0$, much better results are known for MWISR.

The best known approximation for general rectangles was given by \cite{ch12} and has ratio $O\left(\frac{\log n}{\log \log n}\right)$, slightly improving several previous $O(\log n)$-approximations (\cite{avs98,bdmr01,c04}). More recently, a breakthrough by \cite{aw13} showed that a QPTAS is possible. Since only $\NP$-hardness is known as a lower bound, this could suggest that a PTAS is possible for this problem, despite even a constant factor approximation is still not known.

\begin{sloppypar}
	Better results are known for many special cases of the problem. For the unweighted case, \cite{cc09} obtained a $(\log \log n)$-approximation. A PTAS is known for the special case when all the rectangles are squares (\cite{ejs01}). \cite{acw15} gave a PTAS for the relaxation of the problem when rectangles are allowed to be slightly shrunk (more precisely, each rectangle is rescaled by a factor $1-\delta$ for an arbitrarily small $\delta>0$).
\end{sloppypar}

In Chapter~\ref{chap:bagUFP}, we consider a generalization of MWISR where the input rectangles are partitioned into disjoint classes (\emph{bags}), and only one rectangle per class is allowed in a feasible solution.

\paragraph{Storage Allocation Problem (SAP)}\hspace{-10pt} is a problem that arises in the context of resource allocation, but it still has a geometric nature. This problem is similar to \tdk, with two important modifications:
\begin{itemize}
	\item the horizontal position of each rectangle is fixed and given as part of the input (that is, the rectangle is only allowed to be moved vertically);
	\item the target region where each selected rectangle must be packed is the set of points in the positive quadrant that are below a given positive \emph{capacity} curve $c(\cdot)$.
\end{itemize}
See Figure~\ref{fig:SAP} for an example. This problems models the allocation of a resource that has a contiguous range and that varies over time, like bandwidth or disk allocation. Each rectangle represents a request to use the resource, starting at some specified time (the starting $x$-coordinate of the rectangle) and for a duration equal to the width of the rectangle. The height of the rectangle represents the amount of the resource that is requested.
\begin{figure}
	\centering
	\includegraphics[width=14cm]{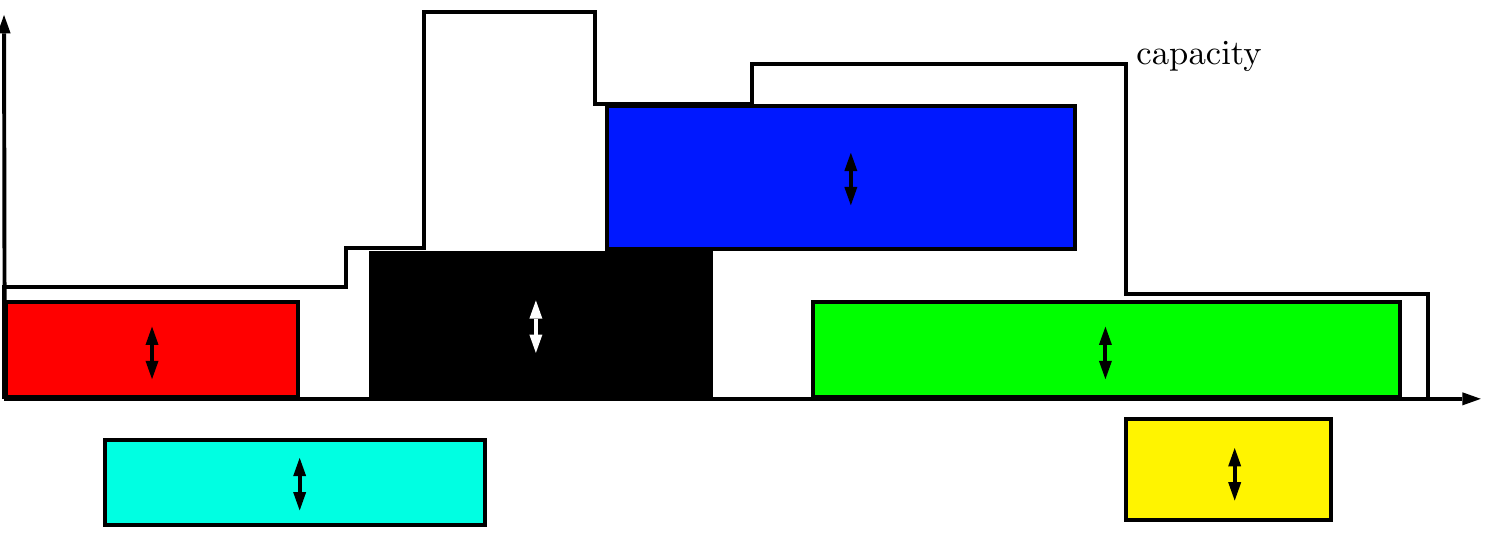}
	\caption{An instance and a feasible solution of the Storage Allocation Problem; two rectangles are not packed. Selected rectangles can only be moved vertically.}
	\label{fig:SAP}
\end{figure}
\cite{bbr13} gave the first constant factor approximation for this problem, with ratio $9 + \eps$; the best known result, with approximation factor $2+\eps$, is due to \cite{mw15}.

See Chapter~\ref{chap:bagUFP} for more related scheduling problems.

\paragraph{$2$-Dimensional Geometric Bin Packing}\hspace{-10pt} is another important geometric packing problem, generalizing the well studied bin packing problem, which is its $1$-dimensional counterpart. Here, we are given an instance as in \tdk, but we want to pack \emph{all} the given rectangles in square knapsacks of size $N \times N$; our goal is to minimize the number of knapsacks used. As for the other problems, both the variants with or without $90^\circ$ rotations are considered.

There is an extensive literature regarding \emph{asymptotic} approximation algorithms. The first results were obtained by \cite{cgj82}, who provided a $2.125$-approximation. The approximation ratio was improved to $(2 + \eps)$ in \cite{kr00}, and then further reduced to $T_\infty + \eps$ by \cite{c02}, where $T_\infty \approx 1.691...$ is the so-called \emph{harmonic constant}. \cite{bcs09} improved the approximation ratio to $1 + \ln T_\infty \approx 1.52$, using a general framework known as Round-and-Approx. Then, \cite{jp13} obtained a $1.5$-approximation, which was improved to $1 + \ln 1.5 \approx 1.405$ by \cite{bk14} with the Round-and-Approx framework. \cite{bcks06} proved that an asymptotic PTAS is impossible, unless $\mathbf{P} = \mathbf{NP}$; \cite{cc06} extended this result to the case with rotations, and also proved explicit lower bounds, showing that it is impossible to approximate $2$-Dimensional Geometric Bin Packing with an asymptotic ratio smaller than $1 + 1/3792$ for the version without rotations, and smaller than $1 + 1/2196$, for the versions with rotations.

For the special case of squares, \cite{fmw98} obtained a $1.998$-approximation, which was improved to a $(14/9 + \eps)$-approximation by \cite{kmrw04} and \cite{sv02}. An algorithm analyzed by \cite{c02} is shown to have an approximation ratio between $1.490$ and $1.507$, although the proof conditionally depends on a certain conjecture. \cite{ev04b} obtained a $(16/11 + \eps)$-approximation. \cite{bcks06} proved that an APTAS is possible for squares (and, more generally, for the $d$-dimensional generalization of the problem). They gave an \emph{exact} algorithm for the relaxation of the problem on rectangles with resource augmentation, that is, where the bins are augmented by a factor $1 + \eps$. Moreover, they provided a PTAS on the related problem of placing a set of rectangles in a minimum-area enclosing rectangle.

Fewer results ar known regarding absolute approximations. As shown by \cite{ltwyc90}, even if the problem is restricted to squares, it is $\mathbf{NP}$-hard to distinguish if $1$ or $2$ bins are needed. \cite{z05} obtained a $3$-approximation for the case without rotations, and \cite{v04} showed that a $2$-approximation is possible for the case of squares. Then, \cite{hv12} showed that a $2$-approximation is possible for general rectangles if rotations are allowed. Note that, due to the $\mathbf{NP}$-hardness mentioned above, an absolute ratio of $2$ is the best possible for any of these problems.

\cite{bls05} proved the surprising result that an asymptotic PTAS is possible for the version of the problem with guillotine cuts.

Many results are also known in the online version of the above problems; we omit them here. A recent review of these and other related problems can be found in \cite{ckpt17}.

\section{Summary of our results and outline of this thesis}

In Chapter~\ref{chap:preliminaries}, we review some preliminary results and develop some useful tools for rectangle packing problems.

Many results in this area of research are obtained by showing that a profitable solution exists that presents a special, simplified structure: namely, the target area where the rectangles are packed is partitioned into a constant number of rectangular regions, whose sizes are chosen from a polynomial set of possibilities. Thus, such simplified structures can be guessed\footnote{As it is common in the literature, by \emph{guessing} we mean trying out all possibilities. In order to lighten the notation, it is often easier to think that we can guess a quantity and thus assume that it is known to the algorithm; it is straightforward to remove all such guessing steps and replace them with an enumeration over all possible choices.} efficiently, and then used to guide the subsequent rectangle selection and packing procedure.

Following the same general approach, we will define one such special structure that we call \emph{container packings}, which we define in Section~\ref{sec:container_packings}. Then, in Sections~\ref{sec:GAP} and \ref{sec:rounding_containers} we show that there is a relatively simple PTAS for container packings based on Dynamic Programming. Notably, this PTAS does not require the solution of any linear program and is purely combinatorial, and it is straightforward to adapt it to the case with $90^\circ$ rotations.

Container packings are a valuable black box tool for rectangle packing problems, generalizing several other such simplified structures used in literature. In Section \ref{sec:2dgk-resource-augmentation} we reprove a known result on \tdk with Resource Augmentation in terms of container packings. We use this modified version in several of our results.

In Chapter~\ref{chap:StripPacking}, we present our new $(4/3 + \eps)$-approximation for Strip Packing in pseudo-polynomial time. The results of this chapter are based on \textbf{\cite*{ggik16}}, published at the \emph{36th IARCS Annual Conference on Foundations of Software Technology and Theoretical Computer Science, December 13--15 2016, Chennai, India}. This result improves on the previously best known $(1.4 + \eps)$-approximation by \cite{nw16} with a novel \emph{repacking lemma} which is described in Section~\ref{sec:repack}. Moreover by using container packings, we obtain a purely combinatorial algorithm (that is, no linear program needs to be solved), and we can easily adapt the algorithm to the case with $90^\circ$ rotations.

Chapters~\ref{chap:2dgk-norot} and~\ref{chap:2dgk-rot} present our results on \tdk and \tdkr from \textbf{\cite*{gghikw17}}, published at the \emph{58th Annual IEEE Symposium on Foundations of Computer Science, October 15--17 2017, Berkeley, California}.

In Chapter~\ref{chap:2dgk-norot}, we present the first polynomial time approximation algorithm that breaks the barrier of $2$ for the \tdk problem (without rotations). Although we still use container packings, a key factor in this result is a PTAS for another special packing problem that we call \fontL\emph{-packing}, which does not seem to be possible to solve with a purely container-based approach. Combining this PTAS with container packing will yield our main result, which is an algorithm with approximation ratio $\frac{17}{9} + \eps < 1.89$. For the unweighted case, we can further refine the approximation ratio to $\frac{558}{325}+\eps < 1.717$ by means of a quite involved case analysis.

In Chapter~\ref{chap:2dgk-rot}, we focus on \tdkr, and we show a relatively simple polynomial time $(3/2 + \eps)$-approximation for the general case; for the unweighted case, we improve it to $4/3 + \eps$ in Section~\ref{sec:cardRot}. In this case, we do not need \fontL-packings, and we again use a purely container-based approach.

In Chapter~\ref{chap:bagUFP}, we study a generalization of the MWISR problem that we call bagMWISR, and we use it to obtain improved approximation for a scheduling problem called bagUFP. A preliminary version of these results was contained in \textbf{\cite*{giu15}}, published in the \emph{13th Workshop on Approximation and Online Algorithms, September 14--18 2015, Patras, Greece}. In Section~\ref{sec:weightedBagUFP} we show that we can generalize the $\left(\frac{\log \log n }{\log n}\right)$-approximation for MWISR by \cite{ch12}, which is based on the \emph{randomized rounding with alterations} framework, to the more general bagMWISR problem, obtaining in turn an improved approximation ratio for bagUFP. Then, in Section~\ref{sec:unweightedBagUFP}, we turn our attention to the unweighted case of bagUFP, and we show that the special instances of bagMWISR that are generated by the reduction from bagUFP have a special structure that can be used to obtain a $O(1)$-approximation. This result builds on the work of \cite{aglw14}.

\chapter{Preliminaries}\label{chap:preliminaries}
In this chapter, we introduce some important tools and preliminary results that we will use extensively in the following chapters. Some results are already mentioned in literature, and we briefly review them; some others, while they are not published in literature to the best of our knowledge, are based on standard techniques; since they are crucial to our approaches, we provide full proofs of them.

\section{Next Fit Decreasing Height}\label{sec:nfdh}

One of the most recurring tools used as a subroutine in countless results on geometric packing problems is the Next Fit Decreasing Height (NFDH) algorithm, which was originally analyzed in \cite{cgjt80} in the context of Strip Packing. We will use a variant of this algorithm to pack rectangles inside a box.

Suppose you are given a box $C$ of size $w\times h$, and a set of rectangles $I'$ each one fitting in the box (without rotations). NFDH computes in polynomial time a packing (without rotations) of $I''\subseteq I'$ as follows. It sorts the rectangles $i\in I'$ in non-increasing order of height $h_i$, and considers rectangles in that order $i_1,\ldots,i_n$. Then the algorithm works in rounds $j\geq 1$. At the beginning of round $j$ it is given an index $n(j)$ and a horizontal segment $L(j)$ going from the left to the right side of $C$. Initially $n(1)=1$ and $L(1)$ is the bottom side of $C$. In round $j$ the algorithm packs a maximal set of rectangles $i_{n(j)},\ldots,i_{n(j+1)-1}$, with bottom side touching $L(j)$ one next to the other from left to right (a \emph{shelf}). The segment $L(j+1)$ is the horizontal segment containing the top side of $i_{n(j)}$ and ranging from the left to the right side of $C$. The process halts at round $r$ when either all the rectangles have being packed or $i_{n(r+1)}$ does not fit above $i_{n(r)}$. See Figure~\ref{fig:nfdh} for an example of packing produced by this algorithm.

\begin{figure}
	\centering
	\includegraphics[width=6cm]{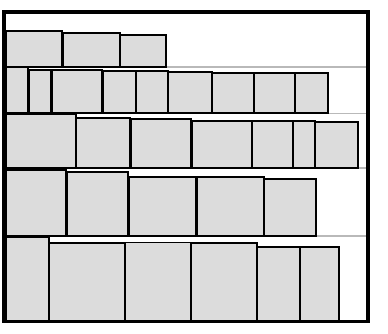}
	\caption{A packing obtained using the classical Next Fit Decreasing Height algorithm.}
	\label{fig:nfdh}
\end{figure}

We prove the following known result:
\begin{lemma}\label{lem:nfdhPack}
	Let $C$ be a given box of size $w \times h$, and let $I'$ be a set of rectangles. Assume that, for some given parameter $\eps\in (0,1)$, for each $i\in I'$ one has $w_i\leq \eps w$ and $h_i\leq \eps h$. Then NFDH is able to pack in $C$ a subset $I''\subseteq I'$ of total area at least $a(I'')\geq \min\{a(I'),(1-2\eps)w\cdot h\}$. In particular, if $a(I')\leq (1-2\eps)w\cdot h$, all the rectangles in $I'$ are packed.
\end{lemma}
\begin{proof}
	The claim trivially holds if all rectangles are packed.  Thus suppose that this is not the case.
	Observe that $\sum_{j=1}^{r+1}h_{i_{n(j)}} > h$, otherwise the rectangle $i_{n(r+1)}$ would fit in the next shelf above $i_{n(r)}$; hence $\sum_{i=2}^{r+1}h_{i_{n(j)}} > h - h_{i_{n(1)}} \geq (1 - \eps)h$. Observe also that the total width of the rectangles packed in each round $j$ is at least $w-\eps w = (1 - \eps)w$, since $i_{n(j+1)}$, of width at most $\eps w$, does not fit to the right of $i_{n(j+1)-1}$. It follows that the total area of the rectangles packed in round $j$ is at least $(w-\eps w)h_{n(j+1)-1}$, and thus:
	\[
		a(I'')\geq \sum_{j=1}^{r}(1-\eps)w \cdot h_{n(j+1)-1}\geq (1-\eps )w\sum_{j=2}^{r+1}h_{n(j)}\geq(1-\eps)^2w\cdot h\geq (1 - 2\eps)w\cdot h.
	\]
\end{proof}

We can use NFDH as an algorithm for Strip Packing, where we work on a box of fixed with $W$ and unbounded height, and we pack all elements. In this context, the following result holds:

\begin{lemma}[\cite{cgjt80}]\label{lem:nfdhStripPacking}
	Given a strip packing instance $(\R,W)$, the NFDH algorithm gives a packing of height at most $h_{\max}(\R) + \frac{2a(\R)} {W}$.
\end{lemma}

\section{Steinberg's algorithm}\label{sec:Steinberg}

The following theorem gives sufficient conditions to pack a given set of rectangles in a rectangular box. We denote $x_+:=\max(x,0)$.

\begin{theorem}[\cite{s97}] \label{thm:steinberg}
	Suppose that we are given a set of rectangles $I'$ and a box $Q$ of size $w \times h$.
	Let $w_{max}\leq w$ and $h_{max}\leq h$ be the maximum width and maximum height among the rectangles in $I'$ respectively.
	If 
	$$
	2a(I') \le wh-(2w_{max}-w)_+(2h_{max}-h)_+
	$$
	then all the rectangles in $I'$ can be packed into $Q$ in polynomial time.
\end{theorem}

In particular, we will use the following simpler corollary:
\begin{corollary}Suppose that we are given a set of rectangles $I'$ and a box $Q$ of size $w \times h$. Moreover, suppose that each rectangle in $I'$ has width at most $w/2$ (resp. each rectangle in $I'$ has height at most $h/2$), and $a(I') \leq wh/2$. Then all the rectangles in $I'$ can be packed into $Q$ in polynomial time.
\end{corollary}

\section{The Maximum Generalized Assignment Problem} \label{sec:GAP}
In this section we show that there is a PTAS for the Maximum Generalized Assignment Problem (GAP) if the number of bins is constant.

GAP is a powerful generalization of the Knapsack problem. We are given a set of $k$ bins, where bin $j$ has capacity $c_j$, and a set of $n$ items. Let us assume that if item $i$ is packed in bin $j$, then it requires size  $s_{ij} \in \mathbb{Z}$ and has profit $p_{ij} \in \mathbb{Z}$. Our goal is to select a maximum profit subset of items, while respecting the capacity constraints on each bin.

GAP is known to be APX-hard and the best known polynomial time approximation algorithm has ratio $(1-1/e+\eps)$ (\cite{fgms11, fv06}). 
In fact, for an arbitrarily small constant $\delta > 0$ (which can even be a function of $n$) GAP remains APX-hard even on the following very restricted instances: bin capacities are identical, and for each item $i$ and bin $j$, $p_{ij} = 1$, and $s_{ij} = 1$ or $s_{ij}=1+ \delta$ (\cite{ck05}). The complementary case, where item sizes do not vary across bins but profits do, is also APX-hard.
However, when all profits and sizes are the same across all the bins (that is, $p_{ij}=p_{ik}$ and $s_{ij}=s_{ik}$ for all bins $j,k$), the problem is known as multiple knapsack problem (MKP) and it admits a PTAS.

Let $p(OPT)$ be the cost of the optimal assignment. First, we show that we can solve GAP \emph{exactly} in pseudopolynomial time.

\begin{lemma}
	\label{lem:GAP-PPT}
	There is an algorithm running in time $O\left(n D^k\right)$ that finds an optimal solution for the maximum generalized assignment problem, where there are $k$ bins and each of them has capacity at most $D$.
\end{lemma}
\begin{proof}
	For each $i \in [n]$ and $d_j \in [c_j]$ for $j \in [k]$, let $S_{i, d_1,d_2, \dots, d_k}$ denote a subset of the set of items $\{1, 2, \dots, i\}$ packed into the bins such that the profit is maximized and the capacity of bin $j$ is at most $d_j$. 
	Let $P[i, d_1, d_2, \dots, d_k]$ denote the profit of $S_{i, d_1, d_2, \dots, d_k}$.
	Clearly $P[1, d_1, d_2, \dots, d_k]$ is known for all $d_j \in [c_j]$ for $j \in [k]$. Moreover, for convenience we define $P[i, d_1, d_2, \dots, d_k] = 0$ if $d_j < 0$ for any $j \in [k]$.
	We can compute the value of $P[i, d_1, d_2, \dots, d_k]$ by using a dynamic program (DP), that exploits the following recurrence:
	\begin{align*}
	P[i, d_1, d_2, \dots, d_k] = \max\{&P[i-1, d_1, d_2, \dots, d_k],\\
	& \max_j \{p_{ij}+ P[i-1, d_1, \dots, d_j -s_{ij},  \dots, d_k]\}\}
	\end{align*}
	With a similar recurrence, we can easily compute a corresponding set $S_{i, d_1,d_2, \dots, d_k}$.\\
	Clearly, this dynamic program can be \sal{executed} in time $O\Big(n \prod\limits_{j=1}^k c_j\Big) = O\left(n D^k\right)$.
\end{proof}

The following lemma shows that we can also solve GAP optimally even in polynomial time, if we are allowed a slight violation of the capacity constraints (that is, in the resource augmentation model).
\begin{lemma}
	\label{lem:GAPresaug}
	There is a $O\left({\left(\frac{1 + \eps}{\eps}\right)}^k n^{k+1}\right)$ time algorithm for the maximum generalized assignment problem with $k$ bins, which returns a solution with profit at least $p(OPT)$ if we are allowed to augment the bin capacities by a $(1+\eps)$-factor for any fixed $\eps>0$.
\end{lemma}
\begin{proof}
	In order to obtain a polynomial time algorithm from Lemma~\ref{lem:GAP-PPT}, we want to construct a modified instance where each capacity $c_j$ is polynomially bounded.
	
	For each bin $j$, let $\mu_j=\frac{\eps c_j}{n}$.
	For item $i$ and bin $j$, define the modified size  $s'_{ij}= \left\lceil \frac{s_{ij}}{\mu_j} \right\rceil = \left\lceil \frac{n  s_{ij}}{\eps c_j} \right\rceil$
	and $c'_{j}= \left\lfloor \frac{(1+\eps)c_j}{\mu_j} \right\rfloor$.
	Note that $c'_j = \left\lfloor\frac{(1+\eps)n}{\eps}\right\rfloor \leq \frac{(1 + \eps)n}{\eps}$, so the algorithm from Lemma~\ref{lem:GAP-PPT} requires time at most $O\left(n \cdot \left(\frac{(1 + \eps)n}{\eps}\right)^k\right)$.
	
	Let $OPT_{modified}$ be the solution found for the modified instance.
	Now consider the optimal solution for the original instance (that is, with the original item and bin sizes) $OPT$. 
	If we show that the same assignment of items to the bins is a feasible solution  (with modified bin sizes and item sizes) for the modified instance, we obtain that $p(OPT_{modified}) \ge p(OPT)$ and that will conclude the proof.
	
	Let $S_j$ be the set of items packed in bin $j$ in $OPT$.
	Since it is feasible, we have that $\sum_{i \in S_j} s_{ij} \le c_j$. Hence,
	\begin{align*}
		\sum_{i \in S_j} s'_{ij} 
		& \le  \left\lfloor \sum_{i \in S_j} \left(\frac{s_{ij}}{\mu_j}+1\right) \right\rfloor  \le  \left\lfloor \frac{1}{\mu_j}\left(\sum_{i \in S_j} s_{ij} +|S_j| \mu_j \right)\right\rfloor\\
		&\le  \left\lfloor \frac{1}{\mu_j}(c_j + n \mu_j) \right\rfloor = \left\lfloor \frac{(1+\eps)c_j}{\mu_j} \right\rfloor=c'_j
	\end{align*}
	Thus $OPT$ is a feasible solution for the modified instance and the above algorithm will return a packing with profit at least $p(OPT)$ under $\eps$-resource augmentation.
\end{proof}

Now we can show how to employ this result to obtain a PTAS for GAP without violating the bin capacities. We first prove the following technical lemma.

\begin{lemma}\label{lem:GAP-bin-shifting-argument}
	If a set of items $R_j$ is packed in a bin $j$ with capacity $c_j$, then there exists a set of  at most $1/\eps^2$ items $X_j$, and a set of items $Y_j$ with $p(Y_j) \le \eps p(R_j)$ such  that all items in $R_j\setminus (X_j \cup Y_j)$ have size at most $\eps (c_j - \sum_{i \in X_j} s_{ij})$. 
\end{lemma}
\begin{proof}
	Let $Q_1$ be the set of items $i$ with $s_{ij} > \eps c_j $. If $p(Q_1) \le \eps p(R_j)$, we are done by taking $Y_j=Q_1$ and $X_j=\phi$.
	Otherwise, define $X_j:=Q_1$ and we continue the next iteration with the remaining items.
	Let $Q_2$ be the items with size greater than $\eps (c_j - \sum_{i \in X_j} s_{ij})$ in $R_j \setminus X_j$. If $p(Q_2) \le \eps p(R_j)$, we are done by taking $Y_j=Q_2$.
	Otherwise define $X_j:= Q_1 \cup Q_2$ and we continue with further iterations till  
	we get a set $Q_t$ with $p(Q_t) \le \eps p(R_j)$. Note that we need at most $\frac{1}{\eps}$ iterations, since the sets $Q_i$ are disjoint.
	Otherwise:
	\[
	p(R_j) \ge \sum\limits_{i=1}^{1/\eps} p(Q_i) > \sum\limits_{i=1}^{1/\eps} \eps p(R_j) \ge p(R_j)
	\]
	which is a contradiction.
	Thus, consider $Y_j=Q_t$ and $X_j=\bigcup_{l=1}^{t-1} Q_l$. One has $|X_j| \le 1/\eps^2$ and $p(Y_j)\le \eps p(R_j)$.  On the other hand, after removing $Q_t$, the remaining  items have size smaller than $\eps (c_j - \sum_{i \in X_j} s_{ij})$. 
\end{proof}

\begin{lemma}\label{lem:GAP}
	There is an algorithm for the maximum generalized assignment problem with $k$ bins that runs in time $O\left(\left(\frac{1 + \eps}{\eps}\right)^k n^{k/\eps^2 + k + 1}\right)$ and returns a solution that has profit at least $(1 - 3\eps)p(OPT)$, for any fixed $\eps > 0$.
\end{lemma}
\begin{proof}
	Consider a bin $j$ that contains the set $R_j$ of items in the optimal solution OPT, and let $X_j$ and $Y_j$ the sets given by Lemma~\ref{lem:GAP-bin-shifting-argument}. Let $c'_j = c_j - \sum_{i \in X_j} s_{ij}$ be the residual capacity, so that each element in the set $R'_j := R_j \setminus (X_j \cup Y_j)$ has size in $j$ smaller than $\eps c'_j$. We divide the residual space into $1/\eps$ equally sized intervals $S_{j,1}, S_{j,2}, \dots, S_{j,{1/\eps}}$ of lengths $\eps c'_j$. Let $R'_{j,l}$ be the set of items intersecting the interval $S_{j,l}$.
	As each packed item can belong to at most two such intervals, the cheapest  set $R''$ among $\{ R'_{j,1 }, \dots, R'_{j,1/\eps} \}$ has profit at most $2 \eps p(R'_j)$. Thus we can remove this set $R''$ and reduce the bin size by a factor of $(1-\eps)$.
	
	Now consider the packing of the $k$ bins in the optimal packing $OPT$. Let $R_j$ be the set of items packed in bin $j$. 

	The algorithm first guesses all $X_j$'s, a constant number of items, in all $k$ bins. We assign them to corresponding bins, implying a $O(n^{k/\eps^2})$ factor in the running time. 
	Then for bin $j$ we are left with capacity $c'_j$.
	From the previous discussion, we know that there is packing of $R''_j \subseteq R_j \setminus X_j$ of profit $(1-2\eps)p(R_j \setminus X_j)$ in a bin with capacity $(1-\eps)c'_j$.
	Thus  we can use the algorithm for GAP with resource augmentation provided by Lemma~\ref{lem:GAPresaug} to pack the remaining items in $k$ bins where for bin $j$ we use the original capacity to be $(1-\eps)c'_j$ for $j \in [k]$; note that $(1 - \eps)(1 + \eps) \leq 1$, so the solution is feasible with the capacities $c'_j$.
	As Lemma~\ref{lem:GAPresaug} returns the optimal packing on this modified bin sizes, we obtain a total profit of at least $(1-2\eps)(1-\eps)p(OPT) \geq (1 - 3\eps)p(OPT)$. \sal{The running time is the same as in Lemma~\ref{lem:GAPresaug} multiplied by the $O(n^{k/\eps^2})$ factor for the initial item guessing and assignment.}
\end{proof}

\section{Container packings}\label{sec:container_packings}

In this section we define the main concept of our framework: a \emph{container}.

Many of the literature results on geometric packing problems follow (implicitly or explicitly) the following approach: since the number of possible packings is too big to be enumerated, a search algorithm is performed only on a restricted family of packings that have a special structure. Thus, the theoretical analysis aims to prove that there exists such a restricted packing that has a high profit.

We follow the same general framework, by defining what we call a \emph{container packing}.


By \emph{container} we mean a special kind of box to which we assign a set of rectangles that satisfy some constraints, as follows (see Figure~\ref{fig:boxescontainers}):

\begin{itemize}
	\item A \emph{horizontal container} is a box such that any horizontal line overlaps at most one rectangle packed in it.
	\item A \emph{vertical container} is a box such that any vertical line overlaps at most one rectangle packed in it.
	\item An $\eps$-granular area container is a box, say of size $w \times h$, such that all the rectangles that are packed inside have width at most $\eps w$ and height at most $\eps h$. We will simply talk about an \emph{area container} when the value of $\eps$ is clear from the context. 
\end{itemize}

\begin{figure}
	\resizebox{!}{5cm}{     
		\includegraphics{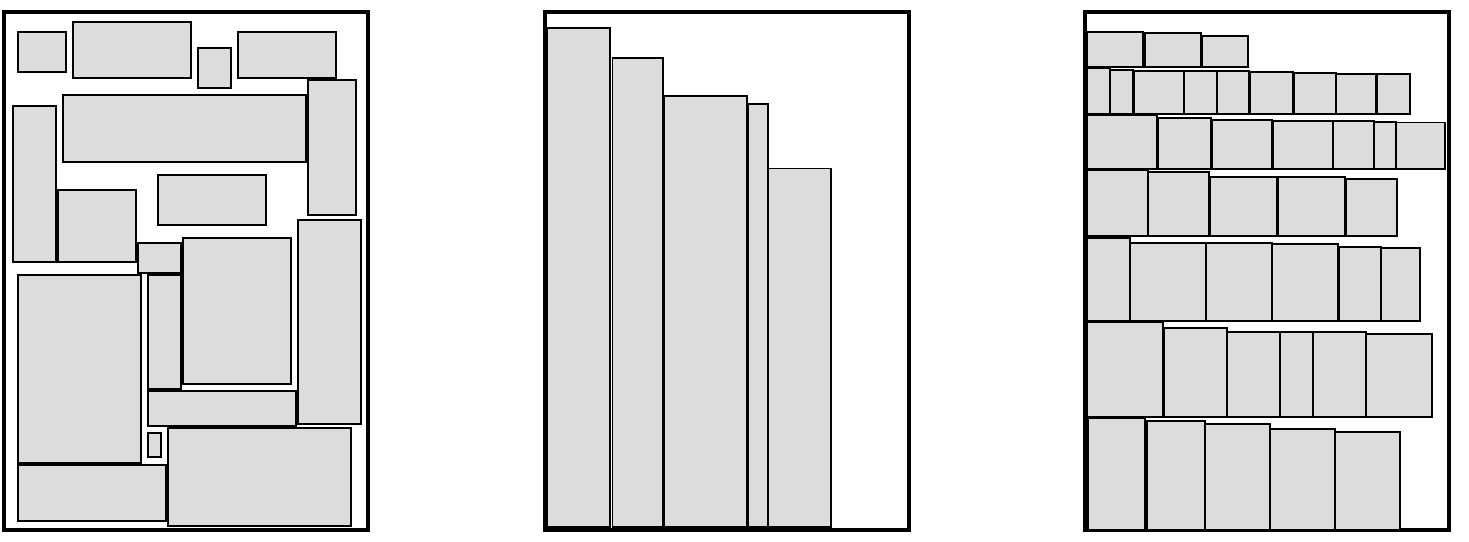}
	}
	\caption{From left to right: a \emph{box}, which is a rectangular region with arbitrary rectangles packed inside; a \emph{vertical container}, which is a box where rectangles are piled from left to right; an \emph{area container}, where width and height of the rectangles inside is much smaller than the width and height of the container, making it easy to pack them by NFDH.}
	\label{fig:boxescontainers}
\end{figure}

A packing such that \emph{all} the packed rectangles are contained in a container and all the area containers are $\eps$-granular is called and \emph{$\eps$-granular container packing}; again, we will simply call it a \emph{container packing} when the choice of $\eps$ is clear from the context.

Observe that for a horizontal or a vertical container, once a set of rectangles that can feasibly be packed is assigned, constructing a packing is trivial. Moreover, the next lemma shows that it is easy to pack almost all the rectangles assigned to an area container:

\begin{lemma}\label{lem:packAreaContainer}
	Suppose that a set $\mathcal{R}$ of rectangles is assigned to an $\eps$-granular area container $C$, and $a(\mathcal{R}) \leq a(C)$. Then it is possible to pack in $C$ a subset of $\mathcal{R}$ of profit at least $(1 - 3\eps)p(\mathcal{R})$.
\end{lemma}
\begin{proof}
	If $a(\mathcal{R}) \leq (1 - 2\eps)a(C)$, then NFDH can pack all the rectangles by Lemma~\ref{lem:nfdhPack}. Suppose that $a(\mathcal{R}) > (1 - 2\eps)a(C)$. Consider the elements of $\mathcal{R}$ by non-increasing order of profit over area ratio. Let $\mathcal{R'}$ be the maximal subset of rectangles of $\mathcal{R}$ in the specified order such that $a(\mathcal{R}') \leq (1 - 2\eps)a(C)$. Since $a(R) \leq \eps^2 a(C)$ for each $R \in \mathcal{R}$, then $a(R')\geq (1 - 2\eps - \eps^2)a(C) \geq (1 - 3\eps)a(C)$, and then $a(R') \geq (1 - 3\eps)a(R)$, which implies $p(\mathcal{R}') \geq (1 - 3\eps)p(\mathcal{R})$ by the choice of $R'$. Since $a(\mathcal{R}') \leq (1 - 2\eps)a(C)$, then NFDH can pack all of $\mathcal{R}'$ inside $C$ by Lemma~\ref{lem:nfdhPack}.
\end{proof}

In the remaining part of this section, we show that container packings are easy to approximate if the number of containers is bounded by some fixed constant.

\subsection{Rounding containers}\label{sec:rounding_containers}
In this subsection we show that it is possible to round down the size of a horizontal, vertical or area container so that the resulting sizes can be chosen from a polynomially sized set, while incurring in a \sal{negligible} loss of profit.

We say that a container $C'$ is smaller than a container $C$ if $w(C') \leq w(C)$ and $h(C') \leq h(C)$. Given a container $C$ and a positive $\eps < 1$, we say that a rectangle $R_j$ is $\eps$-small for $C$ if $w_j \leq \eps w(C)$ and $h_j \leq \eps h(C)$.

\begin{sloppypar}
For a set $\R$ of rectangles, we define $WIDTHS(\R) = \{w_j \, | \, R_j \in \R\} $ and $HEIGHTS(R) = \{h_j \, | \, R_j \in \R\}$.
\end{sloppypar}

Given a finite set $P$ of real numbers and a fixed natural number $k$, we define the set $P^{(k)} = \{(p_1 + p_2 + \dots + p_l) + i p_{l+1} \, | \, p_j \in P \text{ } \forall \, j, l\le k, 0 \leq i \leq n, i\mbox{ integer}\}$; note that if $|P| = O(n)$, then $|P^{(k)}| = O(n^{k+2})$. Moreover,  if $P \subseteq Q$, then obviously $P^{(k)} \subseteq Q^{(k)}$, and if $k' \leq k''$, then $P^{(k')} \subseteq P^{(k'')}$.

\begin{lemma}\label{lem:round_knapsack_container}
	Let $\eps > 0$, and let $\R$ be a set of rectangles packed in a horizontal or vertical container $C$. Then, for any $k \geq 1/\eps$, there is a set $\R' \subseteq \R$ with profit $p(\R') \geq (1 - \eps)p(\R)$ that can be packed in a container $C'$ smaller than $C$ such that $w(C') \in WIDTHS(\R)^{(k)}$ and $h(C') \in HEIGHTS(R)^{(k)}$.
\end{lemma}
\begin{proof}
	Without loss of generality, we prove the thesis for an horizontal container $C$; the proof for vertical containers is symmetric. Clearly, the width of $C$ can be reduced to $w_{max}(\R)$, and $w_{max}(\R) \in WIDTHS(\R) \subseteq {WIDTHS(\R)}^{(k)}$.
	
	If $|\R| \leq 1/\eps$, then $\sum_{R_i \in \R} h_i \in HEIGHTS(\R)^{(k)}$ and there is no need to round the height of $C$ down. Otherwise, let $\R_{TALL}$ be the set of the $1/\eps$ rectangles in $\R$ with largest height (breaking ties arbitrarily), let $R_j$ be the least profitable of them, and let $\R' = \R \setminus \{R_j\}$. Clearly, $p(\R') \geq (1 - \eps)p(\R)$.
	Since each element of $\R' \setminus \R_{TALL}$ has height at most $h_j$, it follows that $h(\R \setminus \R_{TALL}) \leq (n - 1/\eps) h_j$. Thus, letting $i = \left\lceil h(\R' \setminus \R_{TALL}) / h_j \right\rceil \leq n$, all the rectangles in $\R'$ fit in a container $C'$ of width $w_{max}(\R)$ and height $h(C') := h(\R_{TALL}) + i h_j \in {HEIGHTS(R)}^{(k)}$. Since $h(\R_{TALL}) + i h_j \leq h(\R_{TALL}) + h(\R' \setminus \R_{TALL}) + h_j = h(\R) \leq h(C)$, this proves the result.
\end{proof}

\begin{lemma}\label{lem:round_area_container}
	Let $\eps > 0$, and let $\R$ be a set of rectangles that are assigned to an area container $C$. Then there exists a subset $\R' \subseteq \R$ with profit $p(\R') \geq (1 - 3\eps)p(\R)$ and a container $C'$ smaller than $C$ such that: $a(\R') \leq a(C)$, $w(C') \in {WIDTHS(\R)}^{(0)}$, $h(C') \in {HEIGHTS(\R)}^{(0)}$, and each $R_j \in \R'$ is $\dfrac{\eps}{1-\eps}$-small for $C'$. 
\end{lemma}
\begin{proof}
	Without loss of generality, we can assume that $w(C) \leq n w_{max}(\R)$ and $h(C) \leq n h_{max}(\R)$: if not, we can first shrink $C$ so that these conditions are satisfied, and all the rectangles still fit in $C$.
	
	Define a container $C'$ that has width $w(C') = w_{max}(\R) \left\lfloor w(C)/w_{max}(\R) \right\rfloor$ and height $h(C') = h_{max}(\R) \left\lfloor h(C)/h_{max}(\R) \right\rfloor$, that is, $C'$ is obtained by shrinking $C$ to the closest integer multiples of $w_{max}(\R)$ and $h_{max}(\R)$. Observe that $w(C') \in {WIDTHS(\R)}^{(0)}$ and $h(C') \in {HEIGHTS(\R)}^{(0)}$. Clearly, $w(C') \geq w(C) - w_{max}(\R) \geq w(C) - \eps w(C) = (1 - \eps) w(C)$, and similarly $h(C') \geq (1 - \eps) h(C')$. Hence $a(C') \geq (1 - \eps)^2 a(C) \geq (1 - 2\eps) a(C)$.
	
	We now select a set $\R' \subseteq \R$ by greedily choosing elements from $\R$ in non-increasing order of profit/area ratio, adding as many elements as possible without exceeding a total area of $(1 - 2\eps) a(C)$. Since each element of $\R$ has area at most $\eps^2 a(C)$, then either all elements are selected (and then $p(\R') = p(\R)$), or the total area of the selected elements is at least $(1 - 2\eps - \eps^2)a(C) \geq (1 - 3\eps)a(C)$. By the greedy choice, we have that $p(\R') \geq (1 - 3\eps) p(\R)$.
	
	Since each rectangle in $\R$ is $\frac{\eps}{1 - \eps}$-small for $C'$, this proves the thesis.
\end{proof}


\begin{remark}
	Note that in the above, the size of the container is rounded to a family of sizes that depends on the rectangles inside; of course, they are not known in advance in an algorithm that enumerates over all the container packings. On the other hand, if the instance is a set $\mathcal{\mathcal{I}}$ of $n$ rectangles, then for any \sal{fixed} natural number $k$ we have that $WIDTHS(\R)^{(k)} \subseteq WIDTHS(\mathcal{I})^{(k)}$ and $HEIGHTS(\R)^{(k)} \subseteq WIDTHS(\mathcal{\mathcal{I}})^{(k)}$ for any $\R \subseteq \mathcal{I}$; clearly, the resulting set of possible widths and heights has a polynomial size and can be computed from the input.
	
	Similarly, when finding container packings for the case with rotations, one can compute the set $SIZES(\mathcal{I}) := WIDTHS(\mathcal{I}) \cup HEIGHTS(\mathcal{I})$, and consider containers of width and height in $SIZES(\mathcal{I})^{(k)}$ for a sufficiently high constant $k$. 
\end{remark}

\subsection{Packing rectangles in containers}

In this section we prove the main result of this chapter: namely, that there is a PTAS for \tdk for packings into a constant number of containers.

\begin{theorem}\label{thm:container_packing_ptas}
	Let $\eps > 0$, and let $OPT_c$ be the optimal $\eps$-granular container packing for a \tdk instance into some fixed number $K = O_\eps(1)$ of containers. Then there exists a polynomial time algorithm that outputs a packing $APX_c$ such that $p(APX_c) \geq (1 - O(\eps)) p(OPT_c)$. The algorithm works in both the cases with or without $90^\circ$ rotations.
\end{theorem}
\begin{proof}
	Let $OPT'_c$ be the rounded container packing obtained from $OPT_c$ after rounding each container as explained in Lemmas~\ref{lem:round_knapsack_container} and \ref{lem:round_area_container}; clearly, $p(OPT'_c) \geq (1 - O(\eps))p(OPT_c)$. Moreover, the sizes of all the containers in $OPT'_c$ and a feasible packing for them can be guessed in polynomial time.

	Consider first the case without rotations. We construct the following instance of GAP (see Section \ref{sec:GAP} for the notation), where we define a bin for each container of $OPT'_c$.

	For each horizontal (resp., vertical) container $C_j$ of size $a\times b$, we define one bin $j$ with capacity $c_j$ equal to $b$ (resp., $a$). For each area container $C_j$ of size $a\times b$, we define one bin $j$ with capacity $c_j$ equal to $a\cdot b$. For each rectangle $R_i$ we define one element $i$, with profit $p_i$. We next describe a size $s_{i,j}$ for every element-bin pair $(i,j)$. If bin $j$ corresponds to a horizontal (resp., vertical) container $C_j$ of capacity $c_j$, then $s_{i,j}=h_i$ (resp., $s_{i,j}=w_i$) if $w_i\leq a$ (resp., $h_i\leq b$) and $s_{i,j}=+\infty$ otherwise. Instead, if $j$ corresponds to an area container of size $a\times b$, then we set $s_{i,j}=h_i\cdot w_i$ if $w_i\leq \eps a$ and $h_i\leq \eps b$, and $s_{i,j}=+\infty$ otherwise. 
	
	By using the algorithm of Lemma~\ref{lem:GAP}, we can compute a ($1-3\eps$)-approximate solution for the GAP instance. This immediately induces a feasible packing for horizontal and vertical containers. For each area container, we pack the rectangles by using Lemma~\ref{lem:packAreaContainer}, where we lose another $3\eps$-fraction. Overall, we obtain a solution with profit at least $(1 - 3\eps)^2 p(OPT'_c) \geq (1 - O(\eps)) p(OPT_c)$.
	
	In the case with rotations we use the same approach, but defining the GAP instance in a slightly different way. For a horizontal containers $C_j$ of size $a\times b$, we consider the same $s_{i,j}$ as before and update it to $w_i$ if $h_i\leq a$, and $w_i<s_{i,j}$. In the latter case, if element $i$ is packed into bin $j$, then the rectangle $R_i$ is packed rotated inside the container $C_j$. For vertical containers we perform a symmetric assignment.
	
	For an area container $C_j$ of size $a\times b$, if according to the above assignment one has $s_{i,j}=+\infty$, then we update $s_{i,j}$ to $w_i \cdot h_i$ if $w_i\leq \eps b$ and $h_i\leq \eps a$. In the latter case, if element $i$ is packed into bin $j$, then item $i$ is packed rotated inside $C_j$. 
\end{proof} 

If we are allowed pseudo-polynomial time (or if all widths and heights are polynomially bounded), we can obtain the following result, that is useful when we want to pack \emph{all} the given rectangles, but we are allowed to use a slightly larger target region.

\begin{theorem}\label{thm:container_packing_ptas_ppt}
	Let $\eps > 0$, and let $I$ be a set of rectangles that admits a $\eps$-granular container packing with $K = O_\eps(1)$ containers into a $a \times b$ knapsack. Then there exists a PPT algorithm that packs all the rectangles in $I$ into a $(1 + 2\eps)a \times b$ (resp. $a \times (1 + 2\eps)b$) knapsack. The algorithm works in both the cases with or without $90^\circ$ rotations.
\end{theorem}
\begin{proof}
	We prove the result only for the $a \times (1 + 2\eps)b$ target knapsack, the other case being symmetric.
	
	Since we are allowed pseudo-polynomial time, we can guess exactly all the containers and their packing in the $a \times b$ knapsack. Then, we build the GAP instance as in the proof of Theorem~\ref{thm:container_packing_ptas}. This time, in pseudo-polynomial time, we can solve it \emph{exactly} by Lemma~\ref{lem:GAP-PPT}, that is, we can find an assignment of all the rectangles into the containers.
	
	We now enlarge the height of each area container by a factor $1 + 2\eps$, that is, if the container $\mathcal{C}$ has size $w(C) \times h(C)$, we enlarge its size to $w(C) \times (1 + 2\eps)h(C)$; clearly, it is still possible to pack all the containers in the enlarged $a \times (1 + 2\eps)b$ knapsack. The proof is concluded by observing that all the rectangles that are assigned to $\mathcal{C}$ can be packed in the enlarged container by Lemma~\ref{lem:nfdhPack}.
\end{proof}

\section{Packing Rectangles with Resource Augmentation}\label{sec:2dgk-resource-augmentation}

In this chapter we prove that it is possible to pack a high profit subset of rectangles into boxes, if we are allowed to augment one side of a knapsack by a small fraction.

Note that we compare the solution provided by our algorithm with the optimal solution \emph{without} augmentation; hence, the solution that we obtain is not feasible. Still, the violation can be made arbitrarily small, and this result will be a valuable tool for other packing algorithms, and we employ it extensively in the next chapters.

The result that we describe is essentially proved in \cite{js09}, although we introduce some modifications and extensions to obtain the additional properties relative to packing into containers and a guarantee on the area of the rectangles in the obtained packing;  moreover, by using our framework of packing in containers, we obtain a substantially simpler algorithm. For the sake of completeness, we provide a full proof, which follows in spirit the proof of the original result. We will prove the following:

\begin{lemma}[Resource Augmentation Packing Lemma]\label{lem:structural_lemma_augm}
	Let $I'$ be a collection of rectangles that can be packed into a box of size $a\times b$, and $\eps_{ra}>0$ be a given constant. Then there exists an $\eps_{ra}$-granular container packing of $I''\subseteq I'$ inside a box of size $a\times (1+\eps_{ra})b$ (resp., $(1+\eps_{ra})a\times b$) such that:
	\begin{enumerate}\itemsep0pt
		\item $p(I'')\geq (1-O(\eps_{ra}))p(I')$;
		\item the number of containers is $O_{\eps_{ra}}(1)$ and their sizes belong to a set of cardinality $n^{O_{\eps_{ra}}(1)}$ that can be computed in polynomial time;
		\item the total area of the containers is at most $a(I') + \eps_{ra} ab$.
	\end{enumerate}
\end{lemma}

Note that in this result we do not allow rotations, that is, rectangles are packed with the same orientation as in the original packing. However, as an existential result we can apply it also to the case with rotations. Moreover, since Theorem~\ref{thm:container_packing_ptas} gives a PTAS for approximating container packings, this implies a simple algorithm that does not need to solve any LP to find the solution, in both the cases with and without rotations.

For simplicity, in this section we assume that widths and heights are positive real numbers in $(0, 1]$, and $a = b = 1$: in fact, all elements, container and boxes can be rescaled without affecting the property of a packing of being a \emph{container packing} with the above conditions. Thus, without loss of generality, we prove the statement for the augmented $1 \times (1 + \eps_{ra})$ box.

We will use the following Lemma, that follows from the analysis in \cite{kr00}:

\begin{lemma}[\cite{kr00}]\label{lem:Kenyon-Remila}
	Let $\overline{\eps} > 0$, and let $\mathcal{Q}$ be a set of rectangles, each of height and width at most $1$. Let $\mathcal{L} \subseteq \mathcal{Q}$ be the set of rectangles of width at least $\overline{\eps}$, and let $OPT_{SP}(\mathcal{L})$ be the minimum width such that the rectangles in $\mathcal{L}$ can be packed in a box of size $OPT_{SP}(\mathcal{L})\times 1$.
	
	Then $\mathcal{Q}$ can be packed in polynomial time into a box of height $1$ and width $\tilde{w} \le \max\{OPT_{SP}(\mathcal{L}) + \frac{18}{\overline{\eps}^2} w_{\max}, a(\mathcal{Q})(1+\overline{\eps}) + \frac{19}{\overline{\eps}^2}w_{\max}\}$, where $w_{\max}$ is the maximum width of rectangles in $\mathcal{Q}$. Furthermore, all the rectangles with both width and height less than $\overline{\eps}$ are packed into at most $\frac{9}{\overline{\eps}^2}$ boxes, and all the remaining rectangles into at most $\frac{27}{\overline{\eps}^3}$ vertical containers.
\end{lemma}
Note that the boxes containing the rectangles that are smaller than $\overline{\eps}$ are not necessarily packed as containers.

We need the following technical lemma:
\begin{lemma}\label{lem:ra_intermediate}
	Let $\eps>0$ and let $f(\cdot)$ be any positive increasing function such that $f(x)<x$ for all $x$. Then, there exist positive constant values $\delta,\mu \in \Omega_\eps(1)$, with \sal{$\delta \leq f(\eps)$} and \sal{$\mu \leq f(\delta)$} such that the total profit of all the rectangles whose width or height lies in $(\mu, \delta]$ is at most $\eps\cdot p(I')$.
\end{lemma}
\begin{proof}
	Define $k+1=2/\eps+1$ constants $\eps_1,\ldots,\eps_{k+1}$, with 
	\sal{$\eps_1=f(\eps)$}
	and $\eps_{i}=f(\eps_{i+1})$ for each~$i$. Consider the $k$ ranges of widths and heights of type $(\eps_{i+1},\eps_{i}]$. By an averaging argument there exists one index $j$ such that the total profit of the rectangles in $I'$ with at least one side length in the range $(\eps_{j+1}N,\eps_{j}N]$ is at most $2\frac{\eps}{2}p(I')$. It is then sufficient to set $\delta=\eps_j$ and $\mu=\eps_{j+1}$.    
\end{proof}
We use this lemma with $\eps = \eps_{ra}$, and we will specify the function $f$ later. By properly choosing the function $f$, in fact, we can enforce constraints on the value of $\mu$ with respect to $\delta$, which will be useful several times; the exact constraints will be clear from the analysis. Thus, we remove from $I'$ the rectangles that have at least one side length in $(\mu, \delta]$.

We call a rectangle $R_i$ \emph{wide} if $w_i > \delta$, \emph{high} if $h_i > \delta$, \emph{short} if $w_i \le \mu$ and \emph{narrow} if $h_i \le \mu$.\footnote{Note that the classification of the rectangles in this section is different from the ones used in the main results of this thesis, although similar in spirit.}
From now on, we will assume that we start with the optimal packing of the rectangles in $R'$, and we will modify it until we obtain a packing with the desired properties.
We remove from $R'$ all the short-narrow rectangles, initially obtaining a packing. We will show in section \ref{sec:shortnarrow} how to use the residual space to pack them, with a negligible loss of profit.

As a first step, we round up the widths of all the \emph{wide} rectangles in $R'$ to the nearest multiple of $\delta^2$; moreover, we shift them horizontally so that their starting coordinate is an integer multiple of $\delta^2$ (note that, in this process, we might have to shift also the other rectangles in order to make space). Since the width of each wide rectangle is at least $\delta$ and $\frac{1}{\delta}\cdot 2\delta^2 = 2\delta$, it is easy to see that it is sufficient to increase the width of the box to $1 + 2\delta$ to perform such a rounding.

\subsection{Containers for short-high rectangles}\label{sec:containers-short-rectangles}

\begin{figure}
	\centering
	\includegraphics[width=12cm]{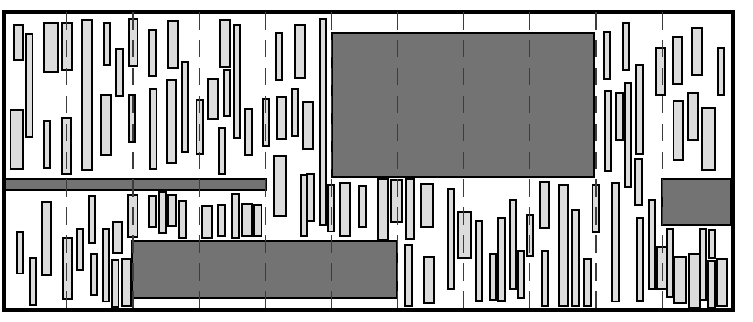}
	\caption{An example of a packing after the short-narrow rectangles have been removed, and the wide rectangles (in dark grey) have been aligned to the $M$ vertical strips. Note that the short-high rectangles (in light gray) are much smaller than the vertical strips.}
	\label{fig:augm-round-x}
\end{figure}

We draw vertical lines across the $1\times (1+2\delta)$ region spaced by $\delta^2$, splitting it into $M := \frac{1+2\delta}{\delta^2}$ vertical strips (see Figure~\ref{fig:augm-round-x}). Consider each maximal rectangular region which is contained in one such strip and does not overlap any wide rectangle; we define a box for each such region that contains at least one short-high rectangle, and we denote the set of such boxes by $\mathcal{B}$.

Note that some short rectangles might intersect the vertical edges of the boxes, but in this case they overlap with exactly two boxes. Using a standard shifting technique, we can assume that no rectangle is cut by the boxes by losing profit at most $\eps_{ra} OPT$: first, we assume that the rectangles intersecting two boxes belong to the leftmost of those boxes. For each box $B \in \mathcal{B}$ (which has width $\delta^2$ by definition), we divide it into vertical strips of width $\mu$. Since there are $\frac{\delta^2}{\mu} > 2/\eps_{ra}$ strips and each rectangle overlaps with at most $2$ such strips, there must exist one of them such that the profit of the rectangles intersecting it is at most $2\mu p(B) \leq \eps_{ra} p(B)$, where $p(B)$ is the profit of all the rectangles that are contained in or belong to $B$. We can remove all the rectangles overlapping such strip, creating in $B$ an empty vertical gap of width $\mu$, and then we can move all the rectangles intersecting the right boundary of $B$ to the empty space, as depicted in Figure~\ref{fig:augm-fix-vertical-box}.
\begin{proposition}
	The number of boxes in $\mathcal{B}$ is at most $\frac{1+2\delta}{\delta^2} \cdot \frac{1}{\delta} \le \frac{2}{\delta^3}$.
\end{proposition}
First, by a shifting argument similar to above, we can reduce the width of each box to $\delta^2 - \delta^4$ while losing only an $\eps_{ra}$ fraction of the profit of the rectangles in $B$. Then, for each $B\in \mathcal{B}$, since the maximum width of the rectangles in $B$ is at most $\mu$, by applying Lemma~\ref{lem:Kenyon-Remila} with $\overline{\eps} = \delta^2/2$ we obtain that the rectangles packed inside $B$ can be repacked into a box $B'$ of height $h(B)$ and width at most $w'(B) \le \max\{\delta^2 - \delta^4 + \frac{72}{\delta^4} \mu, (\delta^2 - \delta^4)(1+\frac{\delta^2}{2}) + \frac{76}{\delta^4}\mu\}\le \delta^2$, which is true if we make sure that $\mu \leq \delta^{10} / 76$. Furthermore, the short-high rectangles in $B$ are packed into at most $\dfrac{216}{\delta^6} \leq \dfrac{1}{\delta^7}$ vertical containers, assuming without loss of generality that $\delta \leq 1/216$. Note that all the rectangles are packed into vertical containers, because rectangles that have both width and height smaller than $\overline{\eps}$ are short-narrow and we removed them. Summarizing:

\begin{figure}
	\centering
	\includegraphics[width=9cm]{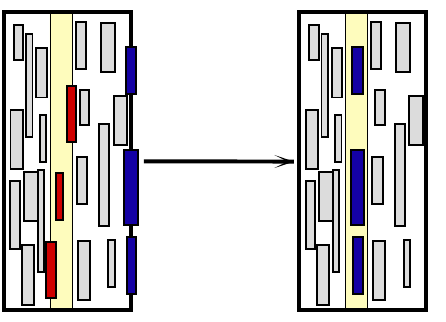}
	\caption{For each vertical box, we can remove a low profit subset of rectangles (red in the picture), to make space for short-high rectangles that cross the right edge of the box (blue).}
	\label{fig:augm-fix-vertical-box}
\end{figure}

\begin{proposition}\label{prop:basic_structure} There is a set $I^+ \subseteq I'$ of rectangles with total profit at least $(1-O(\eps_{ra}))\cdot p(I')$ and a corresponding packing for them in a $1 \times (1+2\delta)$ region such that:
	\begin{itemize}
		\item every wide rectangle in $I^+$ has its length rounded up to the nearest multiple of $\delta^2$ and it is positioned so that its left side is at a position $x$ which is a multiple of $\delta^2$, and
		\item each box $B \in \mathcal{B}$ storing at least one short-high rectangle has width $\delta^2$, and the rectangles inside are packed into at most $1/\delta^7$ vertical containers.
	\end{itemize}
\end{proposition}

\subsection{Fractional packing with \texorpdfstring{$O(1)$}{O(1)} containers}

Let us consider now the set of rectangles $I^+$ and an almost optimal packing $S^+$ for them according to Proposition~\ref{prop:basic_structure}. We remove the rectangles assigned to boxes in $\mathcal{B}$ and consider each box $B \in \mathcal{B}$ as a single pseudoitem. Thus, in the new almost optimal solution there are just pseudoitems from $\mathcal{B}$ and wide rectangles with right and left coordinates that are multiples of $\delta^2$. We will now show that we can derive a fractional packing with the same profit, and such that the rectangles and pseudoitems can be (fractionally) assigned to a constant number of containers. By \emph{fractional packing} we mean a packing where horizontal rectangles are allowed to be sliced horizontally (but not vertically); we can think of the profit as being split proportionally to the heights of the slices.

Let $\mathcal{K}$ be a subset of the horizontal rectangles of size $K$ that will be specified later. By extending horizontally the top and bottom edges of the rectangles in $\mathcal{K}$ and the pseudoitems in $\mathcal{B}$, we partition the knapsack into at most $2(|K| + |\mathcal{B}|) + 1 \leq 2(K + \frac{2}{\delta^3}) + 1 \leq 2(K + \frac{3}{\delta^3})$ horizontal stripes.

Let us focus on the (possibly sliced) rectangles contained in one such stripe of height $h$. For any vertical coordinate $y \in [0,h]$ we can define the \emph{configuration} at coordinate $y$ as the set of positions where the horizontal line at distance $y$ from the bottom cuts a vertical edge of a horizontal rectangle which is not in $\mathcal{K}$. There are at most $2^{M-1}$ possible configurations in a stripe.

\begin{figure}
	\centering
	\includegraphics[width=14.1cm]{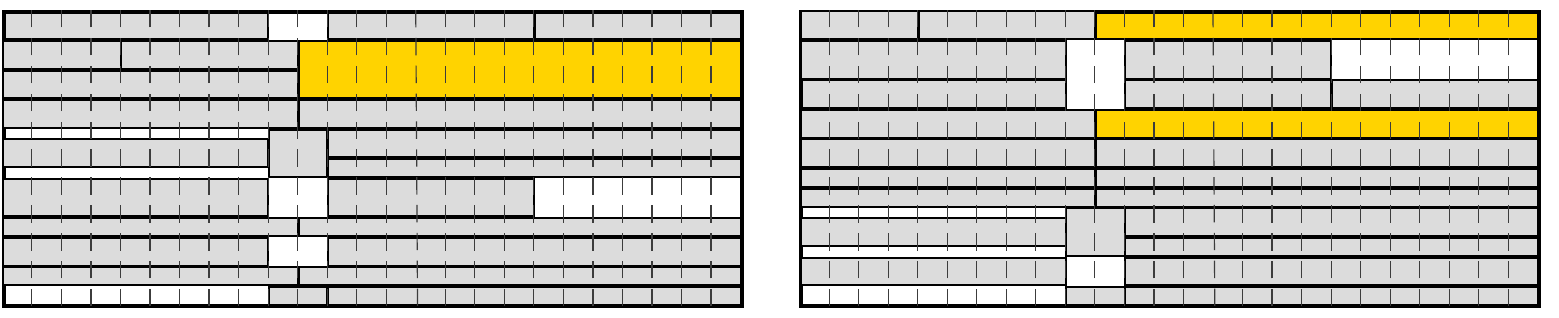}
	\caption{
		Rearranging the rectangles in a horizontal stripe. On the right, rectangles are repacked so that regions with the same configuration appear next to each other. Note that the yellow rectangle has been sliced, since it partakes in two regions with different configurations.
	}
	\label{fig:augm-fractional-repacking}
\end{figure}

We can further partition the stripe in maximal contiguous regions with the same configuration. Note that the number of such regions is not bounded, since configurations can be repeated. But since the rectangles are allowed to be sliced, we can rearrange the regions so that all the ones with the same configuration appear next to each other; see Figure~\ref{fig:augm-fractional-repacking} for an example. After this step is completed, we define up to $M$ horizontal containers per each configuration, where we repack the sliced horizontal rectangles. Clearly, all sliced rectangles are repacked.

Thus, the number of horizontal containers that we defined per each stripe is bounded by $M2^{M-1}$, and the total number overall is at most
\[
  2\left(K + \frac{3}{\delta^3}\right) M 2^{M-1} = \left(K + \frac{3}{\delta^3}\right) M 2^{M}.
\]
\subsection{Existence of an integral packing}\label{sec:integral_packing}

We will now show the existence of an integral packing, at a small loss of profit.

Consider a fractional packing in $N$ containers. Since each rectangle slice is packed in a container of exactly the same width, it is possible to pack all but at most $N$ rectangles integrally by a simple greedy algorithm: choose a container, and greedily pack in it rectangles of the same width, until either there are no rectangles left for that width, or the next rectangle does not fit in the current container. In this case, we discard this rectangle and close the container, meaning that we do not use it further. Clearly, only one rectangle per container is discarded, and no rectangle is left unpacked.

The only problem is that the total profit of the discarded rectangles can be large. To solve this problem, we use the following shifting argument. Let $\mathcal{K}_0 = \emptyset$ and $K_0 = 0$. For convenience, let us define $f(K) = \left(K + \frac{3}{\delta^3}\right) M 2^{M}$.

First, consider the fractional packing obtained by choosing $\mathcal{K} = \mathcal{K}_0$, so that $K = K_0 = 0$. Let $\mathcal{K}_1$ be the set of discarded rectangles obtained by the greedy algorithm, and let $K_1 = |\mathcal{K}_1|$. Clearly, by the above reasoning, the number of discarded rectangles is bounded by $f(K_0)$. If the profit $p(\mathcal{K}_1)$ of the discarded rectangles is at most $\eps_{ra} p(OPT)$, then we remove them and there is nothing else to prove. Otherwise, consider the fractional packing obtained by fixing $\mathcal{K} = \mathcal{K}_0 \cup \mathcal{K}_1$. Again, we will obtain a set $\mathcal{K}_2$ of discarded rectangles such that $K_2 := |\mathcal{K}_2| \leq f(K_0 + K_1)$. Since the sets $\mathcal{K}_1, \mathcal{K}_2, \dots$ that we obtain are all disjoint, the process must stop after at most $1/\eps_{ra}$ iterations. Setting $p := M2^{M}$ and $q := \frac{3}{\delta^3} M2^{M}$, we have that $K_{i+1} \leq p(K_0 + K_1 + \dots K_i) + q$ for each $i \geq 0$. Crudely bounding it as $K_{i+1} \leq i \cdot pq \cdot K_i$, we immediately obtain that $K_i \leq (pq)^i$. Thus, in the successful iteration, the size of $\mathcal{K}$ is at most $K_{1/\eps_{ra}-1}$ and the number of containers is at most $K_{1/\eps_{ra}} \leq (pq)^{1/\eps_{ra}} = (\frac{3}{\delta^2}M^2 2^{2M})^{1/\eps_{ra}} = O_{\eps_{ra}, \delta}(1)$.

\subsection{Rounding down horizontal and vertical containers}\label{sec:shrinking_horiz_vert}

As per the above analysis, the total number of horizontal containers is at most  $(\frac{3}{\delta^2}M^2 2^{2M})^{\eps_{ra}}$ and the total number of vertical containers is at most $\frac{2}{\delta^3} \cdot \frac{1}{\delta^7} = \frac{2}{\delta^{10}}$.

We will now show that, at a small loss of profit, it is possible to replace each horizontal and each vertical container defined so far with a constant number of smaller containers, so that the total area of the new containers is at most as big as the total area of the rectangles originally packed in the container. Note that in each container we consider the rectangles with the original widths (not rounded up). We use the following lemma:

\begin{lemma}\label{lem:shrink_knapsack_container} Let $C$ be a horizontal (resp. vertical) container defined above, and let $\R_C$ be the set of rectangles packed in $C$. Then, it is possible to pack a set $\R'_C \subseteq \R_C$ of profit at least $(1 - 3\eps_{ra})p(\R_C)$ in a set of at most $\left\lceil\log_{1 + \eps_{ra}} (\frac{1}{\delta})\right\rceil / \eps_{ra}^2$ horizontal (resp. vertical) containers that can be packed inside $C$ and such that their total area is at most $a(\R_C)$. 
\end{lemma}
\begin{proof}
	Without loss of generality, we prove the result only for the case of a horizontal container.
	
	Since $w_i \geq \delta$ for each rectangle $R_i \in \R_C$, we can partition the rectangles in $\R_C$ into at most $\left\lceil\log_{1 + \eps_{ra}} (\frac{1}{\delta})\right\rceil$ groups $\R_1, \R_2, \dots$, so that in each $\R_j$ the widest rectangle has width bigger than the smallest by a factor at most $1+\eps_{ra}$; we can then define a container $C_j$ for each group $\R_j$ that has the width of the widest rectangle it contains and height equal to the sum of the heights of the contained rectangles.
	
	Consider now one such $C_j$ and the set of rectangles $\R_j$ that it contains, and let $P := p(\R_j)$. Clearly, $w(C_j) \leq (1+\eps_{ra}) w_i$ for each $R_i \in \R_j$, and so $a(C_j) \leq (1+\eps_{ra})a(\R_j)$. If all the rectangles in $\R_j$ have height at most $\eps_{ra} h(C_j)$, then we can remove a set of rectangles with total height at least $\eps_{ra} h(C)$ and profit at most $2\eps_{ra} p(\R_j)$. Otherwise, let $\mathcal{Q}$ be the set of rectangles of height larger than $\eps_{ra} h(C_j)$, and note that $a(Q) \geq \eps_{ra} h(C_j)w(C_j)/(1 + \eps_{ra})$. If $p(\mathcal{Q}) \leq \eps_{ra} P$, then we remove the rectangles in $\mathcal{Q}$ from the container $C_j$ and reduce its height as much as possible, obtaining a smaller container $C'_j$; since $a(C'_j) \leq a(C_j) - \eps_{ra} a(C_j) = (1 - \eps_{ra})a(C_j) \leq (1 - \eps_{ra})(1 + \eps_{ra})a(\R_j) < a(\R_j)$, then the proof is finished. Otherwise, we define one container for each of the rectangles in $\mathcal{Q}$ (which are at most $1/\eps_{ra}$) of exactly the same size, and we still shrink the container with the remaining rectangles as before; note that there is no lost area for each of the newly defined container. Since at every non-terminating iteration a set of rectangles with profit larger than $\eps_{ra} P$ is removed, the process must end within $1/\eps_{ra}$ iterations.
	
	Note that the total number of containers that we produce for each initial container $C_j$ is at most $1/\eps_{ra}^2$, and this concludes the proof.
\end{proof}

Thus, by applying the above lemma to each horizontal and each vertical container, we obtain a modified packing where the total area of the horizontal and vertical containers is at most the area of the rectangles of $\R'$ (without the short-narrow rectangles, which we will take into account in the next section), while the number of containers increases at most by a factor $\left\lceil\log_{1 + \eps_{ra}} (\frac{1}{\delta})\right\rceil / \eps_{ra}^2$.

\subsection{Packing short-narrow rectangles}\label{sec:shortnarrow}

Consider the integral packing obtained from the previous section, which has at most $K' := \left(\frac{2}{\delta^{10}} + (\frac{3}{\delta^2}M^2 2^{2M})^{\eps_{ra}}\right)\left\lceil\log_{1 + \eps_{ra}} (\frac{1}{\delta})\right\rceil / \eps_{ra}^2$ containers. We can create a non-uniform grid extending each side of the containers until they hit another container or the boundary of the knapsack. Moreover, we also add horizontal and vertical lines spaced at distance $\eps_{ra}$. We call \emph{free cell} each face defined by the above lines that does not overlap a container of the packing; by construction, no free cell has a side bigger than $\eps_{ra}$. The number of free cells in this grid plus the existing containers is bounded by $K_{TOTAL} = {(2K' + 1/\eps_{ra})}^2 = O_{\eps_{ra}, \delta}(1)$. We crucially use the fact that this number does not depend on the value of $\mu$.

Note that the total area of the free cells is no less than the total area of the short-narrow rectangles, as a consequence of the guarantees on the area of the containers introduced so far. We will pack the short-narrow rectangles into the free cells of this grid using NFDH, but we only use cells that have width and height at least $\frac{8\mu}{\eps_{ra}}$; thus, each short-narrow rectangle will be assigned to a cell whose width (resp. height) is larger by at least a factor $8/\eps_{ra}$ than the width (resp. height) of the rectangle. Each discarded cell has area at most $\frac{8\mu}{\eps_{ra}}$, which implies that the total area of discarded cells is at most $\frac{8 \mu K_{TOTAL}}{\eps_{ra}}$. Now we consider the selected cells in an arbitrary order and pack short narrow rectangles into them using NFDH, defining a new area container for each cell that is used. Thanks to Lemma~\ref{lem:nfdhPack}, we know that each new container $C$ (except maybe the last one) that is used by NFDH contains rectangles for a total area of at least $(1 - \eps_{ra}/4)a(C)$. Thus, if all rectangles are packed, we remove the last container opened by NFDH, and we call $S$ the set of rectangles inside, that we will repack elsewhere; note that $a(S) \leq \eps_{ra}^2 \leq \eps_{ra}/3$, since all the rectangles in $S$ were packed in a free cell. Instead, if not all rectangles are packed by NFDH, let $S$ be the residual rectangles. In this case, the area of the unpacked rectangles is $a(S) \leq \frac{8 \mu K_{TOTAL}}{\eps_{ra}} + \eps_{ra}/4 \leq \eps_{ra}/3$, assuming that $\mu \leq \frac{\eps_{ra}^2}{96 K_{TOTAL}}$.

In order to repack the rectangles of $S$, we define a new area container $C_S$ of height $1$ and width $\eps_{ra}/2$. Since $a(C_S) = \eps_{ra}/2 \geq (\eps_{ra}/3) / (1 - 2\eps_{ra})$, all elements from $S$ are packed in $C_S$ by NFDH, and the container can be added to the knapsack by further enlarging its width from $1 + 2\delta$ to $1 + 2\delta + \eps_{ra}/2 < 1 + \eps_{ra}$. 

The last required step is to guarantee the necessary constraint on the total area of the area containers, similarly to what was done in Section~\ref{sec:shrinking_horiz_vert} for the horizontal and vertical containers.

Let $D$ be any full area container (that is, any area container except for $C_S$). We know that the area of the rectangles $R_D$ in $D$ is $a(R_D) \geq (1 - \eps_{ra})a(D)$, since each rectangle $R_i$ inside $D$ has width less than $\eps_{ra} w(D)/2$ and height less than $\eps_{ra} h(D)/2$, by construction. We remove rectangles from $R_D$ in non-decreasing order of profit/area ratio, until the total area of the residual rectangles is between $(1 - 4\eps_{ra})a(D)$ and $(1 - 3\eps_{ra})a(D)$ (this is possible, since each element has area at most $\eps_{ra}^2 a(D)$); let $R'_D$ be the resulting set. We have that $p(R'_D) \geq (1 - 4\eps_{ra}) p(R_D)$, due to the greedy choice. Let us define a container $D'$ of width $w(D)$ and height $(1 - \eps_{ra})h(D)$. It is easy to verify that each rectangle in $R_D$ has width (resp. height) at most $\eps_{ra} w(D')$ (resp. $\eps_{ra} h(D')$). Moreover, since $a(R'_D) \leq (1 - 3\eps_{ra})a(D) \leq (1 - 2\eps_{ra})(1 - \eps_{ra})a(C) \leq (1 - 2\eps_{ra})a(C')$, then all elements in $R'_D$ are packed in $D'$. By applying this reasoning to each area container (except $C_S$), we obtain property (3) of Lemma~\ref{lem:structural_lemma_augm}.

Note that the constraints on $\mu$ and $\delta$ that we imposed are $\mu \leq \frac{\delta^{10}}{76}$ (from Section~\ref{sec:containers-short-rectangles}), and $\mu \leq \frac{\eps_{ra}^2}{96 K_{TOTAL}}$. It is easy to check that both of them are satisfied if we choose $\sal{f(x) = (\eps_{ra}x)^C}$ for a big enough constant $C$ that depends only on $\eps_{ra}$. This concludes the proof.

\chapter{A PPT \texorpdfstring{$(4/3 + \eps)$}{(4/3 + eps)}-approximation for~Strip Packing}\label{chap:StripPacking}


In this chapter, we present a $(\frac43+\eps)$-approximation for Strip Packing running in pseudo-polynomial time. These results are published in \cite*{ggik16}.

Our approach refines the technique of \cite{nw16}. Like many results in the field, it is based on finding some particular type of packings that have a special structure, and that can consequently be approximated more efficiently.

Let $\alpha\in [1/3,1/2)$ be a proper constant parameter, and define a rectangle $R_i$ to be \emph{tall} if $h_i> \alpha\cdot OPT$, where $OPT$ is the height of an optimal packing. They prove that the optimal packing can be structured into a constant number of axis-aligned rectangular regions (\emph{boxes}), that occupy a total height of $OPT' \leq (1+\eps) OPT$ inside the vertical strip. Some rectangles are not fully contained into one box (they are \emph{cut} by some box). Among them, tall rectangles remain in their original position. All the other cut rectangles are repacked on top of the boxes: part of them in a horizontal box of size $W\times O(\eps)OPT$, and the remaining ones in a vertical box of size $O(\eps W)\times \alpha\,OPT$ (that we next imagine as placed on the top-left of the packing under construction).

Some of these boxes contain only relatively high rectangles (including tall ones) of relatively small width.
The next step is a rearrangement of the rectangles inside one such \emph{vertical} box $\overline{B}$ (see Figure \ref{fig_pseudo-rectangles1}), say of size $\overline{w} \times \overline{h}$: they first slice non-tall rectangles into unit width rectangles
(this slicing can be finally avoided with standard techniques). Then they shift tall rectangles to the top/bottom of $\overline{B}$, shifting sliced rectangles consequently (see Figure \ref{fig_pseudo-rectangles2}). Now they discard all the (sliced) rectangles completely contained in a central horizontal region of size $\overline{w}\times (1+\eps-2\alpha)\overline{h}$, and they \emph{nicely rearrange} the remaining rectangles into a constant number of \emph{sub-boxes} (excluding possibly a few more non-tall rectangles, that can be placed in the additional vertical box).

These discarded rectangles can be packed into $2$ extra boxes of size $\frac{\overline{w}}{2}\times(1+\eps-2\alpha)\overline{h}$ (see Figure 
\ref{fig_pseudo-rectangles4}).
In turn, the latter boxes can be packed into two \emph{discarded} boxes of size $\frac{W}{2}\times (1+\eps-2\alpha)OPT'$, that we can imagine as placed, one on top of the other, on the top-right of the packing. See Figure \ref{fig_packing_NW} for an illustration of the final packing. This leads to a total height of $(1+\max\{\alpha,2(1-2\alpha)\}+O(\eps))\cdot OPT$, which is minimized by choosing $\alpha=\frac{2}{5}$. 

\begin{figure}
	\captionsetup[subfigure]{justification=centering}
	\hspace{-20pt}
	\begin{subfigure}[b]{.53\textwidth}
		\resizebox{6cm}{!}{
			\begin{tikzpicture}
			
			
			\draw[fill = gray] (0,0) rectangle (5.3,0.4);
			\draw[fill = gray] (0,0.4) rectangle (5.2,0.7);
			\draw[fill = gray] (0,0.7) rectangle (4.7,1);
			
			\draw[fill = gray] (0,1) rectangle (4.2,1.8);
			\draw[fill = gray] (0,1.8) rectangle (3.8,2.2);
			\draw[fill = gray] (0,2.2) rectangle (3.7,2.5);
			
			\draw[fill = gray] (0,2.5) rectangle (1,7);
			\draw[fill = gray] (1,2.5) rectangle (1.8,7);
			\draw[fill = gray] (1.8,2.5) rectangle (2.6,7);
			
			\fill[fill = gray] (2.6,2.5) rectangle (5.1,5.5);
			
			\draw[fill=gray] (5.5,0) rectangle (6.3,6);
			\draw[fill=gray] (6.3,0) rectangle (7,6);
			
			\fill[fill=gray] (3,6) rectangle (7,7);
			
			\draw[fill=gray] (4,9.25) rectangle (4.5,10.25);
			\draw[fill=gray] (4.5,9.25) rectangle (5,10.25);
			
			\draw[fill=gray] (6,9) rectangle (6.3,11);
			\draw[fill=gray] (6.3,9) rectangle (6.5,11);
			
			\draw[fill=gray] (5.5,12) rectangle (5.9,13.5);
			\draw[fill=gray] (5.9,12) rectangle (6.2,13.5);
			\draw[fill=gray] (6.2,12) rectangle (6.5,13.5);
			
			
			\draw[ultra thick] (0,0) rectangle (7,7);
			
			\draw[ultra thick] (0,0) rectangle (5.5,1);
			\draw[ultra thick] (0,1) rectangle (4.5,2.5);
			\draw[ultra thick] (0,2.5) rectangle (2.6,7);
			\draw[ultra thick] (2.6,2.5) rectangle (5.1,5.5);
			\draw[ultra thick] (3,6) rectangle (7,7);
			\draw[ultra thick] (5.5,0) rectangle (7,6);
			\draw[ultra thick] (3,6) rectangle (7,7);
			
			\draw[ultra thick] (0,7) rectangle (7,7.5);
			\draw[ultra thick] (0,7.5) rectangle (7,8);
			\draw[ultra thick] (0,8) rectangle (7,8.5);
			\draw[ultra thick] (0,8.5) rectangle (7,9);
			
			\draw[ultra thick] (0,9) rectangle (0.5,14);
			\draw[ultra thick] (0.5,9) rectangle (1,14);
			\draw[ultra thick] (1,9) rectangle (1.5,14);
			
			\draw[ultra thick] (3.5,9) rectangle (7,14);
			
			\draw[ultra thick] (4,9.25) rectangle (5,10.25);
			\draw[ultra thick] (6,9) rectangle (6.5,11);
			
			\draw[ultra thick] (5.5,12) rectangle (6.5,13.5);
			
			
			\draw (3.6,13.35) node[anchor=west] {\large \textbf{$B_{disc, top}$}};
			\draw (3.6,10.85) node[anchor=west] {\large \textbf{$B_{disc, bot}$}};
			\draw (0.25,11.5) node [rotate=90] {\large \textbf{$B_{M, ver}$}};
			\draw (0.75,11.5) node [rotate=90] {\large \textbf{$B_{V, cut}$}};
			\draw (1.25,11.5) node [rotate=90] {\large \textbf{$B_{V, round}$}};
			\draw (3.5,8.75) node {\large \textbf{$B_S$}};
			\draw (3.5,8.25) node {\large \textbf{$B_{H, cut}$}};
			\draw (3.5,7.75) node {\large \textbf{$B_{H, round}$}};
			\draw (3.5,7.25) node {\large \textbf{$B_{M, hor}$}};
			
			\draw[dashed] (3.5,11.5) -- (7,11.5);
			
			\draw[thick] (0,0) -- (-0.5,0);
			\draw (-0.6,0) node[anchor = east] {\large $0$}; 
			
			\draw[thick] (0,0) -- (0,-0.5);
			\draw (0,-1) node {\large $0$}; 
			
			\draw[thick] (7,0) -- (7,-0.5);
			\draw (7,-1) node {\large $W$}; 
			
			\draw[thick] (0,6.6) -- (-0.5,6.6);
			\draw (-0.6,6.6) node[anchor = east] {\large $OPT$}; 
			
			\draw[thick] (0,7) -- (-0.5,7);
			\draw (-0.5,7) node[anchor = east] {\large $OPT'$}; 
			
			
			\draw [thick,decorate,decoration={brace,amplitude=8pt}] 
			(0,7) -- (0,9); 
			\draw (-0.4,8) node [anchor = east] {\large $O(\eps) \cdot OPT'$}; 
			
			\draw [thick,decorate,decoration={brace,amplitude=8pt}] 
			(3.5,9) -- (3.5,14); 
			\draw (3.1,11.5) node [anchor = south, rotate=90] {\large $2(1+\eps-2\alpha) \cdot OPT'$}; 
			
			\draw [thick,decorate,decoration={brace,amplitude=8pt}] 
			(0,14) -- (1.5,14); 
			\draw (0.75,14.3) node [anchor = south] {\large $O(\eps) \cdot W$}; 
			
			\draw [thick,decorate,decoration={brace,amplitude=8pt}] 
			(3.5,14) -- (7,14); 
			\draw (5.25,14.3) node [anchor = south] {\large $0.5 \cdot W$};
			
			\draw [thick,decorate,decoration={brace,amplitude=8pt}] 
			(0,9) -- (0,14); 
			\draw (-0.5,11.5) node [anchor = east] {\large $\alpha \cdot OPT'$}; 
			
			\end{tikzpicture}}
		\caption{Final packing obtained by \\ \cite{nw16}.}
		\label{fig_packing_NW}
	\end{subfigure}%
	\hspace{-18pt}
	\begin{subfigure}[b]{.6\textwidth}
		\resizebox{8.3cm}{!}{
			\begin{tikzpicture}
			
			
			\draw[fill = gray] (0,0) rectangle (5.3,0.4);
			\draw[fill = gray] (0,0.4) rectangle (5.2,0.7);
			\draw[fill = gray] (0,0.7) rectangle (4.7,1);
			
			\draw[fill = gray] (0,1) rectangle (4.2,1.8);
			\draw[fill = gray] (0,1.8) rectangle (3.8,2.2);
			\draw[fill = gray] (0,2.2) rectangle (3.7,2.5);
			
			\draw[fill = gray] (0,2.5) rectangle (1,7);
			\draw[fill = gray] (1,2.5) rectangle (1.8,7);
			\draw[fill = gray] (1.8,2.5) rectangle (2.6,7);
			
			\fill[fill = gray] (2.6,2.5) rectangle (5.1,5.5);
			
			\draw[fill=gray] (5.5,0) rectangle (6.3,6);
			\draw[fill=gray] (6.3,0) rectangle (7,6);
			
			\fill[fill=gray] (3,6) rectangle (7,7);
			
			\draw[fill=gray] (2,9.25) rectangle (2.5,10.25);
			\draw[fill=gray] (2.5,9.25) rectangle (3,10.25);
			
			\draw[fill=gray] (6,9) rectangle (6.3,11);
			\draw[fill=gray] (6.3,9) rectangle (6.5,11);
			
			\draw[fill=gray] (4,9) rectangle (4.4,10.5);
			\draw[fill=gray] (4.4,9) rectangle (4.7,10.5);
			\draw[fill=gray] (4.7,9) rectangle (5,10.5);
			
			
			\draw[ultra thick] (0,0) rectangle (7,7);
			
			\draw[ultra thick] (0,0) rectangle (5.5,1);
			\draw[ultra thick] (0,1) rectangle (4.5,2.5);
			\draw[ultra thick] (0,2.5) rectangle (2.6,7);
			\draw[ultra thick] (2.6,2.5) rectangle (5.1,5.5);
			\draw[ultra thick] (3,6) rectangle (7,7);
			\draw[ultra thick] (5.5,0) rectangle (7,6);
			\draw[ultra thick] (3,6) rectangle (7,7);
			
			\draw[ultra thick] (0,7) rectangle (7,7.5);
			\draw[ultra thick] (0,7.5) rectangle (7,8);
			\draw[ultra thick] (0,8) rectangle (7,8.5);
			\draw[ultra thick] (0,8.5) rectangle (7,9);
			
			\draw[ultra thick] (0,9) rectangle (0.5,12);
			\draw[ultra thick] (0.5,9) rectangle (1,12);
			\draw[ultra thick] (1,9) rectangle (1.5,12);
			
			\draw[ultra thick] (1.5,9) rectangle (7,12);
			
			\draw[ultra thick] (2,9.25) rectangle (3,10.25);
			\draw[ultra thick] (4,9) rectangle (5,10.5);
			\draw[ultra thick] (6,9) rectangle (6.5,11);
			
			\draw (1.8,11.35) node[anchor=west] {\large \textbf{$B_{disc}$}};
			\draw (0.25,10.5) node [rotate=90] {\large \textbf{$B_{M, ver}$}};
			\draw (0.75,10.5) node [rotate=90] {\large \textbf{$B_{V, cut}$}};
			\draw (1.25,10.5) node [rotate=90] {\large \textbf{$B_{V, round}$}};
			\draw (3.5,8.75) node {\large \textbf{$B_S$}};
			\draw (3.5,8.25) node {\large \textbf{$B_{H, cut}$}};
			\draw (3.5,7.75) node {\large \textbf{$B_{H, round}$}};
			\draw (3.5,7.25) node {\large \textbf{$B_{M, hor}$}};
			
			\draw[thick] (0,0) -- (-0.5,0);
			\draw (-0.6,0) node[anchor = east] {\large $0$}; 
			
			\draw[thick] (0,0) -- (0,-0.5);
			\draw (0,-1) node {\large $0$}; 
			
			\draw[thick] (7,0) -- (7,-0.5);
			\draw (7,-1) node {\large $W$}; 
			
			\draw[thick] (0,6.6) -- (-0.5,6.6);
			\draw (-0.6,6.6) node[anchor = east] {\large $OPT$}; 
			
			\draw[thick] (0,7) -- (-0.5,7);
			\draw (-0.5,7) node[anchor = east] {\large $OPT'$}; 
			
			
			\draw [thick,decorate,decoration={brace,amplitude=8pt}] 
			(0,7) -- (0,9); 
			\draw (-0.4,8) node [anchor = east] {\large $O(\eps) \cdot OPT'$}; 
			
			\draw [thick,decorate,decoration={brace,amplitude=8pt}] 
			(0,9) -- (0,12); 
			\draw (-0.5,10.5) node [anchor = east] {\large $\alpha \cdot OPT'$}; 
			
			\draw [thick,decorate,decoration={brace,amplitude=8pt}] 
			(0,12) -- (1.5,12); 
			\draw (0.75,12.3) node [anchor = south] {\large $\gamma \cdot W$}; 
			
			\draw [thick,decorate,decoration={brace,amplitude=8pt}] 
			(1.5,12) -- (7,12); 
			\draw (4.75,12.3) node [anchor = south] {\large $(1-\gamma) \cdot W$};
			
			\draw [thick,decorate,decoration={brace,amplitude=8pt}] 
			(7,12) -- (7,9); 
			\draw (7.5,10.5) node [anchor = west] {\large $(1+\eps - 2\alpha) \cdot OPT'$}; 
			
			\end{tikzpicture}}
		\caption{Final packing obtained in this work. \\Here $\gamma$ is a small constant depending on $\eps$.}
		\label{fig_our_packing}
	\end{subfigure}
	\caption{Comparison of final solutions.}
	\label{fig_structure}
\end{figure}
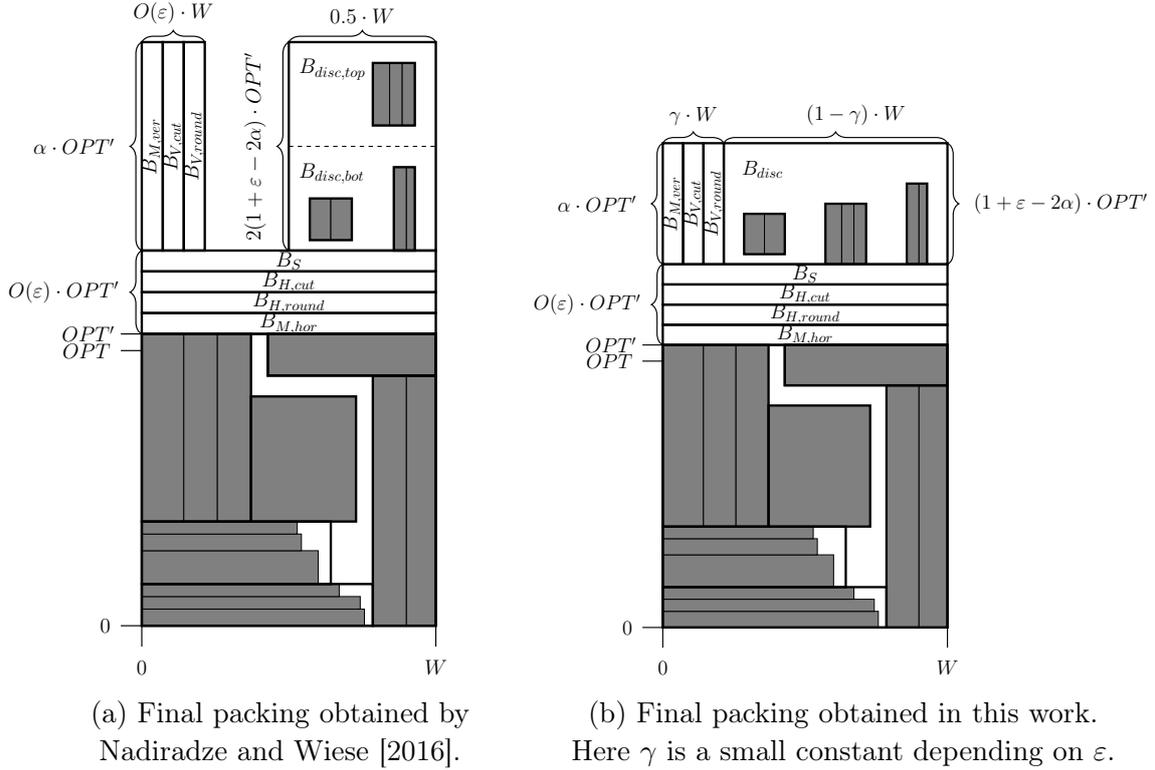

Our main technical contribution is a repacking lemma that allows one to repack a small fraction of the discarded rectangles of a given box inside the free space left by the corresponding sub-boxes (while still having $O_{\eps}(1)$ many sub-boxes in total). This is illustrated in Figure \ref{fig_pseudo-rectangles5}. This way we can pack all the discarded rectangles into a \emph{single} discarded box of size $(1-\gamma)W\times (1+\eps-2\alpha)OPT'$, where $\gamma$ is a small constant depending on $\eps$, that we can place on the top-right of the packing. The vertical box where the remaining rectangles are packed still fits to the top-left of the packing, next to the discarded box. See Figure \ref{fig_our_packing} for an illustration. Choosing $\alpha=1/3$ gives the claimed approximation factor. 

We remark that the basic approach by Nadiradze and Wiese strictly requires that at most $2$ tall rectangles can be packed one on top of the other in the optimal packing, hence imposing $\alpha\geq 1/3$. Thus in some sense we pushed their approach to its limit.

The algorithm by \cite{nw16} is not directly applicable to the case when $90^\circ$ rotations are allowed. 
In particular, they use a linear program to pack some rectangles. When rotations are allowed, it is unclear how to decide which rectangles are packed by the linear program.
We use a combinatorial \emph{container}-based approach to circumvent this limitation, which allows us to pack all the rectangles using dynamic programming. In this way we obtain a PPT $(4/3+\eps)$-approximation for strip packing with rotations, breaking the polynomial-time approximation barrier of 3/2 for that variant as well.

\section{Preliminaries and notations}\label{sec:prelim}

We will follow the notation from \cite{nw16}, which will be explained as it is needed.

Recall that $OPT\in \mathbb{N}$ denotes the height of the optimal packing for instance $\R$. By trying all the pseudo-polynomially many possibilities, we can assume that $OPT$ is known to the algorithm. Given a set $\mathcal{M}\subseteq\R$ of rectangles, $a(\mathcal{M})$ will denote the total area of rectangles in $\mathcal{M}$, i.e., $a(\mathcal{M}) = \sum_{R_i\in \mathcal{M}}{h_i \cdot w_i}$, and $h_{\max}(\mathcal{M})$ (resp. $w_{\max}(\mathcal{M})$) denotes the maximum height (resp. width) of rectangles in $\mathcal{M}$. Throughout this work, a \emph{box} of size $a\times b$ means an axis-aligned rectangular region of width $a$ and height $b$. 



\subsection{Classification of rectangles}

Let $0 < \eps < \alpha$, and assume for simplicity that $\frac{1}{\eps} \in \mathbb{N}$. We first classify the input rectangles into six  groups according to parameters $\delta_h, \delta_w, \mu_h, \mu_w$ satisfying $\eps\geq \delta_h > \mu_h > 0$ and $\eps\geq \delta_w > \mu_w > 0$, whose values will be chosen later (see also Figure~\ref{fig:SP_rectangles_classification}). A rectangle $R_i$ is:

\begin{figure}
	\centering
	\resizebox{!}{8cm}{\begin{tikzpicture}
		
		\fill[pattern color = lightgray, pattern = north east lines] (0,5.25) rectangle (3,8);
		\fill[pattern color = lightgray, pattern = north west lines] (3,3.5) rectangle (6,8);
		\fill[pattern color = lightgray, pattern = north west lines] (0,3.5) rectangle (1.5,5.25);
		\fill[color = lightgray] (1.5,0) rectangle (3,1.75);
		\fill[color = lightgray] (0,1.75) rectangle (6,3.5);
		\fill[color = lightgray] (1.5,3.5) rectangle (3,5.25);
		\fill[pattern color = lightgray, pattern = north east lines] (0,0) rectangle (1.5,1.75);
		\fill[pattern color = lightgray, pattern = north east lines] (3,0) rectangle (6,1.75);
		
		
		\draw[dashed] (1.5,0) -- (1.5,8.5);
		\draw[dashed] (3,0) -- (3,8.5);
		\draw[dashed] (6,0) -- (6,8.5);
		
		\draw[dashed] (0,1.75) -- (6.5,1.75);
		\draw[dashed] (0,3.5) -- (6.5,3.5);
		\draw[dashed] (0,5.25) -- (6.5,5.25);
		\draw[dashed] (0,8) -- (6.5,8);
		
		
		\draw[->] (0,0) -- (6.5,0);
		\draw[->] (0,0) -- (0,8.5);
		
		
		\draw[ultra thick] (0,0) rectangle (1.5,1.75);
		\draw[ultra thick] (1.5,0) -- (1.5,1.75) -- (0,1.75) -- (0,3.5) -- (1.5,3.5) -- (1.5,5.25) -- (3,5.25) -- (3,3.5) -- (6,3.5) -- (6,1.75) -- (3,1.75) -- (3,0) -- (1.5,0);
		\draw[ultra thick] (3,0) rectangle (6,1.75);
		\draw[ultra thick] (0,3.5) rectangle (1.5,5.25);
		\draw[ultra thick] (0,5.25) rectangle (3,8);
		\draw[ultra thick] (3,3.5) rectangle (6,8);
		
		
		\draw (0,0) -- (0,-0.25);
		\draw (0,-0.25) node[anchor=north] {$0$};
		\draw (1.5,0) -- (1.5,-0.25);
		\draw (1.5,-0.25) node[anchor=north] {$\mu_w W$};
		\draw (3,0) -- (3,-0.25);
		\draw (3,-0.25) node[anchor=north] {$\delta_w W$};
		\draw (6,0) -- (6,-0.25);
		\draw (6,-0.25) node[anchor=north] {$W$};
		
		\draw (0,0) -- (-0.25,0);
		\draw (-0.25,0) node[anchor=east] {$0$};
		\draw (0,1.75) -- (-0.25,1.75);
		\draw (-0.25,1.75) node[anchor=east] {$\mu_h OPT$};
		\draw (0,3.5) -- (-0.25,3.5);
		\draw (-0.25,3.5) node[anchor=east] {$\delta_h OPT$};
		\draw (0,5.25) -- (-0.25,5.25);
		\draw (-0.25,5.25) node[anchor=east] {$\alpha OPT$};
		\draw (0,8) -- (-0.25,8);
		\draw (-0.25,8) node[anchor=east] {$OPT$};
		
		
		\draw (0.75,0.875) node {\textbf{\scriptsize Small}};
		\draw (0.75,2.625) node {\textbf{\scriptsize Medium}};
		\draw (0.75,4.375) node {\textbf{\scriptsize Vertical}};
		\draw (0.75,6.625) node {\textbf{\scriptsize Tall}};
		
		\draw (2.25,0.875) node {\textbf{\scriptsize Medium}};
		\draw (2.25,2.625) node {\textbf{\scriptsize Medium}};
		\draw (2.25,4.375) node {\textbf{\scriptsize Medium}};
		\draw (2.25,6.625) node {\textbf{\scriptsize Tall}};
		
		\draw (4.5,0.875) node {\textbf{\scriptsize Horizontal}};
		\draw (4.5,2.625) node {\textbf{\scriptsize Medium}};
		\draw (4.5,4.375) node {\textbf{\scriptsize Large}};
		\draw (4.5,6.625) node {\textbf{\scriptsize Large}};
		
		\end{tikzpicture}}
	\caption{Classification of the rectangles. Each rectangle is represented as a point whose $x$ (resp., $y$) coordinate indicates its width (resp., height).}
	\label{fig:SP_rectangles_classification}
\end{figure}
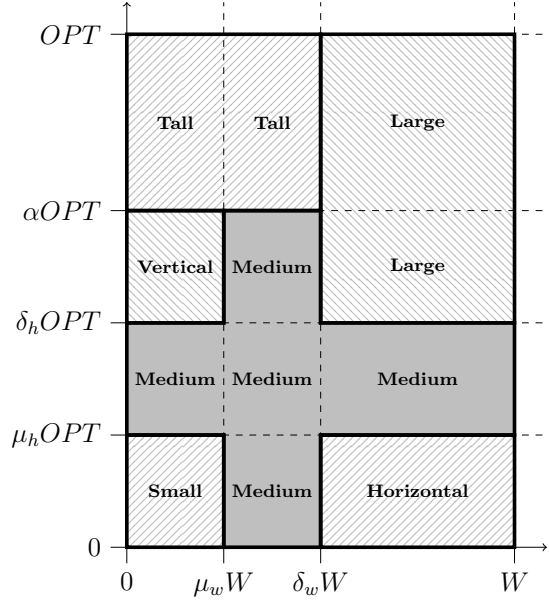

\begin{itemize}
	\item \emph{Large} if $h_i \ge \delta_h OPT$ and $w_i \ge \delta_w W$.
	\item \emph{Tall} if $h_i > \alpha OPT$ and $w_i < \delta_w W$.
	\item \emph{Vertical} if $h_i \in [\delta_h OPT, \alpha OPT]$ and $w_i \le \mu_w W$,
	\item \emph{Horizontal} if $h_i \le \mu_h OPT$ and $w_i \ge \delta_w W$,
	\item \emph{Small} if $h_i \le \mu_h OPT$ and $w_i \le \mu_w W$;
	\item \emph{Medium} in all the remaining cases, i.e., if $h_i \in (\mu_h OPT, \delta_h OPT)$, or $w_i \in (\mu_w W, \delta_w W)$ and $h_i \le \alpha OPT$.
\end{itemize}
We use $L$, $T$, $V$, $H$, $S$, and $M$ to denote the sets of large, tall, vertical, horizontal, small, and medium rectangles, respectively.
We remark that, differently from \cite{nw16}, we need to allow $\delta_h\neq \delta_w$ and $\mu_h\neq\mu_w$ due to some additional constraints in our construction.

Notice that according to this classification, every vertical line across the optimal packing intersects at most two tall rectangles. The following lemma allows us to choose $\delta_h, \delta_w, \mu_h$ and $\mu_w$ in such a way that $\delta_h$ and $\mu_h$ ($\delta_w$ and $\mu_w$, respectively) differ by a large factor, and medium rectangles have small total area.

\begin{lemma}\label{lem:mediumrectanglesarea}
	Given a polynomial-time computable function $f : (0, 1) \rightarrow (0, 1)$, with $f(x) < x$, any constant $\eps\in (0,1)$, and any positive integer $k$, we can compute in polynomial time a set $\Delta$ of $T=2(\frac{1}{\eps})^k$ many positive real numbers upper bounded by $\eps$, such that there is at least one number $\delta_h \in \Delta$ so that 
	$a(M)\leq \eps^k \cdot OPT \cdot W$ by choosing
	$\mu_h = f(\delta_h)$, $\mu_w=\frac{\eps \mu_h}{12}$, and $\delta_w=\frac{\eps \delta_h}{12}$.
	
\end{lemma}
\begin{proof} Let $T = 2(\frac{1}{\eps})^k$. Let $y_1 = \eps$, and, for each $j \in [T]$, define $y_{j+1} = f(y_j)$. Let $x_j = \frac{\eps y_j}{12}$. For each $j \leq T$, let $W_j = \{R_i \in \R : w_i \in [x_{i+1}, x_i)\}$ and similarly $H_j = \{R_i \in \R : h_i \in [y_{i+1}, y_i)\}$. Observe that $W_{j'}$ is disjoint from $W_{j''}$ (resp. $H_{j'}$ is disjoint from $H_{j''}$) for every $j' \neq j''$, and the total area of rectangles in $\bigcup W_i$ ($\bigcup H_i$ respectively) is at most $W \cdot OPT$. Thus, there exists a value $\overline{j}$ such that the total area of the elements in $W_{\overline{j}} \cup H_{\overline{j}}$ is at most $\dfrac{2 OPT\cdot W}{T} = \eps^k \cdot OPT\cdot W$. Choosing $\delta_h = y_{\overline{j}}$, $\mu_h = y_{\overline{j}+1}$, $\delta_w = x_{\overline{j}}$, $\mu_w = x_{\overline{j}+1}$ verifies all the conditions of the lemma.
\end{proof}

The function $f$ and the constant $k$ will be chosen later. From now on, assume that $\delta_h, \delta_w, \mu_h$ and $\mu_w$ are chosen according to Lemma~\ref{lem:mediumrectanglesarea}.  

\subsection{Overview of the construction}

We next overview some of the basic results in \cite{nw16} that are needed in our result. We define the constant $\gamma := \frac{\eps \delta_h}{2}$, and we assume without loss of generality that $\gamma\cdot OPT \in \mathbb{N}$.

Let us forget for a moment about the small rectangles in $S$. We will pack all the remaining rectangles $L\cup H\cup T\cup V\cup M$ into a sufficiently small number of boxes embedded into the strip. By standard techniques, as in \cite{nw16}, it is then possible to pack $S$ (essentially using NFDH in a proper grid defined by the above boxes) while increasing the total height at most by $O(\eps)OPT$. See Section~\ref{sec:smallpack} for more details on packing of small rectangles.

The following lemma from \cite{nw16} allows one to round the heights and positions of rectangles of large enough height, without increasing much the height of the packing. 
\begin{lemma}[\cite{nw16}]\label{lem:verticalrounding}
	There exists a feasible packing of height $OPT'\leq (1+\eps)OPT$ where: (1) the height of each rectangles in $L\cup T\cup V$ is rounded up to the closest integer multiple of $\gamma \cdot OPT$ and (2) their $x$-coordinates are as in the optimal solution and their $y$-coordinates are integer multiples of $\gamma \cdot OPT$.
\end{lemma}

We next focus on rounded rectangle heights (that is, we implicitly replace $L\cup T\cup V$ with their rounded version) and on this slightly suboptimal solution of height $OPT'$.

The following lemma helps us to pack rectangles in $M$.
\begin{lemma}\label{lem:mediumrectanglesrepacking}
	If $k$ in Lemma~\ref{lem:mediumrectanglesarea} is chosen sufficiently large, all the rectangles in $M$ can be packed in polynomial time into a box $B_{M,hor}$ of size $W \times O(\eps)OPT$ and a box $B_{M,ver}$ of size $(\frac{\gamma}{3} W)  \times (\alpha OPT)$. Furthermore, there is one such packing using $\frac{3\eps}{\mu_h}$ vertical containers in $B_{M,hor}$ and $\frac{\gamma}{3\mu_w}$ horizontal containers in $B_{M,ver}$.
\end{lemma}
\begin{proof}
	We first pack rectangles in $\mathcal{A} := \{R_i \in M : h_i \in (\mu_h OPT, \delta_h OPT )\}$ using NFDH. From Lemma~\ref{lem:nfdhStripPacking} we know that the height of the packing is at most $h_{\max}(\mathcal{A}) +  \frac{2\cdot a(\mathcal{A})}{W}$. Since $h_{max}(\mathcal{A}) \leq \delta_h OPT < \eps OPT$ and $a(\mathcal{A}) \le a(M) \le \eps^k \cdot OPT \cdot W \le \eps \cdot OPT \cdot W$, because of Lemma~\ref{lem:mediumrectanglesarea}, the resulting packing fits into a box $B_{M,hor}$ of size $W \times (3\eps \cdot OPT)$. As $h_i \ge \mu_h OPT$, the number of shelves used by NFDH is at most $\frac{3\eps}{\mu_h}$: this also bounds the number of vertical containers needed.

	We next pack $\mathcal{A}' := M \setminus \mathcal{A}$ into a box $B_{M,ver}$ of size $(\frac{\gamma}{3} W)  \times (\alpha OPT)$. Recall that $\gamma := \frac{\eps \delta_h}{2}$. Note that, for each $R_i \in \mathcal{A}'$, we have $w_i \in (\mu_w W, \delta_w W)$ and $h_i \leq \alpha OPT$.  By ideally rotating the box and the rectangles by $90^\circ$, we can apply the NFDH algorithm. Lemma~\ref{lem:nfdhStripPacking} implies that we can pack all the rectangles if the width of the box is at least $w_{\max}(\mathcal{A}') + \frac{2a(\mathcal{A}')}{\alpha OPT}$. Now observe that
	$$
	w_{\max}(\mathcal{A}') \leq \delta_w W = \frac{\eps \delta_h}{12} W = \frac{\gamma}{6}W
	$$
	and also, since $\alpha \geq 1/3$:
	$$
	\frac{2a(\mathcal{A}')}{\alpha OPT} \leq \frac{6a(\mathcal{A}')}{OPT} \leq 6 \eps^k W \le \frac{\gamma}{6}W,
	$$
	where the last inequality is true for any $k \geq \log_{1/\eps}{(36 / \gamma)}$. Similarly to the previous case, the number of shelves is at most $ \frac{\gamma}{3 \mu_w}$. Thus all the rectangles can be packed into at most $\frac{\gamma}{3 \mu_w}$  horizontal containers.
\end{proof}

We say that a rectangle $R_i$ is \emph{cut} by a box $B$ if both $R_i\setminus B$ and $B\setminus R_i$ are non-empty (considering both $R_i$ and $B$ as open regions with an implicit embedding on the plane). We say that a rectangle $R_i \in H$ (resp. $R_i \in T \cup V$) is \emph{nicely cut} by a box $B$ if $R_i$ is cut by $B$ and their intersection is a rectangular region of width $w_i$ (resp. height $h_i$). Intuitively, this means that an edge of $B$ cuts $R_i$ along its longest side (see Figure~\ref{fig:cut_rectangles}).

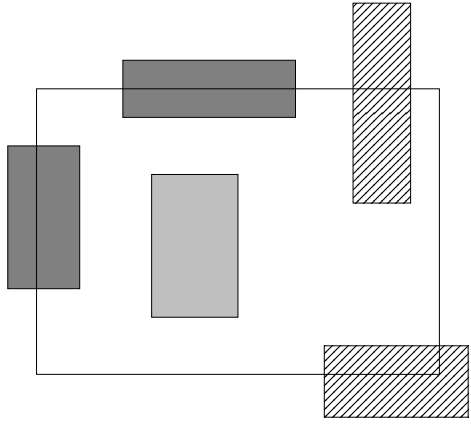
\begin{figure}
	\centering
	\resizebox{!}{5.5cm}{
		\begin{tikzpicture}
		
		
		\draw[solid, fill=gray] (-0.5,1.5) rectangle (0.75,4);
		\draw[solid, fill=gray] (1.5,4.5) rectangle (4.5,5.5);
		
		
		\draw[pattern=north east lines, pattern color=black] (5.5,3) rectangle (6.5,6.5);
		\draw[pattern=north east lines, pattern color=black] (5,-0.75) rectangle (7.5,0.5);
		
		
		\draw[solid, fill=lightgray] (2,1) rectangle (3.5,3.5);
		
		
		\draw[solid] (0,0) rectangle (7,5);
		
		\end{tikzpicture}}
	\caption{A box and some overlapping rectangles. Dark gray rectangles are \emph{nicely cut} by the box; hatched rectangles are \emph{cut} but \emph{not nicely cut} by the box, and light gray rectangle is not cut by the box.}
	\label{fig:cut_rectangles}
\end{figure}

Now it remains to pack $L \cup H\cup T \cup V$: The following lemma, taken from \cite{nw16} modulo minor technical adaptations, describes an almost optimal packing of those rectangles. 

\begin{lemma}\label{lem:boxpartition}
	There is an integer $K_B=(\frac{1}{\eps})(\frac{1}{\delta_w})^{O(1)}$ such that, assuming $\mu_h \le \frac{\eps \delta_w }{K_B}$, there is a partition of the region $B_{OPT'}:=[0,W]\times [0,OPT']$ into a set $\mathcal{B}$ of at most $K_B$ boxes and a packing of the rectangles in $L\cup T\cup V\cup H$ such that:
	\begin{itemize}
		\item each box has size equal to the size of some $R_i\in L$ (\emph{large box}), or has height at most $\delta_h OPT'$ (\emph{horizontal box}), or has width at most $\delta_w W$ (\emph{vertical box});
		\item each $R_i \in L$ is contained into a large box of the same size;
		\item each $R_i\in H$ is contained into a horizontal box or is cut by some box. Furthermore, the total area of horizontal cut rectangles is at most $W\cdot O(\eps)OPT'$;
		\item each $R_i\in T\cup V$ is contained into a vertical box or is nicely cut by some vertical box. 
	\end{itemize}
\end{lemma}
\begin{proof}
	We apply Lemma~6 in~\cite{nw16}, where we set the parameter $\delta$ to $\delta_w$. Recall that $\delta_w < \delta_h$; by requiring that $\mu_h < \delta_w$, and since rectangles with height in $[\delta_w, \delta_h)$ are in $M$, we have that $\{R_i \in \mathcal{R} \setminus M \,:\, w_i \geq \delta_h W \mbox{ and } h_i \geq \delta_h OPT \} = \{R_i \in \mathcal{R} \setminus M \,:\, w_i \geq \delta_w W \mbox{ and } h_i \geq \delta_h OPT \}$.
	
	Let $H_{cut}\subseteq H$ be the set of horizontal rectangles that are nicely cut by a box. Since rectangles in $H_{cut}$ satisfy $w_i \ge \delta_w W$, at most $\frac{2}{\delta_w}$ of them are nicely cut by a box, and there are at most $K_B$ boxes. Hence, their total area is at most $\frac{\mu_h OPT \cdot W \cdot 2K_B}{\delta_w}$, which is at most $2\eps \cdot OPT\cdot W$, provided that $\mu_h \le \eps\cdot \frac{\delta_w}{K_B}$. Since Lemma~6 in~\cite{nw16} implies that the area of the cut horizontal rectangles that are not nicely cut is at most $\eps OPT' \cdot W$, the total area of horizontal cut rectangles is at most $3 \eps OPT' \cdot W$.
\end{proof}

We denote the sets of vertical, horizontal, and large boxes by  $\mathcal{B}_V, \mathcal{B}_H$ and $\mathcal{B}_L$, respectively. Observe that $\mathcal{B}$ can be guessed in PPT. We next use $T_{cut}\subseteq T$ and $V_{cut}\subseteq V$ to denote tall and vertical cut rectangles in the above lemma, respectively. Let us also define $T_{box}=T\setminus T_{cut}$ and $V_{box}=V\setminus V_{cut}$.

Using standard techniques (see e.g. \cite{nw16}), we can pack all the rectangles excluding the ones contained in vertical boxes in a convenient manner. This is summarized in the following lemma.
\begin{lemma}\label{lem:structural_boxes}
	Given $\mathcal{B}$ as in Lemma \ref{lem:boxpartition} and assuming $\mu_w \le \frac{\gamma \delta_h}{6K_B(1+\eps)}$, there exists a packing of $L\cup H\cup T \cup V$ such that: 
	\begin{enumerate}
		\item all the rectangles in $L$ are packed in $\mathcal{B}_L$;
		\item all the rectangles in $H$ are packed in $\mathcal{B}_H$ plus an additional box $B_{H,cut}$ of size $W\times O(\eps)OPT$;
		\item all the rectangles in $T_{cut}\cup T_{box}\cup V_{box}$ are packed as in Lemma \ref{lem:boxpartition};
		\item all the rectangles in $V_{cut}$ are packed in an additional vertical box $B_{V,cut}$ of size $(\frac{\gamma}{3} W)  \times (\alpha OPT)$.
	\end{enumerate}
\end{lemma}
\begin{proof}
	Note that there are at most $1/(\delta_w \delta_h)$ rectangles in $L$ and at most $4K_B$ rectangles in $T_{cut}$, since the left (resp. right) side of each box can nicely cut at most $2$ tall rectangles; this is enough to prove points~(1) and~(3).
	
	Thanks to Lemma \ref{lem:boxpartition}, the total area of horizontal cut rectangles is at most $O(\eps OPT' \cdot W)$. By Lemma~\ref{lem:nfdhPack}, we can remove them from the packing and pack them in the additional box $B_{H, cut}$ using NFDH algorithm, proving point~(2).
	
	Each box of the partition can nicely cut at most $\frac{2(1+\eps)}{\delta_h}$ rectangles in $V$; thus, the cut vertical rectangles are at most $\frac{2K_B(1+\eps)}{\delta_h}$. Since the width of each vertical rectangle is at most $\mu_w W$, they can be removed from the packing and placed in $B_{V, cut}$, piled side by side, as long as $\frac{2K_B(1+\eps)}{\delta_h} \cdot \mu_w W\leq \frac{\gamma}{3} W$, which is equivalent to $\mu_w \le \frac{\gamma \delta_h}{6K_B(1+\eps)}$. This proves point~(4).
\end{proof}

We will pack all the rectangles (essentially) as in \cite{nw16}, with the exception of $T_{box}\cup V_{box}$ where we exploit a refined approach. This is the technical heart of our result, and it is discussed in the next section. 

\section{A repacking lemma}
\label{sec:repack}
We next describe how to pack rectangles in $T_{box}\cup V_{box}$. In order to highlight our contribution, we first describe how the approach by \cite{nw16} works.

It is convenient to assume that all the rectangles in $V_{box}$ are sliced vertically into sub-rectangles of width $1$ each\footnote{For technical reasons, slices have width $1/2$ in~\cite{nw16}. For our algorithm, slices of width $1$ suffice.}. Let $V_{sliced}$ be such \emph{sliced} rectangles. We will show how to pack all the rectangles in $T_{box}\cup V_{sliced}$ into a constant number of sub-boxes. Using standard techniques it is then possible to pack $V_{box}$ into the space occupied by $V_{sliced}$ plus an additional box $B_{V,round}$ of size $(\frac{\gamma}{3}W)\times \alpha OPT$.
See Lemma~\ref{lem:fractional_to_integral} for more details.

We next focus on a specific vertical box $\overline{B}$, say of size $\overline{w}\times \overline{h}$ (see Figure \ref{fig_pseudo-rectangles1}). Let $\overline{T}_{cut}$ be the tall rectangles cut by $\overline{B}$. Observe that there are at most $4$ such rectangles ($2$ on the left/right side of $\overline{B}$).
The rectangles in $\overline{T}_{cut}$ are packed as in Lemma \ref{lem:structural_boxes}.
Let also $\overline{T}$ and $\overline{V}$ be the tall rectangles and sliced vertical rectangles, respectively, originally packed completely inside  $\overline{B}$.

They show that it is possible to pack $\overline{T}\cup\overline{V}$ into a constant size set $\overline{\mathcal{S}}$ of sub-boxes contained inside $\overline{B}-\overline{T}_{cut}$, plus an additional box $\overline{D}$ of size $\overline{w}\times (1+\eps-2\alpha)\overline{h}$. Here $\overline{B}-\overline{T}_{cut}$ denotes the region inside $\overline{B}$ not contained in $\overline{T}_{cut}$.
In more detail, they start by considering each rectangle $R_i\in \overline{T}$. Since $\alpha\geq\frac{1}{3}$ by assumption, one of the regions above or below $R_i$ cannot contain another tall rectangle in $\overline{T}$, say the first case applies (the other one being symmetric). Then we move $R_i$ up so that its top side overlaps with the top side of $\overline{B}$. The sliced rectangles in $\overline{V}$ that are covered this way are shifted right below $R$ (note that there is enough free space by construction). At the end of the process all the rectangles in $\overline{T}$ touch at least one of the top and bottom side of $\overline{B}$ (see Figure~\ref{fig_pseudo-rectangles2}). Note that no rectangle is discarded up to this point.

Next, we partition the space inside $\overline{B}-(\overline{T}\cup \overline{T}_{cut})$ into maximal height unit-width vertical stripes. We call each such stripe a \emph{free rectangle} if both its top and bottom side overlap with the top or bottom side of some rectangle in $\overline{T}\cup \overline{T}_{cut}$, and otherwise a \emph{pseudo rectangle} (see Figure \ref{fig_pseudo-rectangles3}). We define the $i$-th free rectangle to be the free rectangle contained in stripe $[i-1,i]\times [0,\overline{h}]$.

Note that all the free rectangles are contained in a horizontal region of width $\overline{w}$ and height at most $\overline{h}-2\alpha OPT \le \overline{h}-2\alpha \frac{OPT'}{1+\eps}
\le \overline{h}(1-\frac{2\alpha}{1+\eps}) \le \overline{h}(1+\eps-2\alpha)$ contained in the central part of $\overline{B}$.
Let $\overline{V}_{disc}$ be the set of  (sliced vertical) rectangles contained in the free rectangles.
Rectangles in $\overline{V}_{disc}$ can be obviously packed inside $\overline{D}$. For each corner $Q$ of the box $\overline{B}$, we consider the maximal rectangular region that has $Q$ as a corner and only contains pseudo rectangles whose top/bottom side overlaps with the bottom/top side of a rectangle in $\overline{T}_{cut}$; there are at most $4$ such non-empty regions, and for each of them we define a \emph{corner sub-box}, and we call the set of such sub-boxes $\overline{B}_{corn}$ (see Figure \ref{fig_pseudo-rectangles3}). The final step of the algorithm is to rearrange horizontally the pseudo/tall rectangles so that pseudo/tall rectangles of the same height are grouped together \emph{as much as possible} (modulo some technical details). The rectangles in $\overline{B}_{corn}$ are not moved.
The \emph{sub-boxes} are induced by maximal consecutive subsets of pseudo/tall rectangles of the same height touching the top (resp., bottom) side of $\overline{B}$
(see Figure \ref{fig_pseudo-rectangles4}).
We crucially remark that, by construction, the height of each sub-box (and of $\overline{B}$) is a multiple of $\gamma OPT$.

By splitting each discarded box $\overline{D}$ into two halves $\overline{B}_{disc,top}$ and $\overline{B}_{disc,bot}$, and replicating the packing of boxes inside $B_{OPT'}$, it is possible to pack all the discarded boxes into two boxes $B_{disc,top}$ and $B_{disc,bot}$, both of size $\frac{W}{2}\times (1+\eps-2\alpha)OPT'$. 

A feasible packing of boxes (and hence of the associated rectangles) of height $(1+\max\{\alpha,2(1-2\alpha)\}+O(\eps))OPT$ is then obtained as follows. We first pack $B_{OPT'}$ at the base of the strip, and then on top of it we pack $B_{M,hor}$, two additional boxes $B_{H, round}$ and $B_{H, cut}$ (which will be used to repack the horizontal items), and a box $B_S$ (which will be used to pack some of the small items). The latter $4$ boxes all have width $W$ and height $O(\eps OPT')$. On the top right of this packing we place $B_{disc,top}$ and $B_{disc,bot}$, one on top of the other. Finally, we pack $B_{M,ver}$, $B_{V,cut}$ and $B_{V,round}$ on the top left, one next to the other. See Figure \ref{fig_packing_NW} for an illustration.
The height is minimized for $\alpha=\frac{2}{5}$, leading to a $7/5+O(\eps)$ approximation.

The main technical contribution of our result is to show how it is possible to repack a subset of $\overline{V}_{disc}$ into the \emph{free} space inside $\overline{B}_{cut}:=\overline{B}-\overline{T}_{cut}$ not occupied by sub-boxes, so that the residual sliced rectangles can be packed into a single discarded box $\overline{B}_{disc}$ of size $(1-\gamma)\overline{w}\times (1+\eps-2\alpha)\overline{h}$ (\emph{repacking lemma}). See Figure 
\ref{fig_pseudo-rectangles5}. This apparently minor saving is indeed crucial: with the same approach as above all the discarded sub-boxes $\overline{B}_{disc}$ can be packed into a single \emph{discarded box} $B_{disc}$ of size $(1-\gamma)W\times (1+\eps-2\alpha)OPT'$. Therefore, we can pack all the previous boxes as before, and $B_{disc}$ on the top right. Indeed, the total width of $B_{M,ver}$, $B_{V,cut}$ and $B_{V,round}$ is at most $\gamma W$ for a proper choice of the parameters. See Figure \ref{fig_our_packing} for an illustration. Altogether the resulting packing has height $(1+\max\{\alpha,1-2\alpha\}+O(\eps))OPT$. This is minimized for $\alpha=\frac{1}{3}$, leading to the claimed $4/3+O(\eps)$ approximation. 

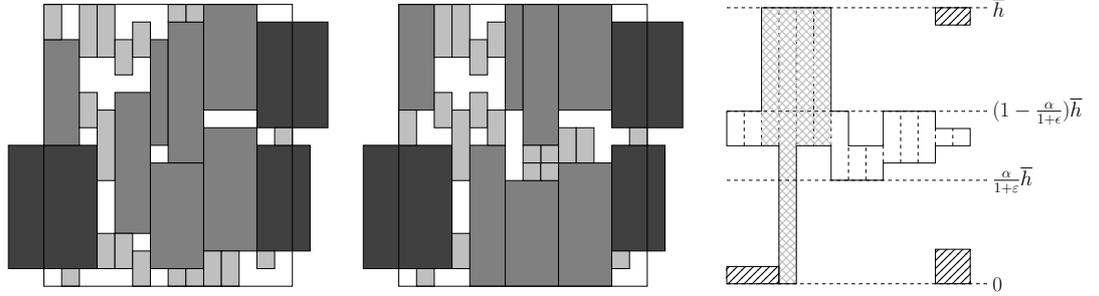
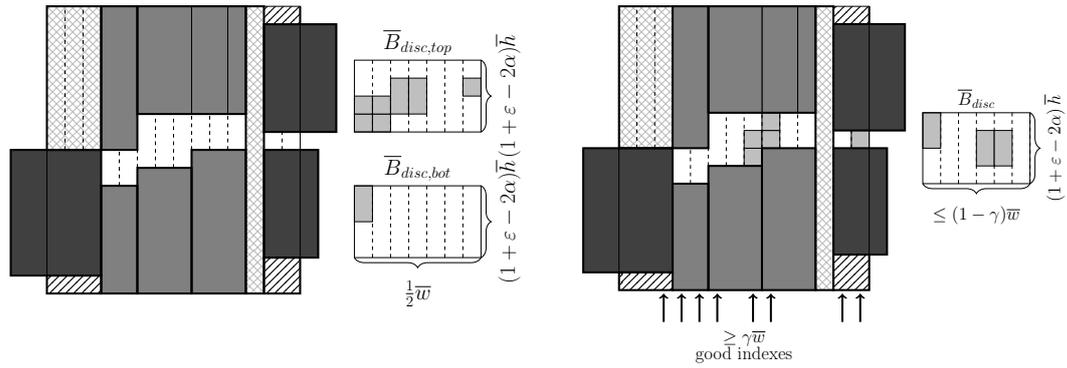
\begin{figure}
	\captionsetup[subfigure]{justification=centering}
	\hspace{-15pt}
	\begin{subfigure}[b]{.32\textwidth}
		\centering
		\resizebox{!}{4cm}{
			\begin{tikzpicture}
			
			\draw[fill=darkgray] (-1,0.5) rectangle (1.5,4);
			\draw[fill=darkgray] (6,1) rectangle (7.5,4);
			\draw[fill=darkgray] (6,4.5) rectangle (8,7.5);
			
			
			\draw[fill=gray] (2,1.5) rectangle (3,5.5);
			\draw[fill=gray] (3,0.5) rectangle (4.5,3.5);
			\draw[fill=gray] (4.5,1) rectangle (6,4.5);
			\draw[fill=gray] (0,4) rectangle (1,7);
			\draw[fill=gray] (3,4) rectangle (3.5,7);
			\draw[fill=gray] (3.5,3.5) rectangle (4.5,7.5);
			\draw[fill=gray] (4.5,5) rectangle (6,8);
			
			
			\draw[fill=lightgray] (0.5,0) rectangle (1,0.5);
			\draw[fill=lightgray] (1.5,0.5) rectangle (2,1.5);
			\draw[fill=lightgray] (2,0.5) rectangle (2.5,1.5);
			\draw[fill=lightgray] (2.5,0) rectangle (3,1);
			\draw[fill=lightgray] (3.5,0) rectangle (4,0.5);
			\draw[fill=lightgray] (4.5,0) rectangle (5,1);
			\draw[fill=lightgray] (5,0) rectangle (5.5,1);
			\draw[fill=lightgray] (6,0.5) rectangle (6.5,1);
			\draw[fill=lightgray] (4,0) rectangle (4.5,0.5);
			\draw[fill=lightgray] (4,7.5) rectangle (4.5,8);
			\draw[fill=lightgray] (6.5,4) rectangle (7,4.5);
			
			
			\draw[fill=lightgray] (0,7) rectangle (0.5,8);
			\draw[fill=lightgray] (1,6.5) rectangle (1.5,8);
			\draw[fill=lightgray] (1,4.5) rectangle (1.5,5.5);
			\draw[fill=lightgray] (1.5,6.5) rectangle (2,8);
			\draw[fill=lightgray] (1.5,3) rectangle (2,5);
			\draw[fill=lightgray] (2,6) rectangle (2.5,7);
			\draw[fill=lightgray] (2.5,6.5) rectangle (3,7.5);
			\draw[fill=lightgray] (3.5,7.5) rectangle (4,8);
			
			
			\draw[thick] (0,0) rectangle (7,8);
			\draw[color=white] (0,-0.25) rectangle (0,-0.1);
			\draw[color=white] (0,8.3) rectangle (0,8.1);
			
			\end{tikzpicture}}
		\caption{Original packing in a vertical box $\overline{B}$ after removing $V_{cut}$. Gray rectangles correspond to $\overline{T}$, dark gray rectangles to $\overline{T}_{cut}$ and light gray rectangles to $\overline{V}$.}
		\label{fig_pseudo-rectangles1}
	\end{subfigure}%
	\hspace{2pt}
	\begin{subfigure}[b]{.32\textwidth}
		\centering
		\resizebox{!}{4cm}{
			\begin{tikzpicture}
			
			\draw[fill=darkgray] (-1,0.5) rectangle (1.5,4);
			\draw[fill=darkgray] (6,1) rectangle (7.5,4);
			\draw[fill=darkgray] (6,4.5) rectangle (8,7.5);
			
			
			\draw[fill=gray] (2,0) rectangle (3,4);
			\draw[fill=gray] (3,0) rectangle (4.5,3);
			\draw[fill=gray] (4.5,0) rectangle (6,3.5);
			\draw[fill=gray] (0,5) rectangle (1,8);
			\draw[fill=gray] (3,5) rectangle (3.5,8);
			\draw[fill=gray] (3.5,4) rectangle (4.5,8);
			\draw[fill=gray] (4.5,5) rectangle (6,8);
			
			
			\draw[fill=lightgray] (0.5,0) rectangle (1,0.5);
			\draw[fill=lightgray] (1.5,0.5) rectangle (2,1.5);
			\draw[fill=lightgray] (2,4.5) rectangle (2.5,5.5);
			\draw[fill=lightgray] (2.5,4) rectangle (3,5);
			\draw[fill=lightgray] (3.5,3) rectangle (4,3.5);
			\draw[fill=lightgray] (4.5,3.5) rectangle (5,4.5);
			\draw[fill=lightgray] (5,3.5) rectangle (5.5,4.5);
			\draw[fill=lightgray] (4,3) rectangle (4.5,3.5);
			\draw[fill=lightgray] (4,3.5) rectangle (4.5,4);
			\draw[fill=lightgray] (6.5,4) rectangle (7,4.5);
			\draw[fill=lightgray] (6,0.5) rectangle (6.5,1);
			
			
			\draw[fill=lightgray] (0,4) rectangle (0.5,5);
			\draw[fill=lightgray] (1,6.5) rectangle (1.5,8);
			\draw[fill=lightgray] (1,4.5) rectangle (1.5,5.5);
			\draw[fill=lightgray] (1.5,6.5) rectangle (2,8);
			\draw[fill=lightgray] (1.5,3) rectangle (2,5);
			\draw[fill=lightgray] (2,6) rectangle (2.5,7);
			\draw[fill=lightgray] (2.5,6.5) rectangle (3,7.5);
			\draw[fill=lightgray] (3.5,3.5) rectangle (4,4);
			
			
			\draw[thick] (0,0) rectangle (7,8);
			\draw[color=white] (0,-0.25) rectangle (0,-0.1);
			\draw[color=white] (0,8.3) rectangle (0,8.1);
			
			\end{tikzpicture}}
		\caption{Rectangles in $\overline{T}$ are shifted vertically so that \\they touch either the top \\or the bottom of box $\overline{B}$, shifting also slices in $\overline{V}$ accordingly.\\~}
		\label{fig_pseudo-rectangles2}
	\end{subfigure}
	\hspace{2pt}
	\begin{subfigure}[b]{.36\textwidth}
		\centering
		\resizebox{!}{4cm}{
			\begin{tikzpicture}
			
			
			\draw[thick, dashed] (0.5,4) -- (0.5,5);
			\draw[thick, dashed] (1.5,4) -- (1.5,8);
			\draw[thick, dashed] (2.5,4) -- (2.5,8);
			\draw[thick, dashed] (3.5,3) -- (3.5,4);
			\draw[thick, dashed] (4.5,3.5) -- (4.5,4);
			\draw[thick, dashed] (5.5,3.5) -- (5.5,5);
			\draw[thick, dashed] (6.5,4) -- (6.5,4.5);
			\draw[thick, dashed] (1,4) -- (1,5);
			\draw[thick, dashed] (2,4) -- (2,8);
			\draw[thick, dashed] (3,4) -- (3,5);
			\draw[thick, dashed] (4,3) -- (4,4);
			\draw[thick, dashed] (5,3.5) -- (5,5);
			\draw[thick, dashed] (6,4) -- (6,4.5);
			
			
			\fill [pattern color=lightgray, pattern = crosshatch] (1.5,0) rectangle (2,8);
			\fill [pattern color=lightgray, pattern = crosshatch] (2,4) rectangle (2.5,8);
			\fill [pattern color=lightgray, pattern = crosshatch] (2.5,4) rectangle (3,8);
			
			\fill [pattern color=lightgray, pattern = crosshatch] (1,4) rectangle (1.5,8);
			
			
			\draw (0,0) rectangle (1.5,0.5);
			\fill[thick, pattern = north east lines] (0,0) rectangle (1.5,0.5);
			\draw (6,0) rectangle (7,1);
			\fill[thick, pattern = north east lines] (6,0) rectangle (7,1);
			\draw (6,7.5) rectangle (7,8);
			\fill[thick, pattern = north east lines] (6,7.5) rectangle (7,8);
			
			
			\draw[thick] (1.5,0) -- (1.5,4) -- (0,4) -- (0,5) -- (1,5) -- (1,8) -- (3,8) -- (3,5) -- (3.5,5) -- (3.5,4) -- (4.5,4) -- (4.5,5) -- (6,5) -- (6,4.5) -- (7,4.5) -- (7,4) -- (6,4) -- (6,3.5) -- (4.5,3.5) -- (4.5,3) -- (3,3) -- (3,4) -- (2,4) -- (2,0) -- (1.5,0);
			
			
			\draw[dashed] (0,0) --(7.5,0);
			\draw[dashed] (0,3) --(7.5,3);
			\draw[dashed] (0,5) --(7.5,5);
			\draw[dashed] (0,8) --(7.5,8);
			
			\draw (7.5,0) node[anchor=west] {\Large $0$};
			\draw (7.5,3) node[anchor=west] {\Large $\frac{\alpha}{1+\eps} \overline{h}$};
			\draw (7.5,5) node[anchor=west] {\Large $(1-\frac{\alpha}{1+\epsilon}) \overline{h}$};
			\draw (7.5,8) node[anchor=west] {\Large $\overline{h}$};
			
			\end{tikzpicture}}
		\caption{We define pseudo rectangles and free space in $\overline{B} - (\overline{T} \cup \overline{T}_{cut})$. Crosshatched stripes correspond to pseudo rectangles, empty stripes to free rectangles, and dashed regions to corner sub-boxes.}
		\label{fig_pseudo-rectangles3}
	\end{subfigure}
	
	\vspace{10pt}
	\hspace{-15pt}
	\begin{subfigure}[b]{.5\textwidth}
		\centering
		\resizebox{!}{4.8cm}{
			\begin{tikzpicture}
			
			\draw[fill=darkgray] (-1,0.5) rectangle (1.5,4);
			\draw[fill=darkgray] (6,1) rectangle (7.5,4);
			\draw[fill=darkgray] (6,4.5) rectangle (8,7.5);
			
			
			\draw[fill=gray] (4,0) rectangle (5.5,4);
			\draw[fill=gray] (1.5,0) rectangle (2.5,3);
			\draw[fill=gray] (2.5,0) rectangle (4,3.5);
			\draw[fill=gray] (5,5) rectangle (5.5,8);
			\draw[fill=gray] (4,5) rectangle (5,8);
			\draw[fill=gray] (1.5,4) rectangle (2.5,8);
			\draw[fill=gray] (2.5,5) rectangle (4,8);
			
			
			\fill [pattern color=lightgray, pattern = crosshatch] (5.5,0) rectangle (6,8);
			\fill [pattern color=lightgray, pattern = crosshatch] (1,4) rectangle (1.5,8);
			\fill [pattern color=lightgray, pattern = crosshatch] (0.5,4) rectangle (1,8);
			\fill [pattern color=lightgray, pattern = crosshatch] (0,4) rectangle (0.5,8);
			
			
			\fill[pattern = north east lines] (0,0) rectangle (1.5,0.5);
			\fill[pattern = north east lines] (6,0) rectangle (7,1);
			\fill[pattern = north east lines] (6,7.5) rectangle (7,8);
			
			
			\draw[thick, dashed] (0.5,4) -- (0.5,8);
			\draw[thick, dashed] (1,4) -- (1,8);
			\draw[thick, dashed] (2,3) -- (2,4);
			\draw[thick, dashed] (2.5,3.5) -- (2.5,4);
			\draw[thick, dashed] (3,3.5) -- (3,5);
			\draw[thick, dashed] (3.5,3.5) -- (3.5,5);
			\draw[thick, dashed] (4,4) -- (4,5);
			\draw[thick, dashed] (4.5,4) -- (4.5,5);
			\draw[thick, dashed] (5,4) -- (5,5);
			\draw[thick, dashed] (5.5,4) -- (5.5,5);
			\draw[thick, dashed] (6,4) -- (6,4.5);
			\draw[thick, dashed] (6.5,4) -- (6.5,4.5);
			
			
			\draw[thick] (0,0) rectangle (7,8);
			
			
			\draw[ultra thick] (0,0) rectangle (1.5,0.5);
			\draw[ultra thick] (6,0) rectangle (7,1);
			\draw[ultra thick] (7,7.5) rectangle (7,8);
			
			\draw[ultra thick] (-1,0.5) rectangle (1.5,4);
			\draw[ultra thick] (6,1) rectangle (7.5,4);
			\draw[ultra thick] (6,4.5) rectangle (8,7.5);
			
			\draw[ultra thick] (0,4) rectangle (1.5,8);
			\draw[ultra thick] (5.5,0) rectangle (6,8);
			
			\draw[ultra thick] (1.5,0) rectangle (2.5,3);
			\draw[ultra thick] (2.5,0) rectangle (4,3.5);
			\draw[ultra thick] (4,0) rectangle (5.5,4);
			\draw[ultra thick] (1.5,4) rectangle (2.5,8);
			\draw[ultra thick] (2.5,5) rectangle (5.5,8);
			
			
			\draw[fill=lightgray] (8.5,2) rectangle (9,3);
			
			
			\draw[fill=lightgray] (8.5,4.5) rectangle (9,5);
			\draw[fill=lightgray] (8.5,5) rectangle (9,5.5);
			\draw[fill=lightgray] (9,4.5) rectangle (9.5,5);
			\draw[fill=lightgray] (9,5) rectangle (9.5,5.5);
			\draw[fill=lightgray] (9.5,5) rectangle (10,6);
			\draw[fill=lightgray] (10,5) rectangle (10.5,6);
			\draw[fill=lightgray] (11.5,5.5) rectangle (12,6);
			
			
			\draw[dashed] (9,1) -- (9,3);
			\draw[dashed] (9.5,1) -- (9.5,3);
			\draw[dashed] (10,1) -- (10,3);
			\draw[dashed] (10.5,1) -- (10.5,3);
			\draw[dashed] (11,1) -- (11,3);
			\draw[dashed] (11.5,1) -- (11.5,3);
			
			\draw[dashed] (9,4.5) -- (9,6.5);
			\draw[dashed] (9.5,4.5) -- (9.5,6.5);
			\draw[dashed] (10,4.5) -- (10,6.5);
			\draw[dashed] (10.5,4.5) -- (10.5,6.5);
			\draw[dashed] (11,4.5) -- (11,6.5);
			\draw[dashed] (11.5,4.5) -- (11.5,6.5);
			
			
			\draw (8.5,1) rectangle (12,3);
			\draw (8.5,4.5) rectangle (12,6.5);
			
			
			\draw [thick,decorate,decoration={mirror, brace,amplitude=8pt}] 
			(8.5,1) -- (12,1); 
			\draw (10.25,0.5) node [anchor = north] {\Large $\frac{1}{2}\overline{w}$};
			
			\draw [thick,decorate,decoration={mirror, brace,amplitude=8pt}] 
			(12,1) -- (12,3); 
			\draw (12.25,2) node [rotate=90, anchor = north] {\Large $(1+\eps-2\alpha)\overline{h}$};
			
			\draw [thick,decorate,decoration={mirror, brace,amplitude=8pt}] 
			(12,4.5) -- (12,6.5); 
			\draw (12.25,5.5) node [rotate=90, anchor = north] {\Large $(1+\eps-2\alpha)\overline{h}$};
			
			\draw (10.25,3) node [anchor=south] {\Large $\overline{B}_{disc, bot}$};
			
			\draw (10.25,6.5) node [anchor=south] {\Large $\overline{B}_{disc, top}$};
			
			\draw[color=white] (0,-2.05) rectangle (1,-1);
			
			\end{tikzpicture}}
		\caption{Rearrangement of pseudo and tall rectangles to get $O_\eps(1)$ sub-boxes, and additional packing of $\overline{V}_{disc}$ as in \cite{nw16}.}
		\label{fig_pseudo-rectangles4}
	\end{subfigure}%
	\hspace{5pt}
	\begin{subfigure}[b]{.5\textwidth}
		\centering
		\resizebox{!}{4.8cm}{
			\begin{tikzpicture}
			
			\draw[fill=darkgray] (-1,0.5) rectangle (1.5,4);
			\draw[fill=darkgray] (6,1) rectangle (7.5,4);
			\draw[fill=darkgray] (6,4.5) rectangle (8,7.5);
			
			
			\draw[fill=gray] (4,0) rectangle (5.5,4);
			\draw[fill=gray] (1.5,0) rectangle (2.5,3);
			\draw[fill=gray] (2.5,0) rectangle (4,3.5);
			\draw[fill=gray] (5,5) rectangle (5.5,8);
			\draw[fill=gray] (4,5) rectangle (5,8);
			\draw[fill=gray] (1.5,4) rectangle (2.5,8);
			\draw[fill=gray] (2.5,5) rectangle (4,8);
			
			
			\fill [pattern color=lightgray, pattern = crosshatch] (5.5,0) rectangle (6,8);
			\fill [pattern color=lightgray, pattern = crosshatch] (1,4) rectangle (1.5,8);
			\fill [pattern color=lightgray, pattern = crosshatch] (0.5,4) rectangle (1,8);
			\fill [pattern color=lightgray, pattern = crosshatch] (0,4) rectangle (0.5,8);
			
			
			\draw[fill=lightgray] (3.5,3.5) rectangle (4,4);
			\draw[fill=lightgray] (3.5,4) rectangle (4,4.5);
			\draw[fill=lightgray] (4,4) rectangle (4.5,4.5);
			\draw[fill=lightgray] (4,4.5) rectangle (4.5,5);
			\draw[fill=lightgray] (6.5,4) rectangle (7,4.5);
			
			
			\fill[pattern = north east lines] (0,0) rectangle (1.5,0.5);
			\fill[pattern = north east lines] (6,0) rectangle (7,1);
			\fill[pattern = north east lines] (6,7.5) rectangle (7,8);
			
			
			\draw[thick, dashed] (0.5,4) -- (0.5,8);
			\draw[thick, dashed] (1,4) -- (1,8);
			\draw[thick, dashed] (2,3) -- (2,4);
			\draw[thick, dashed] (2.5,3.5) -- (2.5,4);
			\draw[thick, dashed] (3,3.5) -- (3,5);
			\draw[thick, dashed] (3.5,3.5) -- (3.5,5);
			\draw[thick, dashed] (4,4) -- (4,5);
			\draw[thick, dashed] (4.5,4) -- (4.5,5);
			\draw[thick, dashed] (5,4) -- (5,5);
			\draw[thick, dashed] (5.5,4) -- (5.5,5);
			\draw[thick, dashed] (6,4) -- (6,4.5);
			\draw[thick, dashed] (6.5,4) -- (6.5,4.5);
			
			
			\draw[thick] (0,0) rectangle (7,8);
			
			
			\draw[ultra thick] (0,0) rectangle (1.5,0.5);
			\draw[ultra thick] (6,0) rectangle (7,1);
			\draw[ultra thick] (7,7.5) rectangle (7,8);
			
			\draw[ultra thick] (-1,0.5) rectangle (1.5,4);
			\draw[ultra thick] (6,1) rectangle (7.5,4);
			\draw[ultra thick] (6,4.5) rectangle (8,7.5);
			
			\draw[ultra thick] (0,4) rectangle (1.5,8);
			\draw[ultra thick] (5.5,0) rectangle (6,8);
			
			\draw[ultra thick] (1.5,0) rectangle (2.5,3);
			\draw[ultra thick] (2.5,0) rectangle (4,3.5);
			\draw[ultra thick] (4,0) rectangle (5.5,4);
			\draw[ultra thick] (1.5,4) rectangle (2.5,8);
			\draw[ultra thick] (2.5,5) rectangle (5.5,8);
			
			
			\draw[dashed] (9,3) -- (9,5);
			\draw[dashed] (9.5,3) -- (9.5,5);
			\draw[dashed] (10,3) -- (10,5);
			\draw[dashed] (10.5,3) -- (10.5,5);
			\draw[dashed] (11,3) -- (11,5);
			\draw[dashed] (11.5,3) -- (11.5,5);
			
			
			\draw[fill=lightgray] (8.5,4) rectangle (9,5);
			\draw[fill=lightgray] (10,3.5) rectangle (10.5,4.5);
			\draw[fill=lightgray] (10.5,3.5) rectangle (11,4.5);
			
			
			\draw (8.5,3) rectangle (11.5,5);
			
			
			\draw[ultra thick, ->] (1.25,-0.9) -- (1.25,-0.1);
			\draw[ultra thick, ->] (1.75,-0.9) -- (1.75,-0.1);
			\draw[ultra thick, ->] (2.25,-0.9) -- (2.25,-0.1);
			\draw[ultra thick, ->] (2.75,-0.9) -- (2.75,-0.1);
			\draw[ultra thick, ->] (3.75,-0.9) -- (3.75,-0.1);
			\draw[ultra thick, ->] (4.25,-0.9) -- (4.25,-0.1);
			\draw[ultra thick, ->] (6.25,-0.9) -- (6.25,-0.1);
			\draw[ultra thick, ->] (6.75,-0.9) -- (6.75,-0.1);
			
			\draw (3.5,-1) node[anchor=north] {\large $\ge \gamma \overline{w}$};
			\draw (3.5,-1.45) node[anchor=north] {\large good indexes};
			
			\draw [thick,decorate,decoration={mirror, brace,amplitude=8pt}] 
			(8.5,3) -- (11.5,3); 
			\draw (10,2.5) node [anchor = north] {\large $\le (1-\gamma)\overline{w}$};
			
			\draw [thick,decorate,decoration={mirror, brace,amplitude=8pt}] 
			(11.5,3) -- (11.5,5); 
			\draw (11.75,4) node [rotate=90, anchor = north] {\large $\left(1+\eps-2\alpha\right)\overline{h}$};
			
			\draw (10,5) node[anchor = south] {\large $\overline{B}_{disc}$};
			
			\end{tikzpicture}}
		\caption{Our refined repacking of $\overline{V}_{disc}$ according to Lemma \ref{lem:repack}: some vertical slices are repacked in the free space.\\~}
		\label{fig_pseudo-rectangles5}
	\end{subfigure}	
	\caption{Creation of pseudo rectangles, how to get constant number of sub-boxes and repacking of vertical slices in a vertical box $\overline{B}$.}
	\label{fig_pseudo-rectangles}
\end{figure}

It remains to prove our repacking lemma.
\begin{lemma}[Repacking Lemma]\label{lem:repack}
	Consider a partition of $\overline{D}$ into $\overline{w}$ unit-width vertical stripes. There is a subset of at least $\gamma \overline{w}$ such stripes so that the corresponding sliced vertical rectangles $\overline{V}_{repack}$ can be repacked inside $\overline{B}_{cut}=\overline{B}-\overline{T}_{cut}$ in the space not occupied by sub-boxes.
\end{lemma}
\begin{proof}
	Let $f(i)$ denote the height of the $i$-th free rectangle, where for notational convenience we introduce a degenerate free rectangle of height $f(i)=0$ whenever the stripe $[i-1,i]\times [0,\overline{h}]$ inside $\overline{B}$ does not contain any free  rectangle. This way we have precisely $\overline{w}$ free rectangles. We remark that free rectangles are defined before the horizontal rearrangement of tall/pseudo rectangles, and the consequent definition of sub-boxes. 
	
	Recall that sub-boxes contain tall and pseudo rectangles.
	Now consider the area in $\overline{B}_{cut}$ not occupied by sub-boxes. Note that this area is  contained in the central region of height $\overline{h}(1 - \frac{2\alpha}{1 + \eps})$. Partition this area into maximal-height unit-width vertical stripes as before (\emph{newly free rectangles}). Let $g(i)$ be the height of the $i$-th newly free rectangle, where again we let $g(i)=0$ if the stripe $[i-1,i]\times [0,\overline{h}]$ does not contain any (positive area) free region. Note that, since  tall and pseudo rectangles are only shifted horizontally in the rearrangement, it must be the case that:
	$$
	\sum_{i=1}^{\overline{w}}f(i)=\sum_{i=1}^{\overline{w}}g(i).
	$$
	Let $G$ be the (\emph{good}) indexes where $g(i)\geq f(i)$, and $\overline{G}=\{1,\ldots,\overline{w}\}-G$ 
	be the \emph{bad} indexes with $g(i)<f(i)$. Observe that for each $i\in G$, it is possible to pack the $i$-th free rectangle inside the $i$-th newly free rectangle, therefore freeing a unit-width vertical strip inside $\overline{D}$. Thus it is sufficient to show that $|G|\geq \gamma \overline{w}$.
	
	Observe that, for $i\in \overline{G}$, $f(i)-g(i)\geq \gamma OPT\geq \gamma \frac{\overline{h}}{1 + \eps}$: indeed, both $f(i)$ and $g(i)$ must be multiples of $\gamma OPT$ since they correspond to the height of $\overline{B}$ minus the height of one or two tall/pseudo rectangles. On the other hand, for any index $i$, $g(i)-f(i)\leq g(i)\leq (1-\frac{2\alpha}{1+\eps})\overline{h}$, by the definition of $g$.
	Altogether
	\begin{align*}
		\left(1-\frac{2\alpha}{1 + \eps}\right)\overline{h} \cdot |G| &\geq \sum_{i\in G}(g(i)-f(i)) = \sum_{i\in \overline{G}}(f(i)-g(i))\\
		&\geq  \frac{\gamma \overline{h}}{1 + \eps} \cdot |\overline{G}| = \frac{\gamma \overline{h}}{1 + \eps} \cdot (\overline{w}-|G|) 
	\end{align*}
	We conclude that $|G|\geq \frac{\gamma}{1+\eps-2\alpha+\gamma}\overline{w}$. The claim follows since by assumption $\alpha>\eps\geq \gamma$. 
\end{proof}

\section{Proof of Lemma \ref{lem:structural_containers}}

The original algorithm in \cite{nw16} uses standard LP-based techniques, as in \cite{kr00}, to pack the horizontal rectangles. We can avoid that via a refined structural lemma: here boxes and sub-boxes are further partitioned into vertical (resp., horizontal) containers. Rectangles are then packed into such containers as mentioned earlier: one next to the other from left to right (resp., bottom to top). Containers define a multiple knapsack instance, that can be solved optimally in PPT via dynamic programming. This approach has two main advantages:
\begin{itemize}\itemsep0pt
	\item It leads to a simpler algorithm.
	\item It can be easily adapted to the case with rotations.
\end{itemize}

In this section we prove the following structural result.

\begin{lemma}\label{lem:structural_containers}
	By choosing $\alpha = 1/3$, there is an integer $K_F \leq \left(\frac{1}{\eps \delta_w}\right)^{O(1/(\delta_w \eps))}$ such that, assuming $\mu_h \leq \frac{\eps}{K_F}$ and $\mu_w \leq \frac{\gamma}{3K_F}$, there is a packing of $\mathcal{R} \setminus S$ in the region $[0, W] \times [0, (4/3 + O(\eps))OPT']$ with the following properties:
	\begin{itemize}
		\item All the rectangles in $\mathcal{R} \setminus S$ are contained in $K_{TOTAL} = O_\eps(1)$ horizontal or vertical containers, such that each of these containers is either contained in or disjoint from $\mathcal{B}_{OPT'}$;
		\item At most $K_F$ containers are contained in $\mathcal{B}_{OPT'}$, and their total area is at most $a(\mathcal{R} \setminus S)$.
	\end{itemize}
\end{lemma}

Given a set $\mathcal{M} \subseteq \mathcal{R}$ of rectangles, we define $h(\mathcal{M}) := \sum_{R_i \in \mathcal{M}} h_i$ and $w(\mathcal{M}) := \sum_{R_i \in \mathcal{M}} w_i$. We start with two preliminary lemmas.

\begin{lemma}\label{lem:rearrange_sliced} Let $q$ and $d$ be two positive integers. If a box $B$ of height $h$ contains only vertical rectangles of width $1$ that have height at least $h/d$, at most $q$ different heights, then there is a packing of all the rectangles in at most $d{(q + 1)}^d$ vertical containers packed in $B$, and the total area of the containers equals the total area of the vertical rectangles in $B$; a symmetrical statement holds for boxes containing only horizontal rectangles of height $1$.
\end{lemma}
\begin{proof}
	Without loss of generality, we only prove the lemma for the case of vertical rectangles. Consider each slice of width $1$ of $B$. In each slice, sort the rectangles by decreasing height, and move them down so that the bottom of each rectangles is touching either the bottom of the box, or the top of another rectangle. Call the \emph{type} of a slice as the set of different heights of rectangles that it contains. It is easy to see that there are at most ${(q+1)}^d$ different types of slices. Sort the slices so that all the slices of the same type appear next to each other. It easy to see that all the rectangles in the slices of a fixed type can be packed in at most $d$ vertical containers, where each container has the same height as the contained rectangles. By repeating this process for all the slices, we obtain a repacking of all the rectangles in at most $d{(q + 1)}^d$ containers.
\end{proof}

\begin{lemma}\label{lem:fractional_to_integral}
	Given a set $\{R_1, R_2, \dots, R_m\}$ of horizontal (resp. vertical) rectangles and a set $\{C_1, C_2, \dots, C_{t}\}$ of horizontal (resp. vertical) containers such that the rectangles can be packed into the containers allowing horizontal (resp. vertical) slicing of height (resp. width) one, then there is a feasible packing of all but at most $t$ rectangles into the same containers.
\end{lemma}
\begin{proof}
	Let us prove the lemma only for the horizontal case (the vertical case is analogous). Without loss of generality, assume that $w(R_1) \ge w(R_2) \ge \dots \ge w(R_m)$ and also $w(C_1) \ge \dots \ge w(C_t)$.
	
	Start assigning the rectangles iteratively to the first container and stop as soon as the total height of assigned rectangles  becomes strictly larger than $h(C_1)$. By discarding the last assigned rectangle, this gives a feasible packing (without slicing) of all the other assigned rectangles in the first container. Then we proceed similarly with the remaining rectangles and following containers.
	
	Now we show that the above procedure outputs a feasible packing of all but at most $t$ rectangles (the discarded ones) into the containers. Due to feasibility of the packing of the sliced rectangles into the containers, we already have, 
	$\sum_{i=1}^m{h(R_i)} \le \sum_{i=1}^{t}{h(C_i)}
	$.
	Note that the non-empty containers (except possibly the last one) are overfilled if we include the discarded rectangles.
	Thus, the above process assigns all the rectangles. 
	
	To finish the proof, we need to show that if $R_j$ is assigned to container $C_k$ by above procedure, then $w(R_j) \le w(C_k)$.
	Now as containers $C_1, \dots, C_{k-1}$ are overfilled including the so far discarded rectangles, we have that 
	\begin{equation}
	\label{eq:horsliced}
	\sum_{i=1}^j{h(R_i)} > \sum_{i=1}^{k-1}{h(C_i)}
	\end{equation}
	
	Now for the sake of contradiction, let us assume $w(R_j) > w(C_k)$.
	Then $w(R_p) > w(C_q)$ for all $p\le j$ and $q \ge k$.
	Thus in every feasible packing, even allowing slicing, rectangles $R_1,\dots, R_j$ must be assigned to containers $C_1, \dots, C_{k-1}$. 
	This contradicts \eqref{eq:horsliced}.
\end{proof}

Consider the packing obtained by applying Lemma~\ref{lem:structural_boxes}. We will refine this packing to obtain the additional structure claimed in Lemma~\ref{lem:structural_containers}.

\subsection{Horizontal rectangles}
In this section, we prove the following lemma:
\begin{lemma}\label{lem:structural_containers_horizontal}
	There is a constant $K_H \leq \left(\frac{1}{\eps \delta_w}\right)^{O(1/(\delta_w \eps))}$ such that, assuming $\mu_h \leq \frac{\eps}{K_H}$, it is possible to pack all the rectangles in $H$ in $K_H \leq \left(\frac{1}{\eps \delta_w}\right)^{O(1/(\delta_w \eps))}$ horizontal containers, so that each container is packed in a box $B \in \mathcal{B}_H \cup \{B_{H, cut}\}$, plus an additional container $B_{H,round}$ of size $W\times O(\eps)OPT'$, and the total area of the containers packed in a box of $\mathcal{B}_H$ is at most $a(H)$.
\end{lemma}
\begin{proof}
	
	Observe that $OPT\cdot W \geq h(H) \delta_w W$. Thus, if $OPT \leq \frac{1}{\eps}$, the statement is immediately proved by defining a container for each rectangle in $H$ (and leaving $B_{H, round}$ empty); since $|H| \leq h(H)$ (being the heights positive integers), this introduces at most $\frac{1}{\eps \delta_w}$ containers. Thus, without loss of generality, we can assume that $OPT > \frac{1}{\eps}$. If $h(H) \leq \left\lceil\eps OPT \right\rceil$, then we can pile all the rectangles in $H$ in $B_{H, round}$, whose height will clearly be at most $h(H) = O(\eps)OPT$. Assume now $h(H) > \left\lceil\eps OPT \right\rceil$.
	
	We use the standard technique of \emph{linear grouping} (\cite{kr00}). Let $j$ be the smallest positive integer such that the set $H_{long}$ of the $j$ horizontal rectangles of maximum width (breaking ties arbitrarily) has height $h(H_{long}) \geq \left\lceil\eps OPT \right\rceil$. Clearly, $h(H_{long}) \leq \left\lceil\eps OPT \right\rceil + \mu_h OPT \leq 3\eps OPT$. We remove the rectangles in $H_{long}$ from the packing. Suppose now that the remaining rectangles are sorted in order of non-increasing width, and that they are sliced in rectangles of unit height. We can form groups $H_1, H_2, \dots, H_t$ of total height exactly $\left\lceil\eps OPT \right\rceil$ (possibly except for the last group, that can have smaller total height). Since $h(H) \leq OPT/\delta_w$, it follows that $t \leq \frac{1}{\eps\delta_w}$. With the convention that $H_0 := H_{long}$, then for each positive integer $i \leq t$ we have that the width of any (possibly sliced) rectangle in $H_t$ is smaller than the width of any rectangle in $H_{t-1}$; round up the widths of each rectangle in $H_i$ to $w_{max}(H_i)$, and let $\overline{H}_i$ be the obtained set of rectangles; let $\overline{H} = \bigcup_{i = 1}^t \overline{H}_i$. By the above observation, for each $i > 0$ it is possible to pack all the rectangles in $\overline{H}_i$ in the space that was occupied in the original packing by the rectangles in $H_{i - 1}$; moreover, $a(\overline{H}) \leq a(H)$.
	
	Consider each box $B \in \mathcal{B}_H \cup \{B_{H,cut}\}$ and the packing of the elements of $\overline{H}$ obtained by the above process. By applying Lemma~\ref{lem:rearrange_sliced} on each box, there is a packing of all the rectangles of $\overline{H}$ in at most $(1/\delta_w){\left(1 + \frac{1}{\eps \delta_w}\right)}^{1/\delta_w} \leq {\left(\frac{1}{\eps\delta_w}\right)}^{3/\delta_w}$ horizontal containers for each box, such that the total area of these containers is at most $a(\overline{H}) \le a(H)$.
	
	By putting back the slices of the original width, we obtain a packing of all the slices of the rectangles in $H_1, H_2, \dots, H_t$. By Lemma~\ref{lem:fractional_to_integral}, there exists a packing of all the rectangles in $H$, except for a set of at most $K_H := K_B {\left(\frac{1}{\eps\delta_w}\right)}^{3/\delta_w}$ horizontal containers. Provided that $\mu_h \leq \frac{\eps}{K_H}$, those remaining rectangles can be piled in $B_{H,round}$, together with rectangles in $H_{long}$, by defining its height as $4\eps OPT$.
	
\end{proof}

\subsection{Vertical and tall rectangles}
In this section we prove the following lemma:
\begin{lemma}\label{lem:structural_containers_vertical}
	There is a constant $K_V \leq \left(\frac{1}{\eps \delta_w}\right)^{O(1/(\delta_w \eps))}$ such that, assuming $\mu_w \leq \frac{\gamma}{3K_V}$, it is possible to pack all the rectangles in $T \cup V$ in at most $K_V$ vertical containers, so that each container is packed completely either:
	\begin{itemize}
		\item in one of the boxes in $\mathcal{B}_V$; 
		\item in the original position of a nicely cut rectangle from Lemma \ref{lem:boxpartition} and containing only the corresponding nicely cut rectangle;
		\item in a box $B_{disc}$ of size $(1 - \gamma)W \times (1+\eps-2\alpha)OPT'$;
		\item in one of two boxes $B_{V, cut}$ and $B_{V, round}$, each of size $\frac{\gamma}{3} W \times \alpha OPT$, which are in fact containers.
	\end{itemize}
	Moreover, the area of the vertical containers packed in $B_{OPT'}$ is at most $a(T \cup V)$.
\end{lemma}
Consider a specific vertical box $\overline{B}$ of size $\overline{w}\times \overline{h}$; as described in Section~\ref{sec:repack}, up to four sub-boxes $\overline{B}_{corn}$ are defined, each one of them only containing slices of rectangles in $V$. Then, the rectangles are repacked so that each rectangle in $\overline{T}$ touches either the top or the bottom edge of $\overline{B}$, and the set $\overline{P}$ of pseudo rectangles is defined. Let $\overline{B}_{rem}:=\overline{B} - \overline{T}_{cut}$.
We now get a rearrangement of this packing applying the following lemma from \cite{nw16}:

\begin{lemma}[follows from the proof of Lemma~10 in \cite{nw16}] \label{lem_rearrangeTUP} There is packing of $\overline{T} \cup \overline{P} \cup \overline{B}_{corn}$ into at most $K_R := 2\frac{1+\eps}{\gamma} \cdot 6^{(1+\eps)/\gamma}$ sub-boxes inside  $\overline{B}_{rem}$, such that:
	\begin{itemize}
		\item each sub-box contains only tall rectangles or only pseudo rectangles, that are all of the same height as the sub-box
		\item each of the sub-boxes is completely occupied by the contained tall rectangles or pseudo rectangles, and the tall or pseudo rectangles that it contains have the same $y$-coordinate as before the rearrangement.
		\item the sub-boxes in $\overline{B}_{corn}$ and the rectangle slices inside them are packed in the same position as before.
	\end{itemize}
\end{lemma}
Consider the packing obtained by the above lemma; partition all the free space in $\overline{B}_{rem}$ which is not occupied by the above defined boxes into at most $2K_R + 1$ empty sub-boxes by considering the maximal rectangular regions that are not intersected by the vertical lines passing through the edges of the sub-boxes. By Lemma~\ref{lem:repack}, the fraction of the rectangles contained in slices of $\overline{D}$ of total width at least $\gamma \overline{w}$ can be repacked inside the empty sub-boxes.

Among the at most $3K_R + 1$ sub-boxes that we defined, some only contain tall rectangles, while the others contain pseudo rectangles. The ones that only contain tall rectangles are in fact already containers. Note also that the box $B_{V,cut}$ defined in the proof of Lemma~\ref{lem:structural_boxes} is indeed already a vertical container. For each sub-box $B'$ that contains pseudo rectangles, we now consider the sliced vertical rectangles that are packed in it. By Lemma~\ref{lem:rearrange_sliced}, there is a packing of all the (sliced) rectangles in $B'$ into at most $\frac{1}{\delta_h}{(1 + 1/\gamma)}^{1/\delta_h}$ containers, and their total area is equal to the total area of the slices of the rectangles they contain. 
There are also at most $4K_B$ containers to pack the tall rectangles that are nicely cut; each of them is packed in his original position in a vertical container of exactly the same size. In total we defined at most $\kappa := (3 K_R + 1) \frac{1}{\delta_h}{(1 + 1/\gamma)}^{1/\delta_h}+4K_B + 1$ containers (where the additional term $1$ is added to take $B_{V, cut}$ into account). We remark that all the tall rectangles are integrally packed, while vertical rectangles are sliced and packed into a container with only slices of vertical rectangles. 
The total area of all the vertical containers packed in $B_{OPT'}$ is at most the sum of the total area of tall items and the total area of the sliced vertical rectangles, i.e., at most $a(T \cup V)$. Finally, by Lemma~\ref{lem:fractional_to_integral}, all but $\kappa$ vertical rectangles can be packed in the containers. With the condition that $\mu_w \leq \frac{\gamma}{3 \kappa}$, these remaining vertical rectangles can be packed in a vertical container $B_{V, round}$ of size $\frac{\gamma W}{3} \times \alpha OPT$. This concludes the proof of Lemma~\ref{lem:structural_containers_vertical} with $K_V := \kappa + 1$.

\subsection{Concluding the proof}

There are at most $K_L := \frac{1}{\delta_h \delta_w}$ many large rectangles. Each such large rectangle is assigned to one container of the same size.

Rectangles in $M$ are packed as described in the proof of  Lemma~\ref{lem:mediumrectanglesrepacking}, using at most $K_M:=\frac{\gamma}{3 \mu_w}+\frac{3 \eps}{\mu_h}$ containers, which are placed in the boxes $B_{M,hor}$ and $B_{M,ver}$.

Horizontal and vertical rectangles are packed as explained in Lemma~\ref{lem:structural_containers_horizontal} and Lemma~\ref{lem:structural_containers_vertical}, respectively. The total number of containers $K_{TOTAL}=K_L+K_M+K_H+K_V$, is clearly $O_\eps(1)$, and each of these containers is either contained or disjoint from $\mathcal{B}_{OPT'}$. Among them, at most $K_F := K_L + K_H + K_V$ are contained in $\mathcal{B}_{OPT'}$. The total area of these $K_F$ containers is at most $a(H)+a(T \cup V)+a(L) \le a(\mathcal{R} \setminus S)$.

By packing the boxes and containers we defined as in Figure~\ref{fig_our_packing}, we obtained a packing in a strip of width $W$ and height $OPT' \cdot (\max\{1 + \alpha, 1 + (1 - 2\alpha)\} + O(\eps))$, which is at most $(4/3 + O(\eps))OPT'$ for $\alpha = 1/3$. This concludes the proof of Lemma~\ref{lem:structural_containers}.

\section{Packing the small rectangles}
\label{sec:smallpack}
It remains to pack the small rectangles $S$. We will pack them in the free space left by containers inside $[0,W] \times [0,OPT']$ plus an additional container $B_S$ of small height as the following lemma states. By placing box $B_S$ on top of the remaining packed rectangles, the final height of the solution increases only by $\eps \cdot OPT'$.
\begin{lemma}\label{lem:smallpacking} Assuming \sal{$\mu_w\mu_h \leq \frac{(1 - 2\eps)\eps^2}{10 K_F^2}$}, it is possible to pack in polynomial time all the rectangles in $S$ into the area $[0,W] \times [0,OPT']$ not occupied by containers plus an additional box $B_S$ of size $W \times \eps OPT'$. Moreover, all the small rectangles are packed into $\eps$-granular area containers.
\end{lemma}
\begin{proof}
We first extend the sides of the containers inside $[0,W] \times [0,OPT']$ in order to define a grid. This procedure partitions the free space in $[0,W] \times [0,OPT']$ into a constant number of rectangular regions (at most ${(2K_F+1)}^2 \leq 5K_F^2$ many) whose total area is at least $a(S)$ thanks to Lemma~\ref{lem:structural_containers}. Let $\mathcal{B}_{small}$ be the set of such rectangular regions with width at least $\sal{\eps} \mu_w W / \sal{\eps}$ and height at least $\mu_h OPT / \sal{\eps}$ (notice that the total area of rectangular regions not in $\mathcal{B}_{small}$ is at most $5K_F^2\mu_w \mu_h \cdot W \cdot OPT / \sal{\eps^2}$). We treat each such region as an $\eps$-granular area container, and we use NFDH to pack a subset of $S$ into the regions in $\mathcal{B}_{small}$. By standard properties of NFDH (see also Lemma~\ref{lem:nfdhPack}), since each region in $\mathcal{B}_{small}$ has size at most $W \times OPT'$ and each item in $S$ has width at most $\mu_w W$ and height at most $\mu_h OPT$, the total area of the unpacked rectangles from $S$ can be bounded above by \sal{$2\eps \cdot 5K_F^2\mu_w \mu_h \cdot W \cdot OPT' / \eps^2 = 10 K_F^2\mu_w \mu_h \cdot W \cdot OPT' / \eps$}.
Therefore, thanks to Lemma~\ref{lem:nfdhPack}, we can pack the latter small rectangles with NFDH in an additional area container $B_S$ of width $W$ and height \sal{$\eps OPT'$}, provided that \sal{$(1 - 2\eps) \eps W \cdot OPT' \geq \frac{10K_F^2 \mu_w \mu_h \cdot W \cdot OPT'}{\eps}$}, that is, \sal{$\mu_w\mu_h \leq \frac{(1 - 2\eps)\eps^2}{10 K_F^2}$}.
\end{proof}

The choice of $f$ and $k$ from Lemma~\ref{lem:mediumrectanglesarea} (and consequently the actual values of $\mu_h$, $\delta_w$, and $\mu_w$) needs to satisfy all the constraints that arise from the analysis, that is:
\begin{table}[ht]
	\captionsetup{labelformat=empty}
	\centering
	{\normalsize
		\begin{tabular}{l l}
			$\bullet$ $\mu_w = \frac{\eps \mu_h}{12}$ and $\delta_w = \frac{\eps \delta_h}{12}$ (Lemma~\ref{lem:mediumrectanglesarea}), \qquad\qquad &
			$\bullet$ $\mu_w \leq \gamma \frac{\delta_h}{6K_B(1+\eps)}$ (Lemma~\ref{lem:structural_boxes}),\\
			
			$\bullet$ $\gamma = \frac{\eps \delta_h}{2}$ (Lemma~\ref{lem:verticalrounding}), &
			$\bullet$ $\mu_w \leq \frac{\gamma}{3 K_F}$ (Lemma~\ref{lem:structural_containers}),\\
			
			$\bullet$ $6 \eps^k \le \frac{\gamma}{6}$ (Lemma \ref{lem:mediumrectanglesrepacking}), &
			$\bullet$ $\mu_h \leq \frac{\eps}{K_F}$ (Lemma \ref{lem:structural_containers}),\\
			
			$\bullet$ $\mu_h \leq \frac{\eps \delta_w}{K_B}$ (Lemma~\ref{lem:boxpartition}), &
			$\bullet$ \sal{$\mu_w\mu_h \leq \frac{(1 - 2\eps)\eps^2}{10 K_F^2}$} (Lemma~\ref{lem:smallpacking}) 
		\end{tabular}
	}
	\label{table:conditions}
	\caption{}
\end{table}
\vspace{-25pt}

It is not difficult to see that all the constraints are satisfied by choosing $f(x) = (\eps x)^{C/(\eps x)}$ for a large enough constant $C$ and $k = \left\lceil\log_{\eps}\left(\frac{\gamma}{36}\right)\right\rceil$.

We proved that there exists a container packing with $O_\eps(1)$ containers with height at most $(\frac{4}{3}+O(\eps))OPT$. Finally, by using the algorithm provided by Lemma~\ref{thm:container_packing_ptas_ppt}, we obtain the main result of this chapter:

\begin{theorem}
	There is a PPT $(\frac{4}{3}+O(\eps))$-approximation algorithm for strip packing, with or without $90^\circ$ rotations.
\end{theorem}

\section{Conclusions}
\label{sec:SP-conclusions}

By refining the techniques of \cite{nw16}, we obtained a Pseudo-Polynomial Time approximation algorithm for Strip Packing that has an approximation ratio of $\frac{4}{3} + O(\eps)$; moreover, we argued that this is the best possible result in this framework, and further improvements would require a substantially different approach.

It is still a very interesting open problem to close the gap between the lower bound of $5/4$ provided by \cite{hjrs17} and the current upper bound of $\frac{4}{3} + O(\eps)$ for the Pseudo-Polynomial Time approximability.

Similarly, there is still a gap for the polynomial time approximability of the problem, with the recent upper bound of $5/3 + \eps$ from \cite{hjpv14}, and the classical lower bound of $3/2$. Closing this gap is a challenging open problem whose solution is required to complete our understanding of the Strip Packing problem.

\chapter[Approximations for 2DGK without Rotations]{Approximations for 2DGK\\without Rotations}
\label{chap:2dgk-norot}

In this chapter, we will discuss our improved approximation algorithms for the \tdk problem, in the setting when the rectangles \emph{cannot} be rotated. We describe the first algorithm that breaks the $2$-approximation barrier for \tdk.

All prior polynomial-time approximation algorithms for \tdk~implicitly or explicitly exploit an approach similar to our container packings. The idea is to partition the knapsack into a constant number of axis-aligned rectangular regions (\emph{containers}). The sizes (and therefore positions) of these containers can be \emph{guessed} in polynomial time. Then items are packed inside the containers in a simple way: either one next to the other from left to right or from bottom to top (similarly to the one-dimensional case), or by means of the simple greedy Next-Fit-Decreasing-Height algorithm.
Indeed, also the QPTAS in~\cite{aw15} can be cast in this framework, with the relevant difference that the number of containers in this case is poly-logarithmic (leading to a quasi-polynomial running time).

One of the major bottlenecks to achieve approximation factors better than $2$ is that items that are high and narrow (\emph{vertical} items) and items that are wide and thin (\emph{horizontal}
items) can interact in a very complicated way. Indeed, consider the following seemingly simple \fontL\emph{-packing} problem: we are given a set of items $i$ with either $\width(i)>N/2$ (horizontal items) or $\height(i)>N/2$ (vertical items). 
Our goal is to pack a maximum profit subset of them inside an $\fontL$-shaped region $\fontL=([0,N]\times [0,\height_\fontL]) \cup ([0,\width_\fontL]\times [0,N])$, so that horizontal (resp. vertical) items are packed in the bottom-right (resp. top-left) of $\fontL$. 

To the best of our knowledge, the best-known approximation ratio for L-packing is $2+\eps$, obtained by the following simple algorithm: remove either all vertical or all horizontal items, and then pack the remaining items by a simple reduction to one-dimensional knapsack (for which an FPTAS is known).
It is unclear whether a container-based packing can achieve a better approximation factor, and we conjecture that this is not the case.

In order to obtain our results, we substantially deviate for the first time from
\emph{pure} container-based packings. Instead, we consider \emph{L\&C-packings} that combine $O_\eps(1)$ containers \emph{plus} one L-packing of the above type (see Figure~\ref{fig:packing+ring}.(a)), and show that one such packing has large enough profit.

While it is easy to pack almost optimally items into containers, the mentioned $2+\eps$ approximation for L-packings is not sufficient to achieve altogether a better than $2$ approximation factor: indeed, the items of the L-packing might carry all the profit! Our main algorithmic contribution is a PTAS for the L-packing problem. 
It is easy to solve this problem optimally in pseudo-polynomial time $(Nn)^{O(1)}$ by means of dynamic programming. We show that a $1+\eps$ approximation can be obtained by restricting the top (resp. right) coordinates of horizontal (resp. vertical) items to a proper set that can be
computed in polynomial time $n^{O_\eps(1)}$. Given that, one can adapt the above dynamic program to run in polynomial time.

\begin{theorem}\label{thm:main:Lpacking}
	There is a PTAS for the L-packing problem. 
\end{theorem}

In order to illustrate the power of our approach, we next sketch a simple $\frac{16}{9}+O(\eps)$ approximation for the cardinality case of \tdk (details in Section~\ref{sec:tdk_car:simple}). By standard arguments\footnote{There can be at most $O_\eps(1)$ such items in any feasible solution, and if the optimum solution contains only $O_\eps(1)$ items we can solve the problem optimally by brute force.} it is possible to discard \emph{large} items with both sides longer than $\eps\cdot N$. The remaining items have height or width smaller than $\eps\cdot N$ (\emph{horizontal} and \emph{vertical} items, resp.). Let us delete all items intersecting a random vertical or horizontal strip of width $\eps\cdot N$ inside the knapsack. We can pack the remaining items into $O_\eps(1)$ containers by Lemma~\ref{lem:structural_lemma_augm} and Theorem~\ref{thm:container_packing_ptas}. A vertical strip deletes vertical items with $O(\eps)$ probability, and horizontal ones with probability roughly proportional to their width, and symmetrically for a horizontal strip. In particular, let us call \emph{long} the items with longer side larger than $N/2$, and \emph{short} the remaining items. Then the above argument gives in expectation roughly one half of the profit $opt_{long}$ of long items, and three quarters of the profit $opt_{short}$ of short ones. This is already good enough unless $opt_{long}$ is large compared to $opt_{short}$.

At this point L-packings and our PTAS come into play. We shift long items such that they form $4$ stacks
at the sides of the knapsack in a \emph{ring-shaped} region, see Figure~\ref{fig:packing+ring}.(b)--(c): this is possible since any vertical long item cannot have
a horizontal long item \emph{both} at its left and at its right, and vice
versa. Next we delete the least profitable of these stacks and rearrange
the remaining long items into an L-packing, see Figure~\ref{fig:packing+ring}.(d). Thus using our PTAS for L-packings, we can compute a solution of profit roughly three quarters of $opt_{long}$. The reader might check that the combination of these two algorithms gives the claimed approximation factor.

Above we used either $O_\eps(1)$ containers or one L-packing: by combining the two approaches together and with a more sophisticated case analysis we achieve the following result (see Section \ref{sec:tdk_car:refined}).

\begin{theorem}\label{trm:tdk_car:refined} There is a polynomial-time
	$\frac{558}{325}+\eps<1.72$ approximation algorithm for cardinality \tdk.
\end{theorem}

\begin{figure}
	\begin{centering}
		\includegraphics[height=3.5cm]{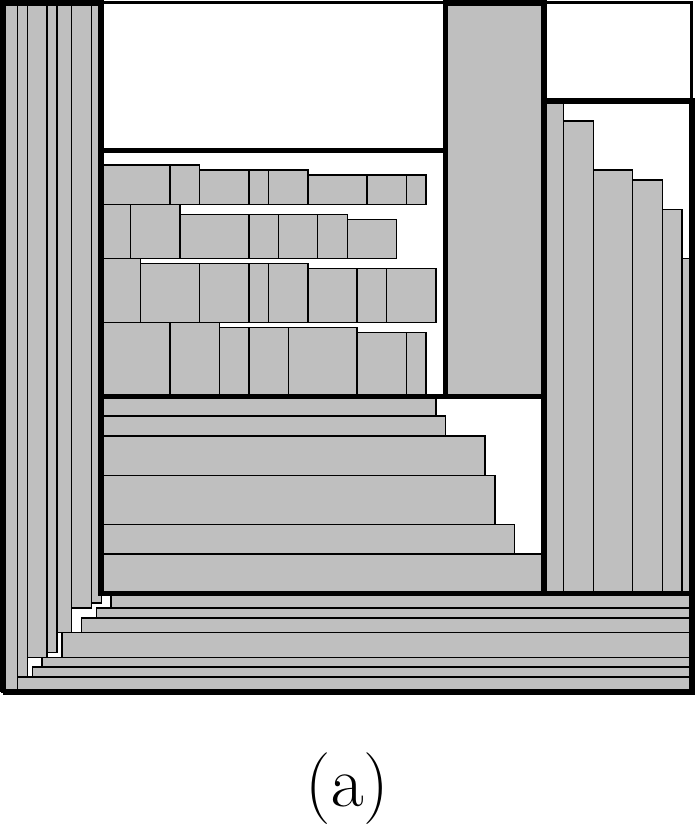}~~~~~~\includegraphics[height=3.5cm]{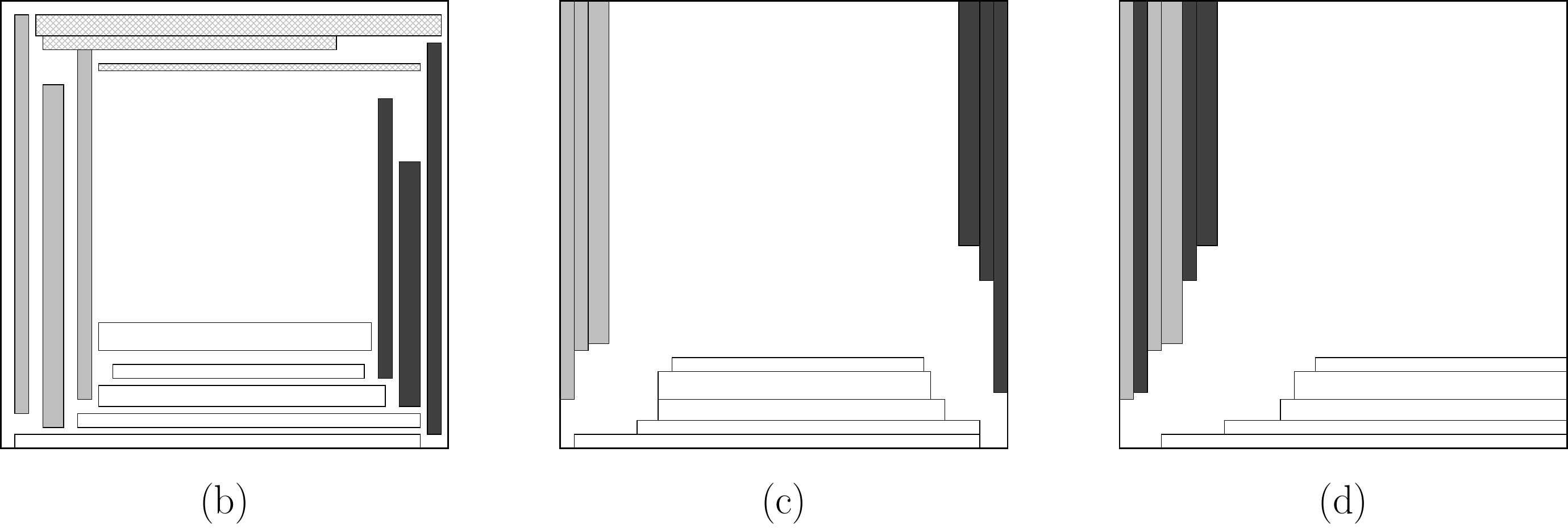}
		\par\end{centering}
	\caption{(a) An L\&C-packing with $4$ containers, where the top-left container is packed by means of Next-Fit-Decreasing-Height. (b) A subset of long items. (c) Such items are shifted into $4$ stacks at the sides of the knapsack, and the top stack is deleted. (d) The final packing into an L-shaped region.}
	\label{fig:packing+ring} 
\end{figure}

For weighted \tdk we face severe technical complications for proving
that there is a profitable L\&C-packing. One key reason is that in the weighted case we cannot discard large items since even one such item might contribute a large fraction to the optimal profit.

In order to circumvent these difficulties, we
exploit the \emph{corridor-partition} at the heart of the QPTAS for
\tdk in \cite{aw15} (in turn inspired by prior
work in \cite{aw13}). Roughly speaking, there exists a partition
of the knapsack into $O_{\eps}(1)$ \emph{corridors}, consisting of the \emph{concatenation} of $O_{\eps}(1)$ (partially overlapping) rectangular regions (\emph{subcorridors}).

In \cite{aw15} the authors partition
the corridors into a \emph{poly-logarithmic} number of containers. Their main algorithm then guesses these containers
in time $n^{\polylog n}$. However, we can only handle a \emph{constant} number of containers in polynomial time.
Therefore, we present a different way to partition the corridors into
containers: here we lose the profit of a set of \emph{thin} items,
which in some sense play the role of long items in the previous discussion.
These thin items fit in a \emph{very narrow} ring at the boundary of the
knapsack and we map them to an \fontL-packing in the same way as
in the cardinality case above. Some of the remaining non-thin items
are then packed into $O_{\eps}(1)$ containers that are placed in
the (large) part of the knapsack not occupied by the \fontL-packing.
Our partition of the corridors is based on a somewhat intricate case
analysis that exploits the fact that \emph{long} consecutive subcorridors are arranged in the shape of \emph{rings} or \emph{spirals}: this is used to show the existence of a profitable L\&C-packing.

\begin{theorem}\label{trm:tdk_weight} There is a polynomial-time
	$\frac{17}{9}+\eps<1.89$ approximation algorithm for (weighted) \tdk.
\end{theorem}

\section{A PTAS for \fontL-packings}
\label{sec:ptasL}

In this section we present a PTAS for the problem of finding an optimal \fontL-packing. In this problem we are given a set of \emph{horizontal} items $I_{hor}$ with width larger than $N/2$, and a set of \emph{vertical} items $I_{ver}$ with height larger than $N/2$. Furthermore, we are given an \fontL-shaped region $\fontL=([0,N]\times [0,\height_\fontL]) \cup ([0,\width_\fontL]\times [0,N])$. Our goal is to pack a  subset $OPT\subseteq I:=I_{hor}\cup I_{ver}$ of maximum total profit $opt=\profit(OPT):=\sum_{i\in OPT}\profit(i)$, such that $OPT_{hor}:=OPT\cap I_{hor}$ is packed inside the \emph{horizontal box} $[0,N]\times [0,\height_\fontL]$ and $OPT_{ver}:=OPT\cap I_{ver}$ is packed inside the \emph{vertical box} $[0,\width_\fontL]\times [0,N]$. We remark that packing horizontal and vertical items independently is not possible due to the possible overlaps in the intersection of the two boxes: this is what makes this problem non-trivial, in particular harder than standard (one-dimensional) knapsack.

Observe that in an optimal packing we can assume without loss of generality that items in $OPT_{hor}$ are pushed as far to the right/bottom as possible: indeed pushing one such item to the right/bottom keeps the packing feasible. Furthermore, the items in $OPT_{hor}$ are packed from bottom to top in non-increasing order of width. Indeed, it is possible to permute any two items violating this property while keeping the packing feasible. A symmetric claim holds for $OPT_{ver}$. See Figure~\ref{fig:packing+ring}.(d) for an illustration.

Given the above structure, it is relatively easy to define a dynamic program (DP) that computes an optimal L-packing in pseudo-polynomial time $(Nn)^{O(1)}$. The basic idea is to scan items of $I_{hor}$ (resp. $I_{ver}$) in decreasing order of width (resp. height), and each time \emph{guess} if they are part of the optimal solution $OPT$. At each step either both the considered horizontal item $i$ and vertical item $j$ are not part of the optimal solution, or there exist a \emph{guillotine cut} separating $i$ or $j$ from the rest of $OPT$. Depending on the cases, one can define a smaller L-packing sub-instance (among $N^2$ choices) for which the DP table already contains a solution.

In order to achieve a $(1+\eps)$-approximation in polynomial time $n^{O_\eps(1)}$, we show that it is possible (with a small loss in the profit) to restrict the possible top coordinates of $OPT_{hor}$ and right coordinates of $OPT_{ver}$ to proper polynomial-size subsets $\cT$ and $\cR$, resp. We call such an L-packing \emph{$(\cT,\cR)$-restricted}. By adapting the above DP one obtains:   

\begin{lemma}\label{lem:DPrestricted}
	An optimal $(\cT,\cR)$-restricted \fontL-packing can be computed in time polynomial in $m:=n+|\cT|+|\cR|$ using dynamic programming.
\end{lemma} 
\begin{proof}
	For notational convenience we assume $0\in \cT$ and $0\in \cR$.
	Let $h_1,\ldots,h_{n(h)}$ be the items in $I_{hor}$ in decreasing order of width and $v_1,\ldots,v_{n(v)}$ be the items in $I_{ver}$ in decreasing order of height (breaking ties arbitrarily). For $\width\in [0,\width_\fontL]$ and $\height\in [0,\height_\fontL]$, let $L(\width,\height)=([0,\width]\times [0,N])\cup ([0,N]\times [0,\height])\subseteq \fontL$. Let also $\Delta \fontL(\width,\height)=([\width,\width_\fontL]\times [\height,N])\cup ([\width,N]\times [\height,\height_\fontL])\subseteq \fontL$. Note that $\fontL=\fontL(\width,\height)\cup \Delta \fontL(\width,\height)$.

	We define a dynamic program table $DP$ indexed by $i\in [1,n(h)]$ and $j\in [1,n(v)]$, by a top coordinate $t \in \cT$, and a right coordinate $r \in \cR$. The value of $DP(i,t,j,r)$ is the maximum profit of a $(\cT,\cR)$-restricted packing of a subset of $\{h_i,\ldots,h_{n(h)}\}\cup \{v_j,\ldots,v_{n(v)}\}$ inside $\Delta \fontL(r,t)$. The value of $DP(1,0,1,0)$ is the value of the optimum solution we are searching for. 
	Note that the number of table entries is upper bounded by $m^4$.
	
	We fill in $DP$ according to the partial order induced by vectors $(i,t,j,r)$, processing larger vectors first. The base cases are given by 
	$(i,j)=(n(h)+1,n(v)+1)$ and $(r,t)=(\width_\fontL,\height_\fontL)$, in which case the table entry has value $0$.
	
	In order to compute any other table entry $DP(i,t,j,r)$, with optimal solution $OPT'$, we take the maximum of the following few values:
	\begin{itemize}\itemsep0pt
		\item[$\bullet$] If $i\leq n(h)$, the value $DP(i+1,t,j,r)$. This covers the case that $h_i\notin OPT'$;
		\item[$\bullet$] If $j\leq n(v)$, the value $DP(i,t,j+1,r)$. This covers the case that $v_j\notin OPT'$;
		\item[$\bullet$] Assume that there exists $t'\in \cT$ such that $t'-\height(h_i)\geq t$ and that $\width(h_i)\leq N-r$. Then for the minimum such $t'$ we consider the value $\profit(h_i)+DP(i+1,t',j,r)$. This covers the case that $h_i\in OPT'$, and there exists a (horizontal) guillotine cut separating $h_i$ from $OPT'\setminus \{h_i\}$.
		\item[$\bullet$] Assume that there exists $r'\in \cR$ such that $r'-\width(v_j)\geq r$ and that $\height(v_j)\leq N-t$. Then for the minimum such $r'$ we consider the value $\profit(v_j)+DP(i,t,j+1,r')$. This covers the case that $v_j\in OPT'$, and there exists a (vertical) guillotine cut separating $v_j$ from $OPT'\setminus \{v_j\}$.
	\end{itemize}
	We observe that the above cases (which can be explored in polynomial time) cover all the possible configurations in $OPT'$. Indeed, if the first two cases do not apply, we have that $h_i,v_j\in OPT'$. Then either the line containing the right side of $v_j$ does not intersect $h_i$ (hence any other item in $OPT'$) or the line containing the top side of $h_i$ does not intersect $v_j$ (hence any other item in $OPT'$). Indeed, the only remaining case is that $v_j$ and $h_i$ overlap, which is impossible since they both belong to $OPT'$.  
\end{proof}

We will show that there exists a $(\cT,\cR)$-restricted \fontL-packing with the desired properties. 
\begin{lemma}\label{lem:Lpacking:structural}
	There exists an $(\cT,\cR)$-restricted \fontL-packing solution of profit at least $(1-2\eps)opt$, where the sets $\cT$ and $\cR$ have cardinality at most $n^{O(1/\eps^{1/\eps})}$ and can be computed in polynomial time based on the input (without knowing $OPT$).
\end{lemma}
Lemmas \ref{lem:DPrestricted} and \ref{lem:Lpacking:structural}, together immediately imply a PTAS for L-packings (showing Theorem \ref{thm:main:Lpacking}). The rest of this section is devoted to the proof of Lemma \ref{lem:Lpacking:structural}.

\begin{figure}
	\begin{centering}
		\includegraphics[height=5cm]{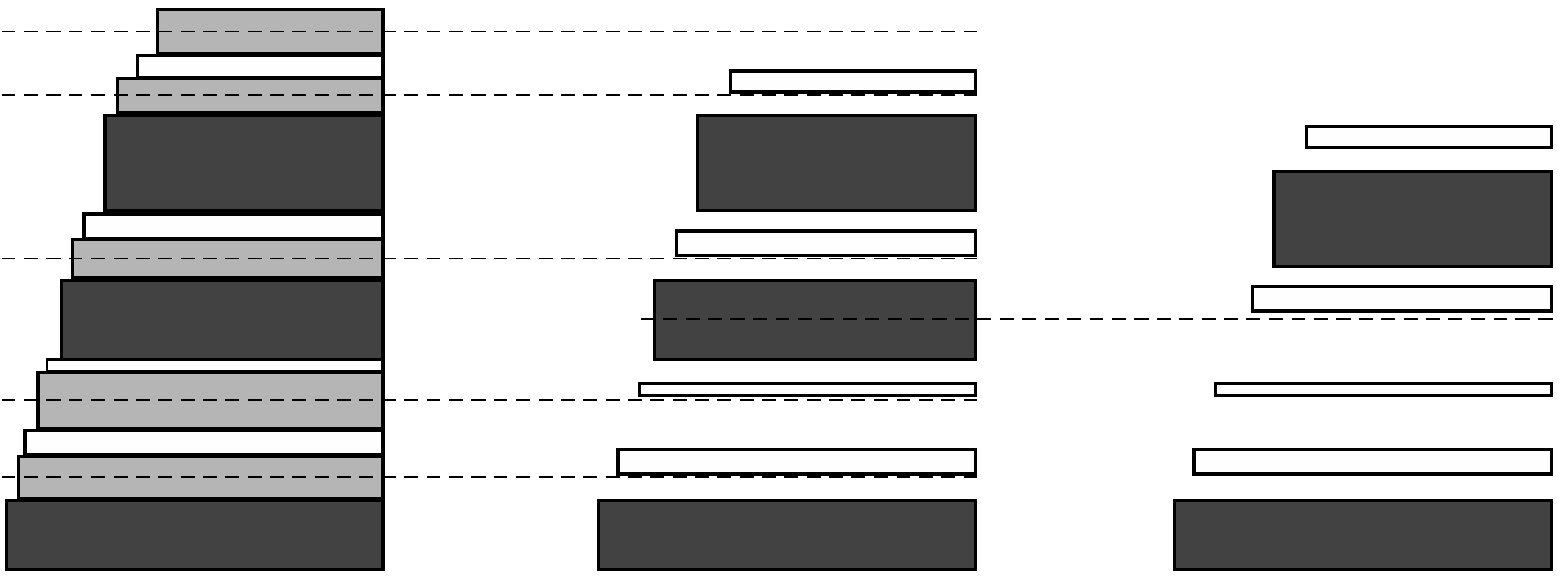}%
		\par\end{centering}
	\caption{Illustration of the {\tt delete\&shift} procedure with $r_{hor}=2$. The dashed lines indicate the value of the new baselines in the different stages of the algorithm. (Left) The starting packing. Dark and light gray items denote the growing sequences for the calls with $r=2$ and $r=1$, resp. (Middle) The shift of items at the end of the recursive calls with $r=1$. Note that light gray items are all deleted, and dark gray items are not shifted. (Right) The shift of items at the end of the process. Here we assume that the middle dark gray item is deleted.}\label{fig:Lpacking}
\end{figure}

We will describe a way to delete a subset of items $D_{hor}\subseteq OPT_{hor}$ with $\profit(D_{hor})\leq 2\eps\profit(OPT_{hor})$, and \emph{shift down} the remaining items $OPT_{hor}\setminus D_{hor}$ so that their top coordinate belongs to a set $\cT$ with the desired properties. Symmetrically, we will delete a subset of items $D_{ver}\subseteq OPT_{ver}$ with $\profit(D_{ver})\leq 2\eps\profit(OPT_{ver})$, and \emph{shift to the left} the remaining items $OPT_{ver}\setminus D_{ver}$ so that their right coordinate belongs to a set $\cR$ with the desired properties. We remark that shifting down (resp. to the left) items of $OPT_{hor}$ (resp. $OPT_{ver}$) cannot create any overlap with items of $OPT_{ver}$ (resp. $OPT_{hor}$). This allows us to reason on each such set separately. 

We next focus on $OPT_{hor}$ only: the construction for $OPT_{ver}$ is symmetric. For notational convenience we let $1,\ldots,n_{hor}$ be the items of $OPT_{hor}$ in non-increasing order of width \emph{and} from bottom to top in the starting optimal packing. We remark that this sequence is not necessarily sorted (increasingly or decreasingly) in terms of item heights: this makes our construction much more complicated.

Let us first introduce some useful notation. Consider any subsequence $B=\{b_{start},\ldots,b_{end}\}$ of consecutive items (\emph{interval}). For any $i\in B$, we define $\topc_B(i):=\sum_{k\in B,k\leq i}\height(k)$ and $\bottomc_B(i)=\topc_B(i)-\height(i)$. The \emph{growing subsequence} $G=G(B)=\{g_1,\ldots,g_h\}$ of $B$ (with possibly non-contiguous elements) is defined as follows. We initially set $g_1=b_{start}$. Given the element $g_i$, $g_{i+1}$ is the smallest-index (i.e., lowest) element in $\{g_i+1,\ldots,b_{end}\}$ such that $\height(g_{i+1})\geq \height(g_i)$. We halt the construction of $G$ when we cannot find a proper $g_{i+1}$. For notational convenience, define $g_{h+1}=b_{end}+1$. We let $B^G_i:=\{g_i+1,\ldots,g_{i+1}-1\}$ for $i=1,\ldots,h$. Observe that the sets $B^G_i$ partition $B\setminus G$. We will crucially exploit the following simple property.
\begin{lemma}\label{lem:propertiesG}
	For any $g_i\in G$ and any $j\in \{b_{start},\ldots,g_{i+1}-1\}$, $\height(j)\leq \height(g_i)$.
\end{lemma}
\begin{proof}
	The items $j\in B^G_i=\{g_i+1,\ldots,g_{i+1}-1\}$ have $\height(j)<\height(g_i)$. Indeed, any such $j$ with $\height(j)\geq \height(g_i)$ would have been added to $G$, a contradiction. 
	
	Consider next any $j\in \{b_{start},\ldots g_i-1\}$. If $j\in G$ the claim is trivially true by construction of $G$. Otherwise, one has $j \in B^G_k$ for some $g_k\in G$, $g_k<g_i$. Hence, by the previous argument  and by construction of $G$, $\height(j)<\height(g_k)\leq \height(g_i)$.
\end{proof}

The intuition behind our construction is as follows. Consider the growing sequence $G=G(OPT_{hor})$, and suppose that $\profit(G)\leq \eps \cdot \profit(OPT_{hor})$. Then we might simply delete $G$, and shift the remaining items $OPT_{hor}\setminus G=\cup_j B^G_j$ as follows. Let $\lceil x\rceil_y$ denote the smallest multiple of $y$ not smaller than $x$. We consider each set $B^G_j$ separately. For each such set, we define a baseline vertical coordinate $\base_j=\lceil \bottomc(g_j)\rceil_{\height(g_j)/2}$, where $\bottomc(g_j)$ is the bottom coordinate of $g_j$ in the original packing. We next round up the height of $i\in B^G_j$ to $\hat{\height}(i):=\lceil \height(i)\rceil_{\height(g_j)/(2n)}$, and pack the rounded items of $B^G_j$ as low as possible above the baseline. The reader might check that the possible top coordinates for rounded items fall in a polynomial size set (using Lemma \ref{lem:propertiesG}). It is also not hard to check that items are \emph{not} shifted up.  

We use recursion in order to handle the case $\profit(G)> \eps \cdot \profit(OPT_{hor})$. Rather than deleting $G$, we consider each $B^G_j$ and build a new growing subsequence for each such set. We repeat the process recursively for $r_{hor}$ many rounds. Let ${\cal G}^r$ be the union of all the growing subsequences in the recursive calls of level $r$. Since the sets ${\cal G}^r$ are disjoint by construction, there must exist a value $r_{hor}\leq \frac{1}{\eps}$ such that $\profit({\cal G}^{r_{hor}})\leq \eps\cdot \profit(OPT_{hor})$. Therefore we can apply the same shifting argument to all growing subsequences of level $r_{hor}$ (in particular we delete all of them). In the remaining growing subsequences we can afford to delete $1$ out of $1/\eps$ consecutive items (with a small loss of the profit), and then apply a similar shifting argument. \\

We next describe our approach in more detail. 
We exploit a recursive procedure {\tt delete\&shift}. This procedure takes in input two parameters: 
an interval $B=\{b_{start},\ldots,b_{end}\}$, and an integer \emph{round parameter} $r\geq 1$. 
Procedure {\tt delete\&shift} returns a set $D(B)\subseteq B$ of deleted items, and a 
shift function $\shift:B\setminus D(B)\rightarrow \mathbb{N}$. Intuitively, $\shift(i)$ is the value of the top coordinate of $i$ in the shifted packing with respect to a proper baseline value which is implicitly defined.  
We initially call {\tt delete\&shift}$(OPT_{hor},r_{hor})$, for a proper $r_{hor}\in \{1,\ldots,\frac{1}{\eps}\}$ to be fixed later. Let $(D,\shift)$ be the output of this call. The desired set of deleted items is $D_{hor}=D$, and in the final packing $\topc(i)=\shift(i)$ for any $i\in OPT_{hor}\setminus D_{hor}$ (the right coordinate of any such $i$ is $N$).

The procedure behaves differently in the cases $r=1$ and $r>1$.
If $r=1$, we compute the growing sequence $G=G(B)=\{g_1=b_{start},\ldots,g_h\}$, and set $D(B)=G(B)$. Considers any set $B^G_j=\{g_{j}+1,\ldots,g_{j+1}-1\}$, $j=1,\ldots,h$. Let 
$\base_j:= \lceil \bottomc_{B}(g_j) \rceil_{\height(g_j)/2}$. We define for any $i\in B^G_j$,
$$
\shift(i)= \base_j+\sum_{k\in B^G_j,k\leq i}\lceil \height(k) \rceil_{\height(g_j)/(2n)}.
$$
Observe that $\shift$ is fully defined since $\cup_{j=1}^{h}B^G_j=B\setminus D(B)$.

If instead $r>1$, we compute the growing sequence $G=G(B)=\{g_1=b_{start},\ldots,g_h\}$. We next delete a subset of items $D'\subseteq G$. If $h<\frac{1}{\eps}$, we let $D'=D'(B)=\emptyset$. Otherwise, let $G_k=\{g_j\in G: j = k \pmod{1/\eps}\}\subseteq G$, for $k\in \{0,\ldots,1/\eps-1\}$. We set $D'=D'(B)=\{d_1,\ldots,d_p\}=G_{x}$ where $x=\arg\min_{k\in \{0,\ldots,1/\eps-1\}}\profit(G_k)$. 
\begin{proposition}\label{pro:deleted}
	One has $\profit(D')\leq \eps\cdot \profit(G)$. Furthermore, any subsequence $\{g_x,g_{x+1},\ldots,g_y\}$ of $G$ with at least $1/\eps$ items contains at least one item from $D'$.
\end{proposition}

Consider each set $B^G_j=\{g_{j}+1,\ldots,g_{j+1}-1\}$, $j=1,\ldots,h$: We run {\tt delete\&shift}$(B^G_j,r-1)$. Let $(D_j,\shift_j)$ be the output of the latter procedure, and $\shift^{max}_j$ be the maximum value of $\shift_j$. We set the output set of deleted items to $D(B)=D'\cup (\cup_{j=1}^{h}D_j)$. 

It remains to define the function $\shift$. Consider any set $B^G_j$, and let $d_q$ be the deleted item in $D'$ with largest index (hence in topmost position) in $\{b_{start},\ldots,g_{j}\}$: define $\base_q = \lceil \bottomc_B(d_q)\rceil_{\height(d_q)/2}$. If there is no such $d_q$, we let $d_q=0$ and $\base_q=0$. For any $i\in B^G_j$ we set: 
$$
\shift(i) = \base_q + \sum_{g_k\in G,d_q< g_k\leq g_j}\height(g_k)+\sum_{g_k\in G,d_q\leq g_k< g_j}\shift^{max}_k+\shift_j(i).
$$ 
Analogously, if $g_j\neq d_q$, we set
$$
\shift(g_j) = \base_q + \sum_{g_k\in G,d_q< g_k\leq g_j}\height(g_k)+\sum_{g_k\in G,d_q\leq g_k< g_j}\shift^{max}_k.
$$ 
This concludes the description of {\tt delete\&shift}. We next show that the final packing has the desired properties. Next lemma shows that the total profit of deleted items is small for a proper choice of the starting round parameter $r_{hor}$.
\begin{lemma}\label{lem:Lpacking:costDeleted}
	There is a choice of $r_{hor}\in \{1,\ldots,\frac{1}{\eps}\}$ such that the final set $D_{hor}$ of deleted items satisfies $\profit(D_{hor})\leq 2\eps\cdot \profit(OPT_{hor})$.
\end{lemma}
\begin{proof}
	Let ${\cal G}^r$ denote the union of the sets $G(B)$ computed by all the recursive calls with input round parameter $r$. Observe that by construction these sets are disjoint.
	Let also ${\cal D}^r$ be the union of the sets $D'(B)$ on those calls (the union of sets $D(B)$ for $r=r_{hor}$). By Proposition \ref{pro:deleted} and the disjointness of sets ${\cal G}^r$ one has
	$$
	\profit(D_{hor})=\sum_{1\leq r\leq r_{hor}}\profit({\cal D}^r)\leq \eps\cdot \sum_{r< r_{hor}}\profit({\cal G}^r)+\profit({\cal D}^{r_{hor}})\leq \eps\cdot \profit(OPT_{hor})+\profit({\cal D}^{r_{hor}}).
	$$ 
	Again by the disjointness of sets ${\cal G}^r$ (hence ${\cal D}^r$), there must exists a value of $r_{hor}\in \{1,\ldots,\frac{1}{\eps}\}$ such that $\profit({\cal D}^{r_{hor}})\leq \eps\cdot \profit(OPT_{hor})$. The claim follows. 
\end{proof}
Next lemma shows that, intuitively, items are only shifted down with respect to the initial packing.
\begin{lemma}\label{lem:Lpacking:shiftDown}
	Let $(D,\shift)$ be the output of some execution of the algorithm {\tt delete\&shift}$(B,r)$. Then, for any $i\in B\setminus D$, $\shift(i)\leq \topc_B(i)$.
\end{lemma}
\begin{proof}
	We prove the claim by induction on $r$. Consider first the case $r=1$. In this case, for any $i\in B^G_j$:
	\begin{align*}
	\shift(i) & =  \lceil \bottomc_{B}(g_j) \rceil_{\height(g_j)/2}+\sum_{k\in B^G_j,k\leq i}\lceil \height(k) \rceil_{\height(g_j)/(2n)} \\
	& \leq \topc_{B}(g_j) - \frac{1}{2}\height(g_j)+\sum_{k\in B^G_j,k\leq i}\height(k)+n\cdot \frac{\height(g_j)}{2n} = \topc_B(i).
	\end{align*}
	Assume next that the claim holds up to round parameter $r-1\geq 1$, and consider round $r$. For any $i\in B^G_j$ with $\base_q = \lceil \bottomc_B(d_q)\rceil_{\height(d_q)/2}$, one has
	\begin{align*}
	\shift(i) & = \lceil \bottomc_B(d_q)\rceil_{\height(d_q)/2} + \sum_{\substack{g_k\in G,\\d_q< g_k\leq g_j}}\height(g_k)+\sum_{\substack{g_k\in G,\\d_q\leq g_k< g_j}}\shift^{max}_k+\shift_j(i) \\
	& \leq \topc_B(d_q)+\sum_{\substack{g_k\in G,\\d_q< g_k\leq g_j}}\height(g_k)+\sum_{\substack{g_k\in G,\\d_q\leq g_k< g_j}}\topc_{B^G_k}(g_{k+1}-1)+\topc_{B^G_j}(i)\\
	& =\topc_B(i).
	\end{align*}
	An analogous chain of inequalities shows that $\shift(g_j)\leq \topc_B(g_j)$ for any $g_j\in G\setminus D'$. A similar proof works for the special case $\base_q=0$.
\end{proof}

It remains to show that the final set of values of $\topc(i)=\shift(i)$ has the desired properties. This is the most delicate part of our analysis. We define a set $\cT^r$ of candidate top coordinates recursively in $r$. Set $\cT^1$ contains, for any item $j\in I_{hor}$, and any integer $1\leq a\leq 4n^2$, the value $a\cdot \frac{\height(j)}{2n}$. Set $\cT^r$, for $r>1$ is defined recursively with respect to to $\cT^{r-1}$. For any item $j$, any
integer $0\leq a\leq 2n-1$, any tuple of $b\leq 1/\eps-1$ items $j(1),\ldots,j(b)$, and any tuple of $c\leq 1/\eps$ values $s(1),\ldots,s(c)\in \cT^{r-1}$, $\cT^r$ contains the sum $a\cdot \frac{\height(j)}{2}+\sum_{k=1}^{b}\height(j(k))+\sum_{k=1}^{c}s(k)$. Note that sets $\cT^r$ can be computed based on the input only (without knowing $OPT$). It is easy to show that $\cT^r$ has polynomial size for $r=O_\eps(1)$.
\begin{lemma}\label{lem:sizeTr}
	For any integer $r\geq 1$, $|\cT^r|\leq (2n)^{\frac{r+2+(r-1)\eps}{\eps^{r-1}}}$.
\end{lemma}
\begin{proof}
	We prove the claim by induction on $r$. The claim is trivially true for $r=1$ since there are $n$ choices for item $j$ and $4n^2$ choices for the integer $a$, hence altogether at most $n\cdot 4n^2<8n^3$ choices. For $r>1$, the number of possible values of $\cT^r$ is at most

	\begin{align*}
		n\cdot 2n \cdot \left(\sum_{b=0}^{1/\eps-1}n^b\right)\cdot \left(\sum_{c=0}^{1/\eps}|\cT^{r-1}|^c\right) &\leq 4n^2\cdot n^{\frac{1}{\eps}-1}\cdot |\cT^{r-1}|^{\frac{1}{\eps}}\\
		&\leq (2n)^{\frac{1}{\eps}+1}\left((2n)^{\frac{r+1+(r-2)\eps}{\eps^{r-2}}}\right)^{\frac{1}{\eps}}\\
		&\leq (2n)^{\frac{r+2+(r-1)\eps}{\eps^{r-1}}}. 
	\end{align*}
\end{proof}

The next lemma shows that the values of $\shift$ returned by {\tt delete\&shift} for round parameter $r$ belong to $\cT^r$, hence the final top coordinates belong to $\cT:=\cT^{r_{hor}}$.
\begin{lemma}\label{lem:Lpacking:possibleHeights}
	Let $(D,\shift)$ be the output of some execution of the procedure {\tt delete\&shift}$(B,r)$. Then, for any $i\in B\setminus D$, $\shift(i)\in \cT^r$.
\end{lemma}
\begin{proof}
	We prove the claim by induction on $r$. For the case $r=1$, recall that for any $i\in B^G_j$ one has 
	\begin{align*}
	\shift(i) & =  \lceil \bottomc_{B}(g_j) \rceil_{\height(g_j)/2}+\sum_{k\in B^G_j,k\leq i}\lceil \height(k) \rceil_{\height(g_j)/(2n)}.
	\end{align*}
	By Lemma \ref{lem:propertiesG}, $\bottomc_B(g_j)=\sum_{k\in B,k<g_j}\height(k)\leq (n-1)\cdot \height(g_j)$. By the same lemma, $\sum_{k\in B^G_j,k\leq i} \height(k)\leq (n-1)\cdot \height(g_j)$. It follows that 
	$$
	\shift(i)\leq 2(n-1)\cdot \height(g_j)+\frac{\height(g_j)}{2}+(n-1)\cdot \frac{\height(g_j)}{2n}\leq 4n^2\cdot \frac{\height(g_j)}{2n}.
	$$
	Hence $\shift(i)=a\cdot \frac{\height(g_j)}{2n}$ for some integer $1\leq a\leq 4n^2$, and $\shift(i)\in \cT^1$ for $j=g_j$ and for a proper choice of $a$.
	
	Assume next that the claim is true up to $r-1\geq 1$, and consider the case $r$. Consider any $i\in B^G_j$, and assume $0<\base_q = \lceil \bottomc_B(d_q)\rceil_{\height(d_q)/2}$. One has:
	\begin{align*}
	\shift(i) & = \lceil \bottomc_B(d_q)\rceil_{\height(d_q)/2} + \sum_{\substack{g_k\in G,\\d_q< g_k\leq g_j}}\height(g_k)+\sum_{\substack{g_k\in G,\\d_q\leq g_k< g_j}}\shift^{max}_k+\shift_j(i) .
	\end{align*}
	By Lemma~\ref{lem:propertiesG}, $\bottomc_B(d_q)\leq (n-1)\height(d_q)$, therefore $\lceil \bottomc_B(d_q)\rceil_{\height(d_q)/2}=a\cdot \frac{\height(d_q)}{2}$ for some integer $1\leq a\leq 2(n-1)+1$. By Proposition~\ref{pro:deleted}, $|\{g_k\in G,d_q< g_k\leq g_j\}|\leq 1/\eps-1$. Hence 
	$\sum_{g_k\in G,d_q< g_k\leq g_j}\height(g_k)$ is a value contained in the set of sums of $b\leq 1/\eps-1$ item heights.  By inductive hypothesis $\shift^{max}_k,\shift_j(i)\in \cT^{r-1}$. Hence by a similar argument the value of $\sum_{g_k\in G,d_q\leq g_k< g_j}\shift^{max}_k+\shift_j(i)$ is contained in the set of sums of $c\leq 1/\eps-1+1$ values taken from $\cT^{r-1}$. 
	Altogether, $\shift(i)\in \cT^r$. A similar argument, without the term $\shift_j(i)$,  shows that $\shift(g_j)\in \cT^r$ for any $g_j\in G\setminus D'$. The proof works similarly in the case $\base_q=0$ by setting $a=0$. The claim follows.
\end{proof}

\begin{proof}[Proof of Lemma \ref{lem:Lpacking:structural}]
	We apply the procedure {\tt deleted\&shift} to $OPT_{hor}$ as described before, and a symmetric procedure to $OPT_{ver}$. In particular the latter procedure computes a set $D_{ver}\subseteq OPT_{ver}$ of deleted items, and the remaining items are shifted to the left so that their right coordinate belongs to a set $\cR:=\cR^{r_{ver}}$, defined analogously to the case of $\cT:=\cT^{r_{hor}}$, for some integer $r_{ver}\in \{1,\ldots,1/\eps\}$ (possibly different from $r_{hor}$, though by averaging this is not critical).
	
	It is easy to see that the profit of non-deleted items satisfies the claim by Lemma \ref{lem:Lpacking:costDeleted} and its symmetric version. Similarly, the sets 
	$\cT$ and $\cR$ satisfy the claim by Lemmas~\ref{lem:sizeTr} and~\ref{lem:Lpacking:possibleHeights}, and their symmetric versions. Finally, with respect to the original packing non-deleted items in $OPT_{hor}$ and $OPT_{ver}$ can be only shifted to the bottom and to the left, resp. by Lemma \ref{lem:Lpacking:shiftDown} and its symmetric version. This implies that the overall packing is feasible.
\end{proof}

\section{A Simple Improved Approximation for Cardinality \tdk}
\label{sec:tdk_car:simple}

In this section we present a simple improved approximation
for the cardinality case of \tdk. We can assume that the optimal
solution $OPT\subseteq I$ satisfies that $|OPT|\geq1/\eps^{3}$
since otherwise we can solve the problem optimally by brute force
in time $n^{O(1/\eps^{3})}$. Therefore, we can discard from
the input all \emph{large} items with both sides larger than $\eps\cdot N$:
any feasible solution can contain at most $1/\eps^{2}$ such items,
and discarding them decreases the cardinality of $OPT$ at most by
a factor $1+\eps$. Let $OPT$ denote this slightly sub-optimal
solution obtained by removing large items.

We will need the following technical lemma, which also works unchanged for the
weighted case. See also Figure~\ref{fig:packing+ring}.(b)--(d).

\begin{lemma}\label{lem:LoftheRing} Let $H$ and $V$ be given subsets
	of items from some feasible solution with width and height strictly
	larger than $N/2$, resp. Let $\height_{H}$ and $\width_{V}$
	be the total height and width of items of $H$ and $V$, resp.
	Then there exists an $\fontL$-packing of a set $APX\subseteq H\cup V$
	with $\profit(APX)\geq\frac{3}{4}(\profit(H)+\profit(V))$ into the
	area $\fontL=([0,N]\times[0,\height_{H}])\cup([0,\width_{V}]\times[0,N])$.
\end{lemma} 

\begin{proof} Let us consider the packing of $H\cup V$. Consider
	each $i\in H$ that has no $j\in V$ to its top (resp. to its bottom)
	and shift it up (resp. down) until it hits another $i'\in H$ or the
	top (resp, bottom) side of the knapsack. Note that, since $\height(j)>N/2$
	for any $j\in V$, one of the two cases above always applies. We iterate
	this process as long as it is possible to move any such $i$. We perform
	a symmetric process on $V$. At the end of the process all items in
	$H\cup V$ are stacked on the $4$ sides of the knapsack\footnote{It is possible to permute items in the left stack so that items appear from left to right in non-increasing order of height, and symmetrically for the other stacks. This is not crucial for this proof, but we implemented this permutation in Figure~\ref{fig:packing+ring}.(c).}.
	
	Next we remove the least profitable of the $4$ stacks: by a simple permutation
	argument we can guarantee that this is the top or right stack. We
	next discuss the case that it is the top one, the other case being
	symmetric. We show how to repack the remaining items in a boundary
	$\fontL$ of the desired size by permuting items in a proper order.
	In more detail, suppose that the items packed on the left (resp.
	right and bottom) have a total width of $\width_{l}$ (resp. total
	width of $\width_{r}$ and total height of $\height_{b}$). We next
	show that there exists a packing into $\fontL'=([0,N]\times[0,\height_{b}])\cup([0,\width_{l}+\width_{r}]\times[0,N])$.
	We prove the claim by induction. Suppose that we have proved it for
	all packings into left, right and bottom stacks with parameters $\width'_{l}$,
	$\width'_{r}$, and $\height'$ such that $\height'<\height_{b}$
	or $\width'_{l}+\width'_{r}<\width_{l}+\width_{r}$ or $\width'_{l}+\width'_{r}=\width_{l}+\width_{r}$
	and $\width'_{r}<\width_{r}$.
	
	In the considered packing we can always find a guillotine cut $\ell$,
	such that one side of the cut contains precisely one \emph{lonely}
	item among the leftmost, rightmost and bottommost items. Let $\ell$
	be such a cut. First assume that the lonely item $j$ is the bottommost
	one. Then by induction the claim is true for the part above $\ell$
	since the part of the packing above $\ell$ has parameters $\width_{l},\width_{r}$,
	and $\height-\height(j)$. Thus, it is also true for the entire packing.
	A similar argument applies if the lonely item $j$ is the leftmost
	one.
	
	It remains to consider the case that the lonely item $j$ is the rightmost
	one. We remove $j$ temporarily and move \emph{all} other items by
	$\width(j)$ to the right. Then we insert $j$ at the left (in the
	space freed by the previous shifting). By induction, the claim is
	true for the resulting packing since it has parameters $\width_{l}+\width(j)$,
	$\width_{r}-\width(j)$, and $\height$, resp. \end{proof}

For our algorithm, we consider the following three packings.
The first uses an $L$ that occupies the full knapsack, i.e., $\width_{\fontL}=\height_{\fontL}=N$.
Let $OPT_{long}\subseteq OPT$ be the items in $OPT$ with height
or width strictly larger than $N/2$ and define $OPT_{short}=OPT\setminus OPT_{long}$.
We apply Lemma~\ref{lem:LoftheRing} to $OPT_{long}$ and hence obtain
a packing for this $L$ with a profit of at least $\frac{3}{4}w(OPT_{long})$.
We run our PTAS for \fontL-packings from Theorem \ref{thm:main:Lpacking}
on this \fontL, the input consisting of all items in $I$ having
one side longer than $N/2$. Hence we obtain a solution with profit
at least \mbox{$(\frac{3}{4}-O(\eps))w(OPT_{long})$}.

For the other two packings we employ the one-sided resource
augmentation PTAS given from Lemma~\ref{lem:structural_lemma_augm} and Theorem~\ref{thm:container_packing_ptas}. We apply this
algorithm to the slightly reduced knapsacks $[0,N]\times[0,N/(1+\eps)]$
and $[0,N/(1+\eps)]\times[0,N]$ such that in both cases it outputs
a solution that fits in the full knapsack $[0,N]\times[0,N]$ and
whose profit is by at most a factor $1+O(\eps)$ worse than the
optimal solution for the respective reduced knapsacks. Will prove
in Theorem~\ref{thm:16/9-apx} that one of these solutions yields
a profit of at least $(\frac{1}{2}-O(\eps))\profit(OPT)+(\frac{1}{4}-O(\eps))\profit(OPT_{short})$
and hence one of our packings yields a $(\frac{16}{9}+\eps)$-approximation.

\begin{theorem}\label{thm:16/9-apx} There is a $\left(\frac{16}{9}+\eps\right)$-approximation
	for the cardinality case of \tdk. \end{theorem} \begin{proof}
	Let $OPT$ be the considered optimal solution with $opt=\profit(OPT)$.
	Recall that there are no large items. Let also $OPT_{vert}\subseteq OPT$
	be the (\emph{vertical}) items with height more than $\eps\cdot N$
	(hence with width at most $\eps\cdot N$), and $OPT_{hor}=OPT\setminus OPT_{ver}$
	(\emph{horizontal} items). Note that with this definition both sides
	of a horizontal item might have a length of at most $\eps\cdot N$.
	We let $opt_{long}=\profit(OPT_{long})$ and $opt_{short}=\profit(OPT_{short})$.
	
	As mentioned above, our $\fontL$-packing PTAS achieves a profit
	of at least $(\frac{3}{4}-O(\eps))opt_{long}$ which can be seen by
	applying Lemma \ref{lem:LoftheRing} with $H=OPT_{long}\cap OPT_{hor}$
	and $V=OPT_{long}\cap OPT_{ver}$. In order to show that the other
	two packings yield a good profit, consider a \emph{random horizontal
		strip} $S=[0,N]\times[a,a+\eps\cdot N]$ (fully contained in the knapsack)
	where $a\in[0,(1-\eps)N)$ is chosen uniformly at random. We remove
	all items of $OPT$ intersecting $S$. Each item in $OPT_{hor}$ and
	$OPT_{short}\cap OPT_{ver}$ is deleted with probability at most $3\eps$
	and $\frac{1}{2}+2\eps$, resp. Therefore the total profit of the
	remaining items is in expectation at least $(1-3\eps)\profit(OPT_{hor})+(\frac{1}{2}-2\eps)\profit(OPT_{short}\cap OPT_{vert})$.
	Observe that the resulting solution can be packed into a restricted
	knapsack of size $[0,N]\times[0,N/(1+\eps)]$ by shifting down the
	items above the horizontal strip. Therefore, when we apply the resource
	augmentation algorithm in~\cite{js07} to the knapsack $[0,N]\times[0,N/(1+\eps)]$,
	up to a factor $1-\eps$, we will find a solution of (deterministically!)
	at least the same profit. In other terms, this profit is at least
	$(1-4\eps)\profit(OPT_{hor})+(\frac{1}{2}-\frac{5}{2}\eps)\profit(OPT_{short}\cap OPT_{vert})$.
	
	By a symmetric argument, we obtain a solution of profit at
	least $(1-4\eps)\profit(OPT_{ver})+(\frac{1}{2}-\frac{5}{2}\eps)\profit(OPT_{short}\cap OPT_{hor})$
	when we apply the algorithm in~\cite{js07} to the knapsack
	$[0,N/(1+\eps)]\times[0,N]$. Thus the best of the latter two solutions
	has profit at least $(\frac{1}{2}-2\eps)opt_{long}+(\frac{3}{4}-\frac{13}{4}\eps)opt_{short}=(\frac{1}{2}-2\eps)opt+(\frac{1}{4}-\frac{5}{4}\eps)opt_{short}$.
	The best of our three solutions has therefore value at least $(\frac{9}{16}-O(\eps))opt$
	where the worst case is achieved for roughly $opt_{long}=3\cdot opt_{short}$.
\end{proof}

\section{Weighted Case Without Rotations\label{sec:weighted}}

In this section we show how to extend the reasoning of the unweighted
case to the weighted case. This requires much more complicated technical
machinery than the algorithm presented in Section~\ref{sec:tdk_car:simple}.

Our strategy is to start with a partition of the knapsack into thin
corridors as defined in~\cite{aw15}. Then, we partition
these corridors into a set of rectangular boxes and an L-packing.
We first present a simplified version of our argumentation in which
we assume that we are allowed to drop $O_{\eps}(1)$ many items
at no cost, i.e., we pretend that we have the right to remove $O_{\eps}(1)$
items from $\opt$ and compare the profit of our computed solution
with the remaining set. Building on this, we give an argumentation
for the general case which will involve some additional shifting arguments.

\subsection{Item classification}

We start with a classification of the input items according to their
heights and widths. For two given constants $1\geq\epsl>\epss>0$,
we classify an item $i$ as:

\itemsep0pt 
\begin{itemize}
	\item[$\bullet$] \emph{small} if $h_{i},w_{i}\leq\epss N$; 
	\item[$\bullet$] \emph{large} if $h_{i},w_{i}>\epsl N$; 
	\item[$\bullet$] \emph{horizontal} if $w_{i}>\epsl N$ and $h_{i}\leq\epss N$; 
	\item[$\bullet$] \emph{vertical} if $h_{i}>\epsl N$ and $w_{i}\leq\epss N$; 
	\item[$\bullet$] \emph{intermediate} otherwise (i.e., at least one side has length
	in $(\epss N,\epsl N]$). 
\end{itemize}
We also call \emph{skewed} items that are horizontal or vertical.
We let $\Rsm$, $\Rla$, $\Rho$, $\Rve$, $\Rsk$, and $\Rin$ be
the items which are small, large, horizontal, vertical, skewed, and
intermediate, respectively. The corresponding intersection with $\opt$
defines the sets $\optsm$, $\optla$, $\optho$, $\optve$, $\optsk$,
$\optin$, respectively. 

Observe that $|\optla|=O(1/\epsl^{2})$ and since we are allowed to
drop $O_{\eps}(1)$ items from now on we ignore $\optla$. The
next lemma shows that we can neglect also $\optin$. 
\begin{lemma}
	\label{lem:item-classification}For any constant $\eps>0$ and positive
	increasing function $f(\cdot)$, $f(x)>x$, there exist constant
	values $\epsl,\epss$, with $\eps\geq\epsl\geq f(\epss)\ge\Omega_{\eps}(1)$
	and $\epss\in\Omega_{\eps}(1)$ such that the total profit of intermediate
	rectangles is bounded by $\eps p(OPT)$. The pair $(\epsl,\epss)$
		is one pair from a set of $O_\eps(1)$ pairs and this set can be computed in polynomial time.
\end{lemma}
\begin{proof}
	Define $k+1=2/\eps+1$ constants $\eps_{1},\ldots,\eps_{k+1}$, such that
	$\eps=f(\eps_{1})$ and $\eps_{i}=f(\eps_{i+1})$ for
		each~$i$. Consider the $k$ ranges of widths and heights of type
	$(\eps_{i+1}N,\eps_{i}N]$. By an averaging argument there exists
	one index $j$ such that the total profit of items in $\opt$ with
	at least one side length in the range $(\eps_{j+1}N,\eps_{j}N]$ is
	at most $2\frac{\eps}{2}p(\opt)$. It is then sufficient to set $\epsl=\eps_{j}$
	and $\epss=\eps_{j+1}$. 
\end{proof}
We transform now the packing of the optimal solution $\opt$. To this
end, we temporarily remove the small items $\optsm$. We will add
them back later. Thus, the reader may now assume that we need to pack
only the skewed items from $\optsk$.

\subsection{Corridors, Spirals and Rings}

We build on a partition of the knapsack into corridors as used in
\cite{aw15}. We define an \emph{open corridor} to
be a face on the 2D-plane bounded by a simple rectilinear polygon
with $2k$ edges $e_{0},\ldots,e_{2k-1}$ for some integer $k\geq2$,
such that for each pair of horizontal (resp. vertical) edges $e_{i},e_{2k-i}$,
$i\in\{1,...,k-1\}$ there exists a vertical (resp. horizontal) line
segment $\ell_{i}$ such that both $e_{i}$ and $e_{2k-i}$ intersect
$\ell_{i}$ and $\ell_{i}$ does not intersect any other edge. Note
that $e_{0}$ and $e_{k}$ are not required to satisfy this property:
we call them the \emph{boundary edges} of the corridor. Similarly
a \emph{closed corridor} (or \emph{cycle}) is a face on the 2D-plane
bounded by two simple rectilinear polygons defined by edges $e_{0},\ldots,e_{k-1}$
and $e'_{0},\ldots,e'_{k-1}$ such that the second polygon is contained
inside the first one, and for each pair of horizontal (resp. vertical)
edges $e_{i},e'_{i}$, $i\in\{0,...,k-1\}$, there exists a vertical
(resp. horizontal) line segment $\ell_{i}$ such that both $e_{i}$
and $e'_{i}$ intersect $\ell_{i}$ and $\ell_{i}$ does not intersect
any other edge. See Figures \ref{fig:packing} and \ref{fig:corridors}
for examples. Let us focus on minimum length such $\ell_{i}$'s:
then the \emph{width} $\alpha$ of the corridor is the maximum
length of any such $\ell_{i}$. We say that an open (resp. closed)
corridor of the above kind has $k-2$ (resp. $k$) \emph{bends}.
A corridor partition is a partition of the knapsack into corridors.

\begin{figure}
	\centering
	\includegraphics[width=.25\textwidth]{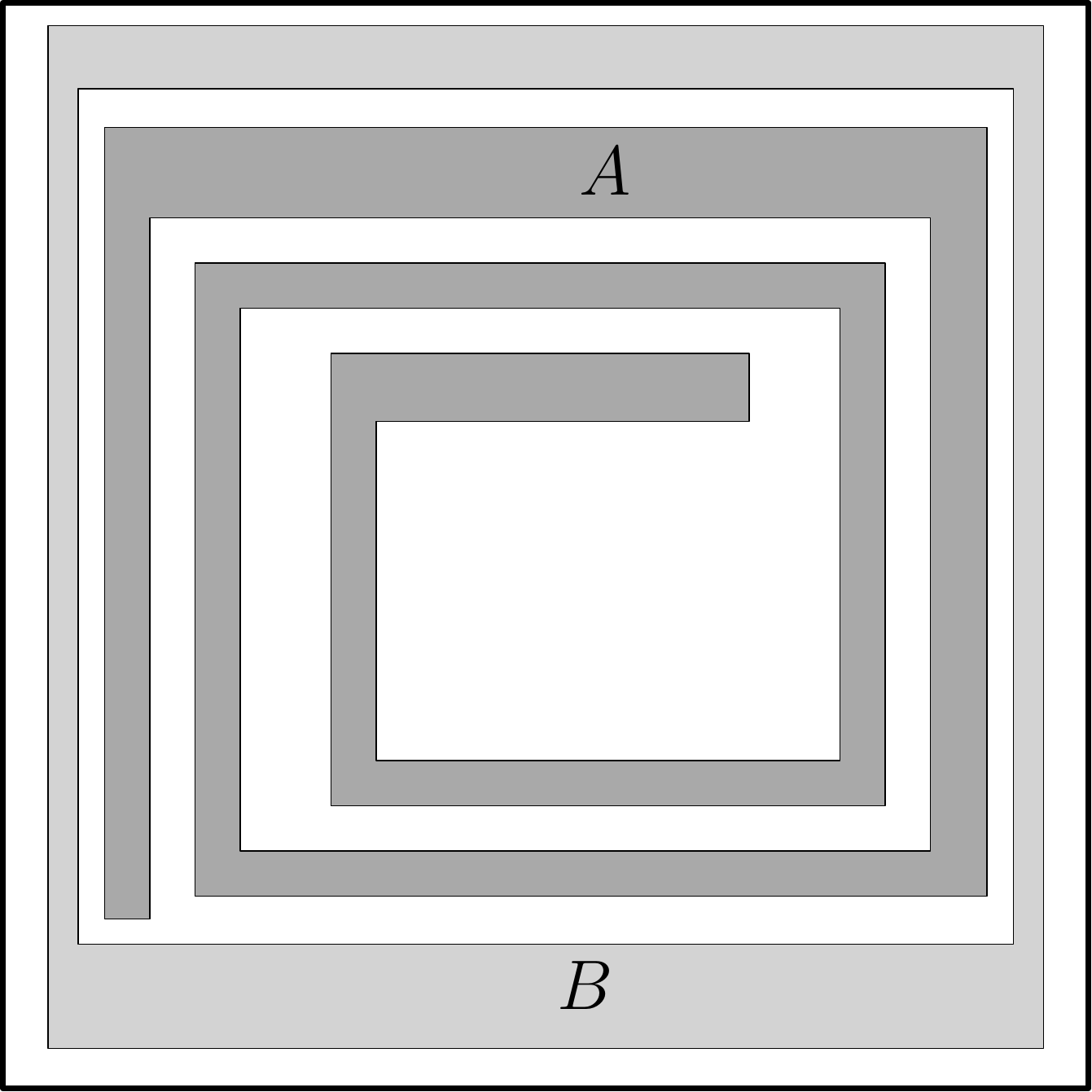}
	\caption{Illustration of two specific types of corridors: spirals (A) and rings (B).\label{fig:packing}.}
\end{figure}

\begin{lemma}[Corridor Packing Lemma, \cite{aw15}]
	\label{lem:corridorPack-weighted}There exists a corridor partition
	and a set $\optco\subseteq\optsk$ such that: 
\end{lemma}
\begin{enumerate}
	\item there is a subset $\optco^{cross}\subseteq\optco$ with$|\optco^{cross}|\le O_{\eps}(1)$
	such that each item $i\in\optco\setminus\optco^{cross}$ is fully contained
	in some corridor, 
	\item $p(\optco)\geq(1-O(\eps))p(\optsk)$,
	\item the number of corridors is $O_{\eps,\epsl}(1)$ and each corridor
	has width at most $\epsl N$ and has at most $1/\eps$ bends. 
\end{enumerate}
Since we are allowed to drop $O_{\eps}(1)$ items from now on
we ignore $\optco^{cross}$. We next identify some structural properties
of the corridors that are later exploited in our analysis. 
Observe that an open (resp. closed) corridor of the above type is
the union of $k-1$ (resp. $k$) boxes, that we next call \emph{subcorridors}
(see also Figure \ref{fig:corridors}). Each such box is a maximally
large rectangle that is contained in the corridor. The subcorridor
$S_{i}$ of an open (resp. closed) corridor of the above kind is
the one containing edges $e_{i},e_{2k-i}$ (resp. $e_{i},e_{i'}$)
on its boundary. The length of $S_{i}$ is the \emph{length} of the
shortest such edge. We say that a subcorridor is \emph{long} if its
length is more than $N/2$, and \emph{short} otherwise. The partition
of subcorridors into short and long will be crucial in our analysis.

We call a subcorridor
\emph{horizontal} (resp. \emph{vertical}) if the corresponding edges
are so. Note that each rectangle in $\optco$ is univocally
	associated with the only subcorridor that fully contains it: indeed,
the longer side of a skewed rectangle is longer than the width of
any corridor. Consider the sequence of consecutive subcorridors $S_{1},\ldots,S_{k'}$
of an open or closed corridor. Consider two consecutive corridors
$S_{i}$ and $S_{i'}$, with $i'=i+1$ in the case of an open corridor
and $i'=(i+1)\pmod{k'}$ otherwise. First assume that $S_{i'}$ is horizontal. We say that $S_{i'}$ is to the
right (resp. left) of $S_{i}$ if the right-most (left-most) boundary of $S_{i'}$ is to the right (left) of the right-most (left-most) boundary of $S_i$. If instead $S_{i'}$ is vertical, then $S_i$ must be horizontal and we say that $S_{i'}$ is to the right (left) of $S_i$ if $S_i$ is to the left (right) of $S_{i'}$. Similarly, if $S_{i'}$ is vertical, we say that $S_{i'}$ is above (below) $S_i$ if the top (bottom) boundary of $S_{i'}$ is above (below) the top (bottom) boundary of $S_i$. If $S_{i'}$ is horizontal, we say that it is above (below) $S_i$ if $S_i$ (which is vertical) is below (above) $S_{i'}$.
We say that the pair $(S_{i},S_{i'})$ forms a clockwise
bend if $S_{i}$ is horizontal and $S_{i'}$ is to its bottom-right
	or top-left, and the complementary cases if $S_{i}$ is vertical.
In all the other cases the pairs form a counter-clockwise bend. Consider
a triple $(S_{i},S_{i'},S_{i''})$ of consecutive subcorridors in
the above sense. It forms a $U$-bend if $(S_{i},S_{i'})$ and $(S_{i'},S_{i''})$
are both clockwise or counterclockwise bends. Otherwise it forms a $Z$-bend. In both cases $S_{i'}$ is the \emph{center} of
the bend, and $S_{i},S_{i''}$ its \emph{sides}. An open corridor
whose bends are all clockwise (resp. counter-clockwise) is a \emph{spiral}.
A closed corridor with $k=4$ is a \emph{ring}. Note that in a ring
all bends are clockwise or counter-clockwise, hence in some sense
it is the closed analogue of a spiral. We remark that a corridor
	whose subcorridors are all long is a spiral or a ring\footnote{We omit the proof, since we	do not explicitly need this claim.}. As we will see, spirals and rings play a crucial role in our analysis.
In particular, we will exploit the following simple fact.

\begin{lemma}\label{lem:spiral} The following properties hold: 
	\itemsep0pt 
	\begin{enumerate}
		\item \label{claim:Zbend} The two sides of a $Z$-bend cannot be long.
		In particular, an open corridor whose subcorridors are all long is
		a spiral. 
		\item \label{claim:Ubends} A closed corridor contains at least $4$ distinct
		(possibly overlapping) $U$-bends.
	\end{enumerate}
\end{lemma} \begin{proof} \eqref{claim:Zbend} By definition of long
subcorridors and $Z$-bend, the $3$ subcorridors of the $Z$-bend
would otherwise have total width or height larger than $N$. \eqref{claim:Ubends}
Consider the left-most and right-most vertical subcorridors, and
the top-most and bottom-most horizontal subcorridors. These $4$ subcorridors
exist, are distinct, and are centers of a $U$-bend. \end{proof}

\subsection{Partitioning Corridors into Rectangular Boxes\label{sec:structural:boxes}}

We next describe a routine to partition the corridors into rectangular
boxes such that each item is contained in one such box. We remark
that to achieve this partitioning we sometimes have to sacrifice a
large fraction of $\optco$, hence we do not achieve a $1+\eps$ approximation
as in \cite{aw13}. On the positive side, we generate only a constant
(rather than polylogarithmic) number of boxes. This is crucial to
obtain a polynomial time algorithm in the later steps.

Recall that each $i\in\optco$ is univocally associated with the only
subcorridor that fully contains it. 
We will say that we \emph{delete} a sub-corridor, when we delete all rectangles univocally associated with the subcorridor. Note that in deletion of a sub-corridor we do not delete rectangles that are partially contained in that subcorridor but  completely contained in a neighbor sub-corridor.
Given a corridor, we sometimes
\emph{delete} some of its subcorridors, and consider the \emph{residual}
corridors (possibly more than one) given by the union of the remaining
subcorridors. 
Note that removing any subcorridor from a closed corridor
turns it into an open corridor. We implicitly assume that items associated
with a deleted subcorridor are also removed (and consequently the
corresponding area can be used to pack other items).

\begin{figure}
	\begin{centering}
		\includegraphics[height=1.5cm]{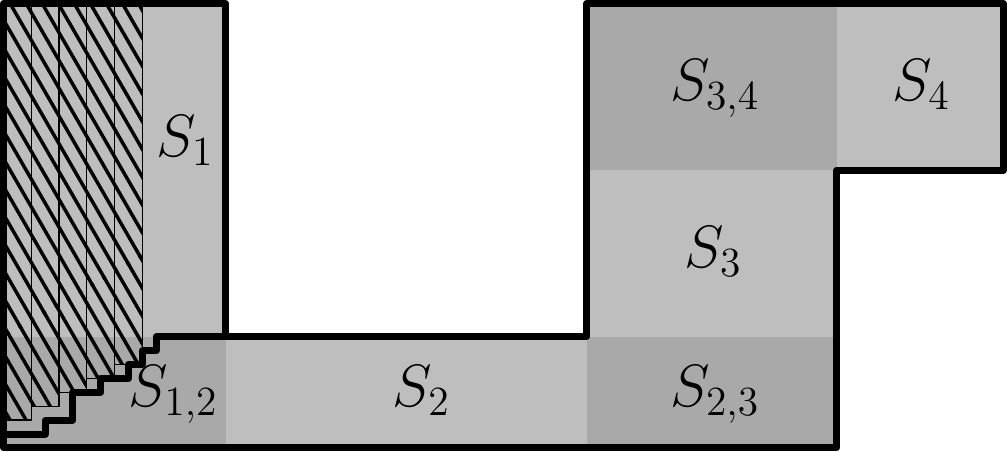}~~~~\includegraphics[height=2cm]{peel-corridor} 
		\par\end{centering}
	\caption{\label{fig:corridors}Left: The subcorridors $S_{1}$ and
			$S_{3}$ are vertical, $S_{2}$ and $S_{4}$ are horizontal. The
		subcorridor $S_{3}$ is on the top-right of $S_{2}$. The
			curve on the bottom left shows the boundary curve between $S_{1}$
			and $S_{2}$. The pair $(S_{3},S_{4})$ forms a clockwise bend and
			the pair $(S_{2},S_{3})$ forms a counter-clockwise bend. The triple
			$(S_{1},S_{2},S_{3})$ forms a $U$-bend and the triple $(S_{2},S_{3},S_{4})$
			forms a $Z$-bend. Right: Our operation that divides a corridor
		into $O_{\eps}(1)$ boxes and $O_{\eps}(1)$ shorter corridors. The
		dark gray items show thin items that are removed in this operation.
		The light gray items are fat items that are shifted to the
			box below their respective original box. The value $\alpha$ denotes
		the width of the depicted corridor. }
\end{figure}

Given two consecutive subcorridors $S_{i}$ and $S_{i'}$, we define
the \emph{boundary curve} among them as follows (see also Figure \ref{fig:corridors}).
Suppose that $S_{i'}$ is to the top-right of $S_{i}$, the other
cases being symmetric. Let $S_{i,i'}=S_{i}\cap S_{i'}$ be the rectangular
region shared by the two subcorridors. Then the boundary curve among
them is any simple rectilinear polygon inside $S_{i,i'}$ that decreases
monotonically from its top-left corner to its bottom-right one and
that does not cut any rectangle in these subcorridors. For a boundary
horizontal (resp. vertical) subcorridor of an open corridor (i.e.,
a subcorridor containing $e_{0}$ or $e_{2k-1}$) we define a dummy
boundary curve given by the vertical (resp. horizontal) side of the
subcorridor that coincides with a (boundary) edge of the corridor. 

\begin{remark}\label{rem:boundaryCurves} Each subcorridor has two
	boundary curves (including possibly dummy ones). Furthermore, all
	its items are fully contained in the region delimited by such curves
	plus the two edges of the corridor associated with the subcorridor
	\emph{(private region)}. \end{remark} 

Given a corridor, we partition its area into a constant number of
boxes as follows (see also Figure \ref{fig:corridors}, and \cite{aw13}
for a more detailed description of an analogous construction). Let
$S$ be one of its boundary subcorridors (if any), or the central
subcorridor of a $U$-bend. Note that one such $S$ must exist (trivially
for an open corridor, otherwise by Lemma \ref{lem:spiral}.\ref{claim:Ubends}).
In the corridor partition, there might be several subcorridors fulfilling
the latter condition. We will explain later in which order to process
the subcorridors, here we explain only how to apply our routine to
\emph{one} subcorridor, which we call \emph{processing} of subcorridor..  

Suppose that $S$ is horizontal with height $b$, with the shorter
horizontal associated edge being the top one. The other cases are
symmetric. Let $\epst>0$ be a sufficiently small constant to be defined
later. If $S$ is the only subcorridor in the considered corridor,
$S$ forms a box and all its items are marked as \emph{fat}.
Otherwise, we draw $1/\epst$ horizontal lines that partition the
private region of $S$ into subregions of height $\epst b$. We mark
as \emph{thin} the items of the bottom-most (i.e., the widest) such
subregion, and as \emph{killed} the items of the subcorridor cut by
these horizontal lines. 
All the remaining items of the subcorridor are marked as \emph{fat}.

For each such subregion, we define an associated (horizontal) box
as the largest axis-aligned box that is contained in the subregion.
Given these boxes, we partition the rest of the corridor into $1/\epst$
corridors as follows. Let $S'$ be a corridor next to $S$, say to
its top-right. Let $P$ be the set of corners of the boxes contained
in the boundary curve between $S$ and $S'$. We project $P$ vertically
on the boundary curve of $S'$ not shared with $S$, hence getting
a set $P'$ of $1/\epst$ points. We iterate the process on
the pair $(S',P')$. At the end of the process, we obtain a set of
$1/\epst$ boxes from the starting subcorridor $S$, plus a collection
of $1/\epst$ new (open) corridors each one having one less bend with
respect to the original corridor. Later, we will also apply this process on the latter
	corridors. Each newly created corridor will have one bend less than the original corridor 
	and thus this process eventually terminates.
Note that, since initially there are $O_{\eps,\epsl}(1)$ corridors
each one with $O(1/\eps)$ bends, the final number of boxes is $O_{\eps,\epsl,\epst}(1)$.
See Figure \ref{fig:corridors} for an illustration.

\begin{remark} Assume that we execute the above procedure on the
	subcorridors until there is no subcorridor left on which we can apply
	it. Then we obtain a partition of $\optco$ into disjoint sets $\optth$,
	$\optfa$, and $\optki$ of thin, fat, and killed items, respectively.
	Note that each order to process the subcorridors leads to different
	such partition. We will define this order carefully in our analysis.
\end{remark} 
\begin{remark}\label{rem:fatPack} By a simple shifting
	argument, there exists a packing of $\optfa$ into the boxes. Intuitively,
	in the above construction each subregion is fully contained
	in the box associated with the subregion immediately below (when no
	lower subregion exists, the corresponding items are thin).
\end{remark}

We will from now on assume that the shifting of items as described in Remark~\ref{rem:fatPack} has been done.

The following lemma summarizes some of the properties of the boxes
and of the associated partition of $\optco$ (independently from the
way ties are broken). Let $\Rho$ and $\Rve$ denote the set
	of horizontal and vertical input items, respectively.
\begin{lemma}\label{lem:boxProperties}
	The following properties hold: 
	
	\itemsep0pt 
	\begin{enumerate}
		\item \label{lem:boxProperties:profitKill} $|\optki|=O_{\eps,\epsl,\epst}(1)$;
		\item \label{lem:boxProperties:thin} For any given constant $\epsr>0$
		there is a sufficiently small $\epst>0$ such that
		the total height (resp. width) of items in $\optth\cap\Rho$ (resp.
		$\optth\cap\Rve$) is at most $\epsr N$. 
	\end{enumerate}
\end{lemma} \begin{proof} \eqref{lem:boxProperties:profitKill} Each
horizontal (resp. vertical) line in the construction can kill at
most $1/\epsl$ items, since those items must be horizontal (resp.
vertical). Hence we kill $O_{\eps,\epsl,\epst}(1)$ items in total.

\eqref{lem:boxProperties:thin} The mentioned total height/width is
at most $\epst N$ times the number of subcorridors, which is $O_{\eps,\epsl}(1)$.
The claim follows for $\epst$ small enough. \end{proof} 

\subsection{Containers\label{sec:structural:containers}}

Assume that we applied the routine described in Section~\ref{sec:structural:boxes}
above until each corridor is partitioned into boxes. We explain how
to partition each box into $O_{\eps}(1)$ subboxes, to which we
refer to as \emph{containers} in the sequel. Hence, we apply the routine
described below to each box.

Consider a box of size $a\times b$ coming from the above construction,
and on the associated set $\optbo$ of items from $\optfa$.
We will show how to pack a set $\optbo'\subseteq\optbo$
with $p(\optbo')\geq(1-\eps)p(\optbo)$ into $O_{\eps}(1)$
containers packed inside the box, such that both the containers and
the packing of $\optbo'$ inside them satisfy some extra properties
that are useful in the design of an efficient algorithm. This part
is similar in spirit to prior work, though here we present a refined
analysis that simplifies the algorithm (in particular, we can avoid
LP rounding).

A \emph{container} is a box labeled as \emph{horizontal}, \emph{vertical},
or \emph{area}. A \emph{container packing} of a set of items $I'$
into a collection of non-overlapping containers has to satisfy the
following properties: 

\itemsep0pt 
\begin{itemize}
	\item[$\bullet$] Items in a horizontal (resp. vertical) container are stacked one
	on top of the other (resp. one next to the other). 
	\item[$\bullet$] Each $i\in I'$ packed in an area container of size $a\times b$
	must have $w_{i}\leq\eps a$ and $h_{i}\leq\eps b$. 
\end{itemize}
Our main building block is the resource augmentation packing Lemma~\ref{lem:structural_lemma_augm}.
Applying it to each box yields the following
lemma.
\begin{lemma}[Container Packing Lemma]
	\label{lem:containerPack} For a given constant $\epsau>0$,
	there exists a set $\optfa^{cont}\subseteq\optfa$ such that there is a
	container packing for all apart from $O_{\eps}(1)$ items in $\optfa^{cont}$
	such that: 
\end{lemma}
\begin{enumerate}
	\item $p(\optfa^{cont})\geq(1-O(\eps))p(\optfa)$; 
	\item The number of containers is $O_{\eps,\epsl,\epst,\epsau}(1)$ and
	their sizes belong to a set of cardinality $n^{O_{\eps,\epsl,\epst,\epsau}(1)}$
	that can be computed in polynomial time.
\end{enumerate}
\begin{proof} Let us focus on a specific box of size $a\times b$
	from the previous construction in Section~\ref{sec:structural:boxes}, and on the items $\optbo\subseteq\optfa$
	inside it. If $|\optbo|=O_{\eps}(1)$ then we can simply
	create one container for each item and we are done. 
	Otherwise, assume without loss of generality that this box (hence its items) is horizontal. We
	obtain a set $\overline{\optbo}$ by removing from $\optbo$
	all items intersecting a proper horizontal strip of height $3\eps b$.
	Clearly these items can be repacked in a box of size $a\times(1-3\eps)b$.
	By a simple averaging argument, it is possible to choose the strip
	so that the items fully contained in it have total profit at most
	$O(\eps)p(\optbo)$. Furthermore, there can be at most $O(1/\epsl)$
	items that partially overlap with the strip (since items are skewed).
	We drop these items and do not pack them.
	
	At this point we can use the Resource Augmentation Lemma~\ref{lem:structural_lemma_augm}
	to pack a large profit subset $\optbo'\subseteq\overline{\optbo}$
	into $O_{\epsau}(1)$ containers that can be packed in a box of size
	$a\times(1-3\eps)(1+\epsau)b\leq a\times(1-2\eps)b$. 
	We perform the above operation on each box of the previous
		construction 
		and define $\optfa^{cont}$ to be the union of the respective
	sets $\optbo'$. The claim follows. \end{proof}

\subsection{A Profitable Structured Packing} \label{sec:structural:lemma}

We next prove our main structural lemma which yields that there exists
a structured packing which is partitioned into $O_{\eps}(1)$
containers and an L. We will refer to such a packing as an L\&C packing
(formally defined below). Note that in the previous section we did
not specify in which order we partition the subcorridors into boxes.
In this section, we give several such orders which will then result
in different packings. The last such packing is special since we will
modify it a bit to gain some space and then reinsert the thin items
that were removed in the process of partitioning the corridors into
containers. Afterwards, we will show that one of the resulting packings
will yield an approximation ratio of $17/9+\eps$.

A \emph{boundary ring} of width $N'$ is a ring having as external
boundary the edges of the knapsack and as internal boundary
the boundary of a square box of size $(N-N')\times(N-N')$ in the
middle of the knapsack. A \emph{boundary $L$} of width $N'$ is the
region covered by two boxes of size $N'\times N$ and $N\times N'$
that are placed on the left and bottom boundaries of the knapsack.

An \emph{L\&C} packing is defined as follows. We are given two integer
parameters $N'\in[0,N/2]$ and $\ell\in(N/2,N]$. We define
a boundary $L$ of width $N'$, and a collection of non-overlapping
containers contained in the space not occupied by the boundary
$L$. The number of containers and their sizes are as in Lemma~\ref{lem:containerPack}.
We let $\ilong\subseteq I$ be the items whose longer side has length
longer than $\ell$ (hence longer than $N/2$), and $\ishort=I\setminus \ilong$
be the remaining items. We can pack only items from $\ilong$ in the
boundary $L$, and only items from $\ishort$ in the containers (satisfying
the usual container packing constraints). See also Figure \ref{fig:packing+ring}. 
\begin{remark}\label{rem:degenerateRing}
	In the analysis sometimes we will not need the boundary $L$. This
	case is captured by setting $N'=0$ and $\ell=N$ (\emph{degenerate
		$L$} case). \end{remark}

\begin{lemma}
	\label{lem:apxNoRotation} Let $\optrc$ be the most profitable solution
	that is packed by an L\&C packing. Then $p(\optrc)\geq(\frac{9}{17}-O(\eps))p(\opt)$. 
\end{lemma}
In the remainder of this section we prove Lemma~\ref{lem:apxNoRotation},
assuming that we can drop $O_{\eps}(1)$ items at no cost. Hence,
formally we will prove that there is an L\&C packing $I'$ and a set
of $O_{\eps}(1)$ items $I_{\mathrm{drop}}$ such that $p(I')+p(I_{\mathrm{drop}})\ge(\frac{9}{17}-O(\eps))p(\opt)$.
Subsequently, we will prove Lemma~\ref{lem:apxNoRotation} in full
generality (without dropping any items).

The proof of Lemma~\ref{lem:apxNoRotation} involves some case analysis.
Recall that we classify subcorridors into short and long, and horizontal
and vertical. We further partition short subcorridors as follows:
let $S_{1},\ldots,S_{k'}$ be the subcorridors of a given corridor,
and let $S_{1}^{s},\ldots,S_{k''}^{s}$ be the subsequence of short
subcorridors (if any). Mark $S_{i}^{s}$ as \emph{even} if $i$ is
so, and \emph{odd} otherwise. Note that corridors are subdivided
into several other corridors during the box construction process (see the right side of Figure \ref{fig:corridors}), 
and these new corridors might have fewer subcorridors than the initial corridor.
However, the marking of the subcorridors (short, long, even, odd, horizontal, vertical) is inherited from the marking of the original subcorridor. 

We will describe now 7 different ways to partition the subcorridors into
boxes, for some of them we delete some of the subcorridors. Each of these different processing orders will give different sets $\optth, \optki$ and $\optfa^{cont}$, and based on these, we will partition the items into three sets. We will then prove three different lower bounds on $p(\optrc)$ with respect to the sizes of these three sets using averaging arguments about the seven cases.

\begin{figure*}[t!]
	\centering

	\includegraphics[width=\textwidth]{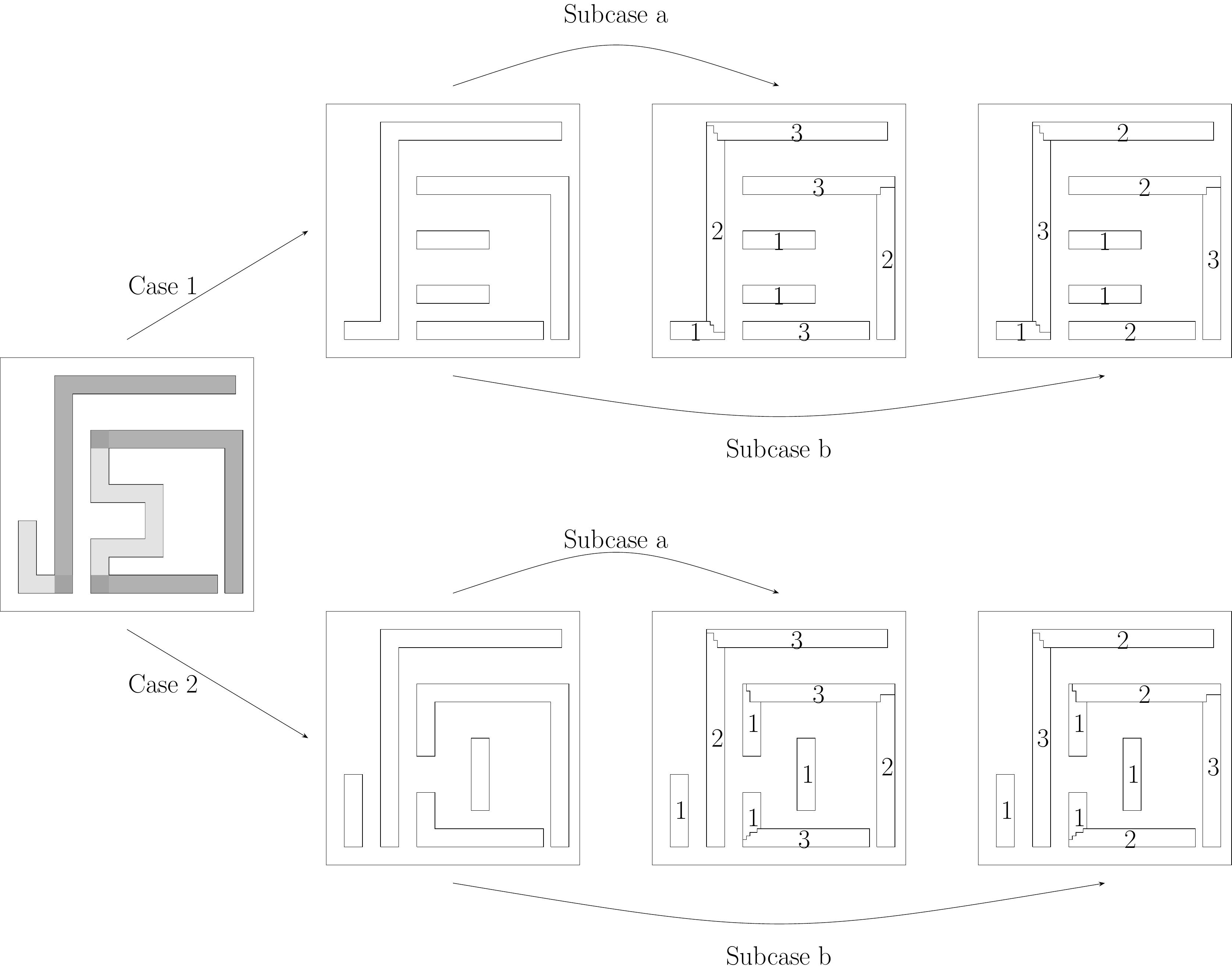}
	
	~
	\caption{Figure for Case 1 and 2. The knapsack on the left contains two corridors, where short subcorridors are marked light gray and long subcorridors are marked dark gray. In case 1, we delete vertical short subcorridors and then consider two processing orders in subcases a and b. In case 2, we delete horizontal short subcorridors and again consider two processing orders in subcases a and b.}
	\label{fig:case1a}
\end{figure*}

\paragraph{Cases 1a, 1b, 2a, 2b: Short horizontal/short vertical subcorridors.}

We delete either all vertical short (case 1) or all horizontal short subcorridors (case 2).
We first process all short subcorridors, then either all vertical (subcases a) or
	horizontal long ones (subcases b), and finally the remaining (horizontal or vertical, resp.) long
ones. We can start by processing all short corridors. Indeed, any
such corridor cannot be the center of a $Z$-bend by Lemma~\ref{lem:spiral}.\ref{claim:Zbend}
since its two sides would be long, hence it must be boundary or the
center of a $U$-bend. After processing short subcorridors, by the
same argument the residual (long) subcorridors are the boundary or the
center of a $U$-bend. So we can process the long subcorridors in
any order. This gives in total four cases.
See Figure~\ref{fig:case1a} for deletion/processing of subcorridors for these cases.

\paragraph{Cases 3a, 3b: Even/odd short subcorridors.}

We delete the odd (or even) short subcorridors and then process
even (resp. odd) short subcorridors last. We exploit the fact that
each residual corridor contains at most one short subcorridor. Then,
if there is another (long) subcorridor, there is also one which is
boundary (trivially for an open corridor) or the center of a $U$-bend
(by Lemma \ref{lem:spiral}, Property~\ref{claim:Ubends}). Hence we can always
process some long subcorridor leaving the unique short subcorridor
as last one. This gives two cases.

\paragraph{Case 4: Fat only.}

Do not delete any short subcorridor. Process subcorridors in any feasible order.

In each of the cases, we apply the procedure described in Section~\ref{sec:structural:containers}
to partition each box into $O_{\eps}(1)$ containers. We next
label items as follows. Consider the classification of items into
$\optfa^{cont}$, $\optth$, and 
$\optki$ in each one of the $7$ cases above.
Then: 

\itemsep0pt

\begin{itemize}
	\item $\T$ is the set of items which are in $\optth$ \emph{in at least
		one case}; 
	\item $\OK$ is the set of items which are in $\optki$ \emph{in
		at least one case}; 
	\item $\F$ is the set of items which are in $\optfa^{cont}$ \emph{in all the
		cases}. 
\end{itemize}

\begin{remark}\label{rem:processedLast} Consider the subcorridor
	of a given corridor that is processed last in one of the above cases. None of its items are
	assigned to $\optth$ 
	in that case
	and thus essentially all its items are packed
	in one of the constructed containers. 
	In particular, for an item in set $\T$, in some of the above cases it might be in such a subcorridor and thus marked fat and packed into a container.
\end{remark}

\begin{lemma}\label{lem:FTprofit} One has:
	\[
	  p(\F\cup \T)+p(\OK)+p(\optla)+p(\optco^{cross})\geq(1-O(\eps))p(\opt).
	\]
\end{lemma} \begin{proof} Let us initialize $\F=\optfa^{cont}$, $\T=\optth$,
and $\OK=\optki$ by considering one of the above cases. Next
we consider the aforementioned cases, hence moving some items in $\F$ to
either $\T$ or $\OK$. Note that initially $p(\F\cup \T)+p(\optki)+p(\optla)+p(\optco^{cross})\geq(1-O(\eps))p(\opt)$
by Lemma~\ref{lem:corridorPack-weighted} and hence we keep this
property.\end{proof}

Let $I_{\mathrm{lc}}$ and $I_{\mathrm{sc}}$ denote the items in long
and short corridors, respectively. We also let $\LF=I_{\mathrm lc}\cap \F$,
and define analogously $\SF$, $\LT$, and $\LF$. The next three
lemmas provide a lower bound on the case of a degenerate L.

\begin{lemma}\label{lem:onlyFat} $p(\optrc)\geq p(\LF)+p(\SF).$ \end{lemma}
\begin{proof} Follows immediately since we pack a superset
	of $\F$ in case 4. \end{proof} 
\begin{lemma}\label{lem:noShort}
	$p(\optrc)\geq p(\LF)+p(\LT)/2+p(\SF)/2.$ \end{lemma}
\begin{proof}
	Consider the sum of the profit of the packed items corresponding to the in total four subcases of cases 1 and 2. Each $i\in \LF$ appears $4$ times in the sum (as items in $\F$ are fat in all cases and all long subcorridors get processed), and
	each $i\in \LT$ at least twice by Remark~\ref{rem:processedLast}: If a long subcorridor $\mathfrak{L}$ neighbors a short subcorridor, the short subcorridor is either deleted or processed first. Further, all neighboring long subcorridors are processed first in case 1a and 2a (if $\mathfrak{L}$ is horizontal, then its neighbors are vertical) or 1b and 2b (if $\mathfrak{L}$ is vertical and its neighbors are horizontal). Thus, $\mathfrak{L}$ is the last processed subcorridor in at least two cases. Additionally, each item $i\in \SF$ also appears twice in the sum, as it gets deleted either in case 1 (if it is vertical) or in case 2 (if it is horizontal) and is fat otherwise.
	
	The claim follows by an averaging argument. \end{proof} \begin{lemma}\label{lem:evenOdd}
	$p(\optrc)\geq p(\LF)+p(\SF)/2+p(\ST)/2.$ \end{lemma} 
\begin{proof}
	Consider the sum of the number of packed items corresponding to cases 3a and 3b. Each $i\in \LF$
	appears twice in the sum as it is fat and all long subcorridors get processed. Each $i\in \SF\cup \ST$ appears at least
	once in the sum by Remark \ref{rem:processedLast}: An item $i\in \SF$ is deleted in one of the two cases (depending on whether it is in an even or odd subcorridor) and otherwise fat. An item $i \in \ST$ is also deleted in one of the two cases and otherwise its subcorridor is processed last. The claim follows
	by an averaging argument. \end{proof}

There is one last (and slightly more involved) case to be considered,
corresponding to a non-degenerate $L$. 

\begin{lemma}\label{lem:ringCase} $p(\optrc)\geq\frac{3}{4}p(\LT)+p(\ST)+\frac{1-O(\eps)}{2}p(\SF).$
\end{lemma} \begin{proof} Recall that $\epsl N$ is the maximum width
of a corridor. We consider an execution of the algorithm with boundary
$L$ width $N'=\epsr N$, and threshold length $\ell=(\frac{1}{2}+2\epsl)N$.
We remark that this length guarantees that items in $\ilong$ are not
contained in short subcorridors.

By Lemma \ref{lem:LoftheRing}, we can pack a subset of $\T\cap \ilong$
of profit at least $\frac{3}{4}p(\T\cap \ilong)$ in a boundary $L$
of width $\epsr N$. 
By Lemma \ref{lem:boxProperties} the remaining items in $\T$ can
be packed in two containers of size $\ell\times\epsr N$ and $\epsr N\times \ell$
that we place on the two sides of the knapsack not occupied by the
boundary $L$.

In the free area we can identify a square region $K''$ with side
length $(1-\eps)N$. We next show that there exists a feasible
solution $\SF'\subseteq \SF$ with $p(\SF')\geq(1-O(\eps))p(\SF)/2$ that
can be packed in a square of side length $(1-3\eps)N$. We can
then apply the Resource Augmentation Lemma~\ref{lem:structural_lemma_augm}
to pack $\SF''\subseteq \SF'$ of cardinality $p(\SF'')\geq(1-O(\eps))p(\SF')$
inside a central square region of side length $(1-3\eps)(1+\epsau)N\leq(1-2\eps)N$
using containers according to Lemma \ref{lem:containerPack}.

Consider the packing of $\SF$ as in the optimum solution. Choose a
random vertical (resp. horizontal) strip in the knapsack of width
(resp. height) $3\eps N$. Delete from $\SF$ all the items
intersecting the vertical and horizontal strips: clearly the remaining
items $\SF'$ can be packed into a square of side length $(1-3\eps)N$.
Consider any $i\in \SF$, and assume $i$ is horizontal (the vertical
case being symmetric). Recall that it has height at most $\epss N\leq\eps N$
and width at most $\ell\leq1/2+2\eps$. Therefore $i$ intersects
the horizontal strip with probability at most $5\eps$ and the vertical
strip with probability at most $1/2+8\eps$. Thus by the union
bound $i\in \SF'$ with probability at least $1/2-13\eps$. The
claim follows by linearity of expectation. 
\end{proof}

Combining the above Lemmas \ref{lem:FTprofit}, \ref{lem:onlyFat}, \ref{lem:noShort},
\ref{lem:evenOdd}, and \ref{lem:ringCase} we achieve the desired
approximation factor, assuming that the (dropped) $O_{\eps}(1)$
items in $\optki\cup\optla\cup\optco^{cross}$ have zero profit. The worst
case is obtained, up to $1-O(\eps)$ factors, for $p(\LT)=p(\SF)=p(\ST)$
and $p(\LF)=5p(\LT)/4$. This gives $p(\LT)=4/17\cdot p(\T\cup \F)$ and
a total profit of $9/17\cdot p(\T\cup \F)$. 

\subsection{Adding small items}

Note that up to now we ignored the small items $\optsm$. In this
section, we explain how to pack a large fraction of these items. 

We described above how to pack a large enough fraction of $\optsk$
into containers. We next refine the mentioned analysis to bound the total area
of such containers. It turns out that the residual area is sufficient
to pack almost all the items of $\optsm$ into a constant number of
area containers (not overlapping with the previous containers) for
$\epss$ small enough.

\sal{To this aim we use the last of the properties given by the Resource Augmentation Lemma~\ref{lem:structural_lemma_augm}: in fact, the lemma
guarantees that the total area of the containers is at most $a(I')+\epsau\,a\cdot b$,
where $a\times b$ and $I'$ are the size of the box and the initial
set of items in the claim of Lemma~\ref{lem:structural_lemma_augm}, respectively.}

\begin{lemma}\label{lem:areaContainer} In the packings due to Lemmas~\ref{lem:onlyFat},
	\ref{lem:noShort}, \ref{lem:evenOdd}, and \ref{lem:ringCase} the
	total area occupied by containers is at most $\min\{(1-2\eps)N^{2},a(\optco)+\epsau N^{2}\}$.\end{lemma} \begin{proof}
	Consider the first upper bound on the area. We have to distinguish
		between the containers considered in Lemma \ref{lem:ringCase} and
		the remaining cases. In the first case, there is a region not occupied
		by the boundary $L$ nor by the containers of area at least $4\eps N^{2}-4\eps^{2}N^{2}-4\epsr N^{2}\geq2\eps N^{2}$
		for $\epsr$ small enough, e.g., $\epsr\le\epsilon^{2}$ suffices.
		The claim follows. For the remaining cases, recall that in each
	horizontal box of size $a\times b$ we remove a horizontal strip of
	height $3\eps b$, and then use the Resource Augmentation Packing
	Lemma to pack the residual items in a box of size $a\times b(1-3\eps)(1+\epsau)\leq a\times b(1-2\eps)$
	for $\epsau\leq\eps$. Thus the total area of the containers is at
	most a fraction $1-2\eps$ of the area of the original box. A symmetric
	argument applies to vertical boxes. Thus the total area of the containers
	is at most a fraction $1-2\eps$ of the total area of the boxes, which
	in turn is at most $N^{2}$. This gives the first upper bound in the
	claim.
	
	For the second upper bound, we just apply the area bound in Lemma
	\ref{lem:structural_lemma_augm} to get that the total area of the
	containers is at most $a(\optco)$ plus $\epsau\,a\cdot b$ for each
	box of size $a\times b$. Summing the latter terms over the boxes
	one obtains at most $\epsau N^{2}$. \end{proof}

We are now ready to state a lemma that provides the desired packing
of small items. By slightly adapting the analysis we can guarantee
that the boundary $L$ that we use to prove the claimed approximation
ratio has width at most $\eps^{2}N$.
\begin{lemma}[Small Items Packing Lemma]\label{lem:smallPack} Suppose
	we are given a packing of non small items of the above type into $k$
	containers of total area $A$ and, possibly, a boundary $L$ of width
	at most $\eps^{2}N$.
	Then for $\epss$ small enough it is possible to define $O_{\epss}(1)$
	area containers of size $\frac{\epss}{\eps}N\times\frac{\epss}{\eps}N$
	neither overlapping with the containers nor with the boundary
	$L$ (if any) such that it is possible to pack $\optsm'\subseteq\optsm$
	of profit $p(\optsm')\geq(1-O(\eps))p(\optsm)$ inside these new area
	containers. \end{lemma} \begin{proof} Let us build a grid of width
	$\eps'N=\frac{\epss}{\eps}\cdot N$ inside the knapsack. We delete
	any cell of the grid that overlaps with a container or with the boundary
	$L$, and call the remaining cells \emph{free}. The new
	area containers are the free cells.
	
	The total area of the deleted grid cells is, for $\epss$ and $\epsau$
	small enough, at most 
	\begin{align*}
	(\eps^{2}N^{2}+A)+\frac{2N+4Nk}{\eps'N}\cdot(\eps'N)^{2}
	&\leq A+2\eps^{2}N^{2} \\
	&\hspace{-10pt}\overset{Lem.\ref{lem:areaContainer}}{\leq}\min\{(1-\eps)N^{2},a(\optco)+3\eps^{2}N^{2}\}
	\end{align*}
	For the sake of simplicity, suppose that
	any empty space in the optimal packing of $\optco\cup\optsm$ is filled
	in with dummy small items of profit $0$, so that $a(\optco\cup\optsm)=N^{2}$.
	We observe that the area of the free cells is at least $(1-O(\eps))a(\optsm)$: 
	Either, $a(\optsm)\geq\eps N^{2}$ and then the area of the free cells is at least 
	$a(\optsm)-3\eps^{2}N^{2}\geq(1-3\eps)a(\optsm)$;
	otherwise, we have that the area of the free cells is at least $\eps N^2 > a(\optsm)$.
	Therefore we can select a subset of small items $\optsm'\subseteq\optsm$,
	with $p(\optsm')\geq(1-O(\eps))p(\optsm)$ and area $a(\optsm)\leq(1-O(\eps))a(\optsm)$
	that can be fully packed into free cells using classical Next Fit
	Decreasing Height algorithm (NFDH) according to Lemma \ref{lem:nfdhPack}
	described later. The key argument for this is that each free
		cell is by a factor $1/\eps$ larger in each dimension than each small
		item.
\end{proof}

Thus, we have proven now that if the items in $\optki\cup\optla\cup\optco^{cross}$
had zero profit, then there is an L\&C-packing for the skewed and
small items with a profit of at least $9/17\cdot p(\optco)+(1-O(\epsilon))(\optsm)\ge(9/17-O(\epsilon))p(\opt)$.

\subsection{Shifting argumentation}

We remove now the assumption that we can drop $O_{\eps}(1)$ items
from $\opt$. We will add a couple of shifting steps to the argumentation
above to prove Lemma~\ref{lem:apxNoRotation} without that assumption.

It is no longer true that we can neglect the large rectangles $\optla$
since they might contribute a large amount towards the objective,
even though their total number is guaranteed to be small. Also, in
the process of constructing the boxes, we killed up to $O_{\eps}(1)$
rectangles (the rectangles in $\optki$). Similarly, we can no longer
drop the constantly many items in $\optco^{cross}$. Therefore, we apply
some careful shifting arguments in order to ensure that we can still
use a similar construction as above, while losing only a
factor $1+O(\eps)$ due to some items that we will discard.

The general idea is as follows: For $t=0,\ldots, k$ (we will later argue that $k<1/\eps$), 
	we define disjoint sets $K(t)$ recursively, each containing at most $O_\eps(1)$ items. 
	Each set $\K(t)=\bigcup_{j=0}^t K(j)$ is used to define a grid $G(t)$ in the knapsack. 
	Based on an item classification that depends on this grid, we identify a set of skewed items 
	and create a corridor partition with respect to these skewed items as described in Lemma \ref{lem:corridorPack-weighted}. 
	We then create a partition of the knapsack into corridors and a constant (depending on $\eps$) 
	number of containers (see Section \ref{sec:weighted:shifting:grid}). 
	Next, we decompose the corridors into boxes (Section \ref{sec:weighted:shifting:boxes}) 
	and these boxes into containers (section \ref{sec:weighted:shifting:containers}) 
	similarly as we did in Sections \ref{sec:structural:boxes} and \ref{sec:structural:containers} 
	(but with some notable changes as we did not delete small items from the knapsack and thus need 
	to handle those as well). 
	In the last step, we add small items to the packing (Section \ref{sec:weighted:shifting:small}). 
	During this whole process, we define the set $K(i+1)$ of items that are somehow ``in our way'' 
	during the decomposition process (e.g., items that are not fully contained in some corridor of 
	the corridor partition), but which we cannot delete directly as they might have large profit. 
	These items are similar to the killed items in the previous argumentation.
	However, using a shifting argument we can simply show that after at most $k<1/\eps$ steps of this 
	process, we encounter a set $K(k)$ of low total profit, that we can remove, at which point we 
	can apply almost the same argumentation as in Lemmas~\ref{lem:onlyFat},
	\ref{lem:noShort}, and \ref{lem:evenOdd} to obtain lower bounds on the profit of an optimal 
	L\&C packing (Section \ref{sec:weighted:shifting:packing}).

We initiate this iterative process as follows: Denote by $K(0)$ a set containing all items that are killed in at
	least one of the cases arising in Section~\ref{sec:structural:lemma}
	(the set $\OK$ in that Section) and additionally the large items $\optla$
	and the $O_{\eps}(1)$ items in $\optco^{cross}$ (in fact $\optla\subseteq\optco^{cross}$).
	Note that $|K(0)|\le O_{\eps}(1)$. If $p(K(0))\le\eps\cdot p(OPT)$
	then we can simply remove these rectangles (losing only a factor of $1+\eps$)
	and then apply the remaining argumentation exactly as above and we are done.
	Therefore, from now on suppose that $p(K(0))>\eps\cdot p(OPT)$.

\subsubsection{Definition of grid and corridor partition}\label{sec:weighted:shifting:grid}

Assume we are in round $t$ of this process, i.e., we defined $K(t)$ in the previous step (unless $t=0$, then $K(t)$ is defined as specified above) and assume that $p(K(t))>\eps \opt$ (otherwise, see Section \ref{sec:weighted:shifting:packing}).
We are now going to define the non-uniform grid $G(t)$ and the induced partition of the knapsack
into a collection of cells $\C_{t}$. The $x$-coordinates ($y$-coordinates) of the grid cells are the
$x$-coordinates ($y$-coordinates, respectively) of the items in
$\K(t)$. This yields a partition of the knapsack into $O_{\eps}(1)$
rectangular cells, such that each item of $\K(t)$ covers one or multiple
cells.
Note that an item might intersect many cells.

Similarly as above, we define constants $1\ge\epsl\ge\epss\ge\Omega_{\eps}(1)$
and apply a shifting step such that we can assume that for each item
$i\in OPT$ touching some cell $\cell$ we have that $w(i\cap\cell)\in(0,\epss w(\cell)]\cup(\epsl w(\cell),w(\cell)]$
and $h(i\cap\cell)\in(0,\epss h(\cell)]\cup(\epsl h(\cell),h(\cell)]$,
where $h(\cell)$ and $w(\cell)$ denote the height and the width
of the cell $\cell$ and $w(i\cap\cell)$ and $h(i\cap\cell)$ denote
the height and the width of the intersection of the rectangle $i$
with the cell $\cell$, respectively. For a cell $\cell$ denote by
$\opt(\cell)$ the set of rectangles that intersect $\cell$ in $\opt$.
We obtain a partition of $\opt(\cell)$ into $\optsm(\cell)$, $\optla(\cell)$,
$\optho(\cell)$, and $\optve(\cell)$: 
\begin{itemize}
	\item $\optsm(\cell)$ contains all items $i\in\opt(\cell)$ with $h(i\cap\cell)\le\epss h(\cell)$
	and $w(i\cap\cell)\le\epss w(\cell)$, 
	\item $\optla(\cell)$ contains all items $i\in\opt(\cell)$ with $h(i\cap\cell)>\epsl h(\cell)$
	and $w(i\cap\cell)>\epsl w(\cell)$, 
	\item $\optho(\cell)$ contains all items $i\in\opt(\cell)$ with $h(i\cap\cell)\le\epss h(\cell)$
	and $w(i\cap\cell)>\epsl w(\cell)$, and 
	\item $\optve(\cell)$ contains all items $i\in\opt(\cell)$ with $h(i\cap\cell)>\epsl h(\cell)$
	and $w(i\cap\cell)\le\epss w(\cell)$. 
\end{itemize}
We call a rectangle $i$ \emph{intermediate }if there is a cell $\cell$
such that $w(i\cap C)\in(\epss w(C),\epsl w(C)]$ or $h(i\cap C)\in(\epss w(C),\epsl w(C)]$.
Note that a rectangle $i$ is intermediate if and only if the last
condition is satisfied for one of the at most four cells that contain
a corner of $i$. 

\begin{lemma} For any constant $\eps>0$ and positive increasing
	function $f(\cdot)$, $f(x)>x$, there exist constant values
	$\epsl,\epss$, with $\eps\geq\epsl\geq f(\epss)\ge\Omega_{\eps}(1)$
	and $\epss\in\Omega_{\eps}(1)$ such that the total profit of intermediate
	rectangles is bounded by $\eps p(OPT)$. \end{lemma}

For each cell $C$ that is not entirely covered by some item in $K(t)$
we add all rectangles in $\optla(C)$ that are not contained in $\mathcal{K}(t)$ to $K(t+1)$.
Note that here, in contrast to before, we do \emph{not} remove small items from the packing but keep them.

Based on the items $\optsk(\C_{t}):=\cup_{\cell\in\C_{t}}\optho(\cell)\cup\optve(\cell)$
we create a corridor decomposition and consequently a box decomposition
of the knapsack. To make this decomposition clearer, we assume that
we first stretch the non-uniform grid into a uniform $[0,1]\times[0,1]$
grid. After this operation, for each cell $C$ and for each element
of $\optho(C)\cup\optve(C)$ we know that its height or width is at
least $\epsl\cdot\frac{1}{1+2|\K(t)|}$. We apply Lemma~\ref{lem:corridorPack-weighted}
on the set $\optsk(\C_{t})$ which yields a decomposition of the $[0,1]\times[0,1]$
square into at most $O_{\eps,\epsl,\K(t)}(1)=O_{\eps,\epsl}(1)$ corridors.
The decomposition for the stretched $[0,1]\times[0,1]$ square corresponds
to the decomposition for the original knapsack, with the same items
being intersected. Since we can assume that all items of $OPT$ are
placed within the knapsack so that they have integer coordinates,
we can assume that the corridors of the decomposition also have integer
coordinates. We can do that, because shifting the edges of the decomposition
to the closest integral coordinate will not make the decomposition
worse, i.e., no new items of $OPT$ will be intersected.

We add all rectangles in $\optsk(\C_{t})$ that are not contained
in a corridor (at most $O_{\eps}(1)$ many) and that are not contained
in $\K(t)$ to $K(t+1)$. 
The corridor partition
has the following useful property: we started with a fixed (optimal)
solution $\opt$ for the overall problem with a \emph{fixed placement
}of the items in this solution. Then we considered the items in $\optsk(\C_{t})$
and obtained the sets $\optco\subseteq\optsk(\C_{t})$ and $\optco^{cross}\subseteq\optco$.
With the mentioned fixed placement, apart from the $O_{\eps}(1)$
items in $\optco^{cross}$, each item in $\optco$ is contained in one corridor.
In particular, the items in $\optco$ do not overlap the items in
$\K(t)$. We construct now a partition of the knapsack into $O_{\eps}(1)$
corridors and $O_{\eps}(1)$ containers where each container contains
exactly one item from $\K(t)$. The main obstacle is that there can
be an item $i\in \K(t)$ that overlaps a corridor (see Figure~\ref{fig:weighted-circumvent}).
We solve this problem by applying the following lemma on each such
corridor. 
\begin{lemma} \label{lem:divide-open-corridors} Let $S$
	be an open corridor with $b(S)$ bends. Let $I'\subseteq OPT$ be
	a collection of items which intersect the boundary of $S$ with $I'\cap\optsk(\C_{t})=\emptyset$.
	Then there is a collection of $|I'|\cdot b(S)$ line segments $\mathcal{L}$
	within $S$ which partition $S$ into corridors with at most $b(S)$
	bends each such that no item from $I'$ is intersected by $\mathcal{L}$
	and there are at most $O_{\eps}(|I'|\cdot b(S))$ items of $\optsk(\C_{t})$
	intersected by line segments in $\L$. \end{lemma} \begin{proof} Let
	$i\in I'$ and assume without loss of generality that $i$ lies within a horizontal
	subcorridor $S_{i}$ of the corridor $S$.
	If the top or bottom edge $e$ of $i$ lies within $S_{i}$, we define
	a horizontal line segment $\ell$ which contains the edge $e$ and
	which is maximally long so that it does not intersect the interior
	of any item in $I'$, and such that it does not cross the boundary
	curve between $S_{i}$ and an adjacent subcorridor, or an edge of
	the boundary of $S$ (we can assume without loss of generality that $e$ does not intersect
	the boundary curve between $S_{i}$ and some adjacent subcorridor).
	We say that $\ell$ \emph{crosses} a boundary curve $c$ (or an edge
	$e$ of the boundary of $S$) if $c\setminus\ell$ (or $e\setminus\ell$)
	has two connected components.
	
	\begin{figure}
		\begin{centering}
			\includegraphics[scale=0.3]{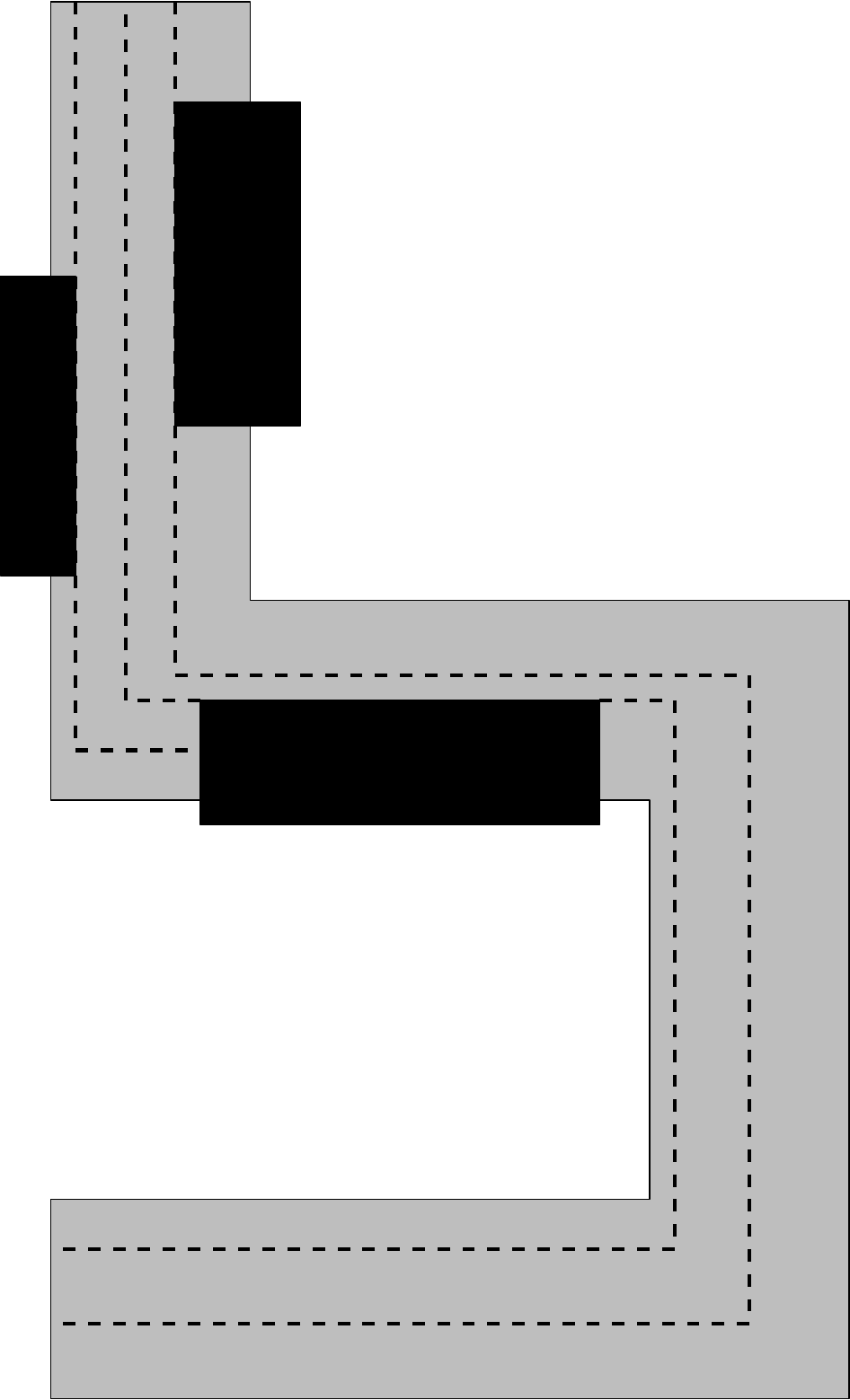} ~~~~~~~~~\includegraphics{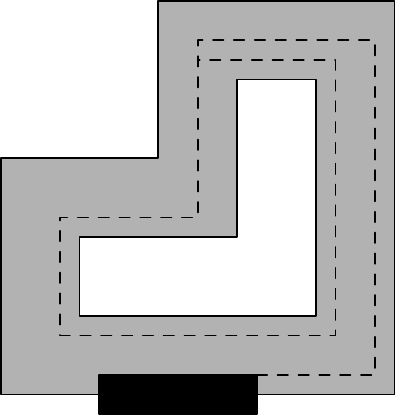}
			\par\end{centering}
		\caption{\label{fig:weighted-circumvent}Circumventing the items in $I'$,
			shown in black. The connected components between the dashed lines
			show the resulting new corridors. }
	\end{figure}
	
	We now ``extend'' each end-point of $\ell$ which does not lie at
	the boundary of some other item of $I'$ or at the boundary of $S$,
	we call such an end point a \emph{loose end}. For each loose end $x$
	of $\ell$ lying on the boundary curve $c_{ij}$ partitioning the
	subcorridors $S_{i}$ and $S_{j}$, we introduce a new line $\ell'$
	perpendicular to $\ell$, starting at $x$ and crossing the subcorridor
	$S_{j}$ such that the end point of $\ell'$ is maximally far away
	from $x$ subject to the constraint that $\ell'$ does not
	cross an item in $I'$, another boundary curve inside $S$, or the
	boundary of $S$. We continue iteratively. Since the corridor has
	$b(S)$ bends, after at most $b(S)$ iterations this operation will
	finish. We repeat the above operation for every item $i\in I'$, and
	we denote by $\L$ the resulting set of line segments, see Figure~\ref{fig:weighted-circumvent}
	for a sketch. Notice that $|\L|=b(S)\cdot|I'|$. By construction,
	if an item $i\in\optsk(\C_{t})$ is intersected by a line in $\L$
	then it is intersected parallel to its longer edge. Thus, each line
	segment in $\L$ can intersect at most $O_{\eps}(1)$ items of $\optsk(\C_{t})$.
	Thus, in total there are at most $O_{\eps,\epsl}(|I'|\cdot b(S))$
	items of $\optsk(\C_{t})$ intersected by line segments in $\L$.
\end{proof} 
We apply Lemma~\ref{lem:divide-open-corridors} to each
open corridor that intersects an item in $\K(t)$. We add
all items of $\Rsk(\C_{t})$ that are intersected by line segments
in $\L$ to $K(t+1)$. This adds $O_{\eps}(1)$ items in total to $K(t+1)$ since $|\K(t)|\in O_{\eps}(1)$
and $b(S)\le1/\eps$ for each corridor $S$. For closed corridors
we prove the following lemma. 
\begin{lemma} \label{lem:divide-closed-corridors}
	Let $S$ be a closed corridor with $b(S)$ bends. Let $\optsk(S)$
	denote the items in $\optsk(\C_{t})$ that are contained in $S$.
	Let $I'\subseteq OPT$ be a collection of items which intersect the
	boundary of $S$ with $I'\cap\optsk(\C_{t})=\emptyset$. Then there
	is a collection of $O_{\eps}(|I'|^{2}/\eps)$ line segments $\mathcal{L}$
	within $S$ which partition $S$ into a collection of closed corridors
	with at most $1/\eps$ bends each and possibly an open corridor with
	$b(S)$ bends such that no item from $I'$ is intersected by $\mathcal{L}$
	and there is a set of items $\optsk'(S)\subseteq\optsk(S)$ with $|\optsk'(S)|\le O_{\eps}(|I'|^{2})$
	such that the items in $\optsk(S)\setminus\optsk'(S)$ intersected
	by line segments in $\L$ have a total profit of at most $O(\eps)\cdot p(\optsk(\C_{t}))$.
\end{lemma} \begin{proof} Similarly as for the case of open corridors,
we take each item $i\in I'$ whose edge $e$ is contained in $S$,
and we create a path containing $e$ that partitions $S$. Here the
situation is a bit more complicated, as our newly created paths could
extend over more than $\frac{1}{\eps}$ bends inside $S$. In this
case we will have to do some shortcutting, i.e., some items contained
in $S$ will be crossed parallel to their shorter edge and we cannot
guarantee that their total number will be small. However, we will
ensure that the total weight of such items is small. We proceed as
follows (see Figure~\ref{fig:weighted-circumvent} for a sketch).

Consider any item $i\in I'$ and assume without loss of generality that $i$ intersects
a horizontal subcorridor $S_{i}$ of the closed corridor $S$. Let
$e$ be the edge of $i$ within $S_{i}$. For each endpoint of $e$
we create a path $p$ as for the case of closed corridors. If after
at most $b(S)\le1/\eps$ bends the path hits an item of $I'$ (possibly
the same item $i$), the boundary of $S$ or another path created
earlier, we stop the construction of the path. Otherwise, if $p$
is the first path inside of $S$ that did not finish after at most
$b(S)$ bends, we proceed with the construction of the path, only
now at each bend we check the total weight of the items of $\optsk(S)$
that would be crossed parallel to their shorter edge, if, instead
of bending, the path would continue at the bend to hit itself. From
the construction of the boundary curves in the intersection of two
subcorridors we know that for two bends of the constructed path, the
sets of items that would be crossed at these bends of the path are
pairwise disjoint. Thus, after $O(|I'|/\eps)$ bends we encounter
a collection of items $\optsk''(S)\subseteq\optsk(S)$ such that $p(\optsk''(S))\le\frac{\eps}{|I'|}p(\optsk(S))$,
and we end the path $p$ by crossing the items of $\optsk''(S)$.
This operation creates an open corridor with up to $O(|I'|/\eps)$
	bends. We divide it into up to $O(|I'|)$ corridors with up to $1/\eps$
	bends each. Via a shifting argument we can argue that this loses at
	most a factor of $1+\eps$ in the profit due to these items. When
we perform this operation for each item $i\in I'$ the total weight
of items intersected parallel to their shorter edge (i.e., due to
the above shortcutting) is bounded by $|I'|\cdot\frac{\eps}{|I'|}p(\optsk(S))=\eps\cdot p(\optsk(S))$.
This way, we introduce at most $O(|I'|^{2}/\eps)$ line segments.
Denote them by $\L$. They intersect at most $O_{\eps}(|I'|^{2})$
items parallel to their respective longer edge, denote them by $\optsk'(S)$.
Thus, the set $\L$ satisfies the claim of the lemma. \end{proof}
Similarly as for Lemma~\ref{lem:divide-open-corridors} we apply
Lemma~\ref{lem:divide-closed-corridors} to each closed corridor.
We add all items in the respective set $\optsk'(S)$ to the set $K(t+1)$ which yields
$O_{\eps}(1)$ many items.
The items in $\optsk(S)\backslash\optsk'(S)$ are removed from the instance, as their total profit is small.

\subsubsection{Partitioning corridors into boxes}\label{sec:weighted:shifting:boxes}

Then we partition the resulting corridors into boxes according to
the different cases described in Section~\ref{sec:structural:lemma}.
There is one difference to the argumentation above: we define that
the set $\optfa$ contains not only skewed items contained in the
respective subregions of a subcorridor, but \emph{all }items contained
in such a subregion. In particular, this includes items that might
have been considered as small items above. Thus, when we move items
from one subregion to the box associated to the subregion below (see
Remark~\ref{rem:fatPack}) then we move \emph{every }item that is
contained in that subregion. If an item is killed in one of the orderings
of the subcorridors to apply the procedure from Section~\ref{sec:structural:boxes}
then we add it to $K(t+1)$. Note that $|K(t+1)|\in O_{\eps,\epsl,\epst}(1)$
and $\K(t)\cap K(t+1)=\emptyset$.
Also note here that we ignore for the moment small items that cross the boundary curves of the subcorridors; they will be taken care of in Section \ref{sec:weighted:shifting:small}.

\subsubsection{Partitioning boxes into containers}\label{sec:weighted:shifting:containers}

Then we subdivide the boxes into containers. We apply Lemma~\ref{lem:containerPack}
to each box with a slight modification. Assume that we apply it to
a box of size $a\times b$ containing a set of items $I_{box}$. Like
above we first remove the items in a thin strip of width $3\eps b$
such that via a shifting argument the items (fully!) contained in
this strip have a small profit of $O(\eps)p(I_{box})$. However, in
contrast to the setting above the set $I_{box}$ contains not only
skewed items but also small items. We call an item $i$ \emph{small
}if there is no cell $\cell$ such that $i\in\optla(\cell)\cup\optho(\cell)\cup\optve(\cell)$
and denote by $\optsm(\C_{t})$ the set of small items. When we choose
the strip to be removed we ensure that the profit of the removed skewed
\emph{and} small items is small. There are $O_{\eps}(1)$ skewed items
that partially (but not completely) overlap the strip whose items
we remove. We add those $O_{\eps}(1)$ items 
to $K(t+1)$. Small items that partially overlap the strip are taken care of later in Section \ref{sec:weighted:shifting:small}, we ignore them for the moment.
Then we apply Lemma~\ref{lem:structural_lemma_augm}. In contrast to the setting
above, we do not only apply it to the skewed items but apply it also
to small items that are contained in the box. Denote by $\optsm'(\C_{t})$
the set of small items that are contained in some box of the box partition.

Thus, we obtain an L\&C packing for the items in $\K(t)$, for a set
of items $\optsk'(\C_{t})\subseteq\optsk(\C_{t})$, and for a set
of items $\optsm''(\C_{t})\subseteq\optsm'(\C_{t})$ with total profit at least $(1-O(\eps))p(\optsk(\C_{t})\cup\optsm'(\C_{t}))$.

\subsubsection{Handling small items}\label{sec:weighted:shifting:small}

So far we ignored the small items in $\optsm''(\C_{t}):=\optsm(\C_{t})\setminus\optsm'(\C_{t})$.
This set consists of small items that in the original packing intersect
a line segment of the corridor partition, the boundary of a box, or
a boundary curve within a corridor. We describe now how to add them
into the empty space of the so far computed packing. First, we assign
each item in $\optsm''(\C_{t})$ to a grid cell. We assign each small
item $i\in\optsm''(\C_{t})$ to the cell $C$ such that in the original
packing $i$ intersects with $C$ and the area of $i\cap C$ is not
smaller than $i\cap C'$ for any cell $C'$ ($i\cap C'$ denotes the
part of $i$ intersecting $C'$ in the original packing for any grid
cell $C'$).

Consider a grid cell $C$ and let $\optsm''(C)$ denote the small
items in $\optsm''(\C_{t})$ assigned to $C$. Intuitively, we want
to pack them into the empty space in the cell $C$ that is not used
by any of the containers, similarly as above. We first prove an analog
of Lemma~\ref{lem:areaContainer} of the setting above.

\begin{lemma}\label{lem:areaContainerWeighted} Let $\cell$ be a cell.
	The total area of $C$ occupied by containers is at most $(1-2\eps)a(\cell)$.
\end{lemma}

\begin{proof}
	
	In our construction of the boxes we moved some of the items (within
	a corridor). In particular, it can happen that we moved some items
	into $C$ that were originally in some other grid cell $C'$. This
	reduces the empty space in $C$ for the items in $\optsm''(C)$. Assume
	that there is a horizontal subcorridor $H$ intersecting $C$ such
	that some items or parts of items within $H$ were moved into $C$
	that were not in $C$ before. Then such items were moved vertically
	and the corridor containing $H$ must intersect the upper or lower
	boundary of $C$. The part of this subcorridor lying within
	$C$ has a height of at most $\epsl\cdot h(C)$. Thus, the total area
	of $C$ lost in this way is bounded by $O(\epsl a(C))$ which includes
	analogous vertical subcorridors.
	
	Like in Lemma~\ref{lem:areaContainer} we argue that in each horizontal
	box of size $a\times b$ we remove a horizontal strip of height $3\eps b$
	and then the created containers lie in a box of height $(1-3\eps)(1+\epsau)b$.
	In particular, if the box does not intersect the top or bottom edge
	of $C$ then within $C$ its containers use only a box of dimension
	$a'\times(1-3\eps)(1+\epsau)b$ where $a'$ denotes the width of the
	box within $C$, i.e., the width of the intersection of the box with
	$C$. If the box intersects the top or bottom edge of $C$ then we
	cannot guarantee that the free space lies within $C$. However, the
	total area of such boxes is bounded by $O(\epsl a(C))$. We can apply
	a symmetric argument to vertical boxes. Then, the total area of $C$
	used by containers is at most $(1-3\eps)(1+\epsau)a(C)+O(\epsl a(C))\le(1-2\eps)a(C)$.
	This gives the claim of the lemma. \end{proof}

Next, we argue that the items in $\optsm''(C)$ have very small total
area. Recall that they are the items intersecting $C$ that are not
contained in a box. The total number of boxes and boundary curves
intersecting $C$ is $O_{\eps,\epsl}(1)$ and in particular, this
quantity does not depend on $\epss$. Hence, by choosing $\epss$
sufficiently small, we can ensure that $a(\optsm''(C))\le\eps a(C)$.
Then, similarly as in Lemma~\ref{lem:smallPack} we can argue that
if $\epss$ is small enough then we can pack the items in $\optsm''(C)$
using NFDH into the empty space within $C$.

\subsubsection{L\&C packings}\label{sec:weighted:shifting:packing}

We iterate the above construction, obtaining
pairwise disjoint sets $K(1),K(2),...$ until we find a set $K(k)$
such that $p(K(k))\le\eps\cdot OPT$. Since the sets $K(0),K(1),...$
are pairwise disjoint there must be such a value $k$ with $k\le1/\eps$.
Thus, $|\K(k-1)|\le O_{\eps}(1)$. We build the grid given by the
$x$- and $y$-coordinates of $\K(k-1)$,
giving a set of cells $\C_{k}$. As described above we define the corridor
partition, the partition of the corridors into boxes (with the different
orders to process the subcorridors as described in Section~\ref{sec:structural:boxes})
and finally into containers. Denote by $\optsm(\C_{k})$
	the resulting set of small items.

We consider the candidate packings based on the grid $\C_{k}$.
For each of the six candidate packings with a degenerate $L$ we can
pack almost all small items of the original packing. We define $I_{\mathrm{lc}}$
	and $I_{\mathrm{sc}}$ the sets of items in long and short subcorridors in the initial
corridor partition, respectively. Exactly as in the cardinality case,
a subcorridor is long if it is longer than $N/2$ and short otherwise.
As before we divide the items into fat and thin items and define the
sets $\SF$, $\LT$, and $\ST$ accordingly. Moreover, we define the
set $\LF$ to contain all items in $I_{\mathrm{lc}}$ that are fat in all candidate
packings \emph{plus} the items in $\K(k-1)$.

Thus, we obtain the respective claims of Lemmas~\ref{lem:onlyFat},
\ref{lem:noShort}, and \ref{lem:evenOdd} in the weighted setting.
For the following lemma let $\optsm:=\optsm(\C_{k})$.
\begin{lemma} \label{lem:weighted-apx}Let $\optrc$ the most profitable
	solution that is packed by an L\&C packing. 
	\begin{enumerate}\renewcommand{\theenumi}{\alph{enumi}}\renewcommand\labelenumi{(\theenumi)}
		\item $p(\optrc)\ge(1-O(\eps))(p(\LF)+p(\SF)+p(\optsm))$ 
		\item $p(\optrc)\geq(1-O(\eps))(p(\LF)+p(\SF)/2+p(\LT)/2+p(\optsm))$ 
		\item $p(\optrc)\geq(1-O(\eps))(p(\LF)+p(\SF)/2+p(\ST)/2+p(\optsm)).$ 
	\end{enumerate}
\end{lemma}
For the candidate packing for the non-degenerate-$L$ case
(Lemma~\ref{lem:ringCase} in Section~\ref{sec:structural:lemma}) we first add the small items as described
above. Then we remove the items in $\K(k-1)$.
Then, like above, with a random shift we delete items touching a horizontal
and a vertical strip of width $3\eps N$. Like before, each item $i$
is still contained in the resulting solution with probability $1/2-15\eps$
(note that we cannot make such a claim for the items in $\K(k-1)$).
For each small item we can even argue that it still contained in the
resulting solution with probability $1-O(\eps)$ (since it is that
small in both dimensions).
We proceed with the construction of the
boundary $L$ and the assignment of the items into it like in the
unweighted case.
\begin{lemma} \label{lem:weighted-apx2}For the solution
	$\optrc$ we have that $p(\optrc)\ge(1-O(\eps))(\frac{3}{4}p(\LT)+p(\ST)+\frac{1-O(\eps)}{2}p(\SF)+p(\optsm))$.
\end{lemma} When we combine Lemmas~\ref{lem:weighted-apx} and \ref{lem:weighted-apx2}
we conclude that $p(\optrc)\ge(17/9+O(\eps))p(OPT)$. Similarly
	as before, the worst case is obtained, up to $1-O(\eps)$ factors,
	when we have that $p(\LT)=p(\SF)=p(\ST)$, $p(\LF)=5p(\LT)/4$, and $p(\optsm)=0$. This completes the proof of Lemma~\ref{lem:apxNoRotation}.

\subsection{Main algorithm}

In this Section we present our main algorithm for the weighted case
of \tdk. It is in fact an approximation scheme for L\&C packings.
Its approximation factor therefore follows immediately from Lemma~\ref{lem:apxNoRotation}. 

Given $\epsilon>0$, we first guess the quantities $\epsl,\epss$,
the proof of Lemma~\ref{lem:item-classification} reveals that there
are only $2/\epsilon+1$ values we need to consider. We choose $\epsr:=\epsilon^{2}$
and subsequently define $\epsb$ according to Lemma~\ref{lem:boxProperties}.
Our algorithm combines two basic packing procedures. The first one
is PTAS to pack items into a constant number
of containers given by Theorem~\ref{thm:container_packing_ptas}.

The second packing procedure is the PTAS for the L-packing problem, see Theorem~\ref{thm:main:Lpacking}.

To use these packing procedures, we first guess whether the optimal
L\&C-packing due to Lemma~\ref{lem:apxNoRotation} uses a non-degenerate
boundary L. 
If yes, we guess a parameter $\ell$ which denotes the
minimum height of the vertical items in the boundary $L$ and the
minimum width of the horizontal items in the boundary $L$. For $\ell$
we allow all heights and widths of the input items that are larger than
$N/2$, i.e., at most $2n$ values. Let $\ilong$ be the items whose
longer side has length at least $\ell$ (hence longer than $N/2$).
We set the width of the boundary $L$ to be $\epsilon^{2}N$ and solve
the resulting instance $(L,\ilong)$ optimally using the PTAS for L-packings
due to Theorem~\ref{thm:main:Lpacking}. Then we enumerate all the possible subsets of
non-overlapping containers in the space not occupied by the boundary
$L$ (or in the full knapsack, in the case of a degenerate $L$),
where the number and sizes of the containers are defined as in Lemma
\ref{lem:containerPack}. In particular, there are at most $O_{\eps}(1)$
containers and there is a set of size $n^{O_{\eps}(1)}$ that
we can compute in polynomial time such that the height and the width
of each container is contained in this set. We compute an approximate
solution for the resulting container packing instance with items $\ishort=I\setminus \ilong$
using the PTAS from Theorem~\ref{thm:container_packing_ptas}. Finally,
we output the most profitable solution that we computed.

\section{Cardinality case without rotations}
\label{sec:tdk_car:refined}
In this section, we present a refined algorithm with approximation factor of $\frac{558}{325}+\eps<1.717$ for the cardinality case when rotations are not allowed. 

\begin{theorem}\label{thm:cardWorot}
	There exists a polynomial-time $\frac{558}{325}+\eps<1.717$-approximation algorithm for cardinality 2DK.
\end{theorem}

Along this section, we will use $\eps, \epsr, \epsl, \epss$ as defined in Section \ref{sec:weighted}.
Let us define $\efl=\sqrt \eps$. Note than $\efl \ge \eps \ge \epsl \ge \epss$.
For simplicity and readability of the section, sometimes we will slightly abuse the notation and for any small constant depending on $\eps, \epsr, \epsl, \epss$, we will just use $O({\efl})$. Also, since the profit of each item is equal to $1$, instead of $\profit(I)$ for a set of items $I$ we will just write $|I|$.
We will again consider the $\fontL\&C$ packing as in Section~\ref{sec:weighted} and use $\il, \is,$ $\LF, \LT,$ $\SF, \ST$ etc. as defined there. 
We will assume as in the proof of Lemma~\ref{lem:ringCase} that $\ell = \left(\frac{1}{2}+2\eps_{large}\right)N$. That way we make sure that no long rectangle belongs to a short subcorridor. 
Let us define $\ilopt:=\il\cap OPT$ and $\isopt:=\is\cap OPT$. 
Now we give a brief informal overview of the refinement and the cases before we go to the details.

\noindent{\bf Overview of refined packing.}
For the refined packing we will consider several $\fontL\&C$ packings.
Some of the packings are just extensions  of previous packings (such as from Theorem~\ref{thm:16/9-apx} and Lemma~\ref{lem:weighted-apx}).
Then we consider several other new $\fontL\&C$ packings where an $\fontL$-region is packed with items from $\il$ and the remaining region is used for packing items from $\is$ using Steinberg's theorem (See Theorem \ref{thm:steinberg}).
Note that in the definition of $\fontL\&C$ packing in Section \ref{sec:weighted}, we assumed the height of  the horizontal part of $L$-region and width of the vertical part of $L$-region are the same. However, for these new packings we will consider $L$-region where the height of  the horizontal part and width of vertical part may not be the same.
Now several cases arise depending on the structure and profit of the $\fontL$-region. To pack items in $\isopt$ we have three options:\\
1. We can pack items in $\is$ using Steinberg's theorem into one rectangular region. Then we need both sides of the region to be greater than $\frac12+2\epsl$. \\
2. We can pack items in $\is$ using Steinberg's theorem such that
vertical and horizontal items are packed separately into different vertical and horizontal rectangular regions inside the knapsack.\\
3. If $a(\isopt)$ is large, we might pack only a small region with items in $\ilopt$, and use the remaining larger space in the knapsack to pack significant fraction of profits from $\isopt$.\\
Now depending on the structure of $L$-packing and $a(\isopt)$, we arrive at several different cases.
If the $\fontL$-region has very small width and height, we have case (1).
Else if $\fontL$-region has very large width (or height), we have case (2B), where we pack nearly 
$\frac12|\ilopt|$  in the $\fontL$-region and then pack items from $\is$ in one large rectangular region.
Otherwise, we have case (2A), where either we pack only items from $\ilthin$ (See Lemma \ref{lem4cardgen}, used in case: (2Ai)) or nearly $3/4|\ilopt|$ (See Lemma \ref{lem4cardgen2}, used in cases (2Aii), (2Aiiia)) or in another case, we pack the vertical and horizontal items in $\isopt$ in two different regions through a more complicated packing (See case (2Aiiib)). The details of these cases will be clearer in the proof of Theorem \ref{thm:cardWorot}.\\

Now first, we start with some extensions of previous packings.
Note that by using analogous arguments as in the proof of Theorem~\ref{thm:16/9-apx}, we can derive the following inequalities leading to a $\left(\frac{16}{9}+O(\efl)\right)$-approximation algorithm.

\begin{equation} \label{lem1card}
|\optrc|\geq \frac{3}{4}|\ilopt|
\end{equation}

\begin{equation} \label{lem2card}
|\optrc|\geq \left(\frac{1}{2} -O(\efl)\right)|\ilopt|+\left(\frac{3}{4}-O(\efl)\right)|\isopt|
\end{equation}

Now from Lemma~\ref{lem:boxProperties}, items in $\isthin$ can be packed into two containers of size $\ell \times \eps_{ring}N$ and $\eps_{ring} N \times \ell$. We can adapt part of the results in Lemma~\ref{lem:weighted-apx} to obtain the following inequalities.
\begin{proposition} The following inequalities hold:
	\begin{equation} \label{lem3card} |\optrc| \geq (1-O(\eps_{\fontL}))(|\ilfat|+|\isfat|).\end{equation}
	\begin{equation} \label{lem4card} |\optrc| \geq (1-O(\eps_{\fontL}))(|\ilfat| + \frac{1}{2}(|\isfat| + |\ilthin|)).\end{equation} \end{proposition}
\begin{proof} Inequality~\eqref{lem3card} follows directly from Lemma~\ref{lem:weighted-apx} since $\LF \cup \SF \cup OPT_{small}= (\ilfat)\cup(\isfat)$ and both groups of sets are disjoint. Inequality~\eqref{lem4card} follows from Lemma~\ref{lem:noShort}: if we consider the sum of the number of packed rectangles corresponding to the $4$ subcases associated with the case ``short horizontal/short vertical'', then every $i \in \ilfat\subseteq \LF$ appears four times, every $i\in \isopt \cap \LF$ appears four times, every $i\in \SF$ appears twice and every $i\in \ilthin$ appears twice. After including a $(1-O(\eps_{\fontL}))$ fraction of $OPT_{small}$, and since $(\isopt\cap \LF) \cup \SF \cup OPT_{small}= \isfat$, the inequality follows by taking average of the four packings \end{proof}

The following lemma is a consequence of Steinberg's Theorem~\ref{thm:steinberg}.
\begin{lemma}\label{lem:smallStein}
	Let $I'$ be a set of items such that $\displaystyle\max_{i \in I'} \height(i) \le \left(\frac{1}{2}+2\eps_{large}\right)N$ and $\displaystyle\max_{i \in I'} \width(i) \le \left(\frac{1}{2}+2\eps_{large}\right)N$.
	Then for any $\alpha, \beta \le \frac{1}{2}-2\eps_{large}$, all of $I'$ can be packed into a knapsack  of width $(1-\alpha)N$ and height  $(1-\beta)N$ if
	$$a(I') \le \Big(\frac12 -(\alpha+\beta)\left(\frac12+2\epsl\right)-8\epsl^2 \Big)N^2.$$
\end{lemma}
\begin{proof}
From Steinberg's Theorem, $I'$ can be packed in a rectangle of size $(1-\alpha)N \times (1-\beta)N$
if

\begin{align*} 2 a(I') & \le (1-\alpha)(1-\beta)N^2 \\
	&- \left(2\left(\frac{1}{2}+2\eps_{large}\right)N  - (1-\alpha)N\right)_+ \left(2\left(\frac{1}{2}+2\eps_{large}\right)N - (1-\beta)N\right)_+ \\
	& = \Big(1 -(\alpha+\beta)(1-4\epsl)-16\epsl^2 \Big)N^2.
\end{align*}
\end{proof}

Now we prove a more general version of Lemma~\ref{lem:ringCase} which holds for the cardinality case. 
\begin{lemma}\label{lem4cardgen}
	If $a(\isfat) \le \gamma N^2$ for any $\gamma \le 1$, then 
	\begin{align*}
	|\optrc| \geq \,\, & \frac{3}{4}|\ilthin|+|\isthin|\\
	                   & +\min\left\{1,\frac{1-O(\efl)}{2 \gamma}\right\}|\isfat|.
	\end{align*}
\end{lemma}
\begin{proof}
	As in Lemma~\ref{lem:ringCase}, we can pack $\frac34|\ilthin|$ and $|\isthin|$ in a boundary $\fontL$-region plus two boxes on the other two sides of the knapsack
	and then a free square region with side length $(1- 3 \eps)N$ can be used to pack items from $\isfat$.
	From Lemma \ref{lem:smallStein}, any subset of rectangles of $\isfat$ with total area at most $(1-O(\efl))N^2/2$ can be packed into that square region of length $(1- 3 \eps)N$.
	Thus we sort rectangles from $\isfat$ in the order of nondecreasing area and iteratively pick them until their total area reaches $(1- O(\efl) -\eps_{small})N^2/2$. Using Steinberg's theorem, there exists a packing of the selected rectangles. If $2\gamma \le 1- O(\efl) -\eps_{small}$ then the profit of this packing is $|\isfat|$, and otherwise the total profit is at least $\frac{1-O(\efl)}{2 \gamma}|\isfat|$. The packing coming from Steinberg's theorem may not be container-based, but we can then use resource augmentation as in Lemma~\ref{lem:ringCase} to obtain an $\fontL\&C$ packing.
\end{proof}

Now the following  lemma will be useful when $a(\isopt)$ is large. 

\begin{lemma}\label{lem4cardgen2}
	If $a(\isopt) > \gamma N^2$ for any $\gamma \ge \frac34+\eps+\eps_{large}$, then 
	$$|\optrc|\geq \frac{3}{4}|\ilopt|+\frac{(3\gamma - 1 -O(\efl))}{4\gamma}|\isopt|.$$
\end{lemma}
\begin{proof}
	Similar to Lemma \ref{lem:LoftheRing} in Section~\ref{sec:tdk_car:simple}, we start from the optimal packing and move all rectangles in $\ilopt$ to the boundary such that all of them are contained in a boundary ring.  Note that unlike the case when we only pack $\ilthin$ in the boundary region, the boundary ring formed by $\ilopt$ may have width or height $>> \eps_{ring} N$.
	Let us call the 4 stacks in the ring to be subrings.
	Let us assume that left and right subrings have width $\alpha_1 N$ and $\alpha_2 N$ respectively and 
	bottom and top subrings have height $\beta_1 N$ and $\beta_2 N$ respectively. 
	It is easy to see that we can arrange the subrings such that $\alpha_1, \alpha_2, \beta_1, \beta_2 \le 1/2$.
	
	As $a(\isopt) > \gamma N^2$, then $a(\ilopt) < (1-\gamma)N^2$.
	Let us define $\alpha=\alpha_1+\alpha_2$ and $\beta=\beta_1+\beta_2$.
	Then $(\alpha+\beta)N \cdot \frac{N}{2}  \le a(\ilopt)$. Hence, $\frac{\alpha+\beta}{2} < 1-\gamma$.
	Thus we get the following two inequalities:
	\begin{equation} \label{steinLar1}
	(\alpha+\beta) \le 2(1-\gamma)
	\end{equation} 
	\begin{equation} \label{steinLar2}
	a(\isopt) \le \left(1- \frac{(\alpha+\beta)}{2}\right)N^2
	\end{equation} 
	Now consider the case when we remove the top horizontal subring and construct a boundary \fontL-region as in Lemma~\ref{lem:LoftheRing}. We will assume that rectangles in the $\fontL$-region are pushed to the left and bottom as much as possible. Then, the boundary $\fontL$-region has width $(\alpha_1+\alpha_2)N$ and height $\beta_1 N$. We will use Steinberg's theorem to show the existence of a packing of rectangles from $\isopt$ in a subregion of the remaining space with width $N-(\alpha_1+\alpha_2+\eps)N$ and height $N-(\beta_1+\eps)N$, and use the rest of the area for resource augmentation to get an $\fontL\&C$-based packing.
	Since $\gamma\ge \frac34+\eps+\eps_{large}$, we have that $\alpha+\beta+2\eps \le 1/2-2\eps_{large}$.
	So $\alpha+\eps \le 1/2-2\eps_{large}$ and $\beta+\eps \le 1/2-2\eps_{large}$.
	Thus from Lemma \ref{lem:smallStein}, in the region with width $N-(\alpha_1+\alpha_2+\eps)N$ and height $N(1-\beta_1-\eps)$ we can pack rectangles from $\isopt$ of total area at most $\Big(\frac12 -\frac{(\alpha_1+\alpha_2+\beta_1)}{2} - O(\efl) \Big)N^2$. 
	Hence, we can take the rectangles in $\isopt$ in the order of nondecreasing area until their total area reaches 
	$\Big(\frac12 -\frac{(\alpha_1+\alpha_2+\beta_1)}{2} -O(\efl) - \eps_{small} \Big)N^2$ and pack at least $|\isopt| \cdot \frac{(\frac 12 - \frac{(\alpha_1 +\alpha_2+\beta_1)}{2}-O(\efl)-\eps_{small})}{(1-\frac{(\alpha+\beta)}{2})}$ using Steinberg's theorem.
	If we now consider all the four different cases corresponding to removal of the four different subrings and take the average of profits obtained in each case, 
	we pack at least \begin{eqnarray*} & & \frac3{4}|\ilopt|+ |\isopt| \cdot \left(\frac{(\frac 12 - \frac38(\alpha_1 +\alpha_2+\beta_1+\beta_2)-O(\efl)}{(1-\frac{(\alpha+\beta)}{2})}\right) \\ & \ge & \frac3{4}|\ilopt|+ |\isopt| \cdot \left(\frac{(\frac 12 - \frac38(\alpha+\beta) -O(\efl))}{(1-\frac{(\alpha+\beta)}{2})}\right) \\ & \ge & \frac3{4}|\ilopt|+ |\isopt| \cdot \frac{3\gamma -1-O(\efl)}{4\gamma},
	\end{eqnarray*} where the last inequality follows from \eqref{steinLar1} and the fact that the expression is decreasing as a function of $(\alpha + \beta)$.\end{proof}
Now we start with the proof of Theorem \ref{thm:cardWorot}.
\begin{proof}[Proof of Theorem~\ref{thm:cardWorot}]
	In the refined analysis, we will consider different solutions and show that the best of them always achieves the claimed approximation guarantee. We will pack some rectangles in a boundary $\fontL$-region (either a subset of only $\ilthin$ or a subset of $\ilopt$) using the PTAS algorithm for boundary $\fontL$-region described in Section~\ref{sec:ptasL}, and in the remaining area of the knapsack (outside of the boundary $\fontL$-region), we will pack a subset of rectangles from $\isopt$.
	
	\begin{figure}[t!]
		\centering
		\includegraphics[width=14cm]{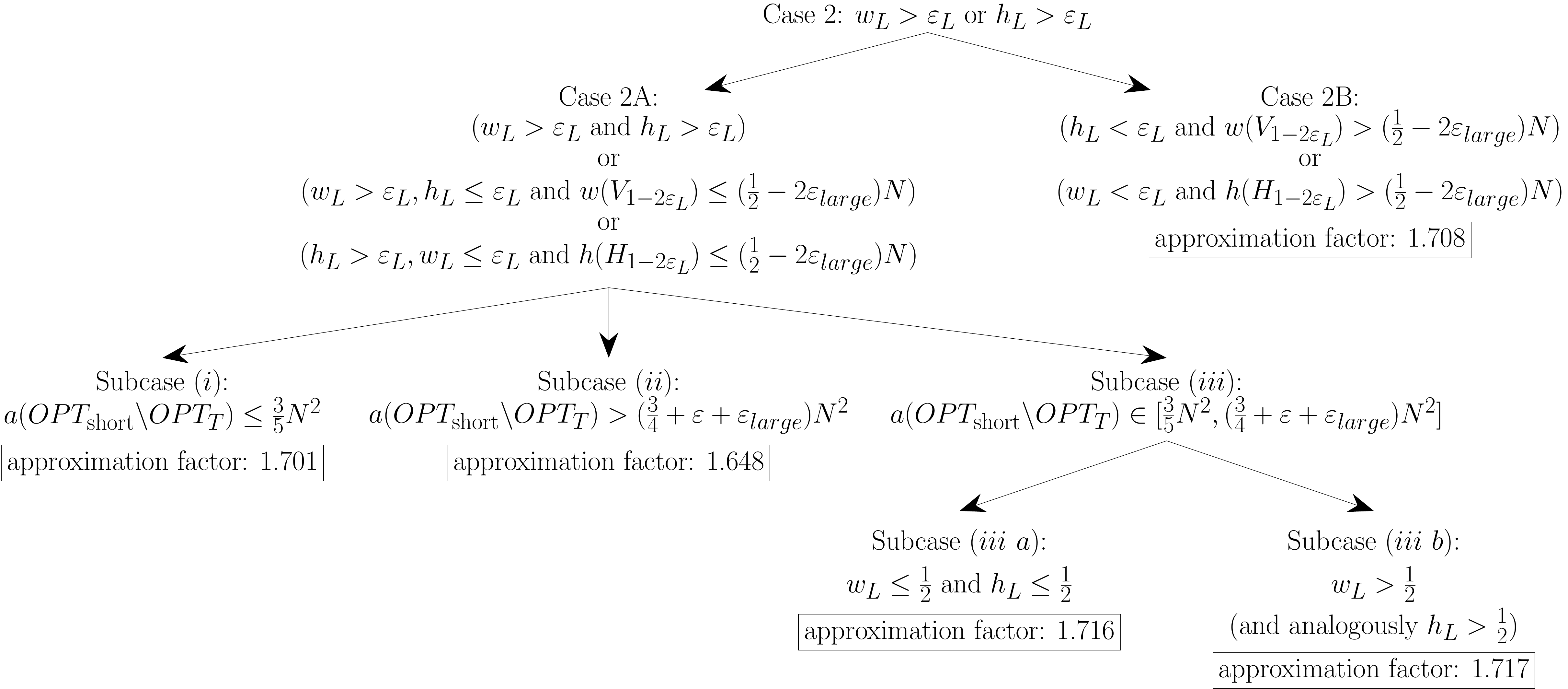}
		\caption{Summary of the cases.}
		\label{fig:cardworotcaseover}
	\end{figure}%

	Consider the ring as constructed in the beginning of the proof of Lemma \ref{lem4cardgen2}.
	Then we remove the least profitable subring and repack the remaining rectangles from $\ilopt$ in a boundary $\fontL$-region.  
	Without loss of generality, assume that the horizontal top subring was the least profitable subring. The other cases are analogous.
	We will use the same notation as in Lemma~\ref{lem4cardgen2}, and also define $w_{\fontL}=(\alpha_1+\alpha_2), h_{\fontL}= \beta_1$.
	Now let us consider two cases (see Figure~\ref{fig:cardworotcaseover} for the overview of the subcases of case 2).\\
	
	\noindent \textbf{$\bullet$ Case 1. $w_{\fontL} \le \eps_{\fontL} , h_{\fontL} \le \eps_{\fontL} $.} \\
	In this case, following the proof of Lemma \ref{lem4cardgen} (using $\gamma=1$), we can pack $\frac{3}{4}|\ilopt|+|\isthin|+\frac{1-O(\eps_{\fontL})}{2}|\isfat|$.
	This along with inequalities \eqref{lem2card}, \eqref{lem3card} and \eqref{lem4card} gives 
	\begin{equation}\label{eqcardcase1}
	|\optrc|\geq \left(\frac{5}{8}-O(\eps_{\fontL})\right)|OPT|
	\end{equation}

	\noindent \textbf{$\bullet$ Case 2. $w_{\fontL} > \eps_{\fontL} $ or  $h_{\fontL} > \eps_{\fontL} $.}\\
	Let $V_{1- 2 \eps_{\fontL}}$ be the set of vertical rectangles having height strictly larger than $(1-2 \eps_{\fontL})N$. Let us define $w(V_{1-2 \eps_{\fontL}})=\sum_{i \in V_{1-2 \eps_{\fontL}}} \width(i)$. 
	Similarly, let $H_{1- 2 \eps_{\fontL}}$ be the set of vertical rectangles of width strictly larger than $(1-2 \eps_{\fontL})N$ and $h(H_{1-2 \eps_{\fontL}})=\sum_{i \in H_{1-2 \eps_{\fontL}}} \height(i)$. \\
	\textit{$\Diamond$  Case 2A.  \Big($w_{\fontL} > \eps_{\fontL} $ and  $h_{\fontL}  > \eps_{\fontL} $\Big) or \Big($w_{\fontL} > \eps_{\fontL} , h_{\fontL} \le \eps_{\fontL} , $ and $w(V_{1-2\eps_{\fontL}})\le \left(\frac{1}{2}-2\eps_{large}\right) N$\Big) or \Big($h_{\fontL}  > \eps_{\fontL} , w_{\fontL} \le \eps_{\fontL} , $ and $h(H_{1-2\eps_{\fontL}})\le \left(\frac{1}{2}-2\eps_{large}\right) N$\Big).}\\
	We show that if any one of the above three conditions is met, then we can pack $\frac{3(1-O(\eps))}{4}|\ilopt|+|\isthin|$ in a boundary $\fontL$-region of width at most $ (w_{\fontL}+\eps_{\fontL}) N$ and height at most $(h_{\fontL}+\eps_{\fontL}) N$ , and 
	then in the remaining area we will pack some rectangles from $\isfat$ using Steinberg's theorem and resource augmentation. 
	
	\begin{figure*}[t!]
		\centering
		\begin{subfigure}[b]{0.4\textwidth}
			\centering
			\includegraphics[width=2.5in]{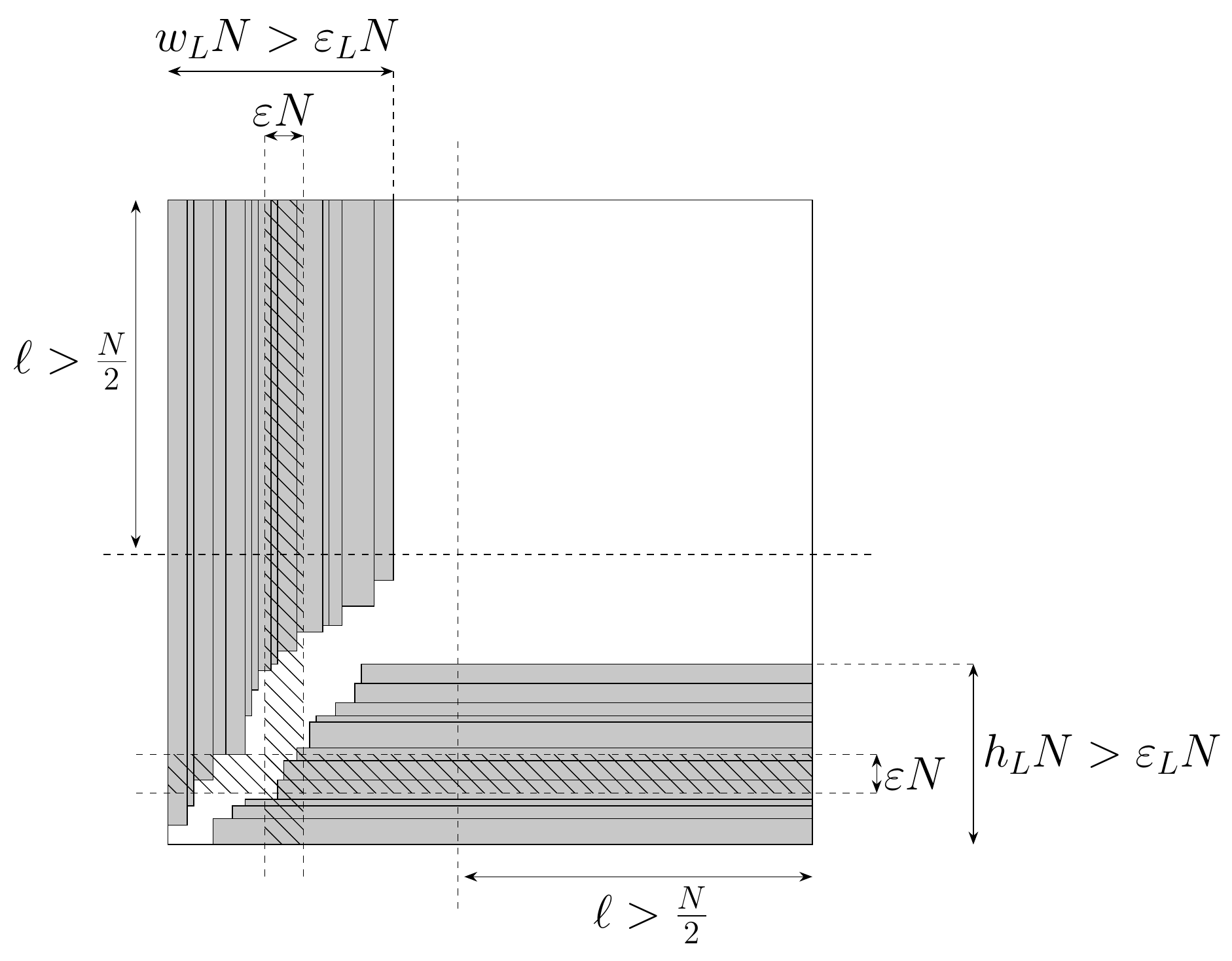}
			\caption{Packing of $\fontL$-region using rectangles from $\ilopt$. Striped strips are cheapest $\eps N$-width and cheapest $\eps N$-height.}
			\label{fig:cardworotcase2a1}
		\end{subfigure}%
		\hspace{40pt}
		\begin{subfigure}[b]{0.35\textwidth}
			\centering
			\includegraphics[width=2.5in]{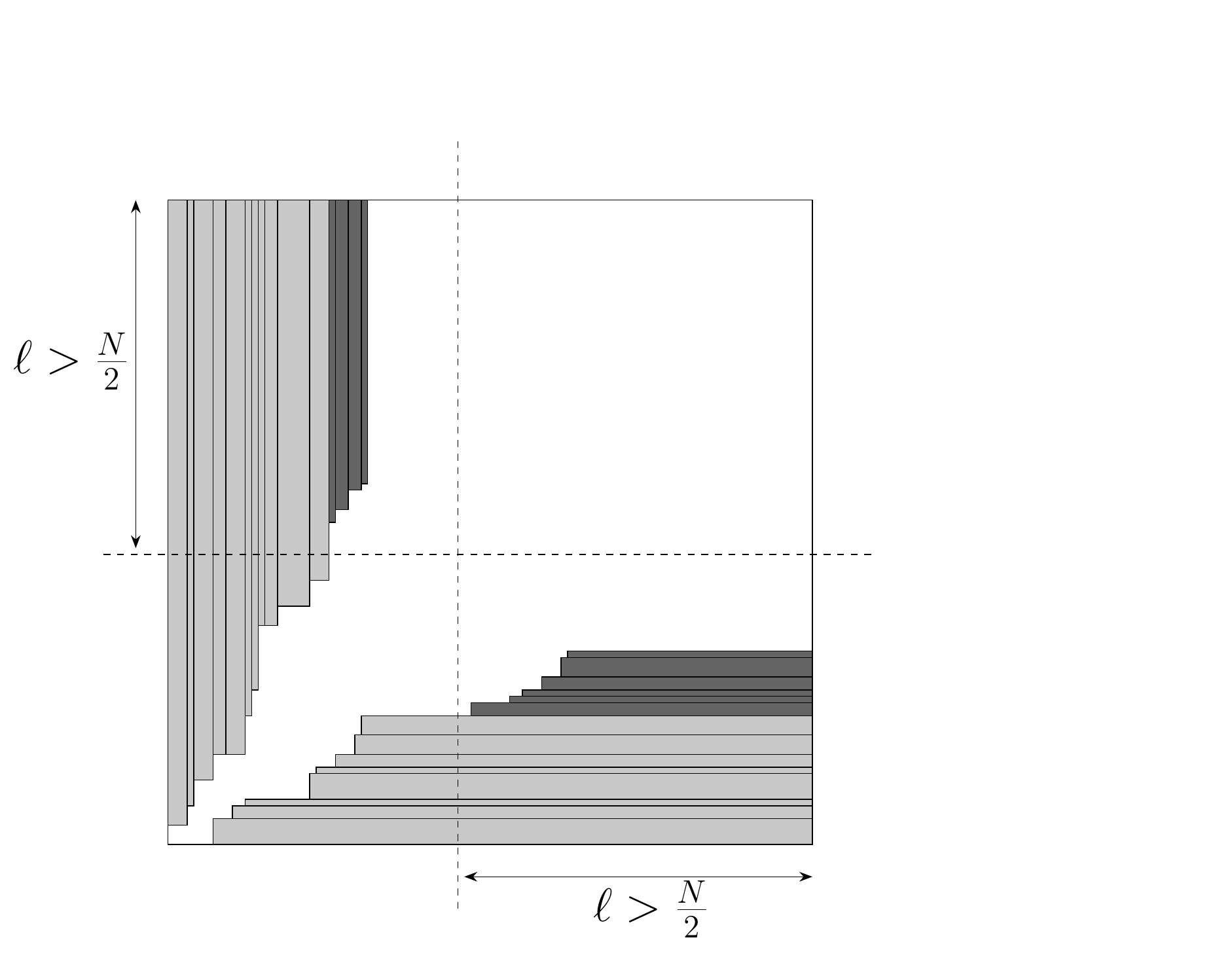}
			\caption{Packing of rectangles in $\ilopt \cup (\isthin)$. Dark gray rectangles are from $\isthin$. }
			\label{fig:cardworotcase2a2}
		\end{subfigure}
		~
		\caption{The case for $w_{\fontL} > \eps_{\fontL} $ and  $h_{\fontL} > \eps_{\fontL}$.}
	\end{figure*}
	
	\begin{figure*}[t!]
		\centering
		\begin{subfigure}[b]{0.4\textwidth}    
			\centering
			\includegraphics[width=2.5in]{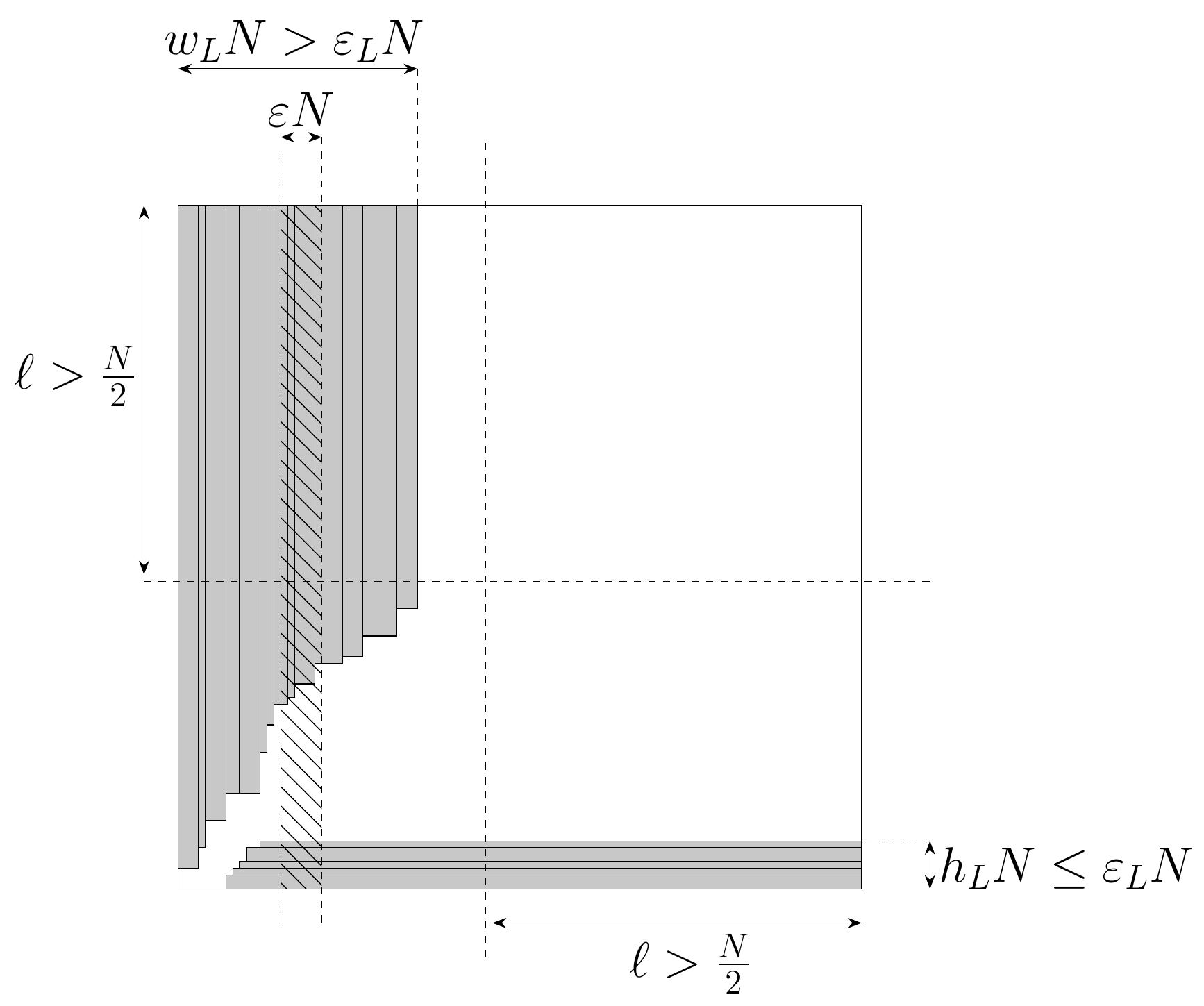}
			\caption{Packing of $\fontL$-region using rectangles from $\ilopt$. Striped strip is the cheapest $\eps N$-width strip.}
			\label{fig:cardworotcase2a3}
		\end{subfigure}%
		\hspace{40pt}
		\begin{subfigure}[b]{0.35\textwidth}   
			\centering
			\includegraphics[width=2.5in]{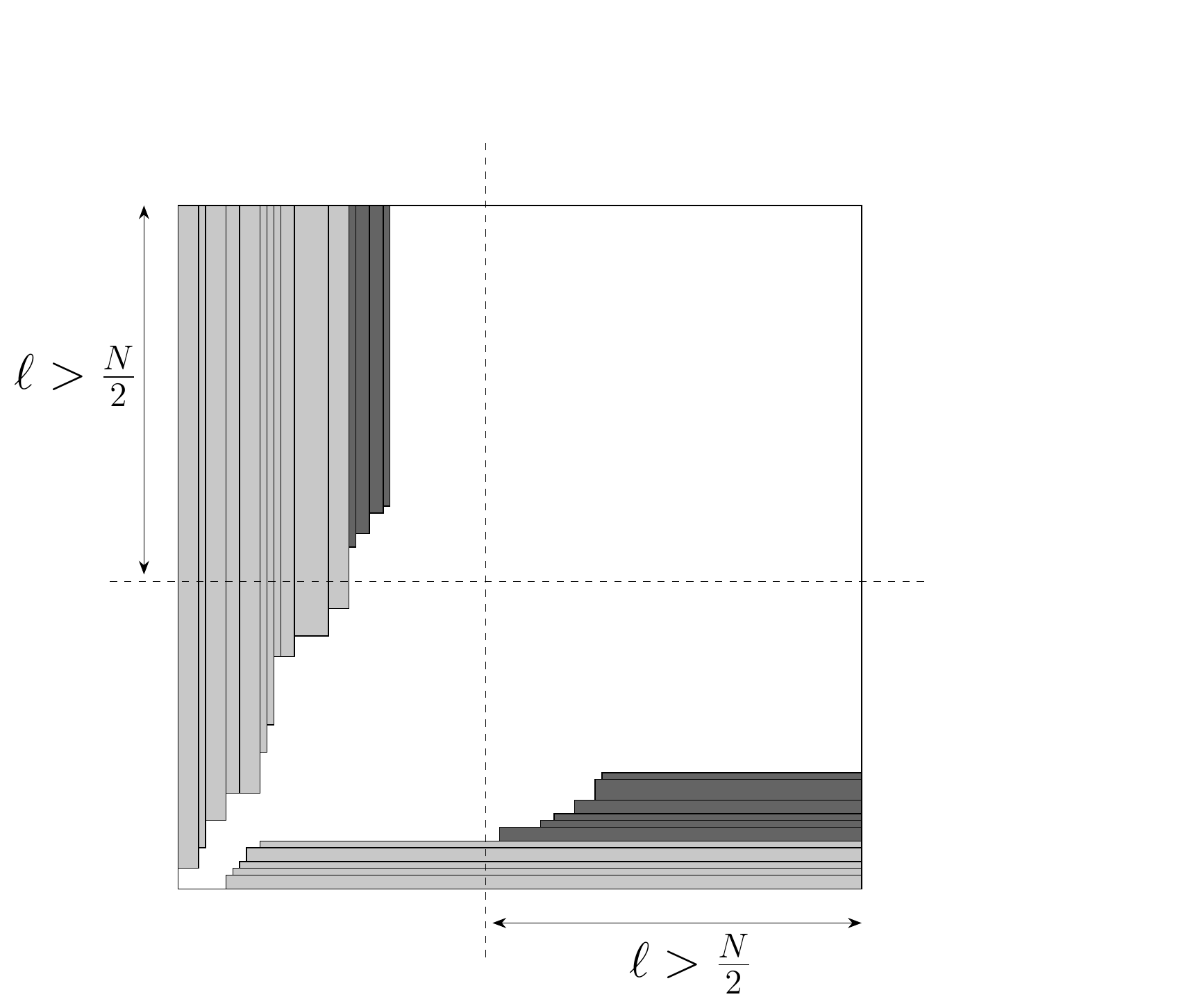}
			\caption{Packing of rectangles in $\ilopt \cup (\isthin)$. Dark gray rectangles are from $\isthin$.}
			\label{fig:cardworotcase2a4}
		\end{subfigure}
		~
		\caption{The case for $w_{\fontL} > \eps_{\fontL} $ and  $h_{\fontL} \le \eps_{\fontL}$.}
	\end{figure*}
	
	\noindent {\bf Packing of subset of rectangles from $\ilopt \cup (\isthin)$ into $\fontL$-region.}\\
	If ($w_{\fontL} > \eps_{\fontL} $ and  $h_{\fontL} > \eps_{\fontL} $), we partition the vertical part of the $\fontL$-region into consecutive strips of width $\eps N$.
	Consider the strip that intersects the least number of vertical rectangles from $\ilopt$ among all strips, and we call it to be the {\em cheapest $\eps N$-width vertical strip} (See Figure \ref{fig:cardworotcase2a1}). 
	Clearly the cheapest $\eps N$-width vertical strip intersects at most a $\frac{\eps+2\epss}{\eps_\fontL} \le O(\efl)$ fraction of the rectangles in the vertical part of $\fontL$-region, so we can remove all such vertical rectangles intersected by the cheapest $\eps N$-width vertical strip at a small loss of profit.
	Similarly, we remove the horizontal rectangles intersected by the cheapest $\eps N$-height horizontal strip in the boundary $\fontL$-region and push the remaining rectangles in the $\fontL$ region to the bottom and left. 
	We now pack the horizontal container for $\isthin$ above the horizontal part of $\fontL$ region and the vertical container for $\isthin$ to the right of the vertical part of the $\fontL$-region (See Figure \ref{fig:cardworotcase2a2}). \\
	In the other case ($w_{\fontL} > \eps_{\fontL} , h_{\fontL} \le \eps_{\fontL} $ and $w(V_{1-2\eps_{\fontL}})\le \left(\frac{1}{2}-2\eps_{large}\right)N$), we can again remove the  cheapest $\eps N$-width vertical strip  in the boundary $\fontL$-region and pack the vertical container for $\isthin$ there (See Figure \ref{fig:cardworotcase2a3}). Now we show how to pack horizontal items from $\isthin$.
	In the packing of the boundary $\fontL$-region, we can assume that the vertical rectangles are sorted non-increasingly by height from left to right and pushed upwards until they touch the top boundary.
	Then, since $w(V_{1- 2 \eps_{\fontL}})\le \left(\frac{1}{2}-2\eps_{large}\right)N$ and ($h_{\fontL} \le \eps_{\fontL}$), the region $\left[\left(\frac{1}{2}-2\eps_{large}\right)N, N\right] \times \left[\eps_{L} N, 2\eps_{L} N\right]$ will be completely empty and thus we will have enough space to pack the horizontal container for $\isthin$ on top of the horizontal part of the $\fontL$-region (See Figure \ref{fig:cardworotcase2a4}). 
	The last case, when $w_{\fontL} \le \eps_{\fontL}$, is analogous.
	Thus we have a packing in  boundary $\fontL$-region of width at most $(w_{\fontL}+\eps_{\fontL}) N$ and height at most $(h_{\fontL}+\eps_{\fontL}) N$ with total profit at least $\frac{3(1-O(\eps))}{4}(|\ilopt|)+|\isthin|$, implying that
	\begin{equation}
	\label{CardCase2a}
	|\optrc| \ge \frac{3(1-O(\eps))}{4}(|\ilopt|)+|\isthin|
	\end{equation} 
	
	\noindent{\bf Packing of subset of rectangles from $\isfat$ into  the remaining region.}\\
	Note that after packing rectangles of cardinality at least $\frac{3(1-O(\eps))}{4}|\ilopt|+|\isthin|$ in the boundary $\fontL$-region, the remaining rectangular region, let us call it to be $\mathcal{R}_{container}$, of width $(1-w_{\fontL}-\eps_{\fontL})N$ and height $(1-h_{\fontL}-\eps_{\fontL})N$ is completely empty.
	Now we will show the existence of a packing of some rectangles from $\isfat$ in the remaining space of the packing (even using some space from the $\fontL$-boundary region).
	Let $\isfhor := (\isfat) \cap \Rho$ and $\isfver := (\isfat) \cap \Rve$. 
	Let us assume without loss of generality that vertical rectangles are shifted as much as possible to the left and top of the knapsack and horizontal are pushed as much as possible to the right and bottom. We divide the analysis in three subcases depending on $a(\isfat)$. \\
	{\em $-$  Subcase (i).} If $a(\isfat) \le \frac35 N^2$, from inequalities \eqref{lem2card}, \eqref{lem3card}, \eqref{lem4card}, \eqref{CardCase2a} and Lemma \ref{lem4cardgen} we get 
	\begin{equation}\label{eqcardcase2a0}
	|\optrc|\geq \left(\frac{127}{216}-O(\eps_{\fontL})\right)|OPT|
	\end{equation}\arir{Gives $1.701$}
	
	\noindent {\em $-$ Subcase (ii).} If $a(\isfat) > (\frac34+\eps+\eps_{large}) N^2$, from inequalities  \eqref{lem2card}, \eqref{lem3card}, \eqref{lem4card}, \eqref{CardCase2a} and Lemma \ref{lem4cardgen2} we get 
	\begin{equation}\label{eqcardcase2a3}
	|\optrc|\geq \left(\frac{17}{28}-O(\efl)\right)|OPT|
	\end{equation} \arir{Gives 1.648. }\\
	
	\noindent{\em $-$ Subcase (iii).} Finally, if $\frac35 N^2 \le a(\isfat) \le (\frac34+\eps+\eps_{large}) N^2$, from inequality \eqref{steinLar1} we get, $\alpha + \beta \le 2(1-\frac35) \le \frac45$. Now we consider two subcases.\\
	{\em $\odot$ Subcase (iii a): $w_{\fontL}\le \frac{1}{2}$ and $h_{\fontL}\le \frac{1}{2}$.}
	Note that in this case if $w_{\fontL}\ge \frac{1}{2}-2\eps_{large}-2\eps_{\fontL}$ or $h_{\fontL}\ge \frac{1}{2}-2\eps_{large}-2\eps_{\fontL}$,
	we can remove the cheapest $2(\eps_{\fontL}+\epsl) N$-width vertical strip and the cheapest $2(\eps_{\fontL}+\epsl) N$-height horizontal strip from the $\fontL$-region. 
	Otherwise we have $w_{\fontL}<\frac{1}{2}-2\eps_{large}-2\eps_{\fontL}$ and $h_{\fontL}< \frac{1}{2}-2\eps_{large}-2\eps_{\fontL}$.
	So there is free rectangular region that has both side lengths at least $N(\frac12+2\eps_{large}+\eps_\fontL)$; we will keep $\eps_\fontL N$ width and $\eps_\fontL N$  height for resource augmentation and use the rest of the rectangular region (with both sides length at least $\left(\frac{1}{2}+2\eps_{large}\right)N$) for showing existence of a packing using Steinberg's theorem.
	Note that this rectangular region has area at least $N^2(1-w_{\fontL}-2\eps_{\fontL})(1-h_{\fontL}-2\eps_{\fontL})$. 
	Thus by using Steinberg's theorem, we can pack either rectangles from $\isfhor$ (by sorting them non-decreasingly by area and picking them iteratively) having total area at least $min\{ a(\isfhor), \frac{N^2(1-w_{\fontL}-2\eps_{\fontL})(1-h_{\fontL}-2\eps_{\fontL})}{2} -\epss\}$ or rectangles from $\isfver$ with total area at least $min\{ a(\isfver), \frac{N^2(1-w_{\fontL}-2\eps_{\fontL}) (1-h_{\fontL}-2\eps_{\fontL})}{2}-\epss\}$. 
	We claim that if we keep the best of the two packings,  we can always pack at least $\left(\frac{7}{48}-O(\eps_L)\right)|\isfat|$. 
	Note that it is sufficient to consider, $\frac{N^2(1-w_{\fontL}-O(\eps_{\fontL})) (1-h_{\fontL}-O(\eps_{\fontL}))}{2} < \min \{a(\isfhor), a(\isfver)\}$, as otherwise, we can pack even more rectangles from $\isfat$ than the claimed fraction.
	In this case, we pack at least
	\begin{align*} & \resizebox{\hsize}{!}{$\frac{N^2}{2}\left( \frac{(1-w_{\fontL}-O(\eps_{\fontL}))(1-h_{\fontL}-O(\eps_{\fontL}))}{2 a(\isfhor)}|\isfhor|+ \frac{(1-w_{\fontL}-O(\eps_{\fontL}))(1-h_{\fontL}-O(\eps_{\fontL}))}{2 a(\isfver)}|\isfver|\right)$} \\ \ge & \resizebox{0.465\hsize}{!}{$\frac{N^2}{2}\left(  \frac{(1-w_{\fontL}-O(\eps_{\fontL}))(1-h_{\fontL}-O(\eps_{\fontL}))}{2 a(\isfat)}|\isfat| \right)$} 
	\end{align*} where the inequality follows from the fact that $\frac{a}{b}+\frac{c}{d} \ge \frac{(a+c)}{(b+d)}$ for $a,b,c,d \ge 0$.
	Since $a(\isfat)\le (1-\frac{\alpha}{2}-\frac{\beta}{2})N^2 \le (1-\frac{w_{\fontL}}{2}-\frac{h_{\fontL}}{2})N^2$
	and $w_\fontL+ h_\fontL \le \alpha+\beta \le \frac45$, we can minimize the expression $$\frac{1}{2}\left(  \frac{N^2(1-w_{\fontL}-O(\eps_{\fontL}))(1-h_{\fontL}-O(\eps_{\fontL}))}{2 a(SF)}|\isfat| \right)$$ 
	over the domain $\{w_{\fontL} + h_{\fontL}\le \frac{4}{5}, 0\le w_{\fontL}\le \frac{1}{2}, 0\le h_{\fontL} \le \frac{1}{2}\}$, obtaining that this is at least $(\frac{7}{48}-O(\eps_{\fontL}))
	|\isfat|$ (value reached at $w_{\fontL}=\frac{1}{2}$ and $h_{\fontL} = \frac{3}{10}$). 
	This together with inequalities \eqref{lem2card}, \eqref{lem3card}, \eqref{lem4card}, \eqref{CardCase2a} and Lemma \ref{lem4cardgen} implies that
	\begin{equation}\label{eqcardcase2a1}
	|\optrc|\geq \left(\frac{215}{369}-O(\eps_{\fontL})\right)|OPT|
	\end{equation}\arir{Gives $1.7163$}
	
	\begin{figure}[t!]
		\centering
		\includegraphics[width=3.5in]{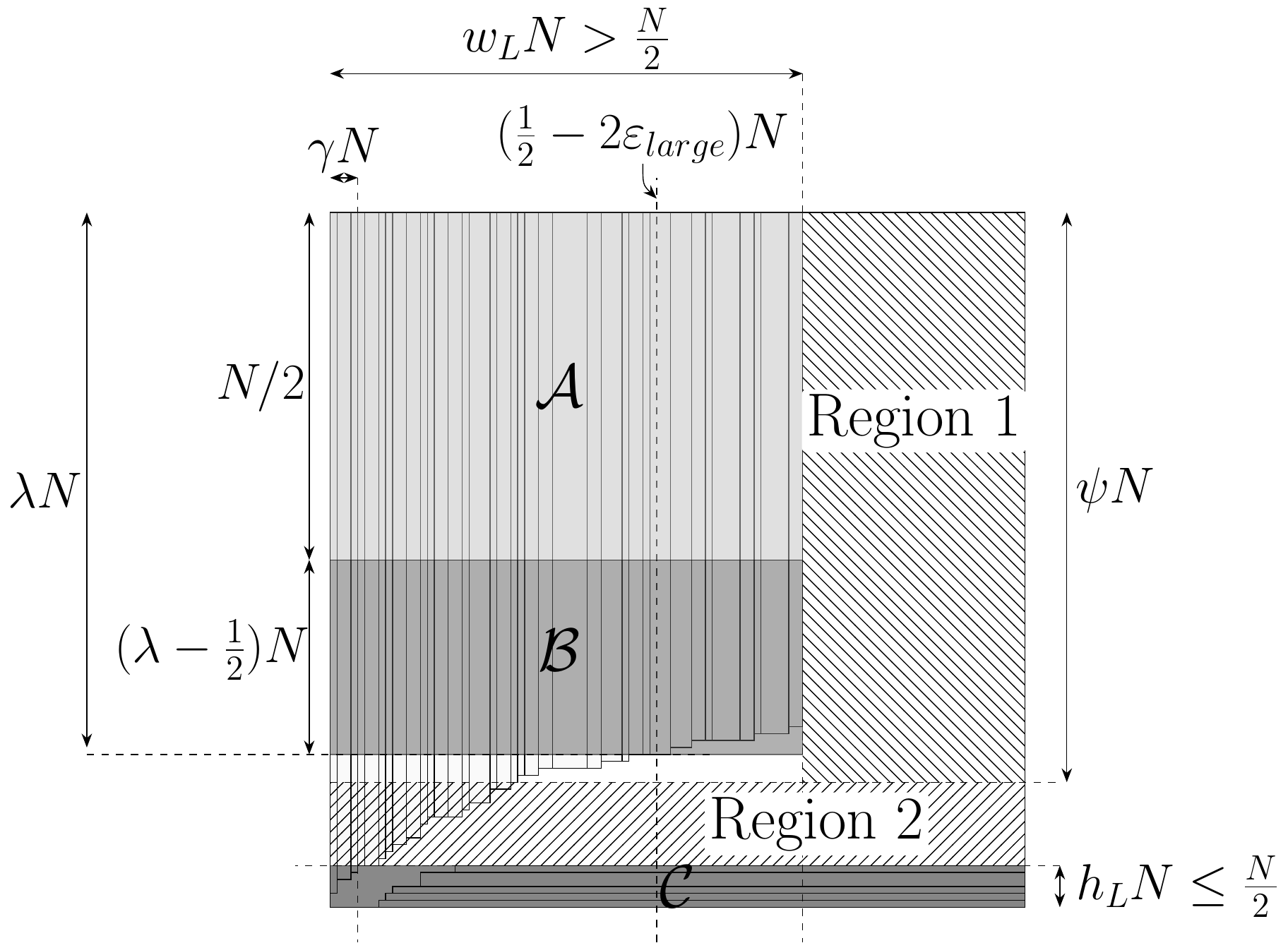}
		\caption{Case 2A(iii)b in the proof of Theorem~\ref{thm:cardWorot}}
		\label{fig:cardworotcase2aiii}
	\end{figure}%
	
	\noindent {\em $\odot$ Subcase (iii b): } $w_{\fontL}>\frac{1}{2}$ (then from inequality \eqref{steinLar1}, $h_{\fontL} \le \frac{3}{10}$).
	Note that $a(\ilopt) \le (1-\frac35)N^2=\frac25 N^2$.\\
	Let us define some parameters from the current packing to simplify the calculations. Let $\lambda N$ be the height of the tallest vertical rectangle in the packing that touches the vertical line $x=\left(\frac{1}{2}-2\eps_{large}\right)N$ and $\gamma N$ be the total width of vertical rectangles having height greater than $(1-h_\fontL)N$. We define also the following three regions in the knapsack: $\mathcal{A}$, the rectangular region of width $w_{\fontL}N$ and height $\frac{1}{2}N$ in the top left corner of the knapsack; $\mathcal{B}$, the rectangular region of width $w_{\fontL} N$ and height $(\lambda-\frac{1}{2})N$ below $\mathcal{A}$ and left-aligned with the knapsack; and $\mathcal{C}$, the rectangular region of width $N$ and height $h_{\fontL}N$ touching the bottom boundary of the knapsack. Notice that $\mathcal{A}$ is fully occupied by vertical rectangles, $\mathcal{B}$ is almost fully occupied by vertical rectangles except for the right region of width $w_{\fontL} N- \left(\frac{1}{2}-2\eps_{large}\right)N$, and at least half of $\mathcal{C}$ is occupied by horizontal rectangles (some vertical rectangles may overlap with this region).
	Our goal is to pack some rectangles from $\isfat$ in the \reflectbox{\ffmfamily{L}}-shaped region outside $\mathcal{A}\cup\mathcal{B}\cup\mathcal{C}$. Let $\psi \in [\lambda, 1-h_\fontL]$ be a parameter to be fixed. We will use, when possible, the following regions for packing items from $\isfat$: Region 1 on the top right corner of the knapsack with width $N(1-w_{\fontL})$ and height $\psi N$ and Region 2 which is the rectangular region $[0,N]\times [h_{\fontL} N, (1-\psi )\cdot N]$ (see Figure~\ref{fig:cardworotcase2aiii}). Region $1$ is completely free but Region 2 may overlap with vertical rectangles. 
	
	We will now divide Region 2 into a constant number of boxes such that: they do not overlap with vertical rectangles, the total area inside Region 2 which is neither overlapping with vertical rectangles nor covered by boxes is at most $O(\eps_\fontL) N^2$ and each box has width at least $\left(\frac{1}{2}+2\eps_{large}\right)N$ and height at least $\eps N$. That way we will be able to pack rectangles from $\isfver$ into the box defined by Region $1$ and rectangles from $\isfhor$ into the boxes defined inside Region 2 using almost completely its free space.
	In order to create the boxes inside Region 2 we first create a monotone chain by doing the following: Let $(x_1,y_1)=(\gamma N, h_{\fontL})$. Starting from position $(x_1,y_1)$, we draw an horizontal line of length $\eps_\fontL N$ and then a vertical line from bottom to top until it touches a vertical rectangle, reaching position $(x_2, y_2)$. From $(x_2, y_2)$ we start again the same procedure and iterate until we reach the vertical line $x=\left(\frac{1}{2}-2\eps_{large}\right)N$ or the horizontal line $y=(1-\psi )N$. 
	Notice that the area above the monotone chain and below $y=(1-\psi) N$ that is not occupied by vertical rectangles, is at most $\sum_i{\eps_{\fontL} N (y_{i+1}-y_i)} \le \eps_\fontL N^2$.
	The number of points $(x_i, y_i)$ defined in the previous procedure is at most $1/\eps_\fontL$. By drawing an horizontal line starting from each $(x_i,y_i)$ up to $(N,y_i)$, together with the drawn lines from the monotone chain and the right limit of the knapsack, we define $k\le 1/\eps_\fontL$ boxes. We discard the boxes having height less than $\eps N$, whose total area is at most $\frac{\eps}{\eps_\fontL}N^2\le \eps_\fontL N^2$, and have all the desired properties for the boxes. 
	
	Since the total area of rectangles in $\ilopt$ is at most $\frac{2}{5}N^2$, the total area not occupied by boxes in Region 2 is at most $N^2(\frac{2}{5} - \frac{1}{2}w_{\fontL} - (\lambda-\frac{1}{2})\frac{1}{2}  - \frac{1}{2}h_{\fontL})= N^2(\frac{13}{20}-\frac{w_{\fontL}}{2}-\frac{\lambda}{2}-\frac{h_{\fontL}}{2})$, which comes from the area we know is occupied for sure in regions $\mathcal{A}, \mathcal{B}$ and $\mathcal{C}$ by rectangles in $\ilopt$. This implies that the total area of the horizontal boxes is at least $N^2(1-\psi -h_{\fontL})-N^2(\frac{13}{20}-\frac{w_{\fontL}}{2}-\frac{\lambda}{2}-\frac{h_{\fontL}}{2})$ and the area of the vertical box is $(1-w_{\fontL})\psi$. 
	These two areas become equal if we can set $\psi =\frac{7+10(w_{\fontL}+\lambda -h_{\fontL})}{40-20w_{\fontL}}$. It is not difficult to verify that in this case $\psi \le 1-h_{\fontL}$.
	If $\frac{7+10(w_{\fontL}+\lambda -h_{\fontL})}{40-20w_{\fontL}} \ge \lambda$, then we set 
	$\psi= \frac{7+10(w_{\fontL}+\lambda -h_{\fontL})}{40-20w_{\fontL}}$. Otherwise we set $\psi=\lambda$.
	
	First, consider when $\psi= \frac{7+10(w_{\fontL}+\lambda -h_{\fontL})}{40-20w_{\fontL}}$.
	Since $\psi \ge \lambda$, the width of the boxes inside Region 2 is at least $\left(\frac{1}{2}+2\eps_{large}\right)N$ and the box in Region $1$ has height at least $\left(\frac{1}{2}+2\eps_{large}\right)N$. By using Steinberg's theorem, we can always pack in these boxes at least \begin{equation*} \resizebox{\hsize}{!}{$\left(\min\left\{1,\frac{\frac{1}{2}(N-w_{\fontL})\psi N}{a(\isfhor)}-\epss \right\} \right)|\isfhor| + \left(\min\left\{1,\frac{\frac{1}{2}(N-w_{\fontL})\psi N}{a(\isfver)}-\epss\right\}\right)|\isfver|.$}\end{equation*} 
	Note that from each box $B'$ of height $h (\ge \eps N)$, we can remove the cheapest $\eps h$-horizontal strip and use resource augmentation to get a container based packing with nearly the same profit as $B'$.
	Thus by performing a similar analysis to the one done in Subcase (iii a), and using the fact that $a(\isfat)\le N^2-(\frac{\alpha}{2} + (\lambda-\frac{1}{2})\frac{1}{2} + \frac{\beta}{2})N^2 \le N^2 - N^2(\frac{w_{\fontL}}{2} +(\lambda-\frac{1}{2})\frac{1}{2} + \frac{h_{\fontL}}{2})$, we can minimize the whole expression over the domain $\{\frac{w_{\fontL}}{2} + (\lambda-\frac{1}{2})\frac{1}{2}+\frac{h_{\fontL}}{2}\le \frac{2}{5}, \frac{1}{2}\le w_{\fontL}\le \frac{4}{5}, \frac{1}{2}\le \lambda \le 1, 0\le \beta \le \frac{1}{2}\}$ and prove that this solution packs at least \[ \left(\frac{3-O(\eps_L)}{4}\right)|\ilopt| + |\isthin| + \left(\frac{5}{36} - O(\eps_\fontL)\right) |\isfat|.\] Thus, using the above inequality along with \eqref{lem2card}, \eqref{lem3card}, \eqref{lem4card} and Lemma \ref{lem4cardgen}, we get 
	\begin{equation}\label{eqcardcase2a2}
	|\optrc|\geq \left(\frac{325}{558}-O(\eps_\fontL)\right)|OPT|
	\end{equation} \arir{Gives 1.7169}
	Finally, if $\psi = \lambda < \frac{7+10(w_{\fontL}+\lambda -h_{\fontL})}{40-20w_{\fontL}}$, we will not get equal area horizontal and vertical boxes for packing of items in $\isfat$. 
	In this case we change the width of the box inside Region $1$ to be $w_{\fontL}'<N(1-w_{\fontL})$ fixed in such a way that the area of this box is equal to the bound we have for the area of the boxes in Region $2$, i.e., $N^2(1-\lambda-h_{\fontL})-(\frac{13}{20}-\frac{w_{\fontL}}{2}-\frac{h_{\fontL}}{2}-\frac{\lambda}{2}+O(\eps_\fontL)) N^2$. Performing the same analysis as before, it can be shown that in this case we pack at least \[\left(\frac{(1-\lambda -h_{\fontL})N^2-(\frac{13}{20}-\frac{w_{\fontL}}{2}-\frac{h_{\fontL}}{2}-\frac{\lambda}{2}) N^2}{2a(\isfat)} - O(\eps_{\fontL})\right)|\isfat|,\] which is at least $(\frac{1}{6}-O(\eps_\fontL))|\isfat|$ over the domain $\{\frac{w_{\fontL}}{2} + (\lambda-\frac{1}{2})\frac{1}{2}+\frac{h_{\fontL}}{2}\le \frac{2}{5}, \psi < \lambda, \frac{1}{2} \le w_{\fontL}\le \frac{4}{5}, \frac{1}{2}\le \lambda \le 1, 0\le \beta \le \frac{3}{10}\}$ (and this solution is then better than \eqref{eqcardcase2a2}).\\
	\begin{figure*}[t!]
		\centering
		\begin{subfigure}[b]{0.35\textwidth}
			\centering
			\includegraphics[width=2.5in]{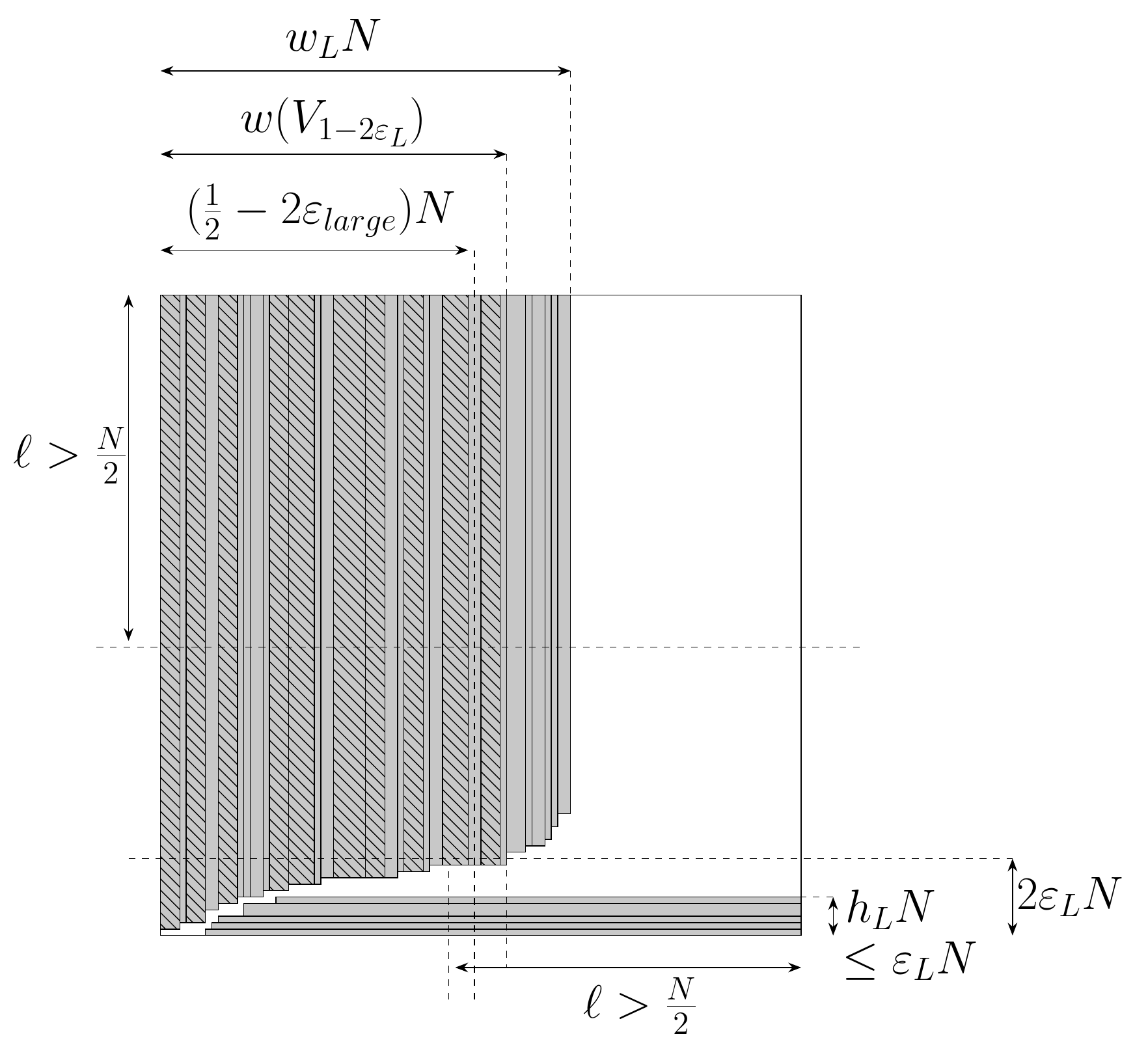}
			\caption{Packing of $\fontL$-region using rectangles from $\ilopt$. Striped rectangles are removed.}
			\label{fig:cardworotcase2b1}
		\end{subfigure}%
		\hspace{35pt}
		\begin{subfigure}[b]{0.4\textwidth}
			\centering
			\includegraphics[width=2.5in]{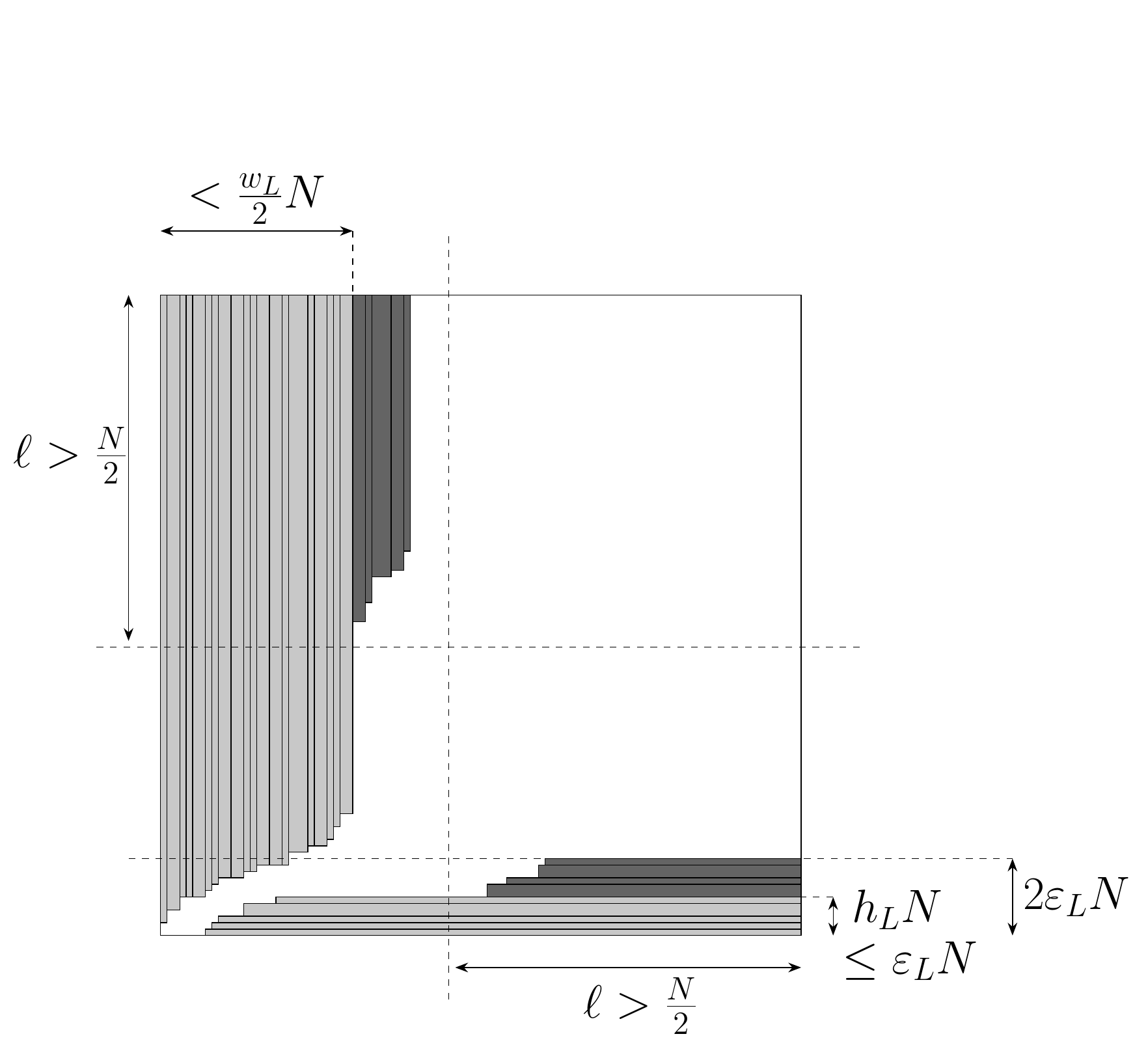}
			\caption{Packing of rectangles in $\ilopt \cup (\isthin)$. Dark gray rectangles are from $\isthin$. }
			\label{fig:cardworotcase2b2}
		\end{subfigure}
		~
		\caption{The case 2B.}
	\end{figure*}
	\textit{$\Diamond$ Case 2B. \Big($h_{\fontL} < \eps_{\fontL} $ and $w(V_{1-2\eps_{\fontL}}) > \left(\frac{1}{2}-2\eps_{large}\right)N$\Big) or \Big($w_{\fontL} < \eps_{\fontL} $ and $h(H_{1-2 \eps_{\fontL}}) > \left(\frac{1}{2}-2\eps_{large}\right)N$\Big)}\\
	In the first case $a(\ilopt) > \left(\frac{1}{2}-2\eps_{large}\right)(1-2\eps_{\fontL})N^2 +(w_\fontL -\frac12)N\cdot \frac{N}{2} \ge (\frac14+\frac{w_\fontL}{2}-\eps_{\fontL}-2\eps_{large})N^2$.
	Thus $a(\isfat) < (\frac34 -\frac{w_\fontL}{2} +\eps_{\fontL}+2\eps_{large})N^2$.\\
	Now consider the vertical rectangles in the boundary $\fontL$-region sorted non-increasingly by width and pick them iteratively
	until their total width crosses $(\frac{w_\fontL}{2}+3\eps_{\fontL}+2\eps_{large})N$.
	Remove these rectangles and push the remaining vertical rectangles in the $\fontL$-region to the left as much as possible.
	This modified $\fontL$-region will have profit at least $(\frac{1}{2} - O(\eps_{\fontL}))|\ilopt|$.
	Now we can put $\eps N$-strip for the vertical items from $\isthin$ next to the vertical part of $\fontL$-region. On the other hand, horizontal part of $\isthin$ can be placed on top of horizontal part of $\fontL$-region.
	The remaining space will be a free rectangular region of height at least $(1-2\eps_{\fontL}) N$ and width $(1-\frac{w_\fontL}{2}+2 \eps_\fontL+2\eps_{large})N$.
	We will use a part of this rectangular region of height $(1-3\eps_{\fontL}) N$ and width $(1-\frac{w_\fontL}{2}+\eps_\fontL)N$ to pack rectangles from $\isfat$ and the rest of the region for resource augmentation.
	Using Lemma~\ref{lem:smallStein}, we can pack small rectangles in this region with profit at least $\Big( \frac{(1-\frac{w_\fontL}{2})/2}{\frac34 -\frac{w_\fontL}{2}}-O(\efl)\Big)|\isfat| \ge (\frac34 -O(\eps_{\fontL}))|\isfat|$ as $w_\fontL \ge \frac12$. Hence, we get, 
	\begin{equation}
		\begin{aligned}
			|\optrc| &\ge \left(\frac{1}{2}-O(\eps_{\fontL})\right)|\ilopt|+|\isthin|\\
			& \quad +\left(\frac34-O(\eps_{\fontL})\right)|\isfat|
		\end{aligned}
		\label{CardCase2b1}
	\end{equation}
	On the other hand, as $a(\isfat) \le (\frac34  -\frac{w_\fontL}{2} +\eps_{\fontL}+2\eps_{large})N^2$ and $w_\fontL>1/2$, we get $a(\isfat) \le \frac12 N^2$ and thus from  Lemma \ref{lem4cardgen} we get, 
	\begin{equation}
	\label{CardCase2b1a}
	|\optrc| \ge \frac{3}{4}|\ilthin|+(1-O(\eps_{\fontL}))|\isopt|
	\end{equation}
	From inequalities \eqref{lem1card}, \eqref{lem3card}, \eqref{lem4card},
	\eqref{CardCase2b1}, \eqref{CardCase2b1a}, we get, 
	\begin{equation}\label{eqcardcase23}
	|\optrc|\geq \left(\frac{24}{41}-O(\eps_{\fontL})\right)|OPT|
	\end{equation}\arir{gives 1.708}\\
	Now we consider the last case when $w_{\fontL} < \eps_{\fontL} $ and $h(H_{1-2 \eps_{\fontL}}) > \left(\frac{1}{2}-2\eps_{large}\right)N$.  
	Note that as we assumed the cheapest subring was the top subring, after removing it we might be left with only $|\ilopt \cap \Rho|/2$ profit in the horizontal part of $\fontL$-region. So, further removal of items from the horizontal part might not give us a good solution.
	Thus we show an alternate good packing.  We restart with the ring packing and delete the cheapest vertical subring instead of the cheapest subring (i.e., the top subring) and create a new boundary $\fontL$-region.
	Here, consider the horizontal rectangles in the boundary $\fontL$-region in non-increasing order of height and take them
	until their total height crosses $(\frac{\beta_1+\beta_2}{2}+\eps_{small}+\eps)N$. Remove these rectangles and push the remaining horizontal rectangles to the bottom as much as possible. 
	Then, following similar arguments as in case 2B, we will obtain the same bounds as in inequality \eqref{eqcardcase23}.

	In summary, from all the cases (see Figure \ref{fig:cardworotcaseover} for the overview of the case analysis) the profit of the obtained solution is at least $$\left(\frac{325}{558}-O(\eps_{\fontL})\right) |\opt|.$$
\end{proof}

\section{Conclusions}

In this chapter, we presented techniques for \tdk that are not entirely based on containers. By leveraging a combination of container packings and \fontL-packings, we obtained the first algorithm that break the barrier of $2$ for the approximation factor of this problem. We obtained an approximation factor of $17/9+\eps < 1.89$ for the weighted case, and  $\frac{558}{325}+\eps<1.72$ in the cardinality case.

Further improvements could come from the study of generalizations of \fontL-packings. For example, a PTAS for  \emph{ring-packing} instances arising by shifting of long items would lead to an improved approximation factor. A PTAS for $O(1)$ simultaneous \fontL-packings would also be an interesting development, and might be a stepping stone for the long-standing problem of finding a PTAS for \tdk, which remains open even in the cardinality case, and even if pseudo-polynomial time is allowed instead of polynomial time.

\chapter[Approximations for 2DGK with Rotations]{Approximations for 2DGK\\with Rotations}
\label{chap:2dgk-rot}

In this chapter, we consider the \tdkr problem, where we are allowed to rotate the rectangles by $90^\circ$.

The possibility to rotate the rectangles makes the problem significantly easier. In fact, we obtain better approximation factors, and we do not need \fontL-packings anymore, obtaining purely container-based algorithms.

The basic idea is that any thin item can now be packed inside a narrow vertical strip
(say at the right edge of the knapsack) by possibly rotating it. This
way we do not lose one quarter of the profit due to the mapping to
an \fontL-packing and instead place all items from the ring into
the mentioned strip (while we ensure that their total width is small).
The remaining short items are packed by means of a novel \emph{resource
	contraction} lemma: unless there is one \emph{huge item} that occupies almost
the whole knapsack (a case that we consider separately), we can pack
almost one half of the profit of non-thin items in a \emph{reduced}
knapsack where one of the two sides is shortened by a factor $1-\eps$
(hence leaving enough space for the vertical strip). Thus, roughly
speaking, we obtain either all profit of non-thin items, or all profit
of thin items plus one half of the profit of non-thin items:
this gives a $3/2+\eps$ approximation. A further refinement of this
approach yields a $4/3+\eps$ approximation in the cardinality case.
We remark that, while resource augmentation is a well-established
notion in approximation algorithms, resource contraction seems to
be a rather novel direction to explore.

\begin{theorem}\label{thm:mainNoRotation} For any constant $\eps>0$,
	there exists a polynomial-time $\frac{3}{2}+\eps$ approximation algorithm
	for \tdkr. In the cardinality case the approximation factor can be
	improved to $\frac{4}{3}+\eps$. \end{theorem} 

First, let us introduce some notations.
Without loss of generality, assume $\width(i) \ge \height(i)$ for all rectangles $R_i \in \R$. Now consider packing of some optimal solution in the knapsack.
Denote the strips of width $N$ and height $\eps N$ at the top and bottom of the knapsack by $S_{T,\eps}:=[0,N]\times[(1-\eps)N,N]$ and 
$S_{B,\eps} :=[0,N]\times[0, \eps N]$. Similarly, denote the strips of height $N$ and width $\eps N$ to the left and right of the bin
by $S_{L,\eps}:=[0, \eps N] \times [0,N]$ and  $S_{R,\eps} :=[(1-\eps)N, N]\times [0,N]$. 
The set of rectangles intersected by one of these strips $S_{K, \eps}, K \in \{T, B, L, R\}$ is denoted by 
$E_{K, \eps}$. 
There are two kinds of rectangles that intersect these strips.
The set of rectangles completely contained inside one of these strips $S_{K, \eps}$, is denoted by 
$C_{K, \eps}$.
The set of remaining rectangles that do not lie completely inside the strip $S_{K, \eps}$ but intersect the strip, is denoted by $D_{K, \eps}$ (see Figure~\ref{fig:unweighted_rotations_def}).
With slight abuse of notations, we will assume $C_{K, \eps}$ also represents the embedding of corresponding rectangles in the strip.

\section{Weighted Case}
\label{sec:weightedRot}

In this section we give a polynomial time $(3/2+\eps)$-approximation algorithm for the weighted 2-dimensional geometric knapsack problem when items are allowed to be rotated by $90^\circ$.
In contrary to the unweighted case, where it is possible to remove a constant number of \emph{large} items, the same is not possible in the weighted case, where an item could have a big profit.

We call an item $i$ \emph{massive} if $\width(i) \geq (1 - \eps)N$ and $\height(i) \geq (1 - \eps)N$. The presence of such a big item in the optimal solution requires a different analysis, that we present below. In both the cases, we can show that there exists a container packing with roughly 2/3 of the profit of the optimal solution.

Let us assume that $\eps < 1/6$. We will prove the following result:

\begin{theorem}\label{lem:structural_lemma_weighted}Let $\eps > 0$ and let $\R$ be a set of items that can be packed into the $N \times N$ knapsack. Then there exists an $\eps$-granular container packing with $O_\eps(1)$ containers of a subset $\R' \subseteq \R$ into the $N \times N$ knapsack such that $\profit(\R') \geq (2/3 - O(\eps)) \profit(\R)$, if rotations are allowed.
\end{theorem}

We start by analyzing the case of a set $\R$ that has a massive item.

\begin{lemma}\label{lem:massiveitem}
	Suppose that a set $\R$ of items can be packed into a $N \times N$ bin and there is a massive item $m \in \R$. Then, there is a container packing with at most $O_\eps(1)$ containers for a subset $\R' \subseteq \R$ such that $\profit(\R') \geq \left(\dfrac{2}{3} - O(\eps) \right) \profit(\R)$. 
\end{lemma}
\begin{proof}
	Assume, without loss of generality, that $1/(3\eps)$ is an integer. Consider the items in $\R \setminus \{m\}$. Clearly, each of them has width or height at most $\eps$; moreover, $a(\R \setminus \{m\}) \leq (1 - (1 - \eps)^2) N^2 = (2\eps - \eps^2)N^2 \leq \frac{N^2}{2(1+\eps)}$, as $\eps < 1/6$; thus, by possibly rotating each element so that the height is smaller than $\eps$, by Theorem~\ref{thm:steinberg} all the items in $\R \setminus \{m\}$ can be packed in a $N \times \frac{N}{1 + \eps}$ bin; then, by Lemma~\ref{lem:structural_lemma_augm}, there is a container packing for a subset of $\R \setminus \{m\}$ with $O_\eps(1)$ containers that fits in the $N \times N$ bin and has profit at least $(1 - O(\eps))\profit(\R\setminus\{m\})$
	
	Consider now the packing of $\R$. Clearly, the region $[\eps N, (1 - \eps)N]^2$ is entirely contained within the boundaries of the massive item $m$. Partition the region with $x$-coordinate between $\eps N$ and $(1 - \eps)N$ in $k = 1/(3\eps)$ strips of width $3\eps(1 - 2\eps)N \geq 2\eps N$ and height $N$, let them be $S_1, \dots, S_{k}$; let $\R(S_i)$ be the set of items in $\R$ such that their left or right edge (or both) are contained in the interior of strip $S_i$. Since each item belongs to at most two of these sets, there exists $i$ such that $\profit(\R(S_i)) \leq 6\eps \profit(\R)$.\\
	Symmetrically, we define $k$ horizontal strips $T_1, \dots, T_k$, obtaining an index $j$ such that $\profit(\R(T_j)) \leq 6\eps \profit(\R)$. Thus, no item in $\overline{\R} := \R \setminus (R(S_i) \cup R(T_j))$ has a side contained in the interior of $S_i$ or $T_j$, and $\profit(\overline{\R}) \geq (1 - 12\eps)\profit(\R)$. Let $M_V$ be the set of items in $\overline{\R} \setminus \{m\}$ that overlap $T_j$, and let $M_H$ be the set of items in $\overline{\R} \setminus \{m\}$ that overlap $S_i$. Clearly, the items in $M_H$ can be packed in a horizontal container with width $N$ and height $N - \height(m)$, and the items in $M_V$ can be packed in a vertical container of width $N - \width(m)$ and height $N$.
	
	Let $H$ be the set of items of $\overline{\R} \setminus M_H$ that are completely above the massive item $m$ or completely below it; symmetrically, let $V$ be the set of items of $\overline{\R} \setminus M_V$ that are completely to the left or completely to the right of $m$. We will now show that there is a container packing for $M_H \cup V \cup \{m\}$. Since all the elements overlapping $T_j$ have been removed, $V$ can be packed in a bin of size $(N - \width(m)) \times (1 - 2\eps)N$ (see Figure~\ref{fig:massive_item}). Since $(1 - 2\eps)N \cdot (1 + \eps) < (1 - \eps)N \leq \height(m)$, Lemma~\ref{lem:structural_lemma_augm} implies that there is a container packing of a subset of $V$ with profit at least $(1 - O(\eps))\profit(V)$ in a bin of size $(N - \width(m)) \times \height(m)$ and using $O_\eps(1)$ containers; thus, by adding a horizontal container of the same size as $m$ and a horizontal container of size $N \times (N - \height(m))$, we obtain a container packing for $M_H \cup V \cup \{m\}$ with $O_\eps(1)$ containers and profit at least $(1 - O(\eps))\profit(M_H \cup V \cup \{m\})$. Symmetrically, there is a container packing for a subset of $M_V \cup H \cup \{m\}$ with profit at least $(1 - O(\eps))\profit(M_V \cup H \cup \{m\})$ and $O_\eps(1)$ containers.
	
	Let $\R_{MAX}$ be the set of maximum profit among the sets $\R\setminus\{m\}$, $M_H \cup V \cup \{m\}$ and $M_V \cup H \cup \{m\}$. By the discussion above, there is a container packing for $\R' \subseteq \R_{MAX}$ with $O_\eps(1)$ containers and profit at least $(1 - O(\eps)) \profit(\R_{MAX})$. Since each element in $\overline{\R}$ is contained in at least two of the above three sets, it follows that:
	
	\begin{align*}
	\profit(\R') & \geq (1 - O(\eps))\profit(\R_{MAX}) \geq \left(1 - O(\eps)\right)\left(\frac{2}{3} \profit(\overline{\R})\right)\\
	& \geq \left(\frac{2}{3} - O(\eps)\right) \profit(\R)
	\end{align*}
\end{proof}

\begin{figure}
	\centering
	\begin{subfigure}[b]{.44\textwidth}
		\resizebox{!}{180pt}{
			\begin{tikzpicture}
			
			
			\draw (0,0) rectangle (10,10);
			
			
			\fill[color=gray] (1,1) rectangle (9.25,9.25);
			\draw (1,1) rectangle (9.25,9.25);
			\draw (5.125,5.125) node {\textbf{\huge $m$}};
			
			
			\fill[color=lightgray, pattern=north east lines] (3.5,0.2) rectangle (6.5,0.4);
			\draw (3.5,0.2) rectangle (6.5,0.4);
			
			\fill[color=lightgray, pattern=north east lines] (3,0.4) rectangle (7.5,0.8);
			\draw (3,0.4) rectangle (7.5,0.8);
			
			\fill[color=lightgray, pattern=north east lines] (2.5,0.8) rectangle (6,1);
			\draw (2.5,0.8) rectangle (6,1);
			
			\fill[color=lightgray, pattern=north east lines] (3.5,9.4) rectangle (8,9.6);
			\draw (3.5,9.4) rectangle (8,9.6);
			
			\fill[color=lightgray, pattern=north east lines] (3,9.6) rectangle (6.7,9.8);
			\draw (3,9.6) rectangle (6.7,9.8);
			
			
			\fill[color=lightgray, pattern=north west lines] (0,3.5) rectangle (0.3,7);
			\draw (0,3.5) rectangle (0.3,7);
			
			\fill[color=lightgray, pattern=north west lines] (0.5,4) rectangle (0.8,7.2);
			\draw (0.5,4) rectangle (0.8,7.2);
			
			\fill[color=lightgray, pattern=north west lines] (9.3,2.5) rectangle (9.45,6.5);
			\draw (9.3,2.5) rectangle (9.45,6.5);
			
			\fill[color=lightgray, pattern=north west lines] (9.5,4) rectangle (9.8,7.5);
			\draw (9.5,4) rectangle (9.8,7.5);
			
			
			\fill[color=lightgray] (0.2,2.5) rectangle (0.5,0.5);
			\draw (0.2,2.5) rectangle (0.5,0.5);
			
			\fill[color=lightgray] (0.7,0.3) rectangle (1,0.9);
			\draw (0.7,0.3) rectangle (1,0.9);
			
			\fill[color=lightgray] (0.2,7.5) rectangle (0.6,9.7);
			\draw (0.2,7.5) rectangle (0.6,9.7);
			
			\fill[color=lightgray] (9.5,0.5) rectangle (9.8,3);
			\draw (9.5,0.5) rectangle (9.8,3);
			
			
			\draw[dashed, ultra thick] (4,0) -- (4,1);
			\draw[dashed, ultra thick] (5.5,0) -- (5.5,1);
			
			\draw[dashed, ultra thick] (4,9.25) -- (4,10);
			\draw[dashed, ultra thick] (5.5,9.25) -- (5.5,10);
			
			\draw[dashed, ultra thick] (0,4.5) -- (1,4.5);
			\draw[dashed, ultra thick] (0,6) -- (1,6);
			
			\draw[dashed, ultra thick] (9.25,4.5) -- (10,4.5);
			\draw[dashed, ultra thick] (9.25,6) -- (10,6);
			
			\draw[dashed, ultra thick] (1,-1) -- (1, 11);
			\draw[dashed, ultra thick] (9.25,-1) -- (9.25, 11);
			
			\draw[dashed, ultra thick] (-1,1) -- (11,1);
			\draw[dashed, ultra thick] (-1,9.25) -- (11,9.25);
			
			
			\draw (1,-0.75) node[anchor=east] {\Large{$V$}};
			
			\draw (9.25,-0.75) node[anchor=west] {\Large{$V$}};
			
			\draw (-0.75,1) node[anchor=north] {\Large{$H$}};
			
			\draw (-0.75,9.25) node[anchor=south] {\Large{$H$}};
			
			\draw (0,5.25) node[anchor=east] {\Large{$M_V$}};
			
			\draw (4.75,0) node[anchor=north] {\Large{$M_H$}};
			
			\end{tikzpicture}}
		\caption{Massive item case. Items intersecting strips $M_H$ and $M_V$ (hatched rectangles) cross them completely.\\~}\label{fig:stripes}
	\end{subfigure}
	\hspace{10pt}
	\begin{subfigure}[b]{.44\textwidth}
		\resizebox{!}{180pt}{
			\begin{tikzpicture}
			
			
			\draw[thick] (0,0) rectangle (10,10);
			
			
			\draw[fill=gray] (5,2.5) rectangle (8,5.5);
			\fill[pattern = vertical lines] (5.5,0.5) rectangle (7.5,2);	
			\draw (5.5,0.5) rectangle (7.5,2);	
			\fill[pattern = horizontal lines] (2,6) rectangle (4,9);
			\fill[pattern = north east lines] (2,6) rectangle (4,9);
			\draw (2,6) rectangle (4,9);
			\fill[pattern = north east lines] (2,3) rectangle (4,5);
			\draw (2,3) rectangle (4,5);
			\fill[pattern = north west lines] (8.5,1.5) rectangle (9.5,4.5);
			\draw (8.5,1.5) rectangle (9.5,4.5);
			\fill[pattern = north west lines] (8.2,6.5) rectangle (9.6,8.5);
			\fill[pattern = horizontal lines] (8.2,6.5) rectangle (9.6,8.5);
			\draw (8.2,6.5) rectangle (9.6,8.5);
			
			
			\draw[dashed] (5,-1) -- (5,11);
			\draw[dashed] (8,-1) -- (8,11);
			\draw[dashed] (-1,2.5) -- (11,2.5);
			\draw[dashed] (-1,5.5) -- (11,5.5);
			
			
			\draw (6.5,4) node {\huge $i$};
			\filldraw (5,2.5) circle (2.5pt);\Large
			\draw (5,2.5) node[anchor = north east] { $(x_i,y_i)$};
			\filldraw (8,5.5) circle (2.5pt);
			\draw (8,5.5) node[anchor = south west] {\Large $(x_i',y_i')$};
			\draw [decorate,decoration={brace,amplitude=7pt}] (5,10) -- (8,10); 
			\draw (6.5,10.3) node [anchor = south] {\Large $w_i$};
			\draw [decorate,decoration={brace,amplitude=7pt}] (10,5.5) -- (10,2.5); 
			\draw (10.3,4) node [anchor = west] {\Large $h_i$};
			\draw (5,-0.6) node[anchor = east] {\Large $Left(i)$};
			\draw (8,-0.6) node[anchor = west] {\Large $Right(i)$};
			\draw (-0.6,2.5) node[anchor = north] {\Large $Bottom(i)$};
			\draw (-0.6,5.5) node[anchor = south] {\Large $Top(i)$};
			
			\end{tikzpicture}}
		\caption{$Bottom(i), Top(i), Left(i), Right(i)$ are represented by vertical, horizontal, north east  and north west stripes respectively.}\label{fig:massive_item1}
	\end{subfigure}
	\caption{}\label{fig:massive_item}
\end{figure}

If there is no massive item, we will show existence of  two {\em container packings} 
and show the maximum of them always packs items with total profit at least $\left(\frac23-O(\eps)\right)$ fraction of the optimal profit. 

First, we follow the corridor decomposition and the classification of items as in Section \ref{sec:weighted} to define
sets $LF, SF, LT, ST, OPT_{small}$.
Let $T:= LT \cup ST$  be the set of thin items.
Also let $APX$ be the best {\em container packing} and $OPT$ be the optimal solution.
Then similar to Lemma~\ref{lem:weighted-apx}, we can show $\profit(\apx)\ge (1-\eps)(\profit(LF)+\profit(SF)+\profit(\optsm))$.
Thus, 
\begin{equation}
\label{rotweightpack1}
\profit(\apx)\ge (1-\eps)\profit(OPT)-\profit(T).
\end{equation}

In the second case, we define the set $T$ as above. 
Then in Resource Contraction Lemma (Lemma~\ref{lem:weightResContract}), we will show that one can pack $1/2$ of the remaining profit  in the optimal solution, i.e., $\profit(OPT \setminus T)/2$ in a 
knapsack of size $N \times (1-\eps/2)N$.
Now, we can pack $T$ in a horizontal container of height $\eps/4$ and using Lemma~\ref{lem:weightResContract} and resource augmentation we can pack $\profit(OPT \setminus T)/2$ in the remaining space $N \times (1-\eps/4)N$.
Thus,
\begin{equation}
\label{rotweightpack2}
\profit(\apx)\ge \profit(T)+(1-\eps)(\profit(OPT)-\profit(T))/2.
\end{equation}
Hence, up to $(1-O(\eps))$ factor, we obtain a packing with profit at least $\max\{(\profit(T)+\profit(OPT \setminus T)/2), \profit(OPT \setminus T)\} \ge 2/3 \cdot  \profit(OPT)$, thus proving Theorem~\ref{lem:structural_lemma_weighted}.

Note that since Theorem~\ref{thm:container_packing_ptas} gives a PTAS for container packings, Theorem~\ref{thm:mainNoRotation} immediately implies a $(2/3 - O(\eps))$-approximation algorithm.

Now to complete the proof of Theorem~\ref{thm:mainNoRotation}, it only remains to prove Lemma~\ref{lem:weightResContract}.
\begin{lemma}(Resource Contraction Lemma)
	\label{lem:weightResContract}
	If a set of items $M$ contains no massive item and can be packed into a $N \times N$ bin, then it is possible to pack a set $M'$ of profit at least $\profit(M) \cdot \frac{1}{2}$ into a $N \times  (1-\frac{\eps}{2})N$ bin (or a $(1-\frac{\eps}{2})N \times N$ bin), if rotations are allowed.
\end{lemma}

\begin{proof}
	Let $\eps_s=\eps/2$.
	We will partition $M$ into two sets $M_1, M \setminus M_1$ and show that both these sets can be packed into $N\times (1-\eps_s)N$ bin.
	If an item $i$ is embedded in position $(x_i, y_i)$, we define $x'_i:=x_i+\width(i), y'_i :=y_i+\height(i)$.
	
	In a packing of a set of items $M$, for item $i$ we define 
	$Left(i):=\{k \in M: x'_k \le x_i \}$, 
	$Right(i):=\{k \in M: x_k \ge x'_i \}$, 
	$Top(i):=\{k \in M: y_k \ge y'_i \}$, 
	$Bottom(i):=\{k \in M: y'_k \le y_i \}$, 
	i.e., the set of items that lie completely on left, right, top and bottom of $i$ respectively.
	Now consider four strips $S_{T,3\eps_s}, S_{B,\eps_s}, S_{L,\eps_s} ,S_{R,\eps_s}$ (see Figure~\ref{fig:resource_contraction}).\\
	
	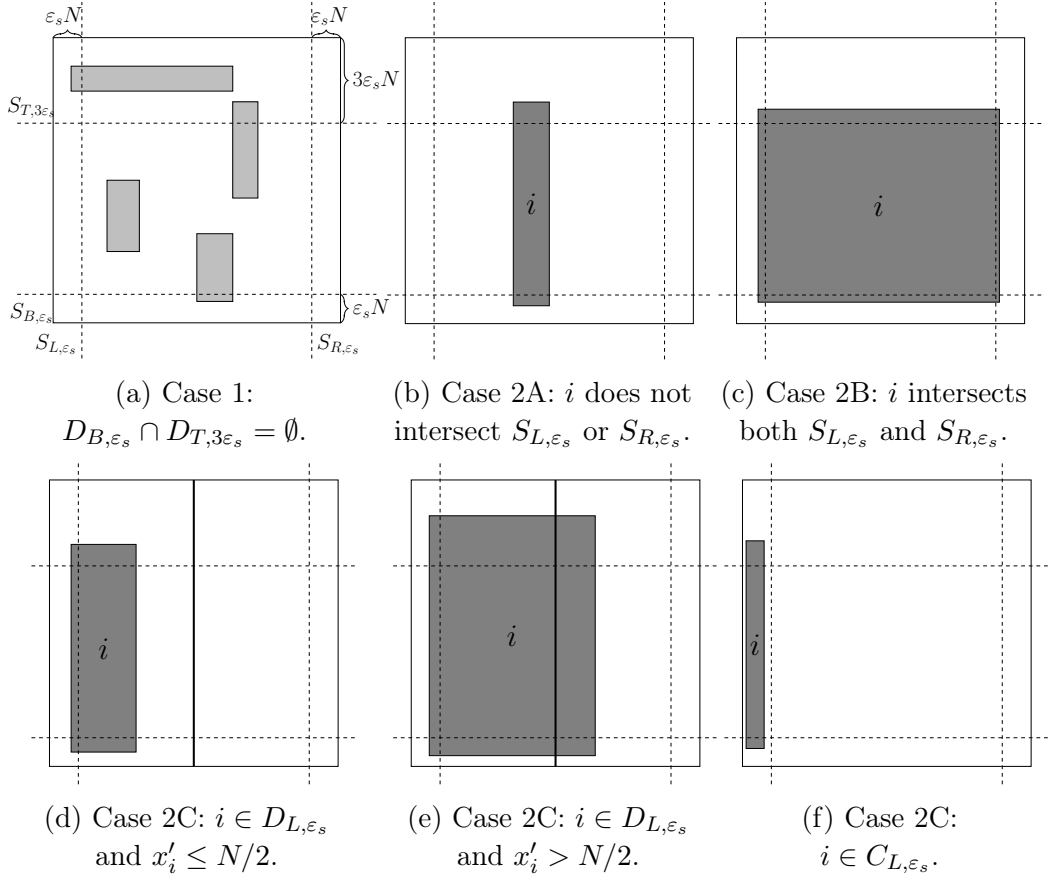
\begin{figure}[ht]
		\captionsetup[subfigure]{justification=centering}
		\begin{subfigure}[b]{.35\textwidth}
			\resizebox{!}{135pt}{
				\begin{tikzpicture}
				
				
				\draw[thick] (0,0) rectangle (8,8);
				
				
				\fill[color=lightgray] (4,0.6)  rectangle (5,2.5);
				\draw (4,0.6)  rectangle (5,2.5);
				\fill[color=lightgray] (5,3.5)  rectangle (5.7,6.2);
				\draw (5,3.5)  rectangle (5.7,6.2);
				\fill[color=lightgray] (0.5,6.5)  rectangle (5,7.2);
				\draw (0.5,6.5)  rectangle (5,7.2);
				\fill[color=lightgray] (1.5,2)  rectangle (2.4,4);
				\draw (1.5,2)  rectangle (2.4,4);
				
				
				\draw[dashed] (-1,0.8) -- (9,0.8);
				\draw[dashed] (-1,5.6) -- (9,5.6);
				\draw[dashed] (0.8,-1) -- (0.8,9);
				\draw[dashed] (7.2,-1) -- (7.2,9);
				
				
				\draw [decorate,decoration={brace,amplitude=6pt}] (8,8) -- (8,5.6); 
				\draw (8.2,6.8) node [anchor = west] {\Large $3\eps_s N$};
				\draw [decorate,decoration={brace,amplitude=6pt}] (8,0.8) -- (8,0); 
				\draw (8.2,0.4) node [anchor = west] {\Large $\eps_s N$};
				\draw [decorate,decoration={brace,amplitude=6pt}] (7.2,8) -- (8,8); 
				\draw (7.7,8.2) node [anchor = south] {\Large $\eps_s N$};
				\draw [decorate,decoration={brace,amplitude=6pt}] (0,8) -- (0.8,8); 
				\draw (0.3,8.2) node [anchor = south] {\Large $\eps_s N$};
				
				\draw (-0.6,0.7) node[anchor=north] {\Large $S_{B,\eps_s}$};
				\draw (0.8,-0.6) node[anchor=east] {\Large $S_{L,\eps_s}$};
				\draw (7.2,-0.6) node[anchor=west] {\Large $S_{R,\eps_s}$};
				\draw (-0.6,5.6) node[anchor=south] {\Large $S_{T,3\eps_s}$};
				\end{tikzpicture}}
			\caption{Case 1: \\$D_{B,\eps_s} \cap D_{T,3\eps_s}=\emptyset$.}
		\end{subfigure}
		\begin{subfigure}[b]{.3\textwidth}
			\resizebox{!}{135pt}{
				\begin{tikzpicture}
				
				
				\draw[thick] (0,0) rectangle (8,8);
				
				
				\draw[fill=gray] (3,0.5) rectangle (4,6.2);
				\draw (3.5,3.35) node {\huge $i$};
				
				
				\draw[dashed] (-0.5,0.8) -- (8.5,0.8);
				\draw[dashed] (-0.5,5.6) -- (8.5,5.6);
				\draw[dashed] (0.8,-1) -- (0.8,9);
				\draw[dashed] (7.2,-1) -- (7.2,9);
				
				\end{tikzpicture}}
			\caption{Case 2A: $i$ does not intersect $S_{L, \eps_s}$ or $S_{R, \eps_s}$.}
		\end{subfigure}
		\begin{subfigure}[b]{.3\textwidth}
			\resizebox{!}{135pt}{
				\begin{tikzpicture}
				
				
				\draw[thick] (0,0) rectangle (8,8);
				
				
				\draw[fill=gray] (0.6,0.6) rectangle (7.3,6);
				\draw (3.95,3.3) node {\huge $i$};
				
				
				\draw[dashed] (-0.5,0.8) -- (8.5,0.8);
				\draw[dashed] (-0.5,5.6) -- (8.5,5.6);
				\draw[dashed] (0.8,-1) -- (0.8,9);
				\draw[dashed] (7.2,-1) -- (7.2,9);
				
				\end{tikzpicture}}
			\caption{Case 2B: $i$ intersects both $S_{L, \eps_s}$ and $S_{R, \eps_s}$.}
		\end{subfigure}
		
		\vspace{3pt}
		\hspace{6.9pt}
		\begin{subfigure}[b]{.3\textwidth}
			\resizebox{!}{121.75pt}{
				\begin{tikzpicture}
				
				
				\draw[thick] (0,0) rectangle (8,8);
				
				
				\draw[fill=gray] (0.6,0.4) rectangle (2.4,6.2);
				\draw (1.5,3.3) node {\huge $i$};
				
				
				\draw[dashed] (-0.5,0.8) -- (8.5,0.8);
				\draw[dashed] (-0.5,5.6) -- (8.5,5.6);
				\draw[dashed] (0.8,-0.5) -- (0.8,8.5);
				\draw[dashed] (7.2,-0.5) -- (7.2,8.5);
				\draw[ultra thick] (4,0) -- (4,8);
				
				\end{tikzpicture}}
			\caption{Case 2C: $i \in D_{L,\eps_s}$ and $x_i'\le N/2$.}
		\end{subfigure}	
		\hspace{7.5pt}
		\begin{subfigure}[b]{.3\textwidth}
			\resizebox{!}{121.75pt}{
				\begin{tikzpicture}
				
				
				\draw[thick] (0,0) rectangle (8,8);
				
				
				\draw[fill=gray] (0.5,0.3) rectangle (5.1,7);
				\draw (2.8,3.65) node {\huge $i$};
				
				
				\draw[dashed] (-0.5,0.8) -- (8.5,0.8);
				\draw[dashed] (-0.5,5.6) -- (8.5,5.6);
				\draw[dashed] (0.8,-0.5) -- (0.8,8.5);
				\draw[dashed] (7.2,-0.5) -- (7.2,8.5);
				\draw[ultra thick] (4,0) -- (4,8);
				
				\end{tikzpicture}}
			\caption{Case 2C: $i \in D_{L,\eps_s}$ and $x_i' > N/2$.}
		\end{subfigure}
		\begin{subfigure}[b]{.3\textwidth}
			\resizebox{!}{121.75pt}{
				\begin{tikzpicture}
				
				
				\draw[thick] (0,0) rectangle (8,8);
				
				
				\draw[fill=gray] (0.1,0.5) rectangle (0.6,6.3);
				\draw (0.35,3.4) node {\huge $i$};
				
				
				\draw[dashed] (-0.5,0.8) -- (8.5,0.8);
				\draw[dashed] (-0.5,5.6) -- (8.5,5.6);
				\draw[dashed] (0.8,-0.5) -- (0.8,8.5);
				\draw[dashed] (7.2,-0.5) -- (7.2,8.5);
				
				\end{tikzpicture}}
			\caption{Case 2C: \\ $i \in C_{L,\eps_s}$.}
		\end{subfigure}
		\caption{Cases for the Resource Contraction Lemma~\ref{lem:weightResContract}.}\label{fig:resource_contraction}
	\end{figure}
	
	\noindent \textbf{Case 1.} $D_{B,\eps_s} \cap D_{T,3\eps_s}=\emptyset$, i.e., no item intersecting $S_{B,\eps_s}$ intersects $S_{T,3\eps_s}$.
	Define $M_1:=E_{T,3\eps_s}$. As these items in $M_1$ do not intersect $S_{B,\eps_s}$,
	$M_1$ can be packed into a $(N, N(1-\eps_s))$ bin.
	For the remaining items, pack $M \setminus (M_1 \cup C_{L, \eps_s} \cup C_{R, \eps_s})$ as it is. 
	Now rotate $C_{L, \eps_s}$ and $C_{R, \eps_s}$ and pack on top of $M \setminus (M_1 \cup C_{L, \eps_s} \cup C_{R, \eps_s})$ into two strips of height $\eps_s N$ and width $N$. This packing will have total height $\le (1-3\eps_s+2\eps_s)N \le (1-\eps_s)N$.\\
	
	\noindent \textbf{Case 2.} $D_{B,\eps_s} \cap D_{T,3\eps_s} \neq \emptyset$, i.e., there is some item  intersecting $S_{B,\eps_s}$ that also crosses $S_{T,3\eps_s}$. Now, there are three subcases:
	\\ \textbf{Case 2A.} \textit{There exists an item $i$ that does neither intersect $S_{L, \eps_s}$ nor $S_{R, \eps_s}$}.
	Then item $i$ partitions the items in $M \setminus ( C_{T, 3\eps_s} \cup C_{B, \eps_s} \cup \{ i \} )$
	into two sets: $Left(i)$ and $Right(i)$.
	Without loss of generality, assume $x_i \le 1/2$. Then remove $Right(i), i, C_{T, 3 \eps_s}$ and $C_{B, \eps_s} $ from the packing.
	Now rotate $C_{T, 3 \eps_s}$ and $C_{B, \eps_s}$ to pack right of $Left(i)$.
	Define this set $M \setminus (Right(i)\cup \{ i \})$ to be $M_1$. Clearly packing of $M_1$ takes height $N$ and width $x_i+4\eps_s N \le (\frac12+4\eps_s)N \le (1-\eps_s)N$ as $\eps_s \le  \frac{1}{10}$.
	As the item $i$ does not intersect the strip $S_{L, \eps_s}$, $(Right(i)\cup \{i\})$ can be packed into height $N$ and width $(1-\eps_s)N$.
	\\ \textbf{Case 2B.} \textit{There exists an item $i$ that intersects both $S_{L, \eps_s}$ and $S_{R, \eps_s}$}. 
	Consider $M_1$ to be $M \setminus (C_{L, \eps_s} \cup C_{R, \eps_s}  \cup Top(i) )$.
	As there is no massive item, $M_1$ is packed in height $(1-\eps_s)N$ and width $N$.
	Now, pack $Top(i)$ and then  rotate $C_{L, \eps_s}$ and $C_{R, \eps_s}$  to pack on top of it.
	These items can be packed into height $(1-y'_i+2\eps_s)N \le 5\eps_s N \le (1-\eps_s)N$ as $\eps_s \le 1/10$.
	\\ \textbf{Case 2C.} \textit{If an item $i$ intersects both $S_{B, \eps_s}$ and $S_{T, 3 \eps_s}$,  then the item $i$ intersects exactly one of $S_{L, \eps_s}$ and $S_{R, \eps_s}$.} 
	Consider the set of items in $D_{B, \eps_s} \cap D_{T, 3 \eps_s}$.\\
	First, consider the case when the set $D_{B, \eps_s} \cap D_{T, 3 \eps_s}$ contains an item $i \in D_{L, \eps_s}$ (similarly one can consider $i \in D_{R, \eps_s}$). Now if $x'_i \le N/2$, take $M_1:=Right(i)$.
	Then, we can rotate $Right(i)$ and pack into height $(1-\eps_s)N$ and width $N$. On the other hand, 
	pack $M \setminus \{  M_1 \cup C_{T,3\eps_s} \cup C_{B,\eps_s} \}$ 
	as it is. Then rotate $C_{T,3\eps_s} \cup C_{B,\eps_s}$ and pack on its side. Total width $\le (1/2+3\eps_s+\eps_s) N \le (1-\eps_s)N$ as $\eps_s \le 1/6$. Otherwise if $x'_i > N/2$ take $M_1:=Left(i)\cup i$.
	Now, consider packing of $M \setminus \{  M_1 \cup C_{T,3\eps_s} \cup C_{B,\eps_s} \}$, rotate $C_{T,3\eps_s} \cup C_{B,\eps_s}$ and pack on its left. Total width $\le (1/2+4\eps_s)N \le (1-\eps_s)N$ as $\eps_s \le 1/10$. 
	\\Otherwise, no items in $S_{B, \eps_s} \cap S_{T, 3 \eps_s}$ are in $D_{L, \eps_s} \cup D_{R, \eps_s}$.
	So let us assume that $i \in C_{L, \eps_s}$ (similarly one can consider $i \in C_{R, \eps_s}$), then  we take $M_1=E_{T,3\eps_s} \setminus (C_{L,\eps_s} \cup C_{R,\eps_s})$. Then we can rotate $C_{L, \eps_s}$ and $C_{R, \eps_s}$ and pack them on top of $M \setminus (M_1 \cup C_{L, \eps_s} \cup C_{R, \eps_s})$ as in Case 1. 
\end{proof}

\section{Unweighted case}
\label{sec:cardRot}
In this section we present our polynomial time $(4/3+\eps)$-approximation algorithm for \tdkr~for the cardinality case.

\begin{figure}
	\centering
	\begin{subfigure}[b]{.34\textwidth}
		\resizebox{!}{135pt}{
			\begin{tikzpicture}
			
			
			\draw[thick] (0,0) rectangle (8,8);
			
			\fill[pattern = north east lines, pattern color = lightgray] (0,0) rectangle (1.5,8);
			\fill[pattern = north east lines, pattern color = lightgray] (6.5,0) rectangle (8,8);
			\fill[pattern = north west lines, pattern color = lightgray] (0,0) rectangle (8,1.5);
			\fill[pattern = north west lines, pattern color = lightgray] (0,6.5) rectangle (8,8);
			
			
			\draw[dashed] (-1,1.5) -- (9,1.5);
			\draw[dashed] (-1,6.5) -- (9,6.5);
			\draw[dashed] (1.5,-1) -- (1.5,9);
			\draw[dashed] (6.5,-1) -- (6.5,9);
			
			
			\draw [decorate,decoration={brace,amplitude=7pt}] (8,8) -- (8,6.5); 
			\draw (8.2,7.25) node [anchor = west] {\Large $\gamma N$};
			\draw [decorate,decoration={brace,amplitude=6pt}] (8,1.5) -- (8,0); 
			\draw (8.2,0.75) node [anchor = west] {\Large $\delta N$};
			\draw [decorate,decoration={brace,amplitude=7pt}] (6.5,8) -- (8,8); 
			\draw (7.25,8.2) node [anchor = south] {\Large $\beta N$};
			\draw [decorate,decoration={brace,amplitude=7pt}] (0,8) -- (1.5,8); 
			\draw (0.75,8.2) node [anchor = south] {\Large $\alpha N$};
			
			\draw (-0.6,1.5) node[anchor=north] {\Large $S_{B,\delta}$};
			\draw (1.5,-0.6) node[anchor=east] {\Large $S_{L,\alpha}$};
			\draw (6.5,-0.6) node[anchor=west] {\Large $S_{R,\beta}$};
			\draw (-0.6,6.5) node[anchor=south] {\Large $S_{T,\gamma}$};
			\end{tikzpicture}}
		\caption{Strips $S_{L,\alpha}, S_{R,\beta}, S_{B,\delta}, S_{T,\gamma}$}
	\end{subfigure}
	\hspace{35pt}
	\begin{subfigure}[b]{.4\textwidth}
		\resizebox{!}{135pt}{
			\begin{tikzpicture}
			
			
			\draw[thick] (0,0) rectangle (8,8);
			
			
			\draw[fill=lightgray] (0.5,6.5)  rectangle (5,7.2);
			\draw[fill=lightgray] (1.5,2)  rectangle (2.4,4);
			\draw[fill=darkgray] (0.5,1.5)  rectangle (1.3,5);
			\draw[fill=gray] (5,1)  rectangle (7,2.5);
			
			
			\draw[dashed] (2,-1) -- (2,9);
			
			
			\draw [decorate,decoration={brace,amplitude=7pt}] (0,8) -- (2,8); 
			\draw (0.75,8.2) node [anchor = south] {\Large $\alpha N$};
			
			\draw (2,-0.6) node[anchor=east] {\Large $S_{L,\alpha}$};
			\filldraw[color=white] (-1,0) rectangle (-0.9,0.1);
			\end{tikzpicture}}
		\caption{ $C_{L,\alpha}, D_{L,\alpha}$ are dark and light gray resp.}
	\end{subfigure}
	\caption{Definitions for cardinality 2DK with rotations.}\label{fig:unweighted_rotations_def}
\end{figure}
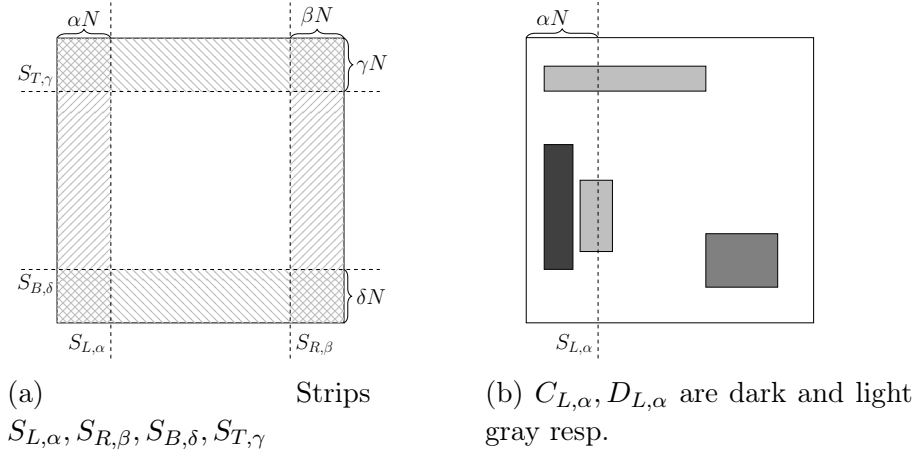 

As in the case without rotations, we will first show the existence of a container packing that packs items with a total profit of $(3/4-O(\eps))|OPT|$ where $OPT$ is the optimal solution.
Let $APX$ be the largest profit solution with container packing. 
As in Section \ref{sec:weighted}, we assume all items to be skewed. Note that small items can be handled with the techniques used in Lemma~\ref{lem:smallPack}.
We start with the corridor partition as in Section \ref{sec:weighted} and define {\em thin}, {\em fat} and {\em killed}
rectangles accordingly.
Let $T$ and $F$ be the set of thin and  fat rectangles respectively.

We will show that 
$|APX|\ge (3/4-O(\eps))|OPT|$.
\begin{lemma}
	\label{uwrot1}
	$|APX| \ge (1-\eps)|F|$.
\end{lemma}
\begin{proof}
	After removal of $T$, we can get a container based packing for almost all items in $F$ as discussed in Lemma~\ref{lem:onlyFat} in Section \ref{sec:weighted}.
\end{proof}

Now the main technical contribution of this section is the following {\em Resource Contraction Lemma}.
\begin{lemma}
	(Resource Contraction Lemma)
	Suppose that there exists a feasible packing of a set of rectangles $M$ in a $N \times N$ bin, for some $\eps > 0$. Then it is possible to pack a subset of $M$ with cardinality at least $\frac23(1-O(\eps))|M|$ into a $N \times \left(1-\eps^{\frac{1}{2\eps}+1}\right)N$ bin (or a $\left(1-\eps^{\frac{1}{2\eps}+1}\right)N \times N$ bin), if rotations are allowed.
	\label{uwrescontr1}
\end{lemma}

We defer the proof of the lemma to the end of this section. First, we show that using Lemma~\ref{uwrescontr1} we can prove the following:

\begin{lemma}
	\label{uwrot2}
	$|APX| \ge |T|+(2/3-\eps)|F|$.
\end{lemma}
\begin{proof}
	First, we use Lemma~\ref{lem:boxProperties} such that the total height of all rectangles in $T$ is at most $\frac{\eps^{\frac{1}{2\eps}+1} N}{2}$. 
	So we pack them in a horizontal container of height $\frac{\eps^{\frac{1}{2\eps}+1} N}{2}$. Then using Lemma~\ref{uwrescontr1} we show the existence of a  packing of $(2/3-O(\eps)) |F|$ in $N \times (1-\eps^{\frac{1}{2\eps}+1})N$. Then we can use resource augmentation to get a container packing of  $(2/3-O(\eps)) |F|$ in the remaining area $N \times \left(1-\frac{\eps^{\frac{1}{2\eps}+1}}{2}\right)N$.
\end{proof}

Thus we get the following theorem:
\begin{theorem}
	$|APX| \ge (3/4-O(\eps))|OPT|$.
\end{theorem}
\begin{proof}
	The claim follows by combining Lemma~\ref{uwrot2} and \ref{uwrot1}. Up to a factor $1-O(\eps)$, the worst case is obtained when $|F|=|T|+2/3 \cdot  |F|$,
	i.e., $|F|=3 |T|$. This gives a total profit of $3/4\cdot |T \cup F|$. 
\end{proof}

It remains to prove Lemma~\ref{uwrescontr1}.
Let $M$ be a set of rectangles that can be packed into a $N \times N$ bin. 
First, we show that 
we can remove a set of rectangles with cardinality  at most $\eps |M|$ such that the  remaining rectangles will either be very tall (so they intersect both $S_{T, \eps^{i+1}}$ and $S_{B, \eps^{i+1}}$) or not so tall (so that they intersect only one of $S_{T, \eps^{i}}$ or $S_{B, \eps^{i}}$) for some $i \in [1/2\eps]$.
\begin{lemma}
	\label{lem:medWeight}
	Given any constant $\eps>0$, there exists a value $i \in [1/2\eps]$ such that all rectangles having height $\in ((1-2\eps^i)N, (1-\eps^{i+1})N]$  have total profit at most $\eps |M|$.
\end{lemma}
\begin{proof}
	Let $K_i$ be the set of rectangles with height $\in ((1-2\eps^i)N, (1-\eps^{i+1})N]$
	for $i \in [1/2\eps]$.
	Clearly a rectangle can belong to at most two such sets. Thus, cheapest set among these $1/2\eps$
	sets has cardinality  at most $\eps |M|$.
\end{proof}

\begin{lemma}
	\label{lem:stripHalf}
	Given any constant $\eps_s \ge \eps_{small}$, 
	either $a(E_{L, \eps_s} \cup E_{R, \eps_s}) \le \frac{(1+8 \eps_s)}{2} N^2$
	or $a(E_{T, \eps_s} \cup E_{B, \eps_s}) \le  \frac{(1+8 \eps_s)}{2} N^2$.
\end{lemma}
\begin{proof}
	Let us define $H := E_{L, \eps_s} \cup E_{R, \eps_s}$ 
	and $V := E_{T, \eps_s} \cup E_{B, \eps_s}$.
	Note that,
	$a(V) + a(H)
	\le 
	a(V \cup H)+
	a(V \cap H)$. 
	Now, $a(V \cup H) \le N^2$ as all rectangles were packed into a $N \times N$ knapsack.
	On the other hand, except possibly four rectangles (the ones that contain at least one of the points $(\eps_s N,\eps_s N),((1-\eps_s) N,\eps_s N),(\eps_s N,(1-\eps_s)N),((1-\eps_s)N, (1-\eps_s)N)$) all other rectangles in $(V \cap H)$ lie entirely within the four $\eps_s N$ strips.
	Thus $a(V \cap H) \le 4\eps_s N^2 +4\eps_{small} N^2 \le 
	8\eps_s N^2$, as $\eps_{small} \le \eps_s$. 
	Thus, $\min\{ a(V),a(H) \} \le  \frac{a(V \cup H)+
		a(V \cap H)}{2} \le \frac{(1+8 \eps_s)}{2} N^2$.
\end{proof}

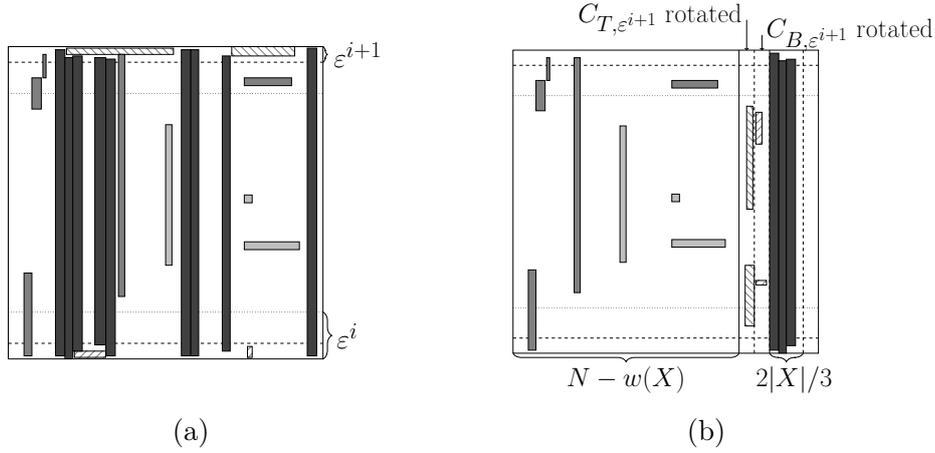
\begin{figure}
	\centering
	\captionsetup[subfigure]{justification=centering}
	\begin{subfigure}[b]{.35\textwidth}
		\resizebox{!}{150pt}{
			\begin{tikzpicture}
			\draw (0,0) -- (10,0) -- (10,10)-- (0,10) -- (0,0);
			\draw[dashed]  (0,0.5) -- (10,0.5);
			\draw[dashed]  (0,9.5) -- (10, 9.5);
			\draw[dotted]  (0,1.5) -- (10,1.5);
			\draw[dotted]  (0,8.5) -- (10, 8.5);
			\filldraw[fill=darkgray, draw=black] (1.5,0.1) rectangle (1.8, 9.9);
			\filldraw[fill=darkgray, draw=black] (1.8,0.005) rectangle (2.05, 9.65);
			\filldraw[fill=darkgray, draw=black] (2.05,0.25) rectangle (2.35, 9.7);
			\filldraw[fill=darkgray, draw=black] (2.75, 0.45) rectangle (3.1, 9.65);
			\filldraw[fill=darkgray, draw=black] (3.1,0.1) rectangle (3.4, 9.6);
			\filldraw[fill=darkgray, draw=black] (5.5,0.1) rectangle (5.8, 9.9);
			\filldraw[fill=darkgray, draw=black] (5.8,0.1) rectangle (6.05, 9.9);
			\filldraw[fill=darkgray, draw=black] (6.8,0.25) rectangle (7.05, 9.7);
			\filldraw[fill=darkgray, draw=black] (9.5,0.1) rectangle (9.8, 9.95);
			\draw[pattern=north west lines, pattern color=gray] (1.85, 9.75) rectangle (5.25, 9.95);
			\draw[pattern=north west lines, pattern color=gray] (7.1, 9.7) rectangle (9.1, 10);
			\draw[pattern=north east lines, pattern color=gray] (2.1, 0.05) rectangle (3.1, 0.25);
			\draw[pattern=north east lines, pattern color=gray] (7.6, 0.05) rectangle (7.75, 0.4);
			\filldraw[fill=gray, draw=black] (0.5,0.1) rectangle (0.75, 2.75);
			\filldraw[fill=gray, draw=black] (0.75,8) rectangle (1.05, 9);
			\filldraw[fill=gray, draw=black] (1.1,9) rectangle (1.2, 9.75);
			\filldraw[fill=gray, draw=black] (3.5,2) rectangle (3.7, 9.75);
			\filldraw[fill=gray, draw=black] (7.5, 8.75) rectangle (9, 9);
			\filldraw[fill=lightgray, draw=black] (5,3) rectangle (5.2, 7.5);
			\filldraw[fill=lightgray, draw=black] (7.5,3.5) rectangle (9.25, 3.75);
			\filldraw[fill=lightgray, draw=black] (7.5,5) rectangle (7.75, 5.25);
			
			\draw [decorate,decoration={brace,amplitude=4pt}, thick] (10,10) -- (10,9.5); 
			\draw (10.2,9.75) node [anchor = west] {\huge $\eps^{i+1}$}; 
			\draw [decorate,decoration={brace,amplitude=8pt},thick] (10,1.5) -- (10,0); 
			\draw (10.3,0.75) node [anchor = west] {\huge $\eps^{i}$};
			
			\fill[color=white] (0,-1.23) rectangle (1,-0.5);
			\fill[color=white] (0,11) rectangle (1,11.5);
			\end{tikzpicture}}
		\caption{\label{f:uwrot1} }
	\end{subfigure}%
	\hspace{45pt}
	\begin{subfigure}[b]{.37\textwidth}
		\resizebox{!}{150pt}{
			\begin{tikzpicture}
			\draw (0,0) -- (10,0) -- (10,10)-- (0,10) -- (0,0);
			\draw (7.4,0) -- (7.4,10);
			\draw[dashed] (7.9,0) -- (7.9,10);
			\draw[dashed] (8.4,0) -- (8.4,10);
			\draw[dashed]  (0,0.5) -- (10,0.5);
			\draw[dashed]  (0,9.5) -- (10, 9.5);
			\draw[dotted]  (0,1.5) -- (10,1.5);
			\draw[dotted]  (0,8.5) -- (10, 8.5);
			\filldraw[fill=gray, draw=black] (0.5,0.1) rectangle (0.75, 2.75);
			\filldraw[fill=gray, draw=black] (0.75,8) rectangle (1.05, 9);
			\filldraw[fill=gray, draw=black] (1.1,9) rectangle (1.2, 9.75);
			\filldraw[fill=gray, draw=black] (2,2) rectangle (2.2, 9.75);
			\filldraw[fill=gray, draw=black] (5.2, 8.75) rectangle (6.7, 9);
			\filldraw[fill=lightgray, draw=black] (3.5,3) rectangle (3.7, 7.5);
			\filldraw[fill=lightgray, draw=black] (5.2,3.5) rectangle (6.95, 3.75);
			\filldraw[fill=lightgray, draw=black] (5.2,5) rectangle (5.45, 5.25);
			
			\draw[pattern=north west lines, pattern color=gray] (7.65, 8.15) rectangle (7.85, 4.75);
			\draw[pattern=north west lines, pattern color=gray] (7.6, 2.9) rectangle (7.9, 0.9);
			\draw[pattern=north east lines, pattern color=gray] (7.95, 7.95) rectangle (8.15, 6.9);
			\draw[pattern=north east lines, pattern color=gray] (7.95,2.4) rectangle (8.3, 2.25);
			
			\filldraw[fill=darkgray, draw=black] (8.4,0.1) rectangle (8.7, 9.9);
			\filldraw[fill=darkgray, draw=black] (8.7,0.005) rectangle (8.95, 9.65);
			\filldraw[fill=darkgray, draw=black] (8.95,0.25) rectangle (9.25, 9.7);
			\draw[dashed]  (9.5,0) -- (9.5, 10);
			
			\draw [decorate,decoration={brace,amplitude=8pt}] (7.4,0) -- (0,0); 
			\draw (3.7,-0.3) node [anchor = north] {\huge $N-w(X)$}; 
			\draw [decorate,decoration={brace,amplitude=8pt}] (9.5,0) -- (8.4,0); 
			\draw (9.2,-0.3) node [anchor = north] {\huge $2|X|/3$};
			\draw[->] (7.65,11) -- (7.65,10); 
			\draw (7.65,11) node [anchor=east] {\huge $C_{T,\eps^{i+1}}$ rotated};
			\draw[->] (8.15,10.5) -- (8.15,10); 
			\draw (8.15,10.5) node [anchor=west] {\huge $C_{B,\eps^{i+1}}$ rotated};
			
			\end{tikzpicture}}
		\caption{ 
		}
	\end{subfigure}
	\caption{Case A for cardinality 2DK with rotations. Dark gray rectangles are $X$, light gray rectangles are $Z$, gray (and hatched) rectangles are $Y$, hatched rectangles are $C_{T,\eps^{i+1}}$ and  $C_{B,\eps^{i+1}}$. Figure (a): original packing in $N \times N$, Figure (b): modified packing leaving space for resource contraction on the right.}
	\label{f:uwrotA}
\end{figure}

\begin{lemma}
	\label{lem:Stein}
	Given a constant $0<\eps_a<1/2$ and a set of rectangles $M:=\{1, \ldots, k\}$ with $\width(i) \ge \epsl N, \height(i) \le \epss N$ for all $i \in M$, if $a(M) \le (1/2+\eps_a)N^2$, then 
	a subset of $M$ with cardinality $(1-2\eps_s-2\eps_a)|M|$ can be packed into a $N \times (1-\eps_s)N$ knapsack.
\end{lemma}
\begin{proof}
	Without loss of generality, assume the rectangles in $M$ are given in nondecreasing order according to their area.
	Note that $a(i) \le \eps_s N^2$ for any $i \in M$.
	Let, $S:=\{1, \ldots, j\}$ such that $\frac{(1-2\eps_s)}{2}N^2 \le \sum_{i=1}^j a(i)  \le \frac{(1-\eps_s)}{2}N^2$ and $\sum_{i=1}^{j +1} a(i) > \frac{(1-\eps_s)}{2}N^2$.
	Then from Theorem~\ref{thm:steinberg}, $S$ can be packed into $N \times (1-\eps_s)N$ bin.
	As we considered rectangles in the order of  nondecreasing area, $\frac{|S|}{|M|} \ge \frac{\left(\frac12-\eps_s\right)}{\left(\frac12+\eps_a\right)}$.
	Thus, $ |S| \ge 
	\left(1- \frac{(\eps_a + \eps_s)}{\left(\frac12+\eps_a\right)}\right)|M|
	\ge (1-2\eps_a -2\eps_s)|M|$.
\end{proof}

From Lemma~\ref{lem:medWeight}, let $V'$ be the set of rectangles having height $\in [(1-2\eps^i)N, (1-\eps^{i+1})N]$  such that $|V'| \le \eps |M|$. 
We remove the rectangles in $V'$.
Now consider strips $S_{L, \eps^{i}}, S_{R, \eps^i} ,S_{T, \eps^i} ,S_{B, \eps^i}$.
From Lemma~\ref{lem:item-classification}, we can always choose $\epss$ such that $\eps^{1/2\eps} > \epss$.
Then as $\eps^i \ge \eps_{small}$, from Lemma~\ref{lem:stripHalf}, without loss of generality assume 
that $a(E_{T, \eps^i} \cup E_{B, \eps^i}) \le  \frac{(1+8 \eps^i)}{2}a(M)$.
Let $X$ be the set of rectangles that intersect both $S_{T, \eps^i}$ and $S_{B, \eps^i}$  and $Y := \{ E_{T, \eps^i} \cup E_{B, \eps^i} \} \setminus X$. 
Define $Z:=M \setminus \{ X \cup Y \cup V \}$ to be the rest of the rectangles.
Let us define $w(X)=\sum_{i \in X} \width(i)$. Now there are two cases.

\noindent \textbf{Case A.} $w(X) \ge 12 \eps^{i+1}N$.
From Lemma~\ref{lem:medWeight}, all rectangles in $X$
intersects both $S_{T, \eps^{i+1}}$ and $S_{B, \eps^{i+1}}$. 
So, removal of  all rectangles in $X, C_{T, \eps^{i+1}}$ and $C_{B, \eps^{i+1}}$ creates many empty strips of height $N$ and total width of $w(X)$. 
Now if in a packing, there are $k$ strips of height (or width) $N$  and width (or height) $w_1, w_2, \dots, w_k$, then one can push the rectangles in the packing to the left or right and create an empty strip of width $\sum_{i=1}^k w_i$.
Thus we can push all the remaining  rectangles in $Y \cup Z$
to the left so that they fit into a strip of width $N-w(X)$. Then we rotate $C_{T, \eps^{i+1}}$ and $C_{B, \eps^{i+1}}$ and pack them in two strips, each
of width $\eps^{i+1}N$. 
Note that $\width(i) \le \eps^{i+1}N$ for all $i \in X$.
Now take rectangles in $X$ by nondecreasing width, till total width is in $[w(X)-4\eps^{i+1}N, w(X)-3\eps^{i+1}N]$ and pack them in the remaining area (see Figure \ref{f:uwrotA}). The cardinality of this set is at least $\frac{(w(X)- 4\eps^{i+1}N)}{w(X)}|X|
\ge \frac{2}{3}|X|$. 
Hence, at least $\frac23|X|+|Y|+|Z|  \ge \frac23(|X|+|Y|+|Z|) $ rectangles are packed into $N \times (1-\eps^{i+1})N$.


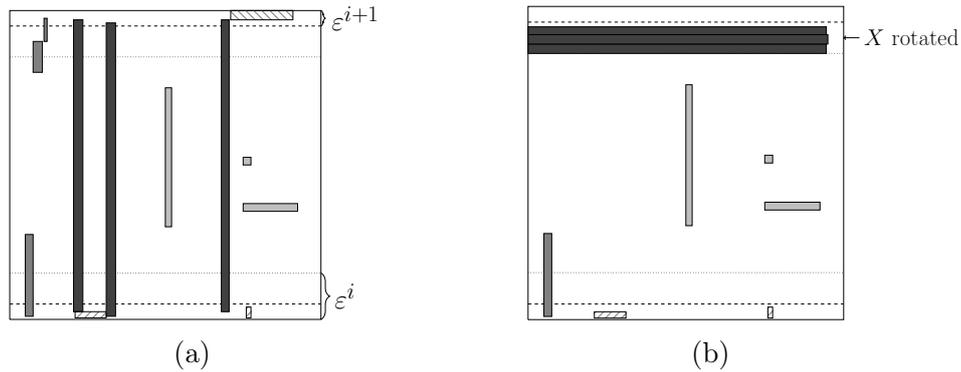
\begin{figure}
	\captionsetup[subfigure]{justification=centering}
	\centering
	\hspace{-20pt}
	\begin{subfigure}[b]{.35\textwidth}
		\resizebox{!}{120pt}{
			\begin{tikzpicture}
			\draw (0,0) -- (10,0) -- (10,10)-- (0,10) -- (0,0);
			\draw[dashed]  (0,0.5) -- (10,0.5);
			\draw[dashed]  (0,9.5) -- (10, 9.5);
			\draw[dotted]  (0,1.5) -- (10,1.5);
			\draw[dotted]  (0,8.5) -- (10, 8.5);
			\filldraw[fill=darkgray, draw=black] (2.05,0.25) rectangle (2.35, 9.7);
			\filldraw[fill=darkgray, draw=black] (3.1,0.1) rectangle (3.4, 9.6);
			\filldraw[fill=darkgray, draw=black] (6.8,0.25) rectangle (7.05, 9.7);
			\draw[pattern=north west lines, pattern color=gray] (7.1, 9.7) rectangle (9.1, 10);
			\draw[pattern=north east lines, pattern color=gray] (2.1, 0.05) rectangle (3.1, 0.25);
			\draw[pattern=north east lines, pattern color=gray] (7.6, 0.05) rectangle (7.75, 0.4);
			\filldraw[fill=gray, draw=black] (0.5,0.1) rectangle (0.75, 2.75);
			\filldraw[fill=gray, draw=black] (0.75,8) rectangle (1.05, 9);
			\filldraw[fill=gray, draw=black] (1.1,9) rectangle (1.2, 9.75);
			\filldraw[fill=lightgray, draw=black] (5,3) rectangle (5.2, 7.5);
			\filldraw[fill=lightgray, draw=black] (7.5,3.5) rectangle (9.25, 3.75);
			\filldraw[fill=lightgray, draw=black] (7.5,5) rectangle (7.75, 5.25);
			
			\draw [decorate,decoration={brace,amplitude=4pt}, thick] (10,10) -- (10,9.5); 
			\draw (10.2,9.75) node [anchor = west] {\huge $\eps^{i+1}$}; 
			\draw [decorate,decoration={brace,amplitude=8pt}, thick] (10,1.5) -- (10,0); 
			\draw (10.3,0.75) node [anchor = west] {\huge $\eps^{i}$}; 
			\end{tikzpicture}}
		\caption{\label{f:uwrot3} 
		}
	\end{subfigure}%
	\hspace{50pt}
	\begin{subfigure}[b]{.35\textwidth}
		\resizebox{!}{120pt}{
			\begin{tikzpicture}
			\draw (0,0) -- (10,0) -- (10,10)-- (0,10) -- (0,0);
			\draw[dashed]  (0,0.5) -- (10,0.5);
			\draw[dashed]  (0,9.5) -- (10, 9.5);
			\draw[dotted]  (0,1.5) -- (10,1.5);
			\draw[dotted]  (0,8.5) -- (10, 8.5);
			
			\filldraw[fill=darkgray, draw=black] (0,8.5) rectangle (9.45, 8.8); 
			\filldraw[fill=darkgray, draw=black] (0,8.8) rectangle (9.5, 9.1); 
			\filldraw[fill=darkgray, draw=black] (0,9.1) rectangle (9.45, 9.35); 
			
			\draw[pattern=north east lines, pattern color=gray] (2.1, 0.05) rectangle (3.1, 0.25);
			\draw[pattern=north east lines, pattern color=gray] (7.6, 0.05) rectangle (7.75, 0.4);
			\filldraw[fill=gray, draw=black] (0.5,0.1) rectangle (0.75, 2.75);
			\filldraw[fill=lightgray, draw=black] (5,3) rectangle (5.2, 7.5);
			\filldraw[fill=lightgray, draw=black] (7.5,3.5) rectangle (9.25, 3.75);
			\filldraw[fill=lightgray, draw=black] (7.5,5) rectangle (7.75, 5.25);
			
			\draw[->] (10.5,9) -- (10,9); 
			\draw (10.5,9) node [anchor=west] {\LARGE $X$ rotated};
			
			\fill[color=white] (0,10.1) rectangle (0.1,10.15);\end{tikzpicture}}
		\caption{
		}
	\end{subfigure}
	\caption{Case B for cardinality 2DK with rotations. Dark gray rectangles are $X$, light gray rectangles are $Z$, gray (and hatched) rectangles are $Y$, hatched rectangles are $C_{T,\eps^{i+1}}$ and  $C_{T,\eps^{i+1}}$. Figure (a): original packing, Figure (b): modified packing leaving space for resource contraction on the top.}
	\label{f:uwrotB}
\end{figure}

\noindent \textbf{Case B.}  $w(X) < 12 \eps^{i+1}N$.
Without loss of generality, assume  $|E_{B, \eps^i} \setminus X| \ge |Y|/2 \ge |E_{T, \eps^i}\setminus X|$.
Then remove $E_{T, \eps^i}$. Pack $X$ on top of $M \setminus E_{T, \eps^i}$
as $12 \eps^{i+1} \le \eps^i - \eps^{i+1}$ (see Figure \ref{f:uwrotB}).
This gives a packing of $|X|+|Z|+\frac{|Y|}{2}$.
On the other hand, as $a(X \cup Y)=a(E_{T, \eps^i} \cup E_{B, \eps^i}) \le  \frac{(1+8 \eps^i)}{2}a(M)$, from Lemma~\ref{lem:Stein}, we claim that 
$(1-2\eps^{i+1}-8\eps^i)(|X \cup Y|)$ can be packed into $N \times (1-\eps^{i+1})N$ bin.

Thus we can always pack a set of rectangles with cardinality at least $\max\{(1-10\eps^i)(|X|+|Y|), |X|+|Z|+\frac{|Y|}{2}\} \ge 
\frac{2}{3}(1-10\eps^i)(1-\eps)|M| $.
This concludes the proof of Lemma~\ref{uwrescontr1}.

Now we can obtain a PTAS for container packing using Theorem~\ref{thm:container_packing_ptas}.
This completes the proof of Theorem~\ref{thm:mainNoRotation} for the cardinality case.

\section{Conclusions}

In this chapter, we exploited the additional flexibility given by the freedom to rotate the rectangles to design improved approximation algorithms for \tdkr. While \fontL-packings do not seem to be useful for this variation of the problem, rotating the rectangles allows us to make use of the \emph{resource contraction} technique. Thus, we proved that simple container packings can obtain an approximation factor of $3/2 + \eps$ in general, and $4/3 + \eps$ for the cardinality case.

Like for \tdk, a PTAS is in the realm of possibilities, making this problem a still interesting target for future research.

\chapter{Unsplittable Flow on a Path with Bags}\label{chap:bagUFP}

In the well-studied {\em Unsplittable Flow on a Path} problem (UFP) we are given a path graph $G=(V,E)$, where $V=\{0,1,\ldots,m\}$ and $E = \{(i, i+1) \,|\, 0 \leq i < m\}$, with positive integer edge capacities $\{u_e\}_{e\in E}$ and a collection $T$ of $n$ tasks. Each task $i\in T$ is associated with a weight $w_i\in \mathbb{N}^+$, a demand $d_i\in \mathbb{N}^+$, and a subpath $P_i$ between nodes $s_i$ and $t_i$. Let $T_e=\{i\in T: e\in P_i\}$ be the tasks \emph{containing} edge $e$. Our goal is to select a subset of tasks $T'\subseteq T$ of maximum total weight $w(T'):=\sum_{i\in T'}w_i$ so that, for each edge $e$, the total demand $d_e(T'):=\sum_{i\in T'\cap T_e}d_i$ of selected tasks using that edge is upper bounded by the corresponding capacity $u_e$. Intuitively, edge capacities model a given resource whose amount varies over time (in a discrete fashion), and each task represents a request to use some specified amount of that resource (the demand) for a given time. By standard reductions, we can assume that $m\leq 2n$ and all edge capacities are distinct.

If restricted to just one edge, the problem is equivalent to the classical (one-dimensional) Knapsack problem, which is weakly NP-hard. On an arbitrary number of edges, UFP is strongly NP-hard even with uniform demands and capacities (\cite{dps10, cwmx12}).

A well-studied case of UFP considers instances that satisfy the \emph{no-bottleneck assumption} (NBA), that is, the assumption that the maximum demand is upper bounded by the minimum capacity. \cite{cckr02} showed a $(2+\eps)$-approximation for the special case with uniform capacities (known in literature as \emph{resource allocation problem}, RAP), improving on earlier $3$-approximations by \cite{bbfns01}. It is interesting to note that $2+\eps$ is the best known polynomial time approximation even in this very restricted special case. \cite{ccgk02} showed the first constant factor approximation for UFP under NBA, with a ratio of $78.51$. Subsequent work by \cite{cks06} obtained a $(2+\eps)$-approximation.

In the general case, the first constant-factor approximation in polynomial time was given in \cite{bsw11}, with a ratio of $7+\eps$. Finally, \cite{aglw14} obtained a $(2+\eps)$-approximation, matching for the general case the best ratio known in the NBA special case. The work of \cite{bces06}, together with the recent improvements by \cite{bgkmw15}, imply a QPTAS in the general case. The recent work of \cite{w17} proved that, while the unweighted version of the problem is $W[1]$-hard when parameterized on the size $k$ of the optimal solution, there is a $(1+\eps)$-approximation in time $2^{O(k \log k)} n^{O_\eps(1)}$; thus, it is unlikely that there exists an EPTAS for UFP. \cite{gmwz17} obtained a PTAS for several interesting special cases, including rooted instances (where all tasks share one edge) and the case where the profit of each task $i$ is proportional to its area $d(i) \cdot |P(i)|$. The existence of a PTAS in the general case, or even any approximation algorithm that breaks the barrier of $2$, is a major open problem.

The {\em UFP with Bags} problem (bagUFP) is a generalization of UFP where tasks are partitioned into a set of $h$ bags $J = \{ \mathcal{B}_1,\ldots,\mathcal{B}_h \}$, and we have the extra constraint that at most one task per bag can be selected. Intuitively, bags model \emph{jobs} that we can execute at different points of time (and at each such time one has a different demand, weight, and processing time). This problem is APX-hard even in the special case of unit demands and capacities (\cite{s99}), that implies that there exists no PTAS. Under the NBA, this problem was studied in \cite{ccs10}, who provided a $120$-approximation; this result was improved by \cite{eggknp12}, who obtained a $65$-approximation. For the general case, \cite{ccgrs14} recently gave a $O(\log n)$-approximation for bagUFP; the approximation factor remains the same in the case of uniform weights.

\section{Related work}

The UFP problem is a very special case of the Unsplittable Flow (UF) problem, where $G$ is a general graph. When $G$ is undirected and all capacities, demands and profits are $1$, the problem is called \emph{maximum Edge Disjoint Path} (EDP) problem. Due to their generality and importance in routing, UF and EDP have been extensively studied. \cite{cks06} proved a $O(\sqrt{n})$-approximation, based on rounding the fractional solution of the corresponding multi-commodity flow relaxation; their result also applies to the UF problem under the no-bottleneck assumption. On the other hand, the best known hardness for EDP is $\Omega\left((\log n)^{\frac12 - \eps}\right)$ under the assumption that $\textbf{NP} \not\subseteq ZPTIME\left(n^{\polylog n}\right)$; very recently, \cite{ckn17} proved a $2^{O(\sqrt{\log n})}$-hardness assuming $\textbf{NP} \not\subseteq PTIME\left(n^{O(\log n)}\right)$. Despite this breakthrough, there is an exponential gap between lower and upper bounds.

The \emph{Unsplittable Flow on a Tree} (UFT) problem is the special case of UF when $G$ is an undirected tree. This removes the difficulty of routing demands, since there exists a unique path between any two nodes. Still, the problem is far from easy: the best known approximation is a $O(\log^2 n)$ by \cite{cek09}, while the problem is only known to be APX-hard. A $O(1)$-approximation is known for the special case where the no-bottleneck assumption holds (\cite{cms07}).

\section{Basic notions}

Let $i \in T$ be a task. We call the \emph{bottleneck} of $i$ the edge $e \in P(i)$ with minimum capacity, and let $b_i$ its capacity. We say that $i$ is \emph{small} if $d_i \leq 1/2 \cdot b_i$; we say that $i$ is \emph{large} otherwise.

\begin{definition}\label{def:capacity_profile}
	For an instance of the UFP (or bagUFP) problem, we call \emph{capacity profile} the closed curve containing, for each edge $e = (i, i+1)$, the segment $(i, u_e)\mbox{--}(i+1, u_e)$, completed with the minimal vertical segments needed in order to obtain a continuous curve.
\end{definition}

\begin{figure}
	\centering
	\includegraphics[width=12cm]{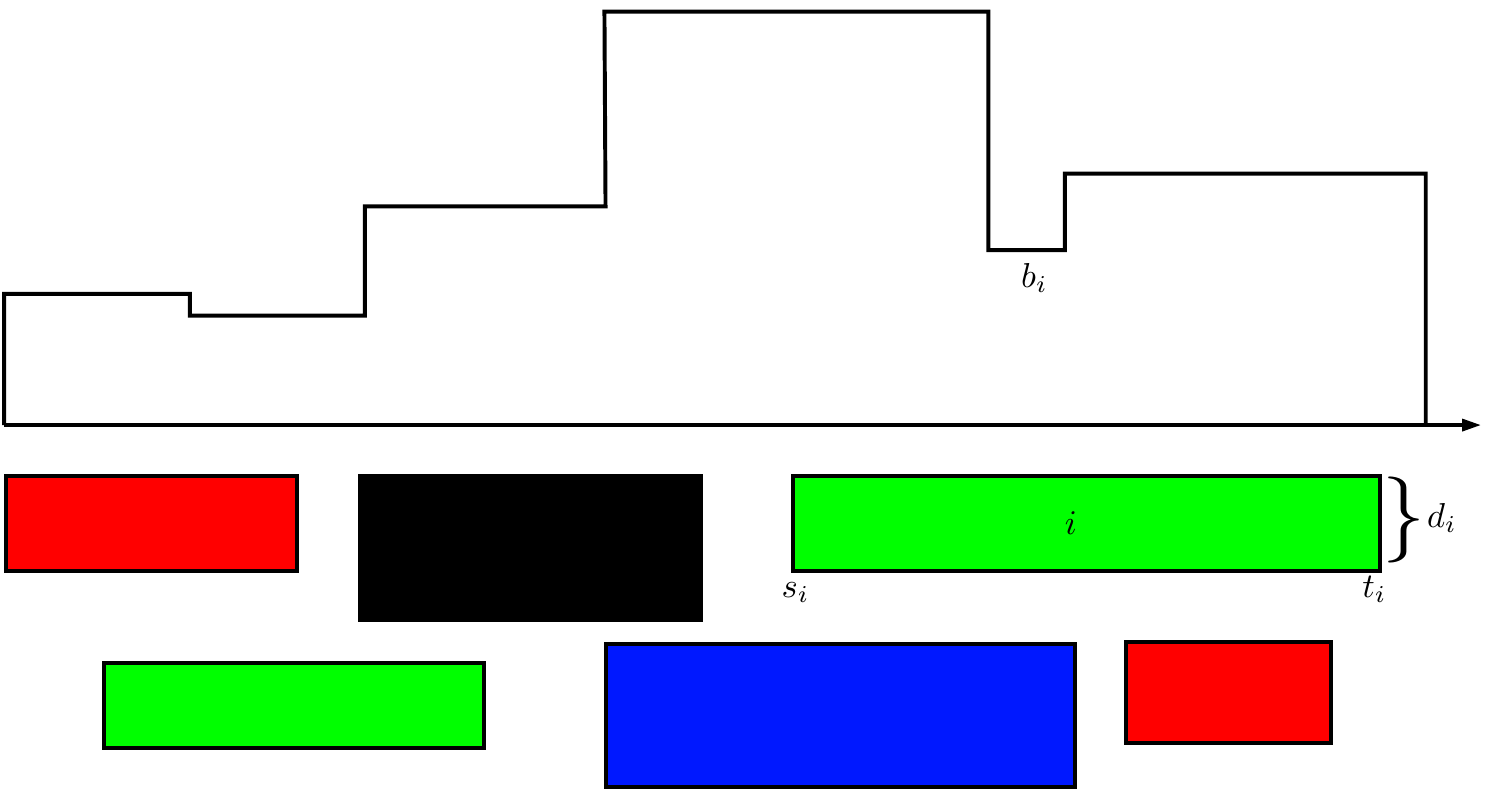}
	\caption{Illustration of the main concept of a bagUFP instance. Each task $i$ is represented as a rectangle, where the starting and ending $x$-coordinate represent $s_i$ and $t_i$, while the height represents the demand of that task; $b_i$ is the bottleneck capacity of task $i$. Tasks in the same bag are represented with the same color.}
	\label{fig:bagUFP}
\end{figure}

\begin{lemma}\label{lem:ufp_partition}
	Suppose that there is an $\alpha$-approximation for the special case when all the tasks are small, and a $\beta$-approximation for the case when all the tasks are large, for some $\alpha, \beta \geq 1$. Then there exists a $\alpha + \beta$ approximation algorithm for the general case.
\end{lemma}
\begin{proof}
	Consider an instance of bagUFP with set of tasks $T$. Suppose that $T = T_S \mathbin{\dot{\cup}} T_L$, where $T_S$ and $T_L$ are the small and large tasks, respectively. Let $OPT$ be the optimal solution, and let $OPT_S = OPT \cap T_S$ and $OPT_L = OPT \cap T_L$.\\
	We run the $\alpha$-approximation on $T_S$ the $\beta$-approximation on $T_L$ to obtain the feasible solutions $APX_S$ and $APX_L$, respectively; then we output the one with largest profit, let it be $APX$.\\
	Clearly, $p(APX_S) \geq p(OPT_S)/\alpha$ and $p(APX_L) \geq p(OPT_L)/\beta$, which implies:
	\begin{align*}
	p(APX) &= \max\left\{p(APX_S), p(APX_L)\right\} \geq \frac{\alpha p(APX_S) + \beta p(APX_L)}{\alpha + \beta}\\
	&\geq \frac{p(OPT_S) + p(OPT_L)}{\alpha + \beta} = \frac{p(OPT)}{\alpha + \beta}
	\end{align*}
\end{proof}
Many results on this problem (and its variations) are indeed based on partitioning the items in the two classes of \emph{small} and \emph{large}, although the exact definition might be different from the one we use here. In fact, when items are relatively small, LP-based techniques have been applied successfully; but they would break down on large items. Vice versa, for relatively large items, good results are typically obtained with Dynamic Programming algorithms. Note that even if algorithms were found that solve \emph{exactly} instances with either only small items or only large items, Lemmsa~\ref{lem:ufp_partition} only guarantees a $2$-approximation overall. Thus, breaking the barrier of $2$, if it is possible, would requires fundamentally new ideas that allow to handle small and large items simultaneously.

Using the primal-dual technique, for instances of bagUFP entirely composed of small tasks, \cite{ccgrs14} proved the following:
\begin{lemma}[\cite{ccgrs14}]\label{lem:ufp_small_apx}
	There exists a $9$-approximation algorithm for the bagUFP problem running in time $O\left(n^2\right)$, if all the tasks are small.
\end{lemma}
Thus, in the following we can focus on instances where all tasks are large.

The natural linear programming relaxation for bagUFP, denoted by LP$_{\textup{bagUFP}}$ is given below. The LP has a variable $x_i$ for each task $i \in T$.
\begin{align*}
\mbox{maximize } \quad &\textstyle \sum_{i: T_i \in \mathcal{T}} w_i x_i & \left(\mbox{LP$_{\textup{bagUFP}}$}\right)\\
\mbox{s.t. }     \quad &\textstyle \sum_{i : e \in T_i} d_i x_i \leq u_e & \forall e \in E\\
\mbox{} &\textstyle \sum_{T_i \in \mathcal{B}_j} x_i \leq 1 & \forall \mathcal{B}_j \in J\\
& x_i \geq 0 & \forall T_i \in \mathcal{T}  
\end{align*}
Note that, if the variables are restricted to $\{0, 1\}$, this is the exact integer programming formulation of the problem.

As it was first done by \cite{bsw11}, for large tasks we exploit a connection to the Maximum~Weight Independent~Set of Rectangles (MWISR) problem, which is the following problem: given a set of axis-aligned rectangles in the plane, each with an associated weight, find the maximum weight set of pairwise non-overlapping rectangles\footnote{We say that two rectangles overlap if their interiors have non-empty intersection.}. The best known polynomial time approximation for this problem is an $O(\log n / \log \log n)$-approximation by \cite{ch12}, although \cite{cc09} gave a $O(\log \log n)$-approximation for the unweighted (cardinality) case. Moreover, a QPTAS was given in \cite{aw13}.

\begin{figure}
	\centering
	\includegraphics[width=12cm]{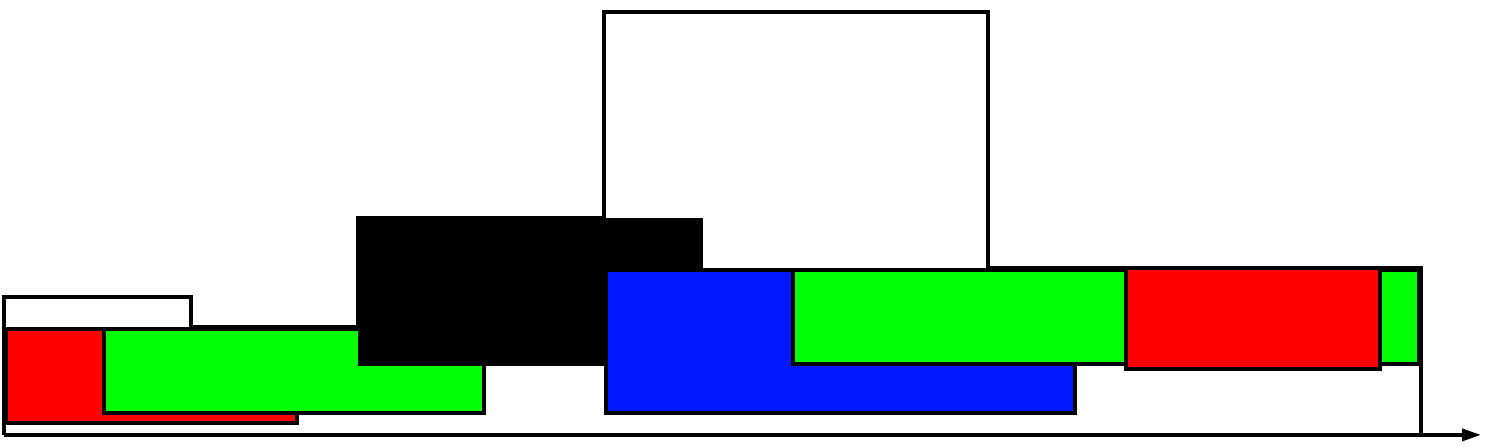}
	\caption{The top-drawn instance corresponding to the tasks in Figure~\ref{fig:bagUFP}.}
	\label{fig:bagUFP-topdrawn}
\end{figure}

For each large task $i$ of a UFP instance, we associate the rectangle $R_i$ with top-left corner $(s_i, b_i)$ and bottom-right corner $(t_i, l_i)$ where $l_i = b_i - d_i$. We call this set of rectangles the \emph{top-drawn} instance corresponding to $T_\textup{large}$ (see Figure~\ref{fig:bagUFP-topdrawn} for an illustration of the capacity profile and a top-drawn instance). Thus, the set of rectangles $\mathcal{R} = \left\{R_i\right\}_i$ with the weight function $w_i$ defines an MWISR instance, and it is easy to see that a feasible solution for the MWISR instance corresponds to a feasible subset of task in the UFP instance. Moreover the following result applies:
\begin{lemma}[implied by Lemma~13 in \cite{bsw11}]\label{lem:color}
	Suppose that all tasks are large. Then any feasible UFP solution $S$ can be partitioned in $4$ sets $S_1, \dots, S_4$ such that each $S_j$ is an independent set of rectangles.
\end{lemma}
A similar reduction works for bagUFP, except that what we obtain is an instance of the natural generalization of the MWISR problem that we call bagMWISR, where we also preserve the bag constraints between rectangles, corresponding to the original bag constraints among tasks.

We define the set $\mathcal{P}$ as follows. We consider the (non-uniform) grid induced by the horizontal and vertical lines containing the edges of the rectangles. For each finite cell in this grid, we add its center to $\mathcal{P}$ if and only if it overlaps with at least one rectangle in $\mathcal{R}$, as shown in Figure~\ref{fig:LP-points}. Clearly, $|\mathcal{P}| = O(n^2)$, since the grid is defined by at most $2n$ horizontal and at most $2n$ vertical lines.

\begin{figure}
	\centering
	\includegraphics[width=12cm]{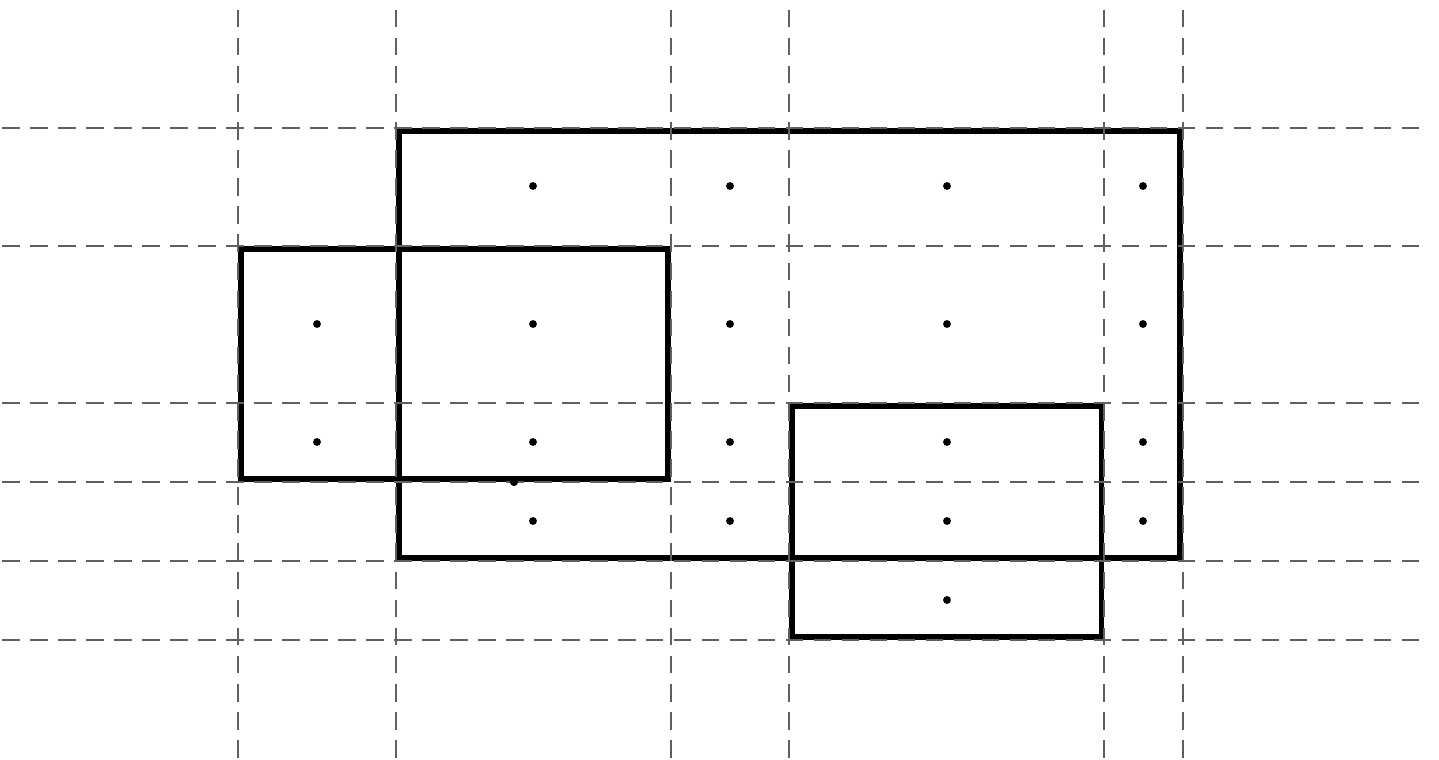}
	\caption{The points in $\mathcal{P}$ in the definition of LP$_{\textup{bagMWISR}}$.}
	\label{fig:LP-points}
\end{figure}

The following is the natural LP relaxation for the bagMWISR problem, which we denote by LP$_{\textup{bagMWISR}}$\footnote{We slightly abuse notation by reusing the symbol $\mathcal{B}_j$ to refer either to the bags of tasks (if we talk about bagUFP) or to the corresponding bags of rectangles of bagMWISR.}. There is a variable $y_i$ for each rectangle $R_i \in \mathcal{R}$, and we define $\mathcal{P}$ as the set of corners of rectangles in $\mathcal{R}$, plus each intersection point of a vertical edge of $R_i$ and a horizontal edge of $R_j$, for any two rectangles $R_i, R_j \in \mathcal{R}$.
\begin{align*}
\mbox{maximize } \quad & \textstyle \sum_{R_i \in \mathcal{R}} w_i y_i & \left(\mbox{LP$_{\textup{bagMWISR}}$}\right)\\
\mbox{s.t. } \quad & \textstyle \sum_{R_i \in \mathcal{R}: p \in R_i} y_i \leq 1  & \forall \, p \in \mathcal{P}\\
\mbox{} & \textstyle \sum_{R_i \in \mathcal{B}_j} y_i \leq 1  & \forall \, \mathcal{B}_j \in  J\\
&y_i \geq 0 & \forall \, R_i \in \mathcal{R} 
\end{align*}
Let $\mathbf{y}$ be the optimal solution and let $opt = \sum_{R_i \in \mathcal{R}} w_i y_i$ be its value. Clearly, if $OPT$ is the value of the optimal integral solution, $opt \geq OPT$. 

\begin{sloppypar}
In the next two sections, we describe approximation algorithms for bagMWISR and, as a corollary, improved approximations for bagUFP. These results have been published in \citet*{giu15}. In the same paper, we also considered a special case of bagUFP, that we call \emph{UFP with time windows} (twUFP). Here, together with the path graph $G$, we are given a collection of \emph{jobs}. Each job $j$ is characterized by a weight, a demand $\tau_j$, and a \emph{time window}, which is the sub-path between two given nodes $s_j$ and $t_j$. For each possible node $\delta_i$ such that $s_j \leq \delta_i \leq t_j - \tau_j$, we define a task $i$ with the same weight and demand as $i$, and whose subpath $P_i$ has the starting point $s_i = \delta_i$ and the endpoint $t_i = s_i + \tau_j$. The tasks corresponding to the same job $j$ define a bag $\mathcal{B}_j$.
\end{sloppypar}

We introduce a limitation to the instances of twUFP that we call \emph{Bounded Time-Window Assumption} (BTWA): namely, we assume that there exists a fixed positive constant $C$ such that $t_j - s_j \leq C t_j$ for each job $j$; that is, the time window size is bigger than the processing time by at most a constant factor. We proved the following result: 

\begin{theorem}
	There is a QPTAS for twUFP under BTWA and assuming that demands are quasi-polynomially bounded in $n$.
\end{theorem}

This result is a generalization of the QPTAS for UFP by \cite{bces06}.

\section{A \texorpdfstring{$O(\log n / \log \log n)$}{O(log n / log log n)}-approximation for bagUFP}
\label{sec:weightedBagUFP}

In this section, we show how to extend the analysis in \cite{ch12} for MWISR to bagMWISR, obtaining an expected $O\left(\log n / \log \log n \right)$-approximation.

First, we consider the general case. We say that two rectangles \emph{cross} if they overlap but no rectangle contains a vertex of the other rectangle. Construct two undirected graphs $G_1$ and $G_2$ with vertex set $\mathcal{R}$ such the edge $(R_i, R_j)$ is in $G_1$ if the rectangles $R_i$ and $R_j$ are in the same bag, or if they overlap but do not cross\footnote{Observe that if two overlapping rectangles are not crossing, then at least one rectangle has a vertex that is contained in the interior of the other rectangle.}; and $(R_i, R_j)$ is in $G_2$ if they cross. Clearly, a set $S \subseteq \mathcal{R}$ is an independent set of rectangles which satisfies the bag constraints if and only if it is an independent set in both $G_1$ and $G_2$. We have the following lemma:

\begin{lemma}\label{lem:color_crossing_graph}
	Suppose that a set $S \subseteq \mathcal{R}$ is an independent set in $G_1$; moreover, suppose that the induced subgraph $G_2[S]$ has maximum clique size $\Delta$. Then it is possible to color in polynomial time each rectangle in $S$ using at most $\Delta$ colors so that no edge of $G_2[S]$ connects two rectangles of the same color.
\end{lemma}
\begin{proof}
	Call a rectangle $R_i$ \emph{thinnest} if no rectangle crosses $R_i$ and has a smaller width. Let $S_1$ be the set of thinnest rectangles. Clearly, they are pairwise not crossing and hence not overlapping, and thus they form an independent set (see Figure~\ref{fig:remove-thinnest}). We assign a color to the rectangles in $S_1$ and we remove them from the graph (together with the incident edges); clearly, the residual graph has maximum clique size at most $\Delta - 1$, and so we can iterate the process to obtain the independent sets $S_1, S_2, \dots, S_\Delta$.
\end{proof}

\begin{figure}
	\centering
	\includegraphics[width=9cm]{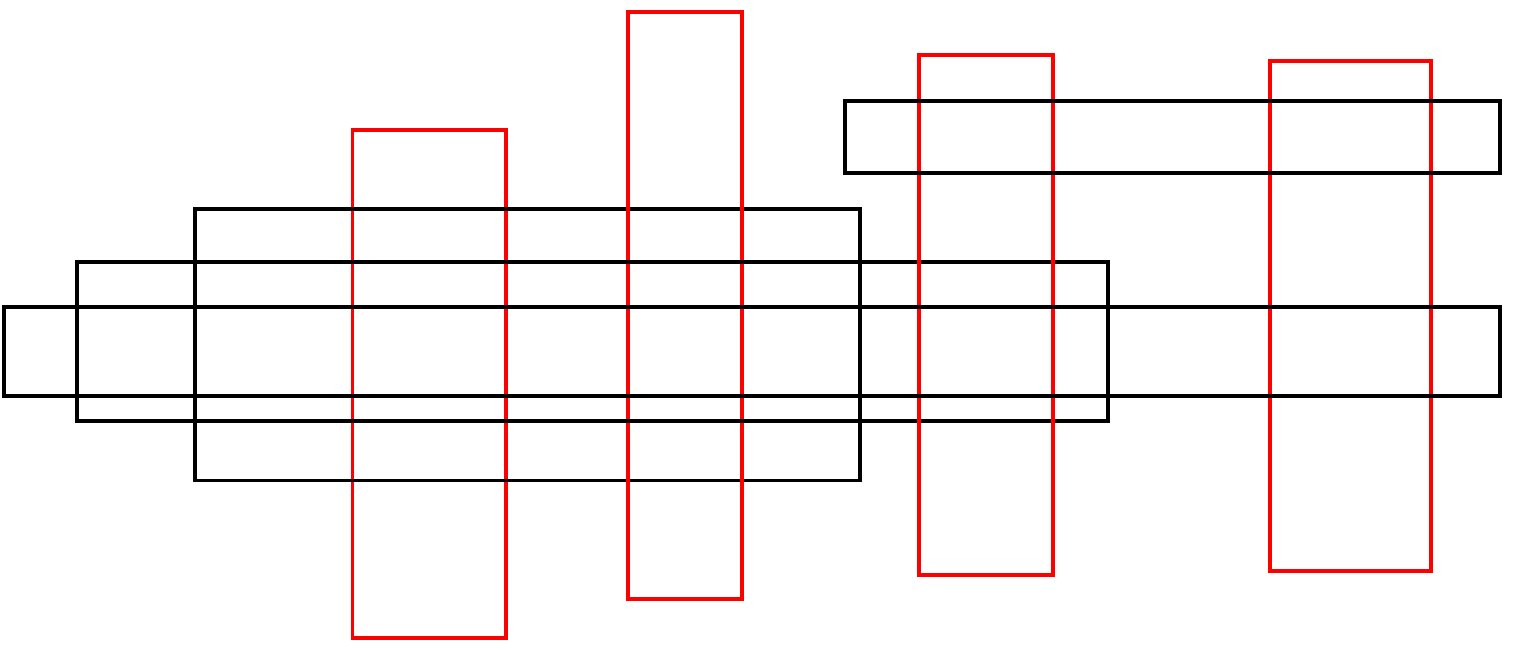}
	\caption{Proof of Lemma~\ref{lem:color_crossing_graph}. The \emph{thinnest} rectangles (in red) are independent, and removing them decreases the maximum clique size.}
	\label{fig:remove-thinnest}
\end{figure}

Consequently, if we find a set $S$ such that $G_1[S]$ is an independent set and $G_2[S]$ has maximum clique size $\Delta$, then we can find an independent set of rectangles $S'$ of profit at least $p(S)/\Delta$.
We will find such a set $S$ with the technique of \emph{randomized rounding} of the LP solution.

\subsection{Algorithm}

For two rectangles $R_i$ and $R_j$, we write that $R_i \otimes R_j$ if the interior of $R_i$ contains at least one corner of $R_j$, or vice versa; that is, $R_i \otimes R_j$ if $R_i$ and $R_j$ overlap but do not cross.

For an arbitrary set $\mathcal{K} \subseteq \mathcal{R}$ and rectangle $R_i \in \mathcal{R}$, we define the \emph{resistance} $\eta$ as:
\[
\eta (R_i, \mathcal{K}) = \sum_{\substack{R_j \in \mathcal{K} \setminus \mathcal{B}(R_i),\\ R_i \otimes R_j}} y_j + \sum_{R_j \in \mathcal{K} \cap \mathcal{B}(R_i) \setminus \{R_i\}} y_j
\]
That is, $\eta(R_i, \mathcal{K})$ is the sum of each $y_j$ for some $R_j \in \mathcal{K}$ such that $R_i \otimes R_j$, or $R_i$ and $R_j$ are in the same bag.

The algorithm works as follows. First, a fractional optimal solution $\mathbf{y}$ of LP$_{\textup{bagMWISR}}$ is computed. Then a permutation $\Pi$ of $\mathcal{R}$ is obtained as follows: if the first $i$ elements are $\Pi_i = \{\pi_1, \dots, \pi_i\}$, then the $(i+1)^{th}$ element $\pi_{i+1}$ is:
\[
\pi_{i+1} = \arg\min_{R_j \in \mathcal{R} \setminus \Pi_i}{\eta(R_j, \mathcal{R}\setminus \Pi_i)}
\]

Let $\eta_i = \eta(\pi_i, \mathcal{R}\setminus \Pi_{i-1})$ be the resistance of the $i^{th}$ rectangle. After the permutation is computed, a candidate set $C$ and an independent set $S$ for $G_1$ are computed with the following process. Initially, $C$ and $S$ are empty. Then, the members of the permutation are scanned in reverse order: at iteration $k$, the rectangle $\pi_{m-(k-1)}$ is added to set $C$ (initially empty) with probability $y(\pi_{m-(k-1)})/\tau$, where $\tau > 1$ is a constant integer parameter that will be defined later. If $\pi_{m-(k-1)}$ is added to $C$ and $S \cup \{\pi_{m-(k-1)}\}$ is an independent set in $G_1$, then $\pi_{m-(k-1)}$ is also added to $S$.

Finally, by using Lemma~\ref{lem:color_crossing_graph}, the algorithm finds a subset $S' \subseteq S$ that is an independent set in $G_2[S]$ such that $p(S') \geq p(S) / \Delta$, where $\Delta$ is the maximum clique size of $G_2[S]$.

\subsection{Analysis}

For a subset $\mathcal{H} \subseteq \mathcal{R}$, let us define $\mathcal{E}(\mathcal{H}) = \sum_{R_i \in \mathcal{H}} y_i$; moreover, let $\mathcal{C}(\mathcal{H}) = \{(p, i, j) \;|\; p\mbox{ is a corner of } R_i, \mbox{ where } R_i, R_j \in \mathcal{R}, i \neq j, p \in R_j\}$.

\begin{lemma} \label{lem:sparsity}
	Let $\mathcal{H}$ be any subset of $\mathcal{R}$. Then $\sum_{(p, i, j) \in \mathcal{C}(\mathcal{H})} y_i y_j \leq 4\mathcal{E}(\mathcal{H})$.
\end{lemma}
\begin{proof}
	$\displaystyle \sum_{(p, i, j) \in \mathcal{C}(\mathcal{H})} y_i y_j = \sum_{i \in \mathcal{H}} y_i \sum_{j:(p, i, j) \in \mathcal{C}(\mathcal{H})} y_j \leq 4 \sum_{i \in \mathcal{H}} y_i = 4\mathcal{E}(\mathcal{H})$.\\
	The inequality follows from the feasibility of $\mathbf{y}$ and the fact that there are four corner points for each rectangle.
\end{proof}
\begin{lemma} 
	\label{lem:lowresistance}
	For any $i$, the resistance of the $i^{th}$ chosen rectangle $\pi_i$ is at most $5$.
\end{lemma}
\begin{proof}
	Fix an $i$, and let $\mathcal{K}=\mathcal{R} \setminus \left\{ \pi_1,\dots,\pi_{i-1} \right\}$. We have
	\begin{align*}
	\sum_{R_k \in \mathcal{K}} y_k \eta(R_k, \mathcal{K}) &= \sum_{R_k \in \mathcal{K}} y_k \sum_{\substack{R_j \in \mathcal{K} \setminus \mathcal{B}(R_k),\\ R_k \otimes R_j}} y_j + \sum_{R_k \in \mathcal{K}} y_k \sum_{R_j \in \mathcal{K} \cap \mathcal{B}(R_k) \setminus \{R_k\}} y_j\\
	&\leq \sum_{R_k \in \mathcal{K}} y_k \sum_{\substack{R_j \in \mathcal{K},\\ R_k \otimes R_j}} y_j + \sum_{R_k \in \mathcal{K}} y_k \sum_{R_j \in  \mathcal{B}(R_k)} y_j \\
	\end{align*}
	We first prove:
	\[
	\sum_{R_k \in \mathcal{K}} y_k \sum_{\substack{R_j \in \mathcal{K}, \\ R_k \otimes R_j}} y_j \leq \sum_{(p, i, j) \in \mathcal{C}(\mathcal{K})} y_i y_j
	\]
	Then, by Lemma~\ref{lem:sparsity}, the right hand side is at most $4\mathcal{E}(\mathcal{H})$.
	Note that, for each pair of distinct rectangles $R_i, R_j \in \mathcal{K}$ such that $R_i \otimes R_j$, the term $y_i y_j$ appears in the left hand side exactly twice. Now we prove the above inequality by showing that the term $y_i y_j$ appears at least twice in the right hand side. If we fix any pair $R_i, R_j \in \mathcal{K}$, there are three cases (see Figure~\ref{fig:corner-intersections}):
	\begin{figure}
		\centering
		\includegraphics[width=8cm]{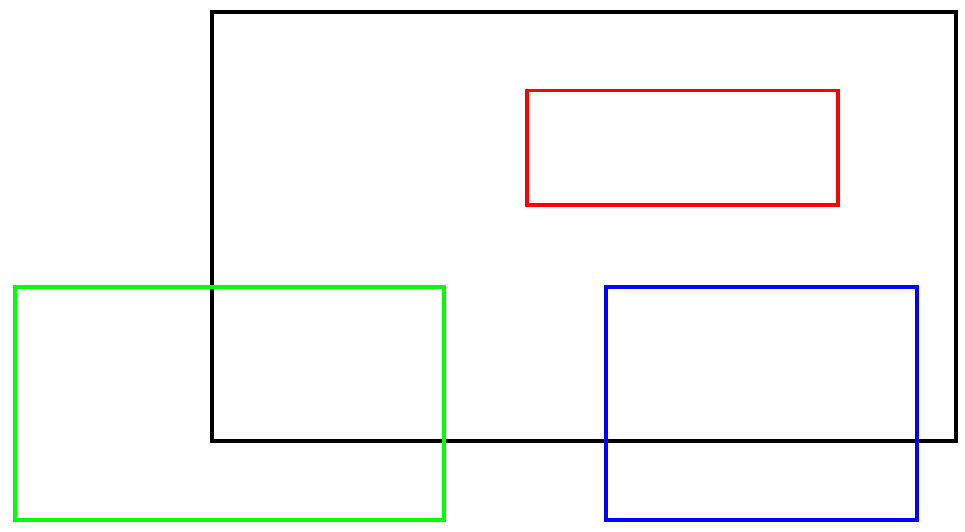}
		\caption{Three types of rectangles that overlap with the black rectangle, without creating a crossing intersection. The black rectangle contains all the 4 corners of the red rectangles, and exactly two corners of the blue rectangle. The green rectangle contains one corner of the black rectangle which, in turn, contains one corner of the green rerctangle.}
		\label{fig:corner-intersections}
	\end{figure}
	\begin{enumerate}
		\item one rectangle, say $R_i$, contains all the 4 corners of the other: then the term $y_i y_j$ appears four times in the right hand side.
		\item one rectangle, say $R_i$, contains exactly 2 corners of the other: then the term $y_i y_j$ appears twice in the right hand side.
		\item each rectangle contains exactly one corner of the other: then again the term $y_i y_j$ appears twice in the right hand side.
	\end{enumerate}
	Moreover, because of the constraints of bagMWISR, it follows that:
	\[
	\sum_{R_k \in \mathcal{K}} y_k \sum_{R_j \in  \mathcal{B}(R_k)} y_j \leq \sum_{R_k \in \mathcal{K}} y_k = \mathcal{E}(\mathcal{K})
	\]
	Thus, we have proved:
	\begin{align*}
	\sum_{R_k \in \mathcal{K}} y_k \eta(R_k, \mathcal{K}) &\leq 4\mathcal{E}(\mathcal{K}) + \mathcal{E}(\mathcal{K}) = 5\mathcal{E}(\mathcal{K}) \\
	&\Longrightarrow \sum_{R_k \in \mathcal{K}} \frac{y_j}{\mathcal{E}(\mathcal{K})} \eta(R_k,\mathcal{K})\leq 5.
	\end{align*}
	Since the left hand side of the last inequality is a convex combination of the elements of the set $\left\{\eta(R_k,\mathcal{K}) \enspace|\enspace R_k \in \mathcal{K} \right\}$, it follows that $\displaystyle{\min_{R_k \, \in \mathcal{K}} \{\eta(R_k,\mathcal{K})\} \leq 5}$. This concludes the proof, because the algorithm always chooses the rectangle that minimizes the resistance.
\end{proof}

\begin{lemma}
	\label{lem:independence1}
	Setting $\tau = 10$, then $\E[w(S)] \geq Opt/20$.
\end{lemma}
\begin{proof}
	Let $j' = \pi(j)$. Observe that the $R_{j'}$, is added to the candidate set $C$ with probability $x_{j'}/\tau$. Let $\mathcal{K}$ be the set of rectangles $R_i$ that were already considered and such that $R_{j'} \otimes R_i$ or $R_i$ is in the same bag as $R_{j'}$. Clearly, $\mathcal{E}(\mathcal{K})$ equals the resistance $\eta_{j}$ of $R_{j'}$ in the moment it is added to the permutation. Thus:
	\[
	\Pr[R_{j'} \in I \,|\, R_{j'} \in C] = \prod_{R_i \in \mathcal{K}} \left(1 - \frac{y_i}{\tau}\right) \geq 1 - \sum_{R_i \in \mathcal{K}} \frac{y_i}{\tau} = 1 - \frac{\mathcal{E}(\mathcal{K})}{\tau} \geq \frac{1}{2}
	\]
	because of Lemma~\ref{lem:lowresistance}. Thus, we have:
	
	\[
	\Pr[R_{j'} \in I] = \Pr[R_{j'} \in S \,|\, R_{j'} \in C]\cdot\Pr[R_{j'} \in C] \geq \frac{y_{j'}}{2\tau}
	\]
	
	\noindent Thus, we have:
	
	\[
	\E[S'] = \sum_{\Pr[R_{j'} \in I]}w_{j'} \geq \sum_{\Pr[R_{j'} \in I]} \frac{y_{j'} w_{j'}}{2\tau} = \frac{Opt}{20}
	\]
\end{proof}

\begin{lemma} 
	\label{lem:independence2}
	With probability at least $1-1/n$, the maximum depth of the rectangles in $I$ is $\displaystyle \Delta = O\left(\frac{\log n}{\log \log n}\right)$.
\end{lemma}
\begin{proof}
	Fix a point $p \in \mathcal{P}$. The number $\operatorname{depth}(p, I)$ of rectangles in $I$ containing $p$ is a sum of independent 0--1 random variables with mean $\mu = \sum_{R_i \ni p} \frac{y_u}{\tau} \leq 1$. Thus, by Chernoff bound:
	\[
	\Pr[\operatorname{depth}(p, I) > (1 + \delta)\mu] < {\left(\frac{e^\delta}{{(1+\delta)}^{1+\delta}} \right)}^\mu
	\]
	for any $\delta > 0$. By choosing $\delta$ such that $t = (1 + \delta)\mu$, it holds that:
	\[
	\Pr[\operatorname{depth}(p, I) > t] < {(e/t)}^t
	\]
	Since $|\mathcal{P}| = O(n^2)$ and $\Delta = \max_{p \in \mathcal{P}} \operatorname{depth}(p, I)$, by the union bound it follows that $\Pr[\Delta > t] = O({(e/t)}^t n^2)$, which is at most $1/n$ for $t = \Theta(\log n / \log \log n)$.
\end{proof}

\begin{lemma}\label{lem:bagMWISR}
	There is an expected $O\left(\dfrac{\log n}{\log \log n}\right)$-approximation for bagMWISR.
\end{lemma}
\begin{proof}
	Clearly, we have that $\sum_{R_i \in S'} w_i \geq \frac{1}{\Delta} \sum_{R_i \in S} w_i$. Thus:
	\begin{align*}
	E\left[\sum_{R_i \in S'} w_i\right] &\geq \left(1 - \frac{1}{n}\right)\Omega\left(\frac{\log \log n}{\log n}\right)E\left[\sum_{R_i \in S} w_i\right] 
	&\geq \Omega\left(\frac{\log \log n}{\log n}\right)Opt
	\end{align*}
\end{proof}

\noindent We finally obtain the following:
\begin{theorem}\label{thm:bagUFP}
	There is an expected $O\left(\dfrac{\log n}{\log \log n}\right)$-approximation for bagUFP.
\end{theorem}
\begin{proof}
	The result follows immediately by combining Lemmas~\ref{lem:ufp_partition}, \ref{lem:ufp_small_apx} and \ref{lem:bagMWISR}.
\end{proof}

\section{A $O(1)$-approximation for the unweighted case}\label{sec:unweightedBagUFP}

In this section we present a $O(1)$-approximation for bagUFP with large tasks, assuming uniform profits. This analysis extends the framework of \cite{agla13}, which only considered the UFP problem without handling bag constraints.

For any horizontal coordinate $x \in [0, m)$, we define for simplicity:
\[
c(x) := c_{(\left\lfloor x\right\rfloor, \left\lfloor x\right\rfloor + 1)}.
\]

We start with an arbitrary maximal feasible set of rectangles $\mathcal{R}_\textup{M} \subseteq \mathcal{R}$.

Let $\mathcal{R}_\textup{bag}$ be the set of all rectangles in $\mathcal{R}$ that belongs to some bag $J \ni R_i$ for each $R_i \in \mathcal{R}_\textup{M}$. Notice that $\mathcal{R}_\textup{bag} \supseteq  \mathcal{R}_\textup{M}$.
\begin{lemma} \label{lem:bag}
	We have that $\displaystyle \sum_{R_j \in \mathcal{R}_\textup{bag}} y_j \leq |\mathcal{R}_\textup{M}|$.
\end{lemma}
\begin{proof}
	Consider the bag $\mathcal{B}_j$ of a rectangle $R_i \in \mathcal{R}_\textup{M}$. From the LP$_{\textup{bagMIWSR}}$ constraint for $\mathcal{B}_j$, we have:
	$\sum_{R_j \in \mathcal{B}_j} y_j \leq 1$,
	Since, $\mathcal{R}_\textup{M}$ is a feasible solution of bagUFP problem, so each $R_i \in \mathcal{R}_\textup{M}$ will belong to a distinct bag. Therefore, if we sum over all bags $\mathcal{B}_j$ such that $ \mathcal{B}_j \ni R_i: R_i \in \mathcal{R}_\textup{M}$, we get the inequality:
	\[ 
	\sum_{R_j \in \mathcal{R}_\textup{bag}} y_j = \sum_{{\mathcal{B}_j \ni R_i: R_i \in \mathcal{R}_\textup{M}}} \sum_{R_j \in \mathcal{B}_j} y_j \leq |\mathcal{R}_\textup{M}|.
	\]
\end{proof}

Let $\mathcal{Q}_\textup{corner}$ be the set of corner points of all the rectangles in $\mathcal{R}_\textup{M}$.
Let $\mathcal{R}_\textup{corner}$ be the set of all rectangles in $\mathcal{R} \setminus \mathcal{R}_\textup{bag}$ that intersects with at least one of the points in $\mathcal{Q}_\textup{corner}$.

\begin{lemma}
	\label{lem:corner}
	We have that $\sum_{R_j \in \mathcal{R}_{\textup{corner}}} y_j \leq 4|\mathcal{R}_\textup{M}|$.
\end{lemma}
\begin{proof}
	Consider the top-left corner $p_i=(s_i,b_i) \in \mathcal{P}$ of a rectangle $R_i \in \mathcal{R}_\textup{M}$. Then $\sum_{{ R_i \in \mathcal{R}_\textup{corner}: p_i \in R_j}} y_j \leq \sum_{{ R_j \in \mathcal{R}: p_i \in R_j}} y_j  \leq 1$, by LP$_\textup{bagMWISR}$ constraints. Hence if $\mathcal{R}_\textup{TL}$ is the set of all the rectangles in $\mathcal{R}_\textup{corner}$ that contain the top-left corner of some rectangle in $\mathcal{R}_\textup{M}$, we have:
	
	\[ 
	\sum_{R_j \in \mathcal{R}_\textup{TL}} y_j \leq \sum_{ R_i \in \mathcal{R}_\textup{M}} \sum_{\substack{ R_j \in \mathcal{R}_\textup{corner}: \\ p_i \in R_j}} y_j \leq \sum_{ R_i \in \mathcal{R}_\textup{M}}  1 \leq |\mathcal{R}_\textup{M}|.
	\]
	
	The result follows by summing up this inequality with the analogous one for the top-right, bottom-left and bottom-right corners.
\end{proof}

Let $\mathcal{R}'=(\mathcal{R} \setminus \mathcal{R}_\textup{M}) \setminus (\mathcal{R}_{\textup{bag}} \cup \mathcal{R}_{\textrm{corner}})$. Note that each rectangle $R_j \in \mathcal{R}'$ overlaps with some rectangle in $\mathcal{R}_\textup{M}$. As illustrated in Figure~\ref{fig:top-left-right-intersecting}, we now classify the rectangles in $\mathcal{R}'$ into three groups (not necessarily disjoint) as follows:

\begin{enumerate}
	\item \emph{Top-intersecting} ($\mathcal{R}_\textup{T}$): We call a rectangle $R_j \in \mathcal{R}'$ \emph{top}-intersecting if there exists a rectangle $R_i \in \mathcal{R}_\textup{M}$ such that $s_i < s_j < t_j < t_i$ and $b_j \geq b_i > l_j$.
	\item \emph{Left-intersecting} ($\mathcal{R}_\textup{L}$): We call a rectangle $R_j \in \mathcal{R}'$ \emph{left}-intersecting if there exists a rectangle $R_i \in \mathcal{R}_\textup{M}$ such that $b_i > b_j > l_j > l_i$ ; $s_j < s_i < t_j$ and $\exists x_B \in [s_j, s_i]$ such that $c(x_B) = b_j$.
	\item \emph{Right-intersecting} ($\mathcal{R}_\textup{R}$): We call a rectangle $R_j \in \mathcal{R}'$ \emph{right}-intersecting if there exists a rectangle $R_i \in \mathcal{R}_\textup{M}$ such that $b_i > b_j > l_j > l_i$ ; $s_j < t_i < t_j$ and $\exists x_B \in [t_i,t_j]$ such that $c(x_B) = b_j$.
\end{enumerate}
\noindent From now on we call a rectangle $R_j \in \mathcal{R}'$ to be \emph{top} (\emph{left}, \emph{right}, resp.) intersecting to a rectangle $R_i \in \mathcal{R}_\textup{M}$ if it satisfies the conditions in 1. (2., 3., resp.).

\begin{figure}
	\centering
	\includegraphics[width=10cm]{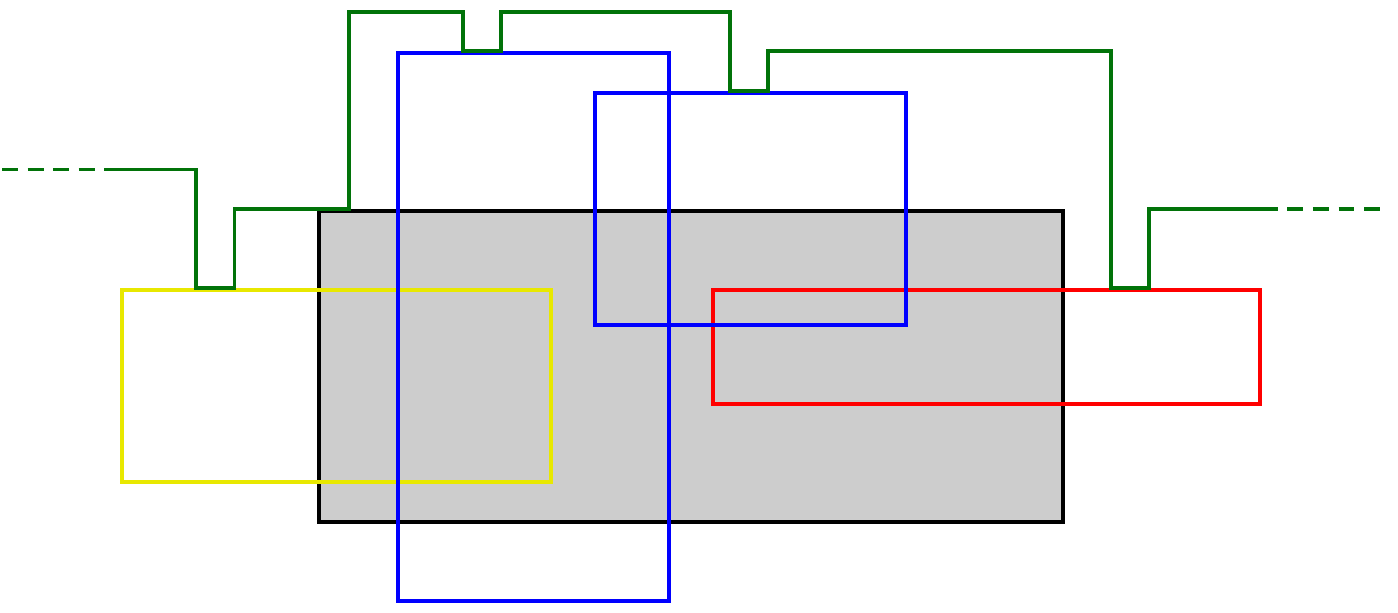}
	\caption{The black rectangle with gray background is in $\mathcal{R}_\textup{M}$. The blue ones are \emph{top-intersecting}; the yellow one is \emph{left-intersecting}; the red one is \emph{right-intersecting}. Note that a rectangle could simultaneously be left-intersecting to a rectangle and right-intersecting to another one.}
	\label{fig:top-left-right-intersecting}
\end{figure}

\begin{figure}
	\centering
	\includegraphics{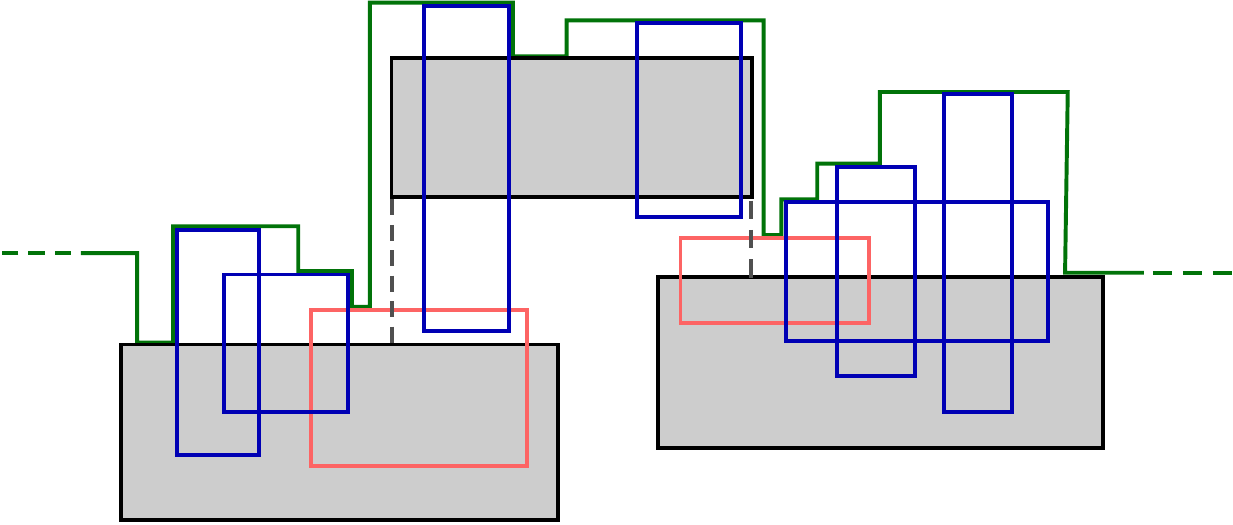}
	\caption{Illustration of some of the concepts for the $O(1)$-approximation for unweighted bagUFP. The green curve is the \textcolor[rgb]{0,0.5,0}{capacity profile}, and the rectangles are part of a \emph{top-drawn} instance of bagUFP. Black rectangles with gray background are in the initial maximal feasible solution $\mathcal{R}_\textup{M}$, while all the other rectangles are in $\mathcal{R}_\textup{T}$ (for simplicity, in this picture there are no rectangles in $\mathcal{R}_\textup{bag}$ and $\mathcal{R}_\textup{corner}$, as well as in $\mathcal{R}_\textup{L}$ and $\mathcal{R}_\textup{R}$). Red rectangles are in \textcolor[rgb]{0.8,0.3,0.3}{$\mathcal{R}^\textup{point}_\textup{top}$}, while the remaining blue rectangles are in \textcolor[rgb]{0,0,0.7}{$\mathcal{R}_{top}$}.}
	\label{fig:ind-rectangles}
\end{figure}

\begin{proposition}
	\label{lem:tlr}
	All the rectangles in $\mathcal{R}'$ are either top-intersecting, left-intersecting, or right-intersecting.
\end{proposition}

In order to bound the profit of rectangles in $\mathcal{R}_\textup{T}$, we divide it into two sets.
First we define a set of points $\mathcal{Q}_\textup{T}$ created according to the following process. For each bottom corner of each rectangle in $\mathcal{R}_\textup{M}$, we project vertically downwards until we meet another rectangle in $\mathcal{R}_\textup{M}$ or we arrive at height 0. More formally, we define $\mathcal{Q}_\textup{T}$ as the set of points $(x, y)$ for which there exist two rectangles  $ R_i, R_j \in \mathcal{R}_\textup{M}$ such that the following conditions hold:
\begin{enumerate}
	\item $l_i > b_j$
	\item $x \in \{s_i, t_i\}$, and $s_j < x < t_j$
	\item $y = b_j$
	\item $\not\exists \, R_k \in \mathcal{R}_\textup{M}$ such that $s'_k < x < t_k$ and $l_i < b_k < b_j\}$
\end{enumerate}


Let $\mathcal{R}^\textup{point}_{\textup{top}}$ be the subset of rectangles in $\mathcal{R}_\textup{T}$ containing at least one point in $\mathcal{Q}_\textup{T}$, and let $\mathcal{R}_\textup{top}=\mathcal{R}_\textup{T} \setminus \mathcal{R}^\textup{point}_{\textup{top}}$. We now bound the total contribution of $\mathcal{R}^\textup{point}_{\textup{top}}$ to the profit of the LP solution.

\begin{lemma}
	\label{lem:point:top}
	We have that $\sum_{R_i \in \mathcal{R}^\textup{point}_{\textup{top}}} y_i \leq 2|\mathcal{R}_\textup{M}|$.
\end{lemma}
\begin{proof}
	Note that for any $R_j \in \mathcal{R}_\textup{M}$, there are at most two points in $\mathcal{Q}_\textup{T}$ corresponding to the projection of the bottom corner points ${(s_j, l_j), (t_j, l_j)}$.The result follow by a similar argument as in the proof of Lemma~\ref{lem:corner}.
\end{proof}

Next we show that the bagMWISR for $\mathcal{R}_\textup{top}$ can be reduced to an independent set of intervals problem with bag constraints.

\begin{lemma} \label{lem:top:overlap}
	Two rectangles $R_i, R_j \in \mathcal{R}_\textup{top}$ overlap if and only if $[s_i,t_i] \cap [s_j,t_j] \neq \phi$.
\end{lemma}
\begin{proof}
	If $[s_i,t_i] \cap [s_j,t_j] = \phi$, then the two rectangles do not overlap.
	
	Assume by contradiction that $[s_i,t_i] \cap [s_j,t_j] \neq \phi$, but the two rectangles do not overlap. Without loss of generality, assume that $R_i$ is above $R_j$ (that is, $l_i \geq b_j$). Since it is a top-drawn instance, $s_j < x < t_j$ for some $x \in \{s_i, t_i\}$. Let $R_{\overline{i}}$ (resp. $R_{\overline{j}}$) be a rectangle in $\mathcal{R}_\textup{M}$ to which $R_i$ (resp. $R_j$) is top-intersecting. Since $s_{\overline{i}} < s_i < t_i < t_{\overline{i}}$, then $l_{\overline{i}} > b_j$ (otherwise $R_j$ would be in $\mathcal{R}_{\textup{corner}}$). Consider the following algorithm. The algorithm keeps track of a rectangle $R$ with the following properties:
	\begin{itemize}
		\item $R \in \mathcal{R}_\textup{M}$;
		\item the coordinates of one of the two bottom corners of $R$, say $(x_R, y_R)$, satisfy $s_j < x_R < t_j$ and $y_R > b_j$.
	\end{itemize}
	The algorithm starts by setting $R \gets R_{\overline{i}}$, which satisfies all the conditions above; at each step, the algorithm chooses a bottom corner of $R$, say $(x_R, y_R)$, such that $s_j < x_R < t_j$. Then it finds the point $P = (x_P, y_P)$ obtained by moving vertically downward starting from $(x_R, y_R)$ till it meets the top edge of some rectangle $R' \in \mathcal{R}_\textup{M}$ (it may be that $R' = R_{\overline{j}}$). Let the bottom left and top right corners of $R'$ be $(s_{R'}, l_{R'})$ and $(t_{R'}, b_{R'})$, respectively. Note that since the instance is top-drawn, one of the coordinates $x \in \{ s_{R'}, t_{R'} \}$ will be in the range $\left(s_j,t_j \right)$, if $R' \neq R_{\overline{j}}$. There are three cases:
	\begin{itemize}
		\item $y_P \leq b_j$: then $P \in R_j$, implying $R_j \in \mathcal{R}^\textup{point}_{\textup{top}}$, which is a contradiction;
		\item $y_P > b_j$ and $l_{R'} \leq b_j$: then for some $x \in \{s_{R'}, t_{R'}\}$ the corner $(x, l_{R'})$ of $R'$ is contained in $R_j$ and so $R_j \in \mathcal{R}_{\textup{corner}}$, which is a contradiction;
		\item $y_P > b_j$ and $l_{R'} > b_j$: then $R'$ satisfies the above conditions, and the algorithm continues by setting $R \gets R'$ and going to the next step.
	\end{itemize}
	
	Note that the $y$-coordinate of the top edge of $R$ is strictly smaller at each iteration, so a contradiction must eventually be found.
\end{proof}

Thus, the problem of finding a maximum set of independent rectangles in $\mathcal{R}_\textup{top}$ with bag constraints is equivalent to solving the corresponding instance of the maximum set of independent intervals problem with bag constraints. There is a $2$-approximation algorithm for this problem, see for example \cite{bgns01} or \cite{bd00}.

\medskip

Similarly we need to bound the contribution of the left-intersecting rectangles ($\mathcal{R}_\textup{L}$) to the profit of the LP solution. For this we need to define a set of points $\mathcal{Q}_L$ (as we did with the set of points $\mathcal{Q}_\textup{T}$). First we define the mapping from one point to another as:
\[
f(x,y) = \left\{
\begin{array}{l l}
\text{min$_{s_i}$ $(s_i,y)$} & \enskip \text{if $\exists R_i \in \mathcal{R}_\textup{M}$: $s_i \geq x$ and $l_i \leq y \leq b_i$} \\ 
& \text{ and $c(x') \geq y$ for all $x < x' \leq s_i$,} \\
\text{min$_{x'}$ $f(x',c(x'))$} & \enskip \text{if no such $R_i$ exists and if $\exists x'>x$:} \\ 
& \text{$c(x')<y$ and $c(x'') \geq y$ for all $ x < x'' < x' $} \\
\bot & \enskip \text{otherwise}
\end{array}
\right.
\]
We now define the set $\mathcal{Q}_L = \{f(t_i, b_i) \, | \, R_i \in \mathcal{R}_\textup{M}\}$. Intuitively, for each rectangle in $\mathcal{R}_\textup{M}$ we start from the top-right corner and we move rightwards, stopping if we meet the left edge of any rectangle in $\mathcal{R}_\textup{M}$; if we meet a vertical edge of the capacity profile, we move downwards until it is possible to move rightwards again, and we continue the process (see Figure~\ref{fig:unweighted-bagUFP-left-intersecting}).

\begin{figure}
	\centering
	\includegraphics[width=13cm]{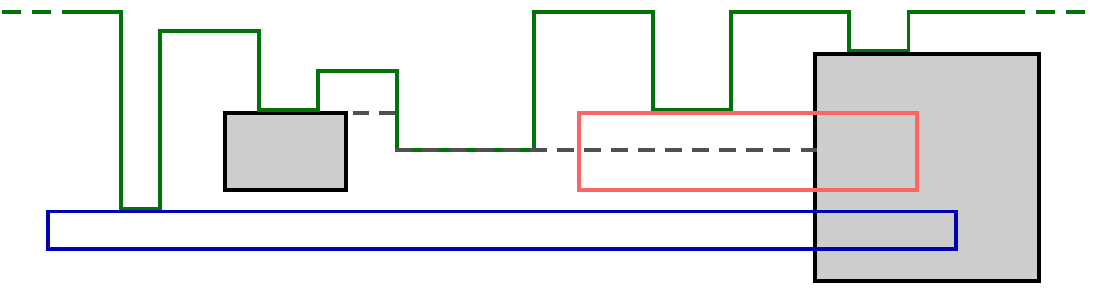}
	\caption{The definition of the set $\mathcal{R}_{left}^{point}$. The rectangles with gray background are in $\mathcal{R}_\textup{M}$; the colored rectangle are left-intersecting rectangle, but the red one is in $\mathcal{R}_{left}^{point}$ and the blue one is not.}
	\label{fig:unweighted-bagUFP-left-intersecting}
\end{figure}

As for the case of top-intersecting rectangles, we let $\mathcal{R}_{left}^{point}$ be the set of rectangles in $\mathcal{R}_L$ containing some point in $\mathcal{Q}_L$. In the next lemma we bound the contribution of rectangles in $\mathcal{R}_{left}^{point}$ to the LP$_{\textup{bagMWISR}}$ solution profit.

\begin{lemma} 
	\label{lem:point:left}
	We have that $\sum_{R_i \in \mathcal{R}^\textup{point}_{\textup{left}}} y_i \leq |\mathcal{R}_\textup{M}|$.
\end{lemma}
\begin{proof}
	Note that for any rectangle $R_j \in \mathcal{R}_\textup{M}$, there is at most one point in $\mathcal{Q}_\textup{L}$, which is the point we obtain by applying function $f$ to its top-right corner point. The result follows by a similar argument as in the proof of Lemma~\ref{lem:corner}.
\end{proof}

Once again, we can reduce the bagMWISR problem for $\mathcal{R}_\textup{left} = \mathcal{R}_L \setminus \mathcal{R}^\textup{point}_{\textup{left}}$ to the maximum independent set of intervals problem with bag constraints.

Let $u_{MAX} := \max_{e \in E}{u_e}$ be the maximum capacity.

We map each rectangle $R_i \in \mathcal{R}_\textup{left}$ to the interval $[u_{MAX} \cdot j + l_i, u_{MAX} \cdot j + b_i]$, where $j$ is defined by:
$s_j = \min_{R_M \in \mathcal{R}_\textup{M}} \{s_M : R_i \mbox{ is left-intersecting to } R_M\}$; that is, $R_j$ is the leftmost starting rectangle in $\mathcal{R}_M$ to which $R_i$ is left-intersecting.


\begin{lemma}\label{lem:left:overlap}
	Let $R_i, R_{i'} \in \mathcal{R}_\textup{left}$. Then $R_i$ and $R_{i'}$ overlap if and only if the segments they are mapped to overlap.
\end{lemma}
\begin{proof}
	Let $R_M$ (resp. $R_{M'})$ be the left-most starting rectangle in $\mathcal{R}_\textup{M}$ such that $R_i$ is left-intersecting to $R_M$ (resp. $R_{i'}$ is left-intersecting to $R_{M'}$). By the above definition of the mapping, $R_i$ is mapped to $[u_{MAX}\cdot M + l_i, u_{MAX}\cdot M + b_i]$ and $R_{i'}$ is mapped to $[u_{MAX}\cdot M' + l_{i'}, u_{MAX}\cdot M' + b_{i'}]$.\\
	If the intervals $[l_i, b_i]$ and $[l_{i'}, b_{i'}]$ do not overlap, then trivially $R_i$ and $R_{i'}$ do not overlap and there is nothing else to prove. Suppose that the above intervals overlap: we will prove that the rectangles $R_i$ and $R_{i'}$ overlap if and only if $R_M = R_{M'}$. 
	\begin{itemize}
		\item[$\Leftarrow$)] Suppose that $R_M = R_{M'}$. Then, $M = M'$ and the maps of the two intervals overlap, since $[l_i, b_i]$ and $[l_{i'}, b_{i'}]$ overlap.
		\item[$\Rightarrow$)] Suppose that $R_i$ and $R_{i'}$ overlap. We will prove that $R_M = R_{M'}$. Suppose by contradiction that $R_M \neq R_{M'}$, and assume without loss of generality that \sal{$s_{M} < s_{M'}$}.\\ 
		Since $R_i$ is left-intersecting to $R_M$, then $[l_i, b_i] \subset [l_M, b_M]$; similarly, $[l_{i'}, b_{i'}] \subset [l_{M'}, b_{M'}]$. Thus, if the intervals $[l_M, b_M]$ and $[l_{M'}, b_{M'}]$ do not overlap, then $[l_i, b_i]$ cannot overlap with $[l_{i'}, b_{i'}]$. Because of the independence of the rectangles $R_M$ and $R_{M'}$, this also implies that $t_M < s_{M'}$, and so $s_i < s_M < t_M < s_{M'} < t_{i'}$.\\
		Let $\overline{y}=\min(b_i, b_{i'})$. Note that for each $x \in [s_i, t_{i'}]$, $c(x) \geq \overline{y}$ .\\
		Consider the bottleneck edge $e = (r, r + 1)$ of $R_{i'}$; since $R_{i'}$ is left-intersecting to $R_{M'}$, then $r + 1 < s_{M'}$. Suppose by contradiction that $r+1 \leq t_{M}$. We consider two cases:
		\begin{itemize}
			\item $r+1 \leq s_{M}$: since $R_{i'}$ does not contain any corner of $R_M$ (otherwise $R_i' \in \mathcal{R}_\textup{corner}$, absurd), but it contains the point $(s_M, \overline{y})$, it follows that $R_{i'}$ is left-intersecting to $R_M$, contradicting the fact that $R_{M'}$ is the left-most starting rectangle to which $R_{i'}$ is left-intersecting.
			\item $s_{M} \leq r$ and $ r+1 \leq t_{M}$: note that for each $x \in [s_{M}, t_{M}]$, we have that $c(x) \geq b_M$; but then $b_{i'} \geq b_M$, so $R_{i'}$ contains the top-right corner of $R_M$ and $R_{i'} \in \mathcal{R}_\textup{corner}$, which is a contradiction.
		\end{itemize}
		Thus, the only case left is that $r \geq t_{M}$. We will show that a contradiction can be obtained with an argument that is similar, in spirit, to the proof of Lemma~\ref{lem:top:overlap}.
		
		We call a point $(x, y)$ \emph{bad} for $R_{i'}$ if it satisfies the following conditions:
		\begin{enumerate}[1)]
			\item $t_M \leq x < s_{M'}$
			\item $y \geq \overline{y}$
			\item $(x, y)$ is the top-right corner of a rectangle $R_j \in \mathcal{R}_\textup{M}$
		\end{enumerate}
		
		Note that the top-right corner of $R_M$ is bad for $R_{i'}$ . If $(x, y)$ is a bad point, then it is easy to see that $f(x, y) = (x', y')$ for some point $(x', y')$ that satisfies the above conditions (1) and (2), and that lies on the left edge of a rectangle $R_j \in \mathcal{R}_\textup{M}$; in fact, since $c(x) \geq \overline{y} \,\,\, \forall x \in [t_{M}, s_{M'}]$, by the definition of $f(x, y)$ we are guaranteed that we will find such a rectangle). If $R_j = R_{M'}$, then $(x', y') \in Q_\textup{L}$, and since $(x', y') \in R_{i'}$ this implies that $R_{i'} \in \mathcal{R}^\textup{points}_\textup{left}$, which is a contradiction. But if $R_j \neq R_{M'}$, then the top-right corner of $R_j$, say $(x'', y'')$ is bad for $R_{i'}$ and $x'' > x$. Thus, by repeating this process, a contradiction must eventually be found.
	\end{itemize}
\end{proof}
In a symmetric fashion, we can define the set of point $\mathcal{Q}_\textup{R}$ and the corresponding sets of rectangles $\mathcal{R}^\textup{points}_\textup{right}$ and $\mathcal{R}_\textup{right}$, and the analogues of lemmas \ref{lem:point:left} and \ref{lem:left:overlap} apply.


In our algorithm, we take $\mathcal{R'} \in {\mathcal{R}_\textup{top}, \mathcal{R}_\textup{left}, \mathcal{R}_\textup{right}}$ and compute the $2$-approximate solution $\textup{APX}(\mathcal{R'})$ for each subproblem. Then we output the maximum sized solution out of $\mathcal{R}_\textup{M}$, $\textup{APX}(\mathcal{R}_\textup{top})$, $\textup{APX}(\mathcal{R}_\textup{left})$ and $\textup{APX}(\mathcal{R}_\textup{right})$, which we call $\textup{APX}$.

We want to bound the cardinality of our solution $\textup{APX}$ in comparison with the optimal fractional solution to LP$_\textup{bagMWISR}$. By Lemma~\ref{lem:tlr} the union of the sets $\mathcal{R}_\textup{bag}$, $\mathcal{R}_\textup{points}$, $\mathcal{R}_\textup{top}$, $\mathcal{R}_\textup{left}$ and $\mathcal{R}_\textup{right}$ equals $\mathcal{R}$.

Let $\mathcal{R}_\textup{points}:= \mathcal{R}_\textup{corner} \cup \mathcal{R}^\textup{points}_\textup{top} \cup \mathcal{R}^\textup{points}_\textup{left} \cup \mathcal{R}^\textup{points}_\textup{right}$. We have:
\begin{lemma} \label{lem:points}
	For any feasible solution $\mathbf{y}$ of LP$_\textup{bagMWISR}$, the fractional profit of the rectangles in $\mathcal{R}_\textup{points}$ is upper bounded by $8 \cdot |\mathcal{R}_\textup{M}|$.
\end{lemma}
\begin{proof}
	It directly follows from lemmas \ref{lem:corner}, \ref{lem:point:top} and \ref{lem:point:left}.
\end{proof}

\begin{lemma} 
	\label{lem:top:left:right}
	For any feasible solution $\mathbf{y}$ to $\textup{LP}_\textup{bagMWISR}$ and any set $\mathcal{R}' \in \{ \mathcal{R}_\textup{top}$, $\mathcal{R}_\textup{left}$,  $\mathcal{R}_\textup{right} \}$ it holds that $\sum_{T_i \in \mathcal{R}'} y_i \leq 2 \cdot \textup{APX}(\mathcal{R}')$.
\end{lemma}
\begin{proof}
	We know that the bagMWISR instance for $\mathcal{R}' \in \{ \mathcal{R}_\textup{top}$, $\mathcal{R}_\textup{left}$,$\mathcal{R}_\textup{right} \}$ is equivalent to the maximum set of independent intervals problem with bag constraints. Let $y'$ be the optimal fractional LP solution for $ \mathcal{R}'$. Since $\mathbf{y}$ is a feasible LP solution for $ \mathcal{R}'$, $\sum_{T_i \in \mathcal{R}_\textup{top}} y_i \leq \sum_{T_i \in \mathcal{R}'} y'_i$. By Theorem 4.2 of \cite{bgns01}, we have that $\sum_{T_i \in \mathcal{R}'} y'_i \leq 2 \cdot |\textup{APX}(\mathcal{R}')|$.
\end{proof}

\begin{lemma}\label{lem:bagUFP_unweighted_large}
	The solution $APX$ satisfies: $\sum_{T_i \in \mathcal{R}} y_i \leq 15 \cdot |\textup{APX}|$.
\end{lemma}
\begin{proof}
	By applying lemmas \ref{lem:bag}, \ref{lem:points} and \ref{lem:top:left:right}, we have:
	\begin{align*}
	\sum_{R_i \in \mathcal{R}} y_i
	&\leq \sum_{R_i \in \mathcal{R}_\textup{bags}} y_i + \sum_{R_i \in \mathcal{R}_\textup{points}} y_i + \sum_{R_i \in \mathcal{R}_\textup{top}} y_i + \sum_{R_i \in \mathcal{R}_\textup{left}} y_i + \sum_{R_i \in \mathcal{R}_\textup{right}} y_i\\
	&\leq |\mathcal{R}_\textup{M}| + 8|\mathcal{R}_\textup{M}| + 2|\textup{APX}(\mathcal{R}_\textup{top})| + 2|\textup{APX}(\mathcal{R}_\textup{left})| + 2|\textup{APX}(\mathcal{R}_\textup{right})| \\
	&\leq 15|\textup{APX}|
	\end{align*}
\end{proof}

Finally, we can prove the following theorem:

\begin{theorem}\label{thm:bagUFP_unweighted}
	There is a $69$-approximation for unweighted bagUFP.
\end{theorem}
\begin{proof}
	By combining Lemmas~\ref{lem:color} and \ref{lem:bagUFP_unweighted_large}, there is an algorithm for $bagUFP$ with approximation factor $4\cdot15=60$ if all the tasks are large. Since there is a $9$-approximation for the case of small tasks by Lemma~\ref{lem:ufp_small_apx}, then the claim follows by Lemma~\ref{lem:ufp_partition}.
\end{proof}

\section{Conclusions}

In this chapter, we obtained improved approximation algorithms for the bagUFP problem. It is an interesting open problem to decide if a constant approximation algorithm is possible for the general case of bagUFP, as we proved for the cardinality case. Moreover, to the best of our knowledge, no explicit lower bound is known for this problem, which is another interesting direction in the attempt to fully understand the approximability of this problem.

\addcontentsline{toc}{chapter}{Conclusions}
\chapter*{Conclusions}\label{chap:conclusions}

In this thesis, we discussed several variations of rectangle packing problems, trying to emphasize the power of the simple structure provided by container packings in different settings.

The techniques to prove the existence of a profitable container packing vary between different problems; many of the structural results from previous results can indeed be rephrased in terms of container packings with a little additional work. We observed that framing the same solutions in term of container packings gives a unifying framework, and simple general algorithms. In fact, instantiations of the algorithms in Theorems~\ref{thm:container_packing_ptas} and~\ref{thm:container_packing_ptas_ppt} give the PTAS for 2-Dimensional Geometric Knapsack with Resource Augmentation, the $(4/3 + O(\eps))$-approximation for Strip Packing in PPT, the $(3/2 + \eps)$-approximation for \tdkr. Due to the simplicity of this approach, it is an interesting research direction to extend this framework to more problems, or to prove its limits, by explicitly showing lower bounds on the approximation ratios that are achievable with container packings.

We also considered an approach that overcomes the main limitation of container packings for \tdk, namely, the inability to cope with the interaction between horizontal and vertical rectangles. By combining container packings with an \fontL-packing, we obtained the first algorithms that beats the factor $2$ also for this problem. While pure container packings do not seem to be powerful enough in this setting, it is possible that further generalizations of L\&C packings could lead to better approximation algorithms, and this is an interesting direction for future research.

Finally, we designed improved approximation algorithms for the Maximum Weight Independent Set of Rectangles with bag constraints, which immediately implies improved approximations for the bagUFP problem.

\backmatter

\bibliographystyle{dcu}
\bibliography{biblio}

\cleardoublepage

\end{document}